\documentclass[11pt]{article}
\usepackage[margin=3cm]{geometry}

\usepackage{amsmath,amsthm,amsfonts,amssymb}
\usepackage{tikz}
\usepackage{lmodern}        
\usepackage{microtype}

\usepackage{multicol}
\usepackage{thmtools}
\usepackage{algorithm}
\usepackage{algpseudocode}
\usepackage{titling}
\usepackage{graphicx}
\usepackage[font={small,it},labelfont=bf]{caption}
\usepackage{appendix}
\usepackage{lipsum}
\usepackage{mathtools}
\usepackage{color}
\definecolor{myurlcolor}{rgb}{0,0.35,0}
\definecolor{mycitecolor}{rgb}{0,0,0.55}
\definecolor{myrefcolor}{rgb}{0.55,0,0}
\usepackage[pagebackref,draft=false]{hyperref}
\hypersetup{colorlinks,
linkcolor=myrefcolor,
citecolor=mycitecolor,
urlcolor=myurlcolor}

\usepackage{cancel}
\usepackage{bbm}
\usepackage[inline]{enumitem}
\setlist[enumerate,1]{label=(\arabic*), ref=\arabic*}
\setlist[enumerate,2]{label=(\roman*), ref=\theenumi.\roman*}
\usepackage{cite}
\usepackage{url}
\usepackage{tikz-cd}
\usepackage{todonotes}
\usepackage{fancyhdr}
\usepackage{forest}
\usepackage{tikz-qtree}

\newcommand{\runtitle}{}
\usepackage{xparse}
\usepackage{extarrows}
\NewDocumentCommand{\Qsym}{}{\mathsf{Q}}        
\newcommand{\Qof}[1]{\Qsym\!\left[#1\right]}

\NewDocumentCommand{\Leaf}{m}{\Qof{#1}}         

\NewDocumentCommand{\DeclareTree}{m m}{\expandafter\def\csname #1\endcsname{#2}}
\DeclareTree{AssemTree}{\Join{\Leaf{S_1}}{\Join{\Leaf{S_2}}{\Leaf{S_3}}{S_{23}}}{S_{1,23}}}

\NewDocumentCommand{\meq}{o m}{%
  \mathrel{\xlongequal{%
    \IfNoValueTF{#1}{#2}{\substack{\text{#1}\\#2}}%
  }}%
}

\let\hat\widehat
\let\tilde\widetilde

\usepackage[capitalize,noabbrev]{cleveref}

\newtheorem{theorem}{Theorem}[section]
\newtheorem{proposition}[theorem]{Proposition}
\newtheorem{lemma}[theorem]{Lemma}
\newtheorem{corollary}[theorem]{Corollary}
\newtheorem{definition}[theorem]{Definition}
\newtheorem{definitionthm}[theorem]{Definition/Theorem}
\newtheorem{question}[theorem]{Question}

\newtheorem{notation}[theorem]{Notation}
\newtheorem*{question*}{Question}
\theoremstyle{definition}
\newtheorem{statement}[theorem]{Statement}
\newtheorem{example}[theorem]{Example}
\newtheorem{remark}[theorem]{Remark}

\newtheorem{condition}[theorem]{Condition}

\AddToHook{cmd/appendix/before}{%
}
\newenvironment{keywords}
{\bgroup\leftskip 20pt\rightskip 20pt \small\noindent{\bf Keywords:} }%
{\par\egroup\vskip 0.25ex}



\newlist{caselist}{enumerate}{1}
\setlist[caselist]{label=\textbf{Case~\arabic*:}, ref=Case~\arabic*, leftmargin=*}

\newtheorem{statementprime}{Statement}

\crefname{definitionthm}{Definition/Theorem}{Definitions/Theorems}
\crefname{notation}{Notation}{Notations}
\crefname{appendix}{Appendix}{Appendices}
\algnewcommand\algorithmicinput{\textbf{Input:}}
\algnewcommand\algorithmicoutput{\textbf{Output:}}
\algnewcommand\Input{\item[\algorithmicinput]}
\algnewcommand\Output{\item[\algorithmicoutput]}
\newcommand{\1}{\mathbf{1}_{|\Sc|}}
\newcommand{\Ac}{\mathcal{A}}

\newcommand{\Cc}{\mathcal{C}}
\newcommand{\Dc}{\mathcal{D}}
\newcommand{\Ec}{\mathcal{E}}

\newcommand{\Hc}{\mathcal{H}}
\newcommand{\Ic}{\mathcal{I}}

\newcommand{\Lc}{\mathcal{L}}
\newcommand{\Mc}{\mathcal{M}}

\newcommand{\Oc}{\mathcal{O}}
\newcommand{\Pc}{\mathcal{P}}
\newcommand{\Qc}{\mathcal{Q}}
\newcommand{\Rc}{\mathcal{R}}
\newcommand{\Sc}{\mathcal{S}}
\newcommand{\Tc}{\mathcal{T}}
\newcommand{\Uc}{\mathcal{U}}
\newcommand{\Vc}{\mathcal{V}}
\newcommand{\Wc}{\mathcal{W}}
\newcommand{\Xc}{\mathcal{X}}
\newcommand{\Yc}{\mathcal{Y}}
\newcommand{\Zc}{\mathcal{Z}}
\newcommand{\Af}{\mathfrak{A}}
\newcommand{\Bf}{\mathfrak{B}}

\newcommand{\Df}{\mathfrak{D}}

\newcommand{\Gf}{\mathfrak{G}}
\newcommand{\Hf}{\mathfrak{H}}

\newcommand{\Mf}{\mathfrak{M}}

\newcommand{\Pf}{\mathfrak{P}}

\newcommand{\Kr}{\mathrm{K}}

\newcommand{\Prb}{\mathrm{P}}
\newcommand{\Qr}{\mathrm{Q}}

\newcommand{\Ab}{\mathbb{A}}
\newcommand{\Bb}{\mathbb{B}}
\newcommand{\Cb}{\mathbb{C}}

\newcommand{\Gb}{\mathbb{G}}

\newcommand{\Mb}{\mathbb{M}}

\newcommand{\Rb}{\mathbb{R}}

\newcommand{\Zb}{\mathbb{Z}}
\newcommand{\Ts}{\mathsf{T}}
\newcommand{\Gs}{\mathsf{G}}
\newcommand{\Ms}{\mathsf{M}}

\newcommand{\pf}{\mathfrak{p}}

\usepackage{centernot}
\usepackage{stmaryrd}

\newcommand{\sepset}{\mathtt{sepset}}

\newcommand{\ol}{\overline}

\newcommand{\cf}{cf.\ }

\newcommand{\ind}{\perp\!\!\!\perp}
\newcommand{\miid}{\,\|\,}
\newcommand{\Do}{\mathsf{do}}
\newcommand{\anc}{\mathsf{Anc}}

\newcommand{\ch}{\mathsf{Ch}}

\newcommand\IM{\mathsf{IM}}
\newcommand{\lb}{\left(}
\newcommand{\rb}{\right)}

\newcommand{\dto}{\dashrightarrow}

\newcommand{\pa}{\mathsf{Pa}}
\newcommand{\de}{\mathsf{De}}
\newcommand{\dr}{\mathrm{d}}

\newcommand{\ant}{\mathsf{Ant}}

\newcommand{\pcc}{\mathsf{Pc}}
\newcommand{\dcc}{\mathsf{Dc}}
\newcommand{\re}{\mathsf{Re}}
\newcommand{\bu}{\mathsf{Bu}}
\newcommand{\id}{\mathrm{id}}
\newcommand{\dist}{\mathsf{dist}}
\newcommand{\pde}{\mathsf{PoDe}}
\newcommand{\pan}{\mathsf{PoAn}}
\newcommand{\pant}{\mathsf{PoAnt}}
\newcommand{\pdet}{\mathsf{PoPot}}
\newcommand{\pch}{\mathsf{PoCh}}
\newcommand{\idp}{\mathrm{sIDP}}

\newcommand{\SCM}{\Mc=(\Ic,\Vc,\Wc,\Xc,\Prb,f)}

\newcommand{\isADMG}{\Af=(\Ic,\Oc,\Sc,\Ec)}
\newcommand{\PAG}{\Pf=(\Ic,\Vc,\Ec)}
\newcommand{\MAG}{\Mf=(\Ic,\Vc,\Ec)}

\newcommand{\ilsADMG}{\Af=(\Ic,\Oc,\Lc,\Sc,\Ec)}
\newcommand{\Edges}{\{\tuh,\hut,\huo,\huh,\ouo,\ouh,\tuo,\out,\tut\}}
\newcommand{\sm}{\setminus}
\newcommand{\dcup}{\,\dot{\cup}\,}
\newcommand{\wrt}{w.r.t.\ }

\newcommand{\st}{s.t.\ }

\newcommand{\Ibbm}{\mathbbm{1} }

\newcommand{\cat}[1]{{\mathsf{#1}}}

\newcommand{\fci}[1]{\mathsf{FCI}\text{-}\mathcal{R}{#1}}
\newcommand{\sep}[2]{\underset{#2}{\stackrel{#1}{\perp}}}

\newcommand{\nsep}[2]{\underset{#2}{\stackrel{#1}{\not\perp}}}
\newcommand{\Ind}[2]{\underset{#2}{\stackrel{#1}{\,\ind\,}}}
\usepackage{tikz}
\usepackage{tikz-cd}
\usetikzlibrary{arrows,arrows.meta,calc,fit}
\usetikzlibrary{positioning,patterns,matrix}
\usetikzlibrary{decorations.markings,decorations.pathreplacing,decorations.pathmorphing}
\usetikzlibrary{shapes,shapes.arrows,shapes.geometric,shapes.multipart}
\tikzstyle{ndout} = [draw, semithick, shape=circle, minimum size=20pt,inner sep=0pt]
\tikzstyle{ndash} = [draw, dashed, shape=circle, minimum size=20pt,inner sep=0pt]
\tikzstyle{ndexo} = [draw, dashed, shape=circle, minimum size=20pt,inner sep=0pt, fill=lightgray]
\tikzstyle{ndlat} = [draw, semithick, shape=circle, minimum size=20pt,inner sep=0pt, fill=lightgray]
\tikzstyle{ndsel} = [draw, semithick, shape=regular polygon, regular polygon sides=3, minimum size=25pt,inner sep=0pt, fill=lightgray,shape border rotate=180]
\tikzstyle{arout} = [style={->,>=Latex}]
\tikzstyle{arlout} = [style={->,>=Latex}]
\tikzstyle{arlat} = [style={<->,>=Latex}]
\newcommand{\arrhead}{{Latex}}
\newcommand{\arrtail}{{}}
\newcommand{\arrstar}{Rays[n=6]}
\newcommand{\arrcirc}{{Circle[open]}}
\newcommand*{\out}[1][]{\mathrel{\tikz [baseline=-0.25ex,\arrcirc-\arrtail, #1] \draw [#1] (0pt,0.5ex) -- (1.3em,0.5ex);}}
\newcommand*{\hut}[1][]{\mathrel{\tikz [baseline=-0.25ex,\arrhead-\arrtail, #1] \draw [#1] (0pt,0.5ex) -- (1.3em,0.5ex);}}
\newcommand*{\tut}[1][]{\mathrel{\tikz [baseline=-0.25ex,\arrtail-\arrtail, #1] \draw [#1] (0pt,0.5ex) -- (1.3em,0.5ex);}}
\newcommand*{\tuh}[1][]{\mathrel{\tikz [baseline=-0.25ex,\arrtail-\arrhead, #1] \draw [#1] (0pt,0.5ex) -- (1.3em,0.5ex);}}
\newcommand*{\tuo}[1][]{\mathrel{\tikz [baseline=-0.25ex,\arrtail-\arrcirc, #1] \draw [#1] (0pt,0.5ex) -- (1.3em,0.5ex);}}
\newcommand*{\huh}[1][]{\mathrel{\tikz [baseline=-0.25ex,\arrhead-\arrhead, #1] \draw [#1] (0pt,0.5ex) -- (1.3em,0.5ex);}}
\newcommand*{\ouo}[1][]{\mathrel{\tikz [baseline=-0.25ex,\arrcirc-\arrcirc, #1] \draw [#1] (0pt,0.5ex) -- (1.3em,0.5ex);}}
\newcommand*{\huo}[1][]{\mathrel{\tikz [baseline=-0.25ex,\arrhead-\arrcirc, #1] \draw [#1] (0pt,0.5ex) -- (1.3em,0.5ex);}}
\newcommand*{\ouh}[1][]{\mathrel{\tikz [baseline=-0.25ex,\arrcirc-\arrhead, #1] \draw [#1] (0pt,0.5ex) -- (1.3em,0.5ex);}}

\newcommand*{\ous}[1][]{\mathrel{\tikz [baseline=-0.25ex,\arrcirc-\arrstar, #1] \draw [#1] (0pt,0.5ex) -- (1.3em,0.5ex);}}
\newcommand*{\hus}[1][]{\mathrel{\tikz [baseline=-0.25ex,\arrhead-\arrstar, #1] \draw [#1] (0pt,0.5ex) -- (1.3em,0.5ex);}}
\newcommand*{\tus}[1][]{\mathrel{\tikz [baseline=-0.25ex,\arrtail-\arrstar, #1] \draw [#1] (0pt,0.5ex) -- (1.3em,0.5ex);}}
\newcommand*{\sut}[1][]{\mathrel{\tikz [baseline=-0.25ex,\arrstar-\arrtail, #1] \draw [#1] (0pt,0.5ex) -- (1.3em,0.5ex);}}
\newcommand*{\suh}[1][]{\mathrel{\tikz [baseline=-0.25ex,\arrstar-\arrhead, #1] \draw [#1] (0pt,0.5ex) -- (1.3em,0.5ex);}}
\newcommand*{\suo}[1][]{\mathrel{\tikz [baseline=-0.25ex,\arrstar-\arrcirc, #1] \draw [#1] (0pt,0.5ex) -- (1.3em,0.5ex);}}
\newcommand*{\sus}[1][]{\mathrel{\tikz [baseline=-0.25ex,\arrstar-\arrstar, #1] \draw [#1] (0pt,0.5ex) -- (1.3em,0.5ex);}}

\tikzstyle{ndint} = [draw, semithick, shape=rectangle, minimum size=20pt,inner sep=0pt]
\tikzstyle{ndout} = [draw, semithick, shape=circle, minimum size=20pt,inner sep=0pt]
\tikzstyle{ndlat} = [draw, semithick, shape=circle, minimum size=20pt,inner sep=0pt, fill=lightgray]
\tikzstyle{arint} = [style={->,>=Latex,thick}] 
\tikzstyle{arout} = [style={->,>=Latex,thick}] 
\tikzstyle{arlout} = [style={->,>=Latex,thick}] 
\tikzstyle{arlat} = [style={<->,>=Latex,thick}] 

\tikzstyle{out} = [style={o->,>=Latex,style=semithick}]
\tikzstyle{hut} = [style={<-,>=Latex,style=semithick}]
\tikzstyle{tut} = [style=semithick]
\tikzstyle{tuh} = [style={->,>=Latex,style=semithick}]
\tikzstyle{tuo} = [style={-o,style=semithick}]
\tikzstyle{huh} = [style={<->,>=Latex,style=semithick}]
\tikzstyle{ouo} = [style={o-o,style=semithick}]
\tikzstyle{huo} = [style={<-o,>=Latex,style=semithick}]
\tikzstyle{ouh} = [style={o->,>=Latex,style=semithick}]
\tikzstyle{ous} = [style={o-{Rays[n=6]},style=semithick}]
\tikzstyle{hus} = [style={<-{Rays[n=6]},>=Latex,style=semithick}]
\tikzstyle{tus} = [style={-{Rays[n=6]},style=semithick}]
\tikzstyle{sut} = [style={{Rays[n=6]}-,style=semithick}]
\tikzstyle{suh} = [style={{Rays[n=6]}->,>=Latex,style=semithick}]
\tikzstyle{suo} = [style={{Rays[n=6]}-o,style=semithick}]
\tikzstyle{sus} = [style={{Rays[n=6]}-{Rays[n=6]},style=semithick}]

\title{Complete Causal Identification from Ancestral Graphs under Selection Bias}
\date{\today}
\long\def\acks#1{\vskip 0.3in\noindent{\large\bf Acknowledgments}\vskip 0.2in
\noindent #1}
\author{Leihao Chen\thanks{Korteweg-de Vries Institute for Mathematics, University of Amsterdam, Amsterdam, the Netherlands; {\tt l.chen2@uva.nl}}
        ~~and~~
        Joris M.~Mooij\thanks{Korteweg-de Vries Institute for Mathematics, University of Amsterdam, Amsterdam, the Netherlands; {\tt j.m.mooij@uva.nl}}}
\begin{document}
\maketitle
\begin{abstract}
  Many causal discovery algorithms, including the celebrated FCI algorithm, output a Partial Ancestral Graph (PAG). PAGs serve as an abstract graphical representation of the underlying causal structure, modeled by directed acyclic graphs with latent and selection variables. This paper develops a characterization of the set of extended-type conditional independence relations that are invariant across all causal models represented by a PAG. This theory allows us to formulate a general measure-theoretic version of Pearl's causal calculus and a sound and complete identification algorithm for PAGs under selection bias. Our results also apply when PAGs are learned by certain algorithms that integrate observational data with experimental data and incorporate background knowledge.
\end{abstract}


\newpage

\tableofcontents

\newpage

\section{Introduction}

Pearl's celebrated causal calculus, together with standard probability calculus, provides a sound and complete method for identifying causal effects from the observational distribution under the assumption that the data are generated according to a causal model (with latent variables) whose (marginalized) graph is available \cite{Pearl1994prob_calculus_action,pearl95causal,pearl2009causality,shpitser2008complete,huang06pearl}. However, there is often insufficient knowledge for specifying a causal graph a priori. This motivates causal discovery, which aims to recover causal structure from (observational/experimental) data under suitable assumptions. For example, under standard assumption, the Fast Causal Inference (FCI) algorithm infers from observational data a partial ancestral graph (PAG), which represents the Markov equivalence class over observed variables induced by underlying causal directed acyclic graphs (DAGs) with the presence of latent variables and selection bias  \cite{spirtes2001causation,spirtes95causal,spirtes99alogorithm,richardson2002ancestral,zhang2006causal,Zhang08complete,Ali09mark_equiv_ancestral_graphs,Colombo14order_ind_constrint_based_causal_structure_learning,Colombo12learning_highdim_dag,Uhler2013Faithfulness}. This immediately raises the following question:
\begin{question}\label{question}
    Is it possible to derive a sound and complete causal identification method for PAGs?
\end{question}

A natural first attempt would be to enumerate all graphs in the Markov equivalence class represented by a learned PAG and then apply a standard identification procedure to each member. Unfortunately, this is infeasible as the number of graphs can grow exponentially with the number of nodes in general. This challenge has attracted substantial attention and has motivated methods that reason directly from PAGs. An early version of sound causal calculus was first developed for maximal ancestral graphs (MAGs) and PAGs by Zhang \cite{zhang2008causal}. Maathuis and Colombo \cite{Maathuis15generalized} and Perković and coauthors \cite{Perkovic15complete,perkovic18complete} considered generalizations of the covariate adjustment criterion and formula for PAGs. Jaber and coauthors further strengthened Zhang's result and derived a complete causal calculus and an identification algorithm for PAGs \cite{jaber19idencom,jaber22causal}. 

There is also a substantial parallel literature on causal identification and effect estimation under incomplete knowledge of the causal structure, not limited to PAGs. A full review is beyond the scope of this paper; we mention only a few illustrative examples \cite{HuVanDerPas2025SelectingValidAdjustmentSets,PerkovicKalischMaathuis2017CPDAGBK,Perkovic2020MPDAGID,WitteHenckelMaathuisDidelez2020EfficientAdjustment,Nandy2017JointInterventions,Hyttinen2015DoCalculusUnknown,vanDerZander2014ConstructingSeparators,Maathuis2009IDA}.

Notably, the above works on PAGs adopt the simplifying assumption of no selection bias. In practice, however, selection mechanisms are pervasive—arising unintentionally or by design—and can fundamentally obstruct causal inference. In particular, conditioning on a common effect of multiple variables can induce spurious dependence among them, even when they are independent in the underlying unselected population (Berkson’s paradox) \cite{cooper95causal_discovery_selecion,Heckman79SampleSB,hernan2004structural,Hernan2020WhatIf,chen2025foundationsstructuralcausalmodels,Elwert2014EndogenousSelectionBias,GreenlandPearlRobins1999CausalDiagrams,Munafo2018ColliderScope}. 

The goal of this paper is to develop a sound and complete causal identification method for PAGs that accommodates selection bias. Conceptually, our work is  inspired by thought-provoking works on statistical causality and (extended) conditional independence by Dawid and coauthors \cite{dawid10dag,Dawid21decision,constantinou17extendedci,Dawid2007Fundamentals,Dawid1979ConditionalIndependence}. Technically, it builds on foundational results on (partial) ancestral graphs, the FCI algorithm, causal identification under Markov equivalence without selection bias, and a specific notion of conditional independence\cite{richardson2002ancestral,zhang2006causal,jaber19idencom,jaber22causal,forre2021transitional}.

\subsection{Idea behind our approach}

From an interventionist perspective, causality concerns controlled \emph{changes} (interventions) and the \emph{invariants} they reveal in a system. Such causal invariants can often be phrased precisely as statements that certain controlled changes are \emph{irrelevant} for certain aspects of the system. In parallel, \emph{conditional independence} is a formal notion of irrelevance. Putting these two ideas side by side already points to a close connection between conditional independence and statistical causality, as stressed by Dawid in his David R.~Cox Foundations of Statistics Lecture \cite{Dawid2025CoxFoundationsLecture}.

Consider a ``causal'' acyclic directed mixed graph (ADMG) $\Af$ (e.g., \cref{fig:admg}). Here ``causal'' means that $\Af$ is endowed with the following interpretations.

\begin{definition}[Causal relation (graphical version)]\label{def:causal_relation_graph}
According to $\Af$:
\begin{enumerate}
    \item \label{item:c1} a directed edge $a\tuh b$ indicates that $X_a$ is a direct cause of $X_b$;
    \item \label{item:c2} a directed path $a\tuh \cdots \tuh b$ indicates that $X_a$ is a cause of $X_b$;
    \item \label{item:c3} a bidirected edge $a\huh b$ in the marginal graph $\Af_{\sm \{a,b\}^c}$, obtained by after marginalizing out all nodes except $a$ and $b$, indicates that $X_a$ and $X_b$ are confounded.
\end{enumerate}
\end{definition}

A key observation is that, once an explicit \emph{regime indicator} is introduced \cite{pearl95causal,spirtes2001causation,Dawid21decision}, these causal notions can be expressed as conditional independence in the sense of Dawid \cite{Dawid21decision}. Concretely, introduce a non-stochastic variable $X_{I_a}$ indicating the data-generating regime: $X_{I_a}=\star_{I_a}$ denotes the observational regime, whereas $X_{I_a}=x_a$ denotes the interventional regime in which $X_a$ is set to $x_a$. Graphically, this variable is represented by a square node $I_a$, indicating its non-stochastic nature, and added as a parent of $a$ (e.g., \cref{fig:admg}). With this augmentation, the causal relations in \cref{def:causal_relation_graph} admit a unified formulation in terms of conditional independence in Dawid's sense, read off graphically via the usual separation criterion for ADMGs. For notational simplicity, we use the term $d$-separation throughout, even in the ADMG setting where the standard term is $m$-separation.
\begin{definition}[Causal relation (conditional independence version)]\label{def:causal_relation_ci}
According to $\Af$:
    \begin{enumerate}
    \item \label{item:c1'} $X_a$ is not a direct cause of $X_b$ if $b \sep{\mathsf{d}}{\Af_{\Do(I_a,\{a,b\}^c)}} I_a$ where $\Af_{\Do(I_a,\{a,b\}^c)}$ is constructed from $\Af$ by adding regime indicator $I_a$ as a parent of $a$, deleting all edges with arrowheads towards nodes outside $\{a,b\}$, and marking all nodes in $\{a,b\}^c$ as square nodes;
    \item \label{item:c2'} $X_a$ is not a cause of $X_b$ if $b \sep{\mathsf{d}}{\Af_{\Do(I_a)}} I_a$;
    \item \label{item:c3'} there is no confounding between $X_a$ and $X_b$ if $b \sep{\mathsf{d}}{\Af_{\Do(I_a)}} I_a\mid a$.
\end{enumerate}
\end{definition}

\begin{example}\label{ex:admg}
    According to the causal ADMG $\Af$ in \cref{fig:admg}, we have: (i) $X_a$ is not a direct cause of $X_c$, (ii) $X_a$ is a cause of $X_c$, and (iii) there is no confounding between $X_a$ and $X_b$, since
\[
    c \sep{\mathsf{d}}{\Af_{\Do(I_a,b)}} I_a, \quad c \nsep{\mathsf{d}}{\Af_{\Do(I_a)}} I_a \quad \text{ and }\quad  b \sep{\mathsf{d}}{\Af_{\Do(I_a)}} I_a\mid a.
\]

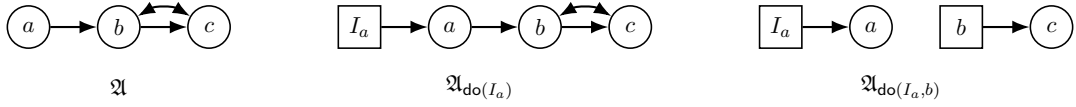
\begin{figure}[ht]
\centering
\begin{tikzpicture}[scale=0.8, transform shape]
\begin{scope}[xshift=0]
    \node[ndout] (a) at (-1.5,0) {$a$};
    \node[ndout] (b) at (0,0) {$b$};
    \node[ndout] (c) at (1.5,0) {$c$};
    \draw[arout] (a) to (b);
    \draw[arout] (b) to (c);
    \draw[arlat, bend left] (b) to (c);
    \node at (0,-1) {$\Af$};
\end{scope}

\begin{scope}[xshift=7cm]
    \node[ndint] (Ia) at (-3,0) {$I_a$};
    \node[ndout] (a) at (-1.5,0) {$a$};
    \node[ndout] (b) at (0,0) {$b$};
    \node[ndout] (c) at (1.5,0) {$c$};
    \draw[arout] (Ia) to (a);
    \draw[arout] (a) to (b);
    \draw[arout] (b) to (c);
    \draw[arlat, bend left ] (b) to (c);
    \node at (-1,-1) {$\Af_{\Do(I_a)}$};
\end{scope}

\begin{scope}[xshift=14cm]
    \node[ndint] (Ia) at (-3,0) {$I_a$};
    \node[ndout] (a) at (-1.5,0) {$a$};
    \node[ndint] (b) at (0,0) {$b$};
    \node[ndout] (c) at (1.5,0) {$c$};
    \draw[arout] (Ia) to (a);
    \draw[arout] (b) to (c);
    \node at (-1,-1) {$\Af_{\Do(I_a,b)}$};
\end{scope}
\end{tikzpicture}
\caption{An ADMG $\Af$ and the derived graphs $\Af_{\Do(I_a)}$  and $\Af_{\Do(I_a,b)}$ in \cref{ex:admg}.}
\label{fig:admg}
\end{figure}
\end{example}

Moreover, Pearl’s causal calculus \cite{pearl95causal,pearl2009causality} admits the following reformulation, as first presented in \cite{spirtes2001causation}:
\begin{enumerate}
    \item If $A\sep{\mathsf{d}}{\Af_{\Do(D)}} B\mid C\cup D$, then $p(x_A\mid x_B,x_C \miid \Do(x_{D}))=p(x_A\mid x_C \miid \Do(x_{D}))$.
    \item If $A\sep{\mathsf{d}}{\Af_{\Do(I_B,D)}} I_B\mid B\cup C\cup D$, then $p(x_A\mid x_C \miid \Do(x_B,x_{D}))=p(x_A\mid x_C ,x_B\miid \Do(x_{D}))$.
    \item If $A\sep{\mathsf{d}}{\Af_{\Do(I_B,D)}} I_B \mid C\cup D$, then $p(x_A\mid x_C \miid \Do(x_B,x_{D}))=p(x_A\mid x_C\miid \Do(x_{D}))$.
\end{enumerate}
Compared to the original formulation in \cite{pearl95causal}, this formulation is syntactically simpler and often easier to apply; with an appropriate Markov property, its proof is more direct \cite{forre2021transitional}.

In summary, conditional independence provides a universal language for causal invariants in causal structures (modeled by ADMGs, possibly augmented with selection variables). To answer \cref{question} without ruling out selection bias, it therefore suffices to understand the causal invariants shared by all ADMGs with selection variables represented by a PAG, which manifest as a particular type of conditional independence statement. This is precisely the characterization delivered by our main results.

\subsection{Overview of the main results}\label{sec:intro_main_result}

We briefly summarize our main results. For illustration, we state a simplified version here; the full generality and rigor are deferred to later sections.

Let $\Gf=(\Vc,\Ec)$ be a MAG or a PAG. Let $D,T\subseteq \Vc$ be disjoint, let $A\subseteq \Vc\sm T$, and let $B,C\subseteq \{I_d\}_{d\in D}\cup T\cup \Vc$ be disjoint. Using an appropriate separation relation $\sep{*}{\bullet}$ together with a graph transformation $(\bullet)_{\Do(I_D,T)}$ (cf.\ \cref{def:sep_PAG,def:int_iADMG,def:int_mag,def:int_pag}), we prove that conditional independence
\[
    A \sep{*}{\Gf_{\Do(I_D,T)}} B \mid C\cup T
\]
is valid if and only if, for every ADMG $\Af$ with selection nodes $\Sc_\Af$ represented by $\Gf$, the following conditional independence holds (cf.\ \cref{thm:sep_hsint_pag_mag,thm:sep_hsint_mag})
\[
    A \sep{*}{\Af_{\Do(I_D,T)}} B \mid C\cup T\cup \Sc_{\Af}.
\]
As a consequence, we obtain a sound and complete causal calculus (\cf \cref{thm:causal_calculus_mag_pag}) for MAGs and PAGs, applicable to general (not necessarily discrete) interventional Markov kernels. Modulo measure-theoretic technicalities, the rules take the following form:
\begin{enumerate}
    \item If $A\sep{*}{\Gf_{\Do(D)}} B\mid C\cup D$, then for every ADMG $\Af$ with selection nodes $\Sc_\Af$ represented by $\Gf$,
    \[
        \Prb(X_A\mid X_B,X_C,X_{\Sc_\Af}=x_{\Sc_\Af} \miid \Do(X_{D}))
        =\Prb(X_A\mid X_C,X_{\Sc_\Af}=x_{\Sc_\Af} \miid \Do(X_{D})).
    \]
    \item If $A\sep{*}{\Gf_{\Do(I_B,D)}} I_B\mid B\cup C\cup D$, then for every such $\Af$,
    \[
        \Prb(X_A\mid X_C,X_{\Sc_\Af}=x_{\Sc_\Af} \miid \Do(X_B,X_{D}))
        =\Prb(X_A\mid X_B,X_C,X_{\Sc_\Af}=x_{\Sc_\Af} \miid \Do(X_{D})).
    \]
    \item If $A\sep{*}{\Gf_{\Do(I_B,D)}} I_B \mid C\cup D$, then for every such $\Af$,
    \[
        \Prb(X_A\mid X_C,X_{\Sc_\Af}=x_{\Sc_\Af} \miid \Do(X_B,X_{D}))
        =\Prb(X_A\mid X_C,X_{\Sc_\Af}=x_{\Sc_\Af} \miid \Do(X_{D})).
    \]
\end{enumerate}

Finally, we provide an algorithm $\Ac$ such that $\Ac(\Gf;A,B)$ returns $\textsc{Fail}$ if the target causal effect $\Prb(X_A\mid X_{\Sc_\Af}=x_{\Sc_\Af}\miid \Do(X_B))$ under selection is non-identifiable for at least one ADMG $\Af$ with selection nodes $\Sc_\Af$ represented by $\Gf$ (cf.\ \cref{thm:idp_complete}). Otherwise, it outputs an identifying functional for the target causal effect, expressed in terms of the observational distribution conditional on selection $\Prb(X_\Vc\mid X_{\Sc_\Af}=x_{\Sc_\Af})$, that is valid for every ADMG $\Af$ with selection nodes $\Sc_\Af$ represented by $\Gf$ (cf.\ \cref{thm:idp_sound}).

\subsection{Outline}

In \cref{sec:causal_mag_pag}, we develop our theory for causal MAGs and PAGs. Our treatment differs from the standard treatment in a way that broadens the scope of the theory. Building on the theory, \cref{sec:reason_PAGs} formulates a rigorous measure-theoretic causal calculus together with an adjustment criterion and formula. In \cref{sec:IDalg}, we give a formal definition of causal identification from PAGs under selection bias, present an identification algorithm, and establish its soundness and completeness. Background on structural causal models, (extended) conditional independence, graphical models, and the FCI algorithm is collected in \cref{sec:preliminary}. Some further discussion related to \cref{sec:causal_mag_pag,sec:IDalg} is deferred to \cref{sec:additional_discussion}. Proofs for the main results of \cref{sec:causal_mag_pag,sec:IDalg} appear in \cref{sec:pf_1,sec:pf_2}, with auxiliary results collected in \cref{sec:pf_auxiliary}.

\section{Causal MAGs and PAGs}\label{sec:causal_mag_pag}

In this section, we develop a basic framework for causal MAGs and for a broader class of graphs, termed SOPAGs with input nodes, with PAGs representing Markov equivalence classes as a special case. We then introduce the graph-manipulation operations and the corresponding notion of graphical separation, which together yield our precise characterization of separations in MAGs and SOPAGs with input nodes.

\subsection{Basic framework}\label{sec:basic_def}

\begin{definition}[Causal ilsADMG]\label{def:isADMG}
    A \textbf{(causal) ilsADMG} is a tuple \[\Af=(\Ic,\Oc,\Lc,\Sc,\Ec)\] whose underlying graph $\hat{\Af}=(\Ic,\Vc,\Ec)$, with $\Vc=\Oc\dcup \Lc\dcup \Sc$, is an ADMG with input nodes $\Ic$. Here, $\Ic$, $\Oc$, $\Lc$, and $\Sc$ are disjoint sets with the interpretation:
    \begin{multicols}{2}
    \begin{enumerate}
        \item $\Ic$ is the set of exogenous input nodes;

        \item $\Oc$ is the set of observed output nodes;

        \item $\Lc$ is the set of latent output nodes;

        \item $\Sc$ is the set of latent selection nodes.
    \end{enumerate}
    \end{multicols}
    If $\Lc=\emptyset$, we call $\Af$ an \textbf{isADMG}. Given an isADMG $\Af$, we often use $\Sc_\Af$ to refer to the set of latent selection nodes of $\Af$.
\end{definition}

\begin{remark}[Causal interpretation of isADMG]
    The rationale behind the definition of an isADMG is as follows:
    We start with a DAG $\Df=(\Ic,\Vc=\Oc\dcup \Lc\dcup \Sc,\Ec)$ interpreted as the augmented causal graph of an acyclic SCM where $\Oc$ is the set of observed endogenous nodes, $\Lc$ is the set of latent exogenous and endogenous nodes and $\Sc$ is the set of latent selection nodes. In other words, $\Df$ is the augmented causal graph of an acyclic s-iSCM (\cf \cref{def:s-scm}). We then marginalize out all latent nodes in $\Lc$ (\cf \cite[Definition~3.2.18]{forre2025mathematical}) to obtain an isADMG $\hat{\Af}=(\Ic,\Oc, \Sc,\hat{\Ec})$. 
\end{remark}

\begin{definition}[Inducing path/walk]\label{def:inducing_path}
  Let $\Gf=(\Ic,\Vc,\Ec)$ be a mixed graph with input nodes $\Ic$, output nodes $\Vc=\Oc\dcup \Lc\dcup \Sc$ and edges $\Ec$ of the types \[\Edges.\] Let $a,b\in \Ic\dcup\Oc$ be distinct nodes. A path/walk $\pi$ from $a$ to $b$ in $\Gf$ is called an \textbf{inducing path/walk from $a$ to $b$ relative to $\Lc$ given $\Sc$}, or simply an $(\Lc,\Sc)$\textbf{-inducing path/walk from $a$ to $b$} if: 
  \begin{enumerate}
      \item every non-endnode on $\pi$ is either in $\Lc$ or a collider, and

      \item every collider on $\pi$ is in $\anc_{\Gf}(\{a,b\}\dcup \Sc)$.
  \end{enumerate}
  We call an $(\emptyset,\Sc)$-inducing path/walk $\pi$ an $\Sc$-\textbf{inducing path/walk}, and an  $(\emptyset,\emptyset)$-inducing path/walk $\pi$ simply an \textbf{inducing path/walk}.
\end{definition}

\begin{definition}[Partial ancestral graphs with inputs (iPAGs)]\label{def:pag}
  A mixed graph $\PAG$ with input nodes $\Ic$, output nodes $\Vc$ and edges $\Ec$ of the types \[\Edges\] is called a \textbf{partial ancestral graph with inputs (iPAG)} if the following conditions hold:
  \begin{enumerate}
    \item there is at most one edge between any two distinct nodes and no edges between a node and itself;
    \item there are no arrowheads towards input nodes and no edges between any two input nodes (input variables are set outside of the system);
    \item the graph does not contain directed cycles, almost directed cycles, or triples of the form $a\suh b\tut c$ (``ancestral'');
    \item there is no inducing path between any two distinct non-adjacent nodes (``maximal'').
  \end{enumerate}
\end{definition}

\begin{remark}
    \begin{enumerate}
        \item Clause~3 also excludes partially directed cycles, i.e., an anterior path from $a$ to $b$ together with an edge $b\tuh a$ \cite[Section~2.4]{richardson2002ancestral}.
        \item If there are no input nodes and no circle edge marks, then this definition is equivalent to the definition of MAGs given by \cite[Definition~3.1, Section~3.7]{richardson2002ancestral}.
        \item In most literature, a PAG represents a Markov equivalence class of MAGs. Our definition of PAGs is more general and allows for a theory with wider applicability (\cf \cref{rem:sopag}).
    \end{enumerate}
\end{remark}

\begin{definition}[Potentially directed/anterior paths and possible graphical relations]\label{def:podirected_path}
Let $\PAG$ be an iPAG and $a,b\in \Vc$ be two distinct nodes. Let \[\pi:a\sus v_1\sus \cdots \sus v_{n-1}\sus b\] be a path in $\Pf$. 
\begin{enumerate}
    \item If none of the edges $v_i\sus v_{i+1}$ is of the form $v_i\hus v_{i+1}$ or $v_i\sut v_{i+1}$, then we call $\pi$ a \textbf{potentially directed path from $a$ to $b$ in $\Pf$}, call node $a$ \textbf{possible ancestor} (when $n=1$, \textbf{possible parent}) of node $b$ in $\Pf$, and call node $b$ \textbf{possible descendant} (when $n=1$, \textbf{possible child}) of node $a$ in $\Pf$.

    \item If none of the edges $v_i\sus v_{i+1}$ is of the form $v_i\hus v_{i+1}$, then we call $\pi$ \textbf{potentially anterior path from $a$ to $b$ in $\Pf$} and call node $a$ \textbf{possible anterior} of node $b$ in $\Pf$. 
\end{enumerate}
The notations for those sets of nodes are $\pan_{\Pf}(b)$, $\mathsf{PoPa}_{\Pf}(b)$, $\pde_{\Pf}(b)$, $\pch_{\Pf}(b)$, and $\pant_{\Pf}(b)$, respectively.
\end{definition}

\begin{definition}[Maximal ancestral graphs with inputs (iMAGs)]\label{def:mag}
  An iPAG $\PAG$ is called a \textbf{maximal ancestral graph with inputs (iMAG)} if all edges in $\Ec$ are of the type \[\{\tuh,\hut,\huh,\tut\}.\]
\end{definition}

One key observation of \cite{richardson2002ancestral} is that a MAG $\Mf$ is an ``appropriate representation'' of an lsADMG $\Af$ if: 
\begin{enumerate}[itemsep=0pt,label=(\roman*)]  
    \item adjacencies in $\Gf$ coincide with $(\Lc,\Sc)$-inducing paths in $\Af$, and
    \item marks of arrowheads and tails encode ancestorship: an arrowhead at a node $v$ forbids $v$ from being an ancestor of the adjacent node or of any selection node in $\Sc_\Af$, whereas a tail at $v$ asserts the opposite.
\end{enumerate} 
We now formalize this as a definition, which plays an essential role in the whole theory.

\begin{definition}[Graph representation]\label{def:graph_rep}
  Let $\PAG$ be an iPAG. Let $\Af=(\tilde{\Ic},\Oc,\Lc, \Sc,\tilde{\Ec})$ be an ilsADMG. We say that \textbf{$\Pf$ $(\Lc,\Sc)$-represents $\Af$} or  \textbf{$\Pf$ represents $\Af$} if
  \begin{enumerate}
    \item $\tilde{\Ic}=\Ic$ and $\Vc=\Oc$;

    \item two distinct nodes $a,b\in \Vc\cup \Ic$ are adjacent in $\Pf$ iff  $\{a,b\}\nsubseteq \Ic$ and there is an $(\Lc,\Sc)$-inducing path from $a$ to $b$ in $\Af$;

    \item if $a\hus b$ is in $\Pf$ then $a\notin \anc_{\Af}(\{b\} \dcup \Sc)$;

    \item if $a\tus b$ is in $\Pf$ then $a\in \anc_{\Af}(\{b\}\dcup \Sc)$.
  \end{enumerate}
\end{definition}

\begin{proposition}[Construction of MAG representation]\label{prop:mag_rep}
  Let $\ilsADMG$ be an ilsADMG. There exists a unique iMAG $\Mf$ that represents $\Af$. We denote this iMAG $\Mf$ by $\cat{MAG}(\Af)$.
\end{proposition}

\begin{remark}
   \begin{enumerate}
       \item It is important to note that the MAG transformation here is slightly different from the one in \cite{richardson2002ancestral}. See \cref{ex:mag_trans}. The difference stems from the presence of exogenous input nodes. If there are no exogenous input nodes, or if we model them as endogenous nodes (\cf \cref{def:endo_graph}), then the two transformations are the same. This is also the main reason why one cannot directly use the results of \cite{richardson2002ancestral} in many cases (such as \cref{prop:sep_mag}) when there are input nodes.

       \item A single isADMG may be represented by multiple iPAGs. All such iPAGs must have the same adjacencies; they may differ only in their marks of endnodes of edges. For example, an edge $a\suo b$ in an iPAG $\Pf^1$ representing $\Af$ could be $a\suh b$ or $a\sut b$ in another iPAG $\Pf^2$ representing $\Af$.
   \end{enumerate}
\end{remark}

\begin{definition}[Endogenized Graphs]\label{def:endo_graph}
    Let $\Gf=(\Ic,\Vc,\Ec)$ be a mixed graph with input nodes $\Ic$, output nodes $\Vc$ and edges $\Ec$ of the types \[\Edges.\] The endogenized graph  of $\Gf$ is the mixed graph $\Gf^*\coloneqq (\emptyset,\Ic\dcup \Vc,\Ec)$, obtained by reclassifying all input nodes as output nodes.
\end{definition}

\begin{example}[MAG transformation with input nodes]\label{ex:mag_trans}
Consider an isADMG $\Af^1$ where $a$ and $b$ are input nodes and an isADMG $\Af^2$ where $a$ and $b$ are output nodes, shown in \cref{fig:mag_trans}, together with their different MAG representations. 
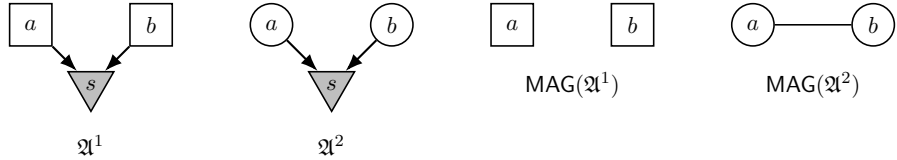
\begin{figure}[ht]
\centering
\begin{tikzpicture}[scale=0.8, transform shape]
\begin{scope}[xshift=0]
    \node[ndint] (a) at (-1,1) {$a$};
    \node[ndint] (b) at (1,1) {$b$};
    \node[ndsel] (s) at (0,0) {$s$};
    \draw[arout] (a) to (s);
    \draw[arout] (b) to (s);
    \node at (0,-1) {$\Af^1$};
\end{scope}

\begin{scope}[xshift=4cm]
    \node[ndout] (a) at (-1,1) {$a$};
    \node[ndout] (b) at (1,1) {$b$};
    \node[ndsel] (s) at (0,0) {$s$};
    \draw[arout] (a) to (s);
    \draw[arout] (b) to (s);
    \node at (0,-1) {$\Af^2$};
\end{scope}

\begin{scope}[xshift=8cm]
    \node[ndint] (a) at (-1,1) {$a$};
    \node[ndint] (b) at (1,1) {$b$};
    \node at (0,0) {$\cat{MAG}(\Af^1)$};
\end{scope}

\begin{scope}[xshift=12cm]
    \node[ndout] (a) at (-1,1) {$a$};
    \node[ndout] (b) at (1,1) {$b$};
    \draw[tut] (a) to (b);
    \node at (0,0) {$\cat{MAG}(\Af^2)$};
\end{scope}

\end{tikzpicture}
\caption{Two isADMGs $\Af^1$ and $\Af^2$ in \cref{ex:mag_trans} are the identical up to node type, whereas their MAG representations are different.}
\label{fig:mag_trans}
\end{figure}    
\end{example}

\begin{lemma}\label{lem:marg_and_mag}
    Let $\Af=(\Ic,\Oc,\Lc,\Sc,\Ec)$ be an ilsADMG. Then we have $\cat{MAG}(\Af_{\sm \Lc})=\cat{MAG}(\Af)$.
\end{lemma}

Together with the properties of iADMG marginalization (\cf \cite[Section~3]{forre2025mathematical}), this lemma tells us that for our theory (\cref{thm:sep_hsint_mag,thm:sep_hsint_pag_mag}) there is no loss of generality in restricting attention to isADMGs rather than general ilsADMGs. The benefit is that we only need to consider $\Sc$-inducing paths/walks, rather than $(\Lc,\Sc)$-inducing paths/walks.

\begin{notation}
\begin{enumerate}
    \item  Let $\Gf$ be an iMAG or iPAG and $\Pf$ be an iPAG. We write
  \[
  \begin{aligned}
      [\Gf]_\Gs\coloneqq& \{\Af : \Af\text{ is an isADMG represented by }\Gf \}\\
    \text{ and } [\Pf]_\Ms\coloneqq & \{\Mf : \Mf \text{ is an iMAG and }  [\Mf]_\Gs\subseteq  [\Pf]_\Gs\}.
  \end{aligned}
\]

  \item  Similar to \cite[Section~6.1]{richardson2002ancestral}, we define the ``canonical isADMG'' $\cat{isADMG}(\Mf)=(\Ic,\Oc,\Sc,\tilde{\Ec})$ of an iMAG $\MAG$ as follows:
\begin{enumerate}
  \item $\Oc\coloneqq \Vc$, $\Sc\coloneqq \{s_{ab}: a\tut b \text{ in } \Mf\}/ \sim$, where $s_{ab}\sim s_{ba}$ for all $a,b\in \Vc$ with $a\tut b$ in $\Ec$; and
  \item replacing $a \tut b$ in $\Mf$ with $a \tuh s_{ab} \hut b$ and all other edges in $\Ec$ are left unchanged.
\end{enumerate}
\end{enumerate}
\end{notation}

 With this notation, we have $\Af\in [\cat{MAG}(\Af)]_\Gs$ and $\cat{isADMG}(\Mf)\in [\Mf]_\Gs$. Moreover, if $\Pf$ is an iPAG that represents an isADMG $\Af$, then $\cat{MAG}(\Af)\in [\Pf]_\Ms$.

\begin{remark}
    Let $\PAG$ be an iPAG and $\Mf^1,\Mf^2\in [\Pf]_\Ms$. Then $\Pf$, $\Mf^1$, and $\Mf^2$ have the same adjacencies. If there is an arrowhead $a\hus b$ or tail $a\tus b$ in $\Pf$, then there must be an arrowhead $a\hus b$ or an tail $a\tus b$ in $\Mf$ for every $\Mf\in [\Pf]_\Ms$. Note that the converse does not hold for arbitrary iPAGs.
\end{remark}

We now define the class of iPAGs that will be our main focus. In the special case without input nodes, it includes the class of CPAGs \cite[Definition~3.2.1]{zhang2006causal}, i.e., PAGs that represent Markov equivalence classes of MAGs.  

\begin{definition}[Completely/Sufficiently oriented PAGs]\label{def:SOPAG}
Let $\Af$ be an isADMG that is represented by an iPAG $\PAG$. We call $\Pf$ a \textbf{completely oriented PAG with inputs (iCOPAG)} and a COPAG when $\Pf$ does not have inputs, if
\begin{enumerate}[label=(P\arabic*), ref=(P\arabic*)]
    \item \label{sopag:p1} $\Pf$ is closed under all FCI orientation rules in \cref{alg:FCI} given the independence model of $\Af$ conditioned on $\Sc_\Af$;
    \item \label{sopag:p2} $\Pf$ does not contain $a\suh b\tuo c$ or $a\suh b\out c$;
    \item \label{sopag:p3} if $a\suh b\ous c$ is in $\Pf$, then $a\suh c$ must be in $\Pf$; and
    \item \label{sopag:p4} the orientation scheme of \cite[Definition~3.3.1 and Lemma~3.3.4]{zhang2006causal} applies to $\Pf$ and yields an iMAG $\Mf\in [\Pf]_\Ms$.
\end{enumerate}
We call $\Pf$ a \textbf{sufficiently oriented PAG with inputs (iSOPAG)} and an SOPAG when $\Pf$ does not have inputs, if $\Pf$ satisfies Properties~\ref{sopag:p1}, \ref{sopag:p2} and \ref{sopag:p3}.

\end{definition}

\begin{remark}[On the definition of COPAGs/SOPAGs]\label{rem:sopag}
Let $\Pf$ be the output of the FCI algorithm \cite{zhang2006causal} without input nodes given an oracle independence model of MAG $\Mf$. Then $\Pf$ represents the Markov equivalence class of $\Mf$ and is a COPAG \cite[Lemmas~A.1 and A.2]{Zhang08complete}.  In contrast, not all SOPAGs represent the Markov equivalence classes of MAGs. Some of them can represent \emph{restricted Markov equivalence classes} of $\Mf$, which occurs when background knowledge is incorporated to exclude some members from the Markov equivalence class of $\Mf$ \cite{bryan20tieredbackground,venkateswaran24completecausalexplanationexpert}. Also note that by allowing for exogenous input nodes in our iSOPAGs, our theory applies not only to causal inference from purely observational data, but also to settings combining observational and experimental data \cite{Mooij20JCI}.
\end{remark}

Zhang \cite{zhang2008causal} first identified the importance of the so-called ``visibility'' of an edge for formulating a causal calculus for MAGs without undirected edges. We generalize \cite[Definition~8]{zhang2008causal} to the case where we have exogenous input nodes and undirected edges.

\begin{definition}[(In)visible directed edges]\label{def:visible_edge}
  Let $\Gf=(\Ic,\Vc,\Ec)$ be a mixed graph. Let $a\in \Ic\cup \Vc$ and $b\in \Vc$. A directed edge $a\tuh b$ in $\Gf$ is called \textbf{visible} if: 
  \begin{enumerate}
      \item $a\in \Ic$, or

      \item $a\in \Vc$ and there is a node $c\in \Ic\cup \Vc$ such that

      \begin{enumerate}
          \item $c$ is not adjacent to $b$, and
          \item  there is either $c\suh a$ or a (definite collider) path $c\suh v_1\huh \cdots \huh v_{n-1} \huh a$ into $a$ in $\Gf$ for some $n\geq 2$ and $v_1\ldots,v_{n-1}\in \pa_{\Gf}(b)$.
      \end{enumerate}
  \end{enumerate}
  Otherwise, we call the edge $a\tuh b$ \textbf{invisible}.
\end{definition}

\subsection{Graph manipulation}

We first introduce hard and soft manipulation operations for iADMGs and iMAGs. For iMAGs, hard manipulation generalizes the upper manipulation operation of \cite[Definition~11]{zhang2008causal} to the setting with exogenous input nodes and selection bias. The definition of soft manipulation is inspired by \cite[Definition~3.2.14]{forre2025mathematical} and by the discussion in \cref{rem:intuition_soft_manipulation}.

\begin{definition}[Hard/soft manipulation on iADMGs]\label{def:int_iADMG}
    Let $\Af=(\Ic,\Vc,\Ec)$ be an iADMG and $A\subseteq \Ic\cup \Vc$. 

    \begin{enumerate}
        \item We define the \textbf{hard-manipulated iADMG} $\Af_{\Do(A)}=(\hat{\Ic},\hat{\Vc},\hat{\Ec})$ by: 
    \begin{enumerate}
        \item  $\hat{\Ic}\coloneqq \Ic\dcup (A\cap \Vc)$ and $\hat{\Vc}\coloneqq \Vc\sm A$; and

        \item  $\hat{\Ec}\coloneqq \Ec\sm \{b\suh a: a\in A\}$. 
    \end{enumerate}

    \item We define the \textbf{soft-manipulated iADMG} $\Af_{\Do(I_A)}=(\tilde{\Ic},\tilde{\Vc},\tilde{\Ec})$ by: 
    \begin{enumerate}
        \item $\tilde{\Ic}\coloneqq \Ic\dcup \{I_a\}_{a\in A\cap \Vc}$ and $\tilde{\Vc}\coloneqq \Vc$; and

        \item $\tilde{\Ec}\coloneqq \Ec\dcup \{I_a\tuh a: a\in A\cap \Vc\}$.
    \end{enumerate}
    \end{enumerate}
\end{definition}

\begin{definition}[Manipulation operations on iMAGs]\label{def:int_mag}
  Let $\MAG$ be an iMAG and $A\subseteq \Ic\cup \Vc$. 
\begin{enumerate}
  \item We define the \textbf{hard-manipulated iMAG $\Mf_{\Do(A)}\coloneqq (\tilde{\Ic},\tilde{\Vc},\tilde{\Ec})$} by:
  \begin{enumerate}
      \item $\tilde{\Ic}\coloneqq \Ic\dcup (A\cap \Vc)$ and  $\tilde{\Vc}\coloneqq \Vc \sm A$; and 
      \item $\tilde{\Ec}\coloneqq  \Ec \sm (\Ec_1\cup \Ec_2)$, where 
      \[
      \Ec_1 \coloneqq \{b\suh a: a\in A\sm \Ic\} \text{ and } \Ec_2\coloneqq \{a\sus b: a,b\in A\cup  \Ic\}.
      \]
  \end{enumerate}

  \item We define the \textbf{soft-manipulated iMAG $\Mf_{\Do(I_A)}\coloneqq (\tilde{\Ic},\tilde{\Vc},\tilde{\Ec})$} by:
  \begin{enumerate}
      \item $\tilde{\Ic}\coloneqq \Ic\dcup \{I_a\}_{a\in A\sm \Ic}$ and  $\tilde{\Vc}\coloneqq \Vc $; and

      \item $\tilde{\Ec}\coloneqq  \Ec \dcup \Ec_1 \dcup \Ec_2\dcup \Ec_3 \dcup \Ec_4 \dcup \Ec_5$, where
    \[
    \begin{aligned}
      \Ec_1&\coloneqq \{I_a\tuh a: a\in A\sm \Ic \text{ and } \exists b\in \Ic\cup \Vc \text{ s.t.\ } a\hus b \text{ in } \Mf\}\\
      \Ec_2&\coloneqq \{I_a\tut a: a\in A\sm \Ic \text{ and } \exists b\in \Ic\cup \Vc \text{ s.t.\ } a\tut b \text{ in } \Mf\}\\
      \Ec_3&\coloneqq\{I_a\tuo a: a\in A\sm \Ic \text{ and } \nexists b\in \Ic\cup \Vc \text{ s.t.\ } a\hus b \text{ or } a\tut b \text{ in } \Mf\}\\
      \Ec_4&\coloneqq\{I_a\tuh b: a\in A\sm \Ic,\ b\in \Vc \text{ and } a\tuh b \text{ is invisible in } \Mf\}\\
      \Ec_5&\coloneqq\{I_a\tut b: a\in A\sm \Ic,\ b\in \Vc \text{ and } a\tut b \text{ is an edge in } \Mf \}.
    \end{aligned}
    \]
  \end{enumerate}
\end{enumerate}
\end{definition}

\begin{remark}
In what follows, we show that a restricted class of separations statements holds in the hard-manipulated iMAG iff the same separations hold in every hard-manipulated isADMGs such that the isADMGs themselves are represented by that original iMAG. Although the construction of hard-manipulated iMAGs is straightforward, the reason that this construction enjoys the above mentioned property is not obvious without examining the proof of \cref{lem:sep_hint_mag}.
\end{remark}

\begin{remark}[On soft manipulation]\label{rem:intuition_soft_manipulation}
    Although the definition of soft-manipulated iMAGs may appear complicated at first glance,\footnote{They are actually iPAGs by definition.}  the idea behind it is quite simple: one runs over all $\Af\in [\Mf]_\Gs$ and takes the ``union'' of all the iMAGs $\cat{MAG}(\Af_{\Do(I_A)})$ to obtain an iPAG $\Pf=(\tilde{\Ic},\tilde{\Vc},\tilde{\Ec})$. To be more precise, we define
\begin{enumerate}
    \item  $\tilde{\Ic}\coloneqq \bigcup_{\Af\in [\Mf]_\Gs}\Ic_{\cat{MAG}(\Af_{\Do(I_A)})}=\Ic $
    \item  $\tilde{\Vc}\coloneqq \bigcup_{\Af\in [\Mf]_\Gs}\Vc_{\cat{MAG}(\Af_{\Do(I_A)})}=\Vc $
    \item  $\tilde{\Ec}\coloneqq ``\bigcup_{\Af\in [\Mf]_\Gs}\Ec_{\cat{MAG}(\Af_{\Do(I_A)})}\text{''}=\Ec \dcup \Ec_1 \dcup \Ec_2\dcup \Ec_3 \dcup \Ec_4 \dcup \Ec_5$,
\end{enumerate}
where in the definition of $\tilde{\Ec}$ we use circle marks when there are ambiguities in endnodes marks of edges from different iMAGs $\cat{MAG}(\Af_{\Do(I_A)})$ for $\Af\in [\Mf]_\Gs$ and we use \cref{lem:sint_mag_I,lem:sint_mag_II}. The ambiguity occurs when there are no arrowheads towards node $a$ and there is no $b\in \Ic\cup \Vc$ such that $a\tut b$ is in $\Mf$, because there could be some isADMG $\Af^1\in [\Mf]_\Gs$ having $a\in \anc_{\Af^1}(\Sc_{\Af^1})$ while some isADMG $\Af^2\in [\Mf]_\Gs$ having $a\notin \anc_{\Af^2}(\Sc_{\Af^2})$ by \cref{lem:anc_selec_I}. In this case, we can have 
\[I_a \tut a \text{ in } \cat{MAG}(\Af^1_{\Do(I_A)}) \quad  \text{ and } \quad I_a\tuh a \text{ in } \cat{MAG}(\Af^2_{\Do(I_A)}).\]  For the purpose of reading off the separation statements relevant here, one may replace circle marks by tails or arrowheads arbitrarily. Ultimately, one can show that this way of adding edges is minimal in the following sense: a separation statement holds in the soft-manipulated iMAG if and only if, for every isADMG represented by the original iMAG, the corresponding separation statement with the selection nodes added to the conditioning set holds in the corresponding soft-manipulated isADMG (cf. \cref{thm:sep_hsint_mag}).
\end{remark}

Let $\MAG$ be an iMAG and $A,B\subseteq \Ic\cup \Vc$ be disjoint. Given a soft-manipulated iMAG $\Mf_{\Do(I_A)}$, we define $(\Mf_{\Do(I_A)})_{\Do(I_B)}$ by first orienting all the circles (if any) as tails and then applying \cref{def:int_mag}.\footnote{Or one can directly apply \cref{def:int_pag}.} Then we have the following result showing that hard manipulations commute pairwise, and so do soft manipulations..

\begin{proposition}[Manipulations commute]\label{prop:int_com}
  Let $\MAG$ be an iMAG and $A,B\subseteq \Ic\cup \Vc$. Then we have
  \begin{enumerate}
    \item $(\Mf_{\Do(A)})_{\Do(B)}=\Mf_{\Do(A\cup B)}=(\Mf_{\Do(B)})_{\Do(A)}$;

    \item $(\Mf_{\Do(I_A)})_{\Do(I_B)}=\Mf_{\Do(I_{A\cup B})}=(\Mf_{\Do(I_B)})_{\Do(I_A)}$.

  \end{enumerate}
\end{proposition}

This proposition shows that there is no ambiguity in iterating hard manipulations or in iterating soft manipulations: the order of application does not matter within each class. Note that even if $A$ and $B$ are disjoint, the soft manipulations and the hard manipulations do not commute, that is, 
\[(\Mf_{\Do(I_A)})_{\Do(B)}\ne(\Mf_{\Do(B)})_{\Do(I_A)}\] (\cf \cref{ex:count_int_com}).  We therefore define 
\[\Mf_{\Do(I_A,B)}\coloneqq (\Mf_{\Do(I_A)})_{\Do(B)},\] 
i.e., first a soft manipulation and then a hard manipulation. This is similar to \cite[Definition~11]{zhang2008causal} if we replace the soft manipulation and the hard manipulation with the lower manipulation and the upper manipulation, respectively.

\begin{example}[Hard manipulation and soft manipulation do not commute]\label{ex:count_int_com}
    Consider a MAG $\Mf$ consisting of $a\huh b\tuh c$. Then  $(\Mf_{\Do(a)})_{\Do(I_b)}$ consists of an isolated node $a$ together with $b\tuh c$ and $b\out I_b\tuh c$, but $(\Mf_{\Do(I_b)})_{\Do(a)}$ consists of an isolated node $a$ together with $I_b\tuh b\tuh c$.
\end{example}

After the discussion on iMAGs, we now turn to iPAGs.

\begin{definition}[Manipulation operations on iPAGs]\label{def:int_pag}
  Let $\PAG$ be an iPAG and $A\subseteq \Ic\cup \Vc$.

\begin{enumerate}
  \item We define the \textbf{hard-manipulated iPAG} $\Pf_{\Do(A)}\coloneqq (\tilde{\Ic},\tilde{\Vc},\tilde{\Ec})$ with: 
  \begin{enumerate}
      \item $\tilde{\Ic}\coloneqq \Ic\cup A$ and $\tilde{\Vc}\coloneqq \Vc \sm A$; and

      \item $\tilde{\Ec}\coloneqq  \left(\Ec\sm (\Ec_1\cup \Ec_2\cup \Ec_3)\right)\dcup \Ec_4$, where 
      \[
      \begin{aligned}
          &\Ec_1 \coloneqq \{b\suh a: a\in A\sm \Ic,\ b\suh a \text{ in } \Ec\},\\
          &\Ec_2\coloneqq \{b\suo a: a\in A\sm \Ic,\ b\in \Vc\sm A,\ b\suo a\text{ in }\Ec \}\\
          &\Ec_3\coloneqq \{a\sus b: a,b\in \Ic\cup A,\ a\sus b \text{ in } \Ec\}, \text{ and }\\
          &\Ec_4\coloneqq \{b\sut a: b\suo a \text{ in } \Ec_2\}.
      \end{aligned}
      \]
  \end{enumerate}
  
  \item We define the \textbf{soft-manipulated iPAG} $\Pf_{\Do(I_A)}\coloneqq (\tilde{\Ic},\tilde{\Vc},\tilde{\Ec})$ with: 
  \begin{enumerate}
      \item $\tilde{\Ic}\coloneqq \Ic\dcup \{I_a\}_{a\in A\sm \Ic} $ and $\tilde{\Vc}\coloneqq \Vc $; and

      \item $\tilde{\Ec}\coloneqq  \Ec \dcup \Ec_1 \dcup \Ec_2\dcup \Ec_3 \dcup \Ec_4 \dcup \Ec_5 \dcup \Ec_6 \dcup \Ec_7$, where
    \[
    \begin{aligned}
      \Ec_1&\coloneqq \{I_a\tuh a: a\in A\sm \Ic \text{ and } \exists b\in \Vc\cup \Ic  \text{ s.t.\ } a\hus b \text{ in } \Pf\}\\
      \Ec_2&\coloneqq \{I_a\tut a: a\in A\sm \Ic \text{ and } \exists b\in \Vc\cup \Ic \text{ s.t.\ } a\tut b \text{ in } \Pf\}\\ 
      \Ec_3&\coloneqq\{I_a\tuo a: a\in A\sm \Ic \text{ and } \nexists b\in \Vc\cup \Ic \text{ s.t.\ } a\hus b \text{ or } a\tut b \text{ in } \Pf
      \}\\
      \Ec_4&\coloneqq\{I_a\tuh b: a\in A\sm \Ic \text{ and } \big((a\tuh b \text{ is invisible in $\Pf$) or } (a\ouh b\text{ in } \Pf)\big)\}\\
      \Ec_5&\coloneqq\{I_a\tut b: a\in A\sm \Ic, b\in \Vc \text{ and } a\tut b  \text{ in } \Pf\}\\
      \Ec_6&\coloneqq\{I_a\tut b: a\in A\sm \Ic, b\in \Vc \text{ and } a\out b \text{ \st } \exists c\tut b  \text{ in } \Pf\}\\
      \Ec_7&\coloneqq\{I_a\tuo b: a\in A\sm \Ic \text{ and }\\
        &\qquad \qquad \big((a\tuo b \text{ in }\Pf)\text{ or }( a\ouo b \text{ in }\Pf) \text{ or } (a\out b \text{ \st } \nexists c\tut b \text{ in } \Pf)\big)\}.
    \end{aligned}
    \]
  \end{enumerate}
\end{enumerate}
\end{definition}

\begin{remark}
    In general $\Pf_{\Do(I_A)}$ is not an iSOPAG even if $\Pf$ is sufficiently oriented. 
\end{remark}

\subsection{Graphical separation}

We now introduce a notion of graphical separation for hard- and soft-manipulated iPAGs, iMAGs, and iADMGs, which we call $id$-separation. It is inspired by \cite[Definition~3.4.3]{forre2025mathematical} and \cite[Definition~4]{zhang2008causal}. It is the appropriate notion for establishing the theory developed in \cref{prop:sep_mag,thm:sep_hsint_mag,thm:sep_hsint_pag_mag,thm:causal_calculus_mag_pag}. Note that this notion of separation is asymmetric. When applied to a MAG or an ADMG (no input nodes), $id$-separation reduces to the usual $d$-separation, i.e., the standard $m$-separation criterion (see, e.g., \cite{pearl2009causality,richar03mark,richardson2002ancestral} for the definition of $m$-separation).

\begin{definition}[Graphical separation]\label{def:sep_PAG}
Let $\Gf=(\Ic,\Vc,\Ec)$ be an iPAG/iMAG/iADMG. Let $D,T \subseteq \Ic\cup \Vc$ be disjoint and $C\subseteq\Ic\cup \Vc\cup \{I_d\}_{d\in D}$ and ($n\geq 0$) \[\pi: v_0 \sus \cdots \sus v_n\]  be a path/walk in $\Hf\coloneqq \Gf_{\Do(I_D,T)}$.

We say that the path/walk $\pi$ is \textbf{$C$-$id$-open} or \textbf{$id$-open given $C$} if:
    \begin{enumerate}
        \item [(i)] $v_0 \notin C$ and $v_n \notin C$; and

        \item [(ii)] every pair of adjacent edges in $\pi$ is of one of the following forms:
            $$
            \begin{array}{rlll}
            id\text{-non-collider: } & v_{i-1} \sut v_i \sus v_{i+1}  & \text { with } & v_i \notin C; \\
            id\text{-non-collider: } & v_{i-1} \sus v_i \tus v_{i+1} & \text { with } & v_i \notin C ; \\
            id\text{-non-collider: } &  v_{i-1} \suo v_i \ous v_{i+1}   & \text { with } & v_i \notin C \text{ and unshielded }; \\
            id\text{-collider: } & v_{i-1} \suh v_i \hus v_{i+1} & \text { with } & v_i \in \anc_{\Hf}(C),\{v_j\}_{j=i-1}^{i+1}\subseteq \Ic\cup \Vc;\\
            id\text{-collider: } & v_{i-1} \suh v_i \hus v_{i+1} & \text { with } & v_i \in \pan_{\Hf}(C\cap \Vc),\{v_j\}_{j=i-1}^{i+1}\nsubseteq \Ic\cup \Vc;\\
            id\text{-collider: } & v_{i-1} \suo v_i \hus v_{i+1} & \text { with } & v_i \in \pan_{\Hf}(C\cap \Vc),v_{i-1}\in \{I_d\}_{d\in D};\\
            id\text{-collider: } & v_{i-1} \suh v_i \ous v_{i+1} & \text { with } & v_i \in \pan_{\Hf}(C\cap \Vc), v_{i+1}\in \{I_d\}_{d\in D}.
            \end{array}
            $$
    \end{enumerate}
We say that the path/walk $\pi$ is \textbf{$C$-$id$-blocked} or \textbf{$id$-blocked by $C$} if it is not $C$-$id$-open. For convenience, we often omit the prefix ‘$id$-’ and simply say that a path/walk is open or blocked.

Let $A, B, C \subseteq \Ic\cup \Vc\cup \{I_d\}_{d\in D}$ (not necessarily disjoint) be subsets of nodes. We then say that:
 \textbf{$A$ is $id$-separated from $B$ given $C$ in $\Gf_{\Do(I_D,T)}$}, in symbols:
$$
A \underset{\Gf_{\Do(I_D,T)}}{\stackrel{\mathsf{id}}{\perp}} B \mid C,
$$
if every path/walk from a node in $A$ to a node in $B\cup \Ic\cup \{I_d\}_{d\in D}\cup T$ is $id$-blocked by $C$.
\end{definition}

\begin{remark}[Potentially directed vs.\ definitely directed]
To motivate \cref{def:sep_PAG} in the case of iPAG $\Pf$, assume first for simplicity $\Ic=D=T=\emptyset$.  At a collider triple $v_{i-1}\suh v_i \hus v_{i+1}$, one may declare the path to be open given $C$ either if (i) $v_i\in\pan_{\Pf}(C)$ or if (ii) $v_i\in\anc_{\Pf}(C)$. The purpose of \cref{def:sep_PAG} is to characterize, soundly and completely, the connecting paths that occur in some $\Mf\in[\Pf]_\Ms$. In the absence of selection bias, it was observed in \cite{zhang2008causal,jaber22causal} that option (i) is only sound but not complete, whereas option (ii) is both sound and complete. The incompleteness of (i) comes from the fact that, when orienting a PAG $\Pf$ to become a MAG $\Mf\in [\Pf]_\Ms$ such that a PAG-open path with multiple colliders becomes open given $C$ in $\Mf$, it may require making multiple potentially directed paths from these colliders to $C$ become definitely directed. Such an orientation target need not be achievable simultaneously in a single $\Mf\in[\Pf]_\Ms$.

When $D\neq\emptyset$, it is tempting to extend option (ii) verbatim (e.g., by treating a collider as open whenever it lies in $\anc_{\Pf_{\Do(I_D)}}(C)$), but this can fail to be sound. Our definition therefore adopts a hybrid criterion, combining potentially directed and definitely directed paths. While multiple potentially directed paths cannot always be made definite directed simultaneously in a single MAG, any single such path always can; see the proof of \cref{thm:sep_hsint_pag_mag}.
\end{remark}

\begin{proposition}[Graphical separation in MAG representation]\label{prop:sep_mag}
    Let $\MAG$ be an iMAG. Then for $A, B,C\subseteq \Ic\cup \Vc$  we have
    \[
        A\sep{\mathsf{id}}{\Mf} B\mid C \quad \Longleftrightarrow \quad \forall \Af\in [\Mf]_\Gs:\ A\sep{\mathsf{id}}{\Af} B\mid C\cup \Sc_\Af.
    \]
\end{proposition}

The following theorem clarifies the relationship between separation in a hard- and soft-manipulated iMAG $\Mf_{\Do(I_D,T)}$ and corresponding separations in the  hard- and soft-manipulated isADMG $\Af_{\Do(I_D,T)}$ where $\Af\in [\Mf]_\Gs$.

\begin{theorem}[Main result I: Separation in manipulated iMAG/isADMG]\label{thm:sep_hsint_mag}
    Let $\MAG$ be an iMAG and $D,T\subseteq \Vc$ be disjoint. Let $A\subseteq \Vc\sm T$ and $B,C\subseteq \Ic\cup \{I_d\}_{d\in D}\cup  \Vc$ be pairwise disjoint. Then we have
  \[
      A \sep{\mathsf{id}}{\Mf_{\Do(I_D,T)}} B\mid C\cup T \quad \Longleftrightarrow \quad \forall \Af\in[\Mf]_\Gs: \  A \sep{\mathsf{id}}{\Af_{\Do(I_D,T)}} B\mid C\cup T\cup \Sc_{\Af}.
  \]
\end{theorem}

This theorem plays a central role in the subsequent development, so we record several remarks.

\begin{remark}\label{rem:sep_hsint_mag}
By \cref{thm:sep_hsint_mag},
\[
\begin{aligned}[t]
A \sep{\mathsf{id}}{\Mf_{\Do(I_D,T)}} B \mid C\cup T
&\ \Longrightarrow\ \forall \Af\in[\Mf]_\Gs:\ 
A \sep{\mathsf{id}}{\Af_{\Do(I_D,T)}} B \mid C\cup T\cup \Sc_\Af,\\
A \nsep{\mathsf{id}}{\Mf_{\Do(I_D,T)}} B \mid C\cup T
&\ \Longrightarrow\ \exists \Af\in[\Mf]_\Gs:\ 
A \nsep{\mathsf{id}}{\Af_{\Do(I_D,T)}} B \mid C\cup T\cup \Sc_\Af.
\end{aligned}
\]
Hence, the associated causal calculus rule for iMAGs is sound and \emph{atomically complete}: if a rule applies in the iMAG, it applies in every represented isADMG after conditioning on $X_{\Sc_\Af}$; if it does not apply, some represented isADMG witnesses this failure after conditioning on $X_{\Sc_\Af}$.
\end{remark}

\begin{remark}
In the second implication of \cref{rem:sep_hsint_mag}, the witnessing graph
$\Af\in[\Mf]_\Gs$ may depend on the particular triple $(A,B,C)$. In general, there is
no single isADMG $\Af$ represented by $\Mf$ such that, for all $A,B,C$,
\[
A \nsep{\mathsf{id}}{\Mf_{\Do(I_D,T)}} B \mid C\cup T
\ \Longrightarrow\
A \nsep{\mathsf{id}}{\Af_{\Do(I_D,T)}} B \mid C\cup T\cup \Sc_\Af.
\]
Equivalently, one may fail to find $\Af\in[\Mf]_\Gs$ with
\[
\mathsf{IM}(\Mf_{\Do(I_D)})=\mathsf{IM}\bigl(\Af_{\Do(I_D)}\mid \Sc_\Af\bigr),
\]
where $\mathsf{IM}(\Gf\mid \Sc)$ denotes the conditional independence model of $\Gf$ given $\Sc$; see \cref{ex:countex_ci} for a counterexample.

Consequently, the preceding result yields at most \emph{atomic completeness} of the causal calculus for iMAGs/iPAGs (\cf \cref{thm:causal_calculus_atomic_complete}). Since identification typically requires applying a \emph{sequence} of calculus rules, atomic completeness does not by itself imply the stronger form of completeness: namely, that failure to identify a functional via the calculus for an iMAG/iPAG guarantees the existence of a represented isADMG in which the functional is not identifiable.
\end{remark}

\begin{remark}\label{re:sint}
In the definition of $id$-separation, all input nodes are implicitly included on the right-hand side of the symbol $\perp$. This convention is essential in \cref{thm:sep_hsint_mag}. If one replaces $id$-separation by ordinary $d$-separation, the completeness direction in \cref{thm:sep_hsint_mag} may fail; in particular,
\[
\Bigl(\forall \Af\in[\Mf]_\Gs:\ A\sep{\mathsf{d}}{\Af_{\Do(I_D)}} B \mid C\cup \Sc_\Af\Bigr)
\;\centernot\Longrightarrow\;
A\sep{\mathsf{d}}{\Mf_{\Do(I_D)}} B \mid C .
\]
See \cref{ex:dvsd*} for a counterexample in which, for every $\Af\in[\Mf]_\Gs$
\[
    A\sep{\mathsf{d}}{\Af_{\Do(I_D)}} B \mid C\cup \Sc_\Af \quad \text{ and }\quad A\nsep{\mathsf{d}}{\Mf_{\Do(I_D)}} B \mid C .
\]
\end{remark}

\begin{remark}\label{rem:zhang}
In connection with \cref{re:sint}, recall that the converse of \cite[Corollary~13]{zhang2008causal} fails in general. By contrast, in our setting we have
\[
\Big(\forall \Af\in[\Mf]_\Gs:\
A \sep{\mathsf{id}}{\Af_{\Do(I_D)}} B \mid C\cup \Sc_\Af\Big)
\ \Longrightarrow\
A \sep{\mathsf{id}}{\Mf_{\Do(I_D)}} B \mid C .
\]
Accordingly, the counterexample in \cite{zhang2008causal} does not carry over; see Figure~\ref{fig:ex_zhang}. We have 
\[
b\nsep{\mathsf{id}}{\Mf_{\Do(I_a)}} c \mid a, \quad b\nsep{\mathsf{id}}{\Af^1_{\Do(I_a)}} c \mid a  \quad \text{ and } \quad b\nsep{\mathsf{id}}{\Af^2_{\Do(I_a)}} c \mid a.
\]

\begin{figure}[ht]
\centering
\begin{tikzpicture}[scale=0.8, transform shape]
\begin{scope}[xshift=0]
    \node[ndout] (a) at (0,0) {$a$};
    \node[ndout] (b) at (-1.5,0) {$b$};
    \node[ndout] (c) at (1.5,0) {$c$};
    \draw[arout] (a) to (b);
    \draw[arout] (a) to (c);
    \node at (0,-1) {$\Mf$};
\end{scope}

\begin{scope}[xshift=5cm]
    \node[ndout] (a) at (0,0) {$a$};
    \node[ndout] (b) at (-1.5,0) {$b$};
    \node[ndout] (c) at (1.5,0) {$c$};
    \draw[arout] (a) to (b);
    \draw[arout] (a) to (c);
    \draw[arlat, bend left] (b) to (a);
    \node at (0,-1) {$\Af^1$};
\end{scope}
\begin{scope}[xshift=10cm]
    \node[ndout] (a) at (0,0) {$a$};
    \node[ndout] (b) at (-1.5,0) {$b$};
    \node[ndout] (c) at (1.5,0) {$c$};
    \draw[arout] (a) to (b);
    \draw[arout] (a) to (c);
    \draw[arlat, bend left] (a) to (c);
    \node at (0,-1) {$\Af^2$};
\end{scope}

\begin{scope}[xshift=0,yshift=-4cm]
    \node[ndint] (I) at (0,1.5) {$I_a$};
    \node[ndout] (a) at (0,0) {$a$};
    \node[ndout] (b) at (-1.5,0) {$b$};
    \node[ndout] (c) at (1.5,0) {$c$};
    \draw[tuo] (I) to (a);
    \draw[arout] (I) to (b);
    \draw[arout] (I) to (c);
    \draw[arout] (a) to (b);
    \draw[arout] (a) to (c);
    \node at (0,-1) {$\Mf_{\Do(I_a)}$};
\end{scope}

\begin{scope}[xshift=5cm,yshift=-4cm]
    \node[ndint] (I) at (0,1.5) {$I_a$};    
    \node[ndout] (a) at (0,0) {$a$};
    \node[ndout] (b) at (-1.5,0) {$b$};
    \node[ndout] (c) at (1.5,0) {$c$};
    \draw[arout] (I) to (a);
    \draw[arout] (a) to (b);
    \draw[arout] (a) to (c);
    \draw[arlat, bend left] (b) to (a);
    \node at (0,-1) {$\Af^1_{\Do(I_a)}$};
\end{scope}

\begin{scope}[xshift=10cm,yshift=-4cm]
    \node[ndint] (I) at (0,1.5) {$I_a$};
    \node[ndout] (a) at (0,0) {$a$};
    \node[ndout] (b) at (-1.5,0) {$b$};
    \node[ndout] (c) at (1.5,0) {$c$};
    \draw[arout] (I) to (a);
    \draw[arout] (a) to (b);
    \draw[arout] (a) to (c);
    \draw[arlat, bend left] (a) to (c);
    \node at (0,-1) {$\Af^2_{\Do(I_a)}$};
\end{scope}
\end{tikzpicture}
\caption{Example about \cite[Corollary~13]{zhang2008causal} in \cref{rem:zhang}: $\Af^1$ and $\Af^2$ are represented by $\Mf$.}
\label{fig:ex_zhang}
\end{figure}
\end{remark}

\begin{example}[Counterexample to equality of conditional independence models]\label{ex:countex_ci}
Consider a MAG $\Mf$ and its soft-manipulated counterpart $\Mf_{\Do(I_a)}$ shown in Figure~\ref{fig:countex_ci}. We have 
\[
d\nsep{\mathsf{id}}{\Mf_{\Do(I_a)}} I_a\mid a\quad  \text{ and }\quad e\nsep{\mathsf{id}}{\Mf_{\Do(I_a)}} I_a\mid a,
\]
but there is no isADMG $\Af\in [\Mf]_\Gs$ such that both
\[
d\nsep{\mathsf{id}}{\Af_{\Do(I_a)}} I_a\mid \{a\}\cup \Sc_\Af \quad \text{ and }\quad  e\nsep{\mathsf{id}}{\Af_{\Do(I_a)}} I_a\mid \{a\}\cup \Sc_\Af
\]  
hold  simultaneously. This implies that \emph{for all} $\Af\in [\Mf]_\Gs$, we have $\mathsf{IM}(\Mf_{\Do(I_a)})\subsetneq\mathsf{IM}(\Af_{\Do(I_a)}\mid \Sc_\Af)$. 
\begin{figure}[ht]
\centering
\begin{tikzpicture}[scale=0.8, transform shape]
\begin{scope}[xshift=0]
    \node[ndout] (a) at (0,0) {$a$};
    \node[ndout] (b) at (-1.5,0) {$b$};
    \node[ndout] (c) at (1.5,0) {$c$};
    \node[ndout] (d) at (-3,0) {$d$};
    \node[ndout] (e) at (3,0) {$e$};
    \draw[arout] (a) to (b);
    \draw[arout] (a) to (c);
    \draw[arout] (b) to (d);
    \draw[arout] (c) to (e);
    \node at (0,-1) {$\Mf$};
\end{scope}

\begin{scope}[xshift=8cm]
    \node[ndint] (I) at (0,1.5) {$I_a$};
    \node[ndout] (a) at (0,0) {$a$};
    \node[ndout] (b) at (-1.5,0) {$b$};
    \node[ndout] (c) at (1.5,0) {$c$};
    \node[ndout] (d) at (-3,0) {$d$};
    \node[ndout] (e) at (3,0) {$e$};
    \draw[arout] (a) to (b);
    \draw[arout] (a) to (c);
    \draw[arout] (b) to (d);
    \draw[arout] (c) to (e);
    \draw[tuo] (I) to (a);
    \draw[arout] (I) to (b);
    \draw[arout] (I) to (c);
    \node at (0,-1) {$\Mf_{\Do(I_a)}$};
\end{scope}

\end{tikzpicture}
\caption{\cref{ex:countex_ci} showing the failure of $\mathsf{IM}(\Mf_{\Do(I_a)})=\mathsf{IM}(\Af_{\Do(I_a)}\mid \Sc_\Af)$.}
\label{fig:countex_ci}
\end{figure}
\end{example}

\begin{example}[Example for \cref{re:sint}]\label{ex:dvsd*}
    Consider a MAG $\Mf$ and its soft-manipulated counterpart $\Mf_{\Do(I_a)}$ shown in \cref{fig:dvsd*}. We have \[d\nsep{\mathsf{d}}{\Mf_{\Do(I_a)}} c\mid a\] as there is an open path $d\hut b\hut I_a\tut c$ from $d$ to $c$ given $a$ in $\Mf_{\Do(I_a)}$. As one can check, for all $\Af\in[\Mf]_\Gs$, we have 
    \[
    d\sep{\mathsf{d}}{\Af_{\Do(I_a)}} c\mid \{a\}\cup \Sc_\Af
    \] 
    since node $a$ cannot be a collider on a path from $d$ to $c$ in $\Af_{\Do(I_a)}$, e.g., $\Af^1_{\Do(I_a)}$ and $\Af^2_{\Do(I_a)}$ shown in \cref{fig:dvsd*}. Note that if $a$ is a collider on a path from $d$ to $c$ in an isADMG with the same set of observed nodes as $\Af$ (e.g., $\Af^3$ with $\Af^3_{\Do(I_a)}$ shown in \cref{fig:dvsd*}) then there must be an edge between $b$ and $c$ in its MAG representation ($\cat{MAG}(\Af^3)$), since there is an inducing path from $b$ to $c$ (there is a collider path from $b$ to $c$ with colliders $a,s\in \anc_{\Af^3}(\{s\})$). 
    
\begin{figure}[ht]
\centering
\begin{tikzpicture}[scale=0.8, transform shape]
\begin{scope}[xshift=0]
    \node[ndout] (a) at (0,0) {$a$};
    \node[ndout] (b) at (-1.5,0) {$b$};
    \node[ndout] (c) at (1.5,0) {$c$};
    \node[ndout] (d) at (-3,0) {$d$};
    \draw[arout] (a) to (b);
    \draw[tut] (a) to (c);
    \draw[arout] (b) to (d);
    \node at (-0.75,-1) {$\Mf$};
\end{scope}

\begin{scope}[xshift=7cm]
    \node[ndint] (I) at (0,1.5) {$I_a$};
    \node[ndout] (a) at (0,0) {$a$};
    \node[ndout] (b) at (-1.5,0) {$b$};
    \node[ndout] (c) at (1.5,0) {$c$};
    \node[ndout] (d) at (-3,0) {$d$};
    \draw[arout] (a) to (b);
    \draw[tut] (a) to (c);
    \draw[arout] (b) to (d);
    \draw[tut] (I) to (a);
    \draw[arout] (I) to (b);
    \draw[tut] (I) to (c);
    \node at (-0.75,-1) {$\Mf_{\Do(I_a)}$};
\end{scope}

\begin{scope}[xshift=0cm,yshift=-4cm]
    \node[ndint] (I) at (0,1.5) {$I_a$};
    \node[ndout] (a) at (0,0) {$a$};
    \node[ndout] (b) at (-1.5,0) {$b$};
    \node[ndout] (c) at (1.5,0) {$c$};
    \node[ndsel] (s) at (0.75,-1.25) {$s$};
    \node[ndout] (d) at (-3,0) {$d$};
    \draw[arout] (a) to (b);
    \draw[arout] (a) to (s);
    \draw[arout] (c) to (s);
    \draw[arout] (b) to (d);
    \draw[arout] (I) to (a);
    \draw[arlat, bend left] (a) to (b);
    \node at (-0.75,-2) {$\Af^1_{\Do(I_a)}$};
\end{scope}

\begin{scope}[xshift=7cm,yshift=-4cm]
    \node[ndint] (I) at (0,1.5) {$I_a$};
    \node[ndout] (a) at (0,0) {$a$};
    \node[ndout] (b) at (-1.5,0) {$b$};
    \node[ndout] (c) at (1.5,0) {$c$};
    \node[ndsel] (s) at (0.75,-1.25) {$s$};
    \node[ndout] (d) at (-3,0) {$d$};
    \draw[arout] (a) to (b);
    \draw[arout] (a) to (s);
    \draw[arout] (c) to (s);
    \draw[arout] (b) to (d);
    \draw[arout] (I) to (a);
    \draw[arlat, bend left] (a) to (s);
    \node at (-0.75,-2) {$\Af^2_{\Do(I_a)}$};
\end{scope}

\begin{scope}[xshift=14cm,yshift=-4cm]
    \node[ndint] (I) at (0,1.5) {$I_a$};
    \node[ndout] (a) at (0,0) {$a$};
    \node[ndout] (b) at (-1.5,0) {$b$};
    \node[ndout] (c) at (1.5,0) {$c$};
    \node[ndsel] (s) at (0.75,-1.25) {$s$};
    \node[ndout] (d) at (-3,0) {$d$};
    \draw[arout] (a) to (b);
    \draw[arout] (a) to (s);
    \draw[arout] (c) to (s);
    \draw[arout] (b) to (d);
    \draw[arout] (I) to (a);
    \draw[arlat, bend left] (a) to (b);
     \draw[arlat, bend left] (a) to (s);
    \node at (-0.75,-2) {$\Af^3_{\Do(I_a)}$};
\end{scope}

\end{tikzpicture}
\caption{MAG $\Mf$ in \cref{ex:dvsd*} such that $d\nsep{\mathsf{d}}{\Mf_{\Do(I_a)}} c\mid a$, but for $\Af^1,\Af^2\in [\Mf]_\Gs$, we have $d\sep{\mathsf{d}}{\Af^1_{\Do(I_a)}} c\mid \{a\}\cup \{s\}$ and $d\sep{\mathsf{d}}{\Af^2_{\Do(I_a)}} c\mid \{a\}\cup \{s\}$. We have $d\nsep{\mathsf{d}}{\Af^3_{\Do(I_a)}} c\mid \{a\}\cup \{s\}$, but $\Af^3\notin [\Mf]_\Gs$.}
\label{fig:dvsd*}
\end{figure}
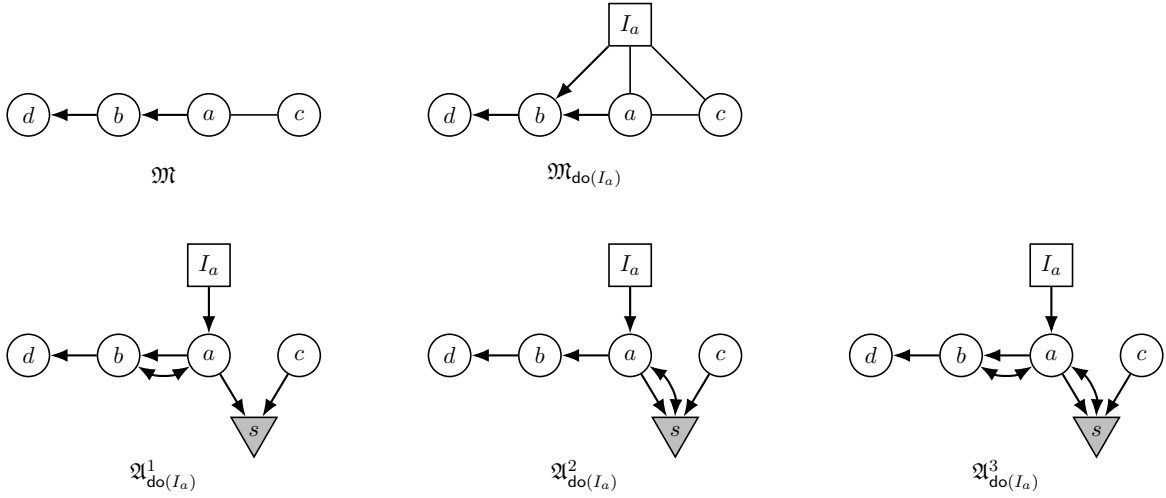
\end{example}

The next theorem connects separations in hard and soft-manipulated iPAGs, iMAGs and iADMGs.

\begin{theorem}[Main result II: Separation in manipulated iSOPAG/iMAG/iADMG]\label{thm:sep_hsint_pag_mag}
    Let $\PAG$ be an iSOPAG representing some isADMG $\Af$ and $D,T\subseteq \Vc$ be disjoint. Let $A\subseteq \Vc \sm T$ and $B,C\subseteq \Ic\dcup \{I_d\}_{d\in D}\dcup \Vc$ be pairwise disjoint. Then we have
  \[
    \begin{aligned}
        A \sep{\mathsf{id}}{\Pf_{\Do(I_D,T)}} B\mid C\cup T \quad &\Longrightarrow \quad \forall \Mf\in[\Pf]_\Ms: \  A \sep{\mathsf{id}}{\Mf_{\Do(I_D,T)}} B\mid C\cup T,\\  
        \text{ and }\quad  A \sep{\mathsf{id}}{\Pf_{\Do(I_D,T)}} B\mid C\cup T \quad &\Longrightarrow \quad \forall \Af\in[\Pf]_\Gs: \  A \sep{\mathsf{id}}{\Af_{\Do(I_D,T)}} B\mid C\cup T\cup \Af_\Sc.
    \end{aligned}      
  \]  
  Furthermore, if $\Pf$ is an iCOPAG, then we have
  \[
    \begin{aligned}
        A \sep{\mathsf{id}}{\Pf_{\Do(I_D,T)}} B\mid C\cup T \quad &\Longleftrightarrow \quad \forall \Mf\in[\Pf]_\Ms: \  A \sep{\mathsf{id}}{\Mf_{\Do(I_D,T)}} B\mid C\cup T,\\  
        \text{ and }\quad  A \sep{\mathsf{id}}{\Pf_{\Do(I_D,T)}} B\mid C\cup T \quad &\Longleftrightarrow \quad \forall \Af\in[\Pf]_\Gs: \  A \sep{\mathsf{id}}{\Af_{\Do(I_D,T)}} B\mid C\cup T\cup \Af_\Sc.
    \end{aligned}      
  \]  
\end{theorem}

\begin{remark}\label{re:hint_mag}
    In \cref{thm:sep_hsint_pag_mag}, the conclusion is established only for a restricted class of separation statements where hard-manipulated targets are always conditioned upon. If hard-manipulated targets are not conditioned on, as \cref{ex:hint_pag} shows, then the conclusion does not hold in general:
    \[
        A \sep{\mathsf{id}}{\Pf_{\Do(T)}} B\mid C \quad \centernot\Longrightarrow \quad \forall \Af\in[\Pf]_{\Gs}: \  A \sep{\mathsf{id}}{\Af_{\Do(T)}} B\mid C \cup \Sc_{\Af}.
    \]
\end{remark}

\begin{example}[Example for \cref{re:hint_mag}]\label{ex:hint_pag}
    Consider an SOPAG $\Pf$ and an sADMG $\Af\in [\Pf]_{\Gs}$ shown in \cref{fig:hint_pag}. We have 
    \[a \sep{\mathsf{id}}{\Pf_{\Do(t)}} b\mid \{c_1,c_2\},\] 
    since the paths 
    \[
    a\ouh c_1\huo c_2\ouo b \quad  \text{ and } \quad  a\ouh c_1 \out t
    \]
    in $\Pf_{\Do(t)}$ are both blocked by $\{c_1,c_2\}$. On the other hand, it follows that
    \[
    a \nsep{\mathsf{id}}{\Af_{\Do(t)}} b\mid \{c_1,c_2\}\cup \{s\},
    \] 
    since there is an open path $a\tuh c_1\hut t$ given $\{c_1,c_2\}\cup \{s\}$ in $\Af_{\Do(t)}$. Note that the path $a\tuh c_1\hut t$ is blocked by $\{c_1,c_2\}\cup \{s\}\cup \{t\}$ in $\Af_{\Do(t)}$.

    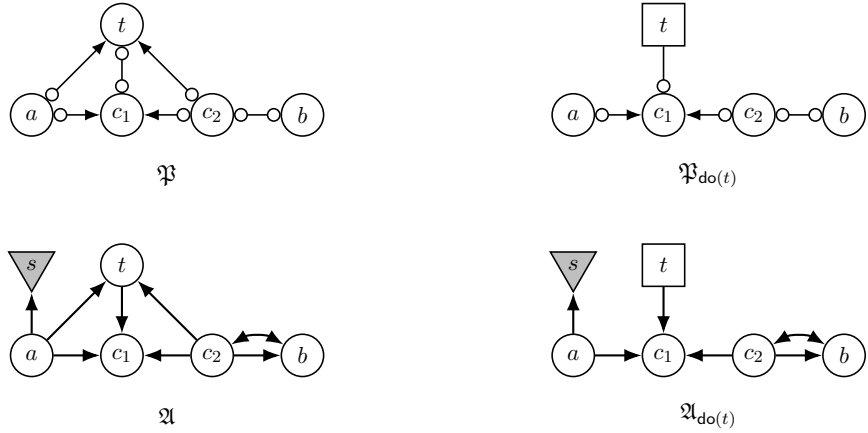
\begin{figure}[ht]
\centering
\begin{tikzpicture}[scale=0.8, transform shape]
\begin{scope}[xshift=0]
    \node[ndout] (a) at (-1.5,0) {$a$};
    \node[ndout] (c1) at (0,0) {$c_1$};
    \node[ndout] (c2) at (1.5,0) {$c_2$};
    \node[ndout] (b) at (3,0) {$b$};
    \node[ndout] (t) at (0,1.5) {$t$};
    \draw[ouh] (a) to (c1);
    \draw[ouh] (c2) to (c1);
    \draw[ouo] (c2) to (b);
    \draw[ouh] (c2) to (t);
    \draw[ouh] (a) to (t);
    \draw[ouo] (t) to (c1);
    \node at (0.75,-1) {$\Pf$};
\end{scope}

\begin{scope}[xshift=9cm]
    \node[ndout] (a) at (-1.5,0) {$a$};
    \node[ndout] (c1) at (0,0) {$c_1$};
    \node[ndout] (c2) at (1.5,0) {$c_2$};
    \node[ndout] (b) at (3,0) {$b$};
    \node[ndint] (t) at (0,1.5) {$t$};
    \draw[ouh] (a) to (c1);
    \draw[ouh] (c2) to (c1);
    \draw[ouo] (c2) to (b);
    \draw[tuo] (t) to (c1);
     \node at (0.75,-1) {$\Pf_{\Do(t)}$};
\end{scope}

\begin{scope}[xshift=0cm,yshift=-4cm]
     \node[ndout] (a) at (-1.5,0) {$a$};
     \node[ndsel] (s) at (-1.5,1.5) {$s$};
    \node[ndout] (c1) at (0,0) {$c_1$};
    \node[ndout] (c2) at (1.5,0) {$c_2$};
    \node[ndout] (b) at (3,0) {$b$};
    \node[ndout] (t) at (0,1.5) {$t$};
    \draw[arout] (a) to (c1);
    \draw[arout] (c2) to (c1);
    \draw[arout] (c2) to (b);
    \draw[arlat,bend left] (c2) to (b);
    \draw[arout] (c2) to (t);
    \draw[arout] (a) to (t);
     \draw[arout] (a) to (s);
    \draw[arout] (t) to (c1);
    \node at (0.75,-1) {$\Af$};
\end{scope}

\begin{scope}[xshift=9cm,yshift=-4cm]
    \node[ndout] (a) at (-1.5,0) {$a$};
    \node[ndsel] (s) at (-1.5,1.5) {$s$};
    \node[ndout] (c1) at (0,0) {$c_1$};
    \node[ndout] (c2) at (1.5,0) {$c_2$};
    \node[ndout] (b) at (3,0) {$b$};
    \node[ndint] (t) at (0,1.5) {$t$};
    \draw[arout] (a) to (c1);
     \draw[arout] (a) to (s);
    \draw[arout] (c2) to (c1);
    \draw[arout] (c2) to (b);
    \draw[arlat,bend left] (c2) to (b);
    \draw[arout] (t) to (c1);
    \node at (0.75,-1) {$\Af_{\Do(t)}$};
\end{scope}

\end{tikzpicture}
\caption{PAG $\Pf$, hard-manipulated PAG $\Pf_{\Do(t)}$, an isADMG $\Af\in [\Pf]_{\Gs}$ and its hard-manipulated isADMG $\Af_{\Do(t)}$ in \cref{ex:hint_pag}.}
\label{fig:hint_pag}
\end{figure}
\end{example}

\section{Causal reasoning with PAGs}\label{sec:reason_PAGs}

Building on the theory developed in \cref{sec:causal_mag_pag}, we derive rigorous causal reasoning rules for iSOPAGs, including a causal calculus and adjustment criteria and formulas.

\subsection{Probability calculus}

We recall several standard operations on Markov kernels that will be used in the sequel.

\begin{definitionthm}[Probability calculus]\label{defthm:prob_calculus}
    Let $\Xc$, $\Yc$, $\Zc$, $\Tc$, $\Uc$, $\Wc$ be standard measurable spaces. Consider Markov kernels 
    \[
    \begin{aligned}
        &\Kr(X,Y\miid T):\,\Tc \dto  \Xc \times \Yc, \\
    &\Kr_1(Z\miid U,X,T):\, \Uc \times \Xc\times \Tc \dto \Zc,\quad  \text{ and } \quad \Kr_2(X,Y\miid T,W):\, \Tc \times \Wc \dto \Xc \times \Yc.
    \end{aligned}
    \]
    
    \begin{enumerate}

    \item We define the \textbf{product Markov kernel} of $\Kr_1$ and $\Kr_2$ as follows:
    \[ 
    \begin{aligned}
        &\Kr_1(Z\miid U,X,T) \otimes \Kr_2(X,Y\miid T,W) :\, \Uc \times  \Tc \times \Wc \dto \Zc \times \Xc \times \Yc, \\
        &\big( \Kr_1(Z\miid U,X,T) \otimes \Kr_2(X,Y\miid T,W)\big) (B;(u,t,w))=\\
        &\int \Ibbm_B(z,x,y)\, \Kr_1(Z \in \dr z\miid U=u,X=x,T=t) \, \Kr_2((X,Y) \in \dr (x,y)\miid T=t,W=w).
    \end{aligned}
    \]

    \item  The \textbf{composition of Markov kernels}
    $\Kr_1(Z \miid U,X,T) \circ \Kr_2(X,Y \miid T,W) :\, \Uc \times  \Tc \times \Wc \dto \Zc $
    is defined using measurable sets $B \subseteq \Zc$ via:
    \[
    \begin{aligned}
    &\Bigl(\Kr_1(Z \miid U,X,T) \circ \Kr_2(X,Y\miid T,W)\Bigr)(B,(u,t,w)) \\
    &= \displaystyle \int \Kr_1(Z \in B\miid U=u,X=x,T=t) \, \Kr_2(X \in \dr x,Y\in \Yc \miid T=t,W=w).
    \end{aligned}
    \] 
    
    \item  We define the \textbf{marginal Markov kernels} of $\Kr(X,Y\miid T)$ over $X$ and $Y$, respectively, as follows:
    \[ 
    \begin{aligned}
        &\Kr(X\miid T)\coloneqq \Kr(X,Y\miid T)^{\sm Y}\coloneqq (\id_X\otimes \varepsilon_Y)\circ \Kr(X,Y\mid T) :\, \Tc \dto \Xc, \text{and }\\
        &\Kr(Y\miid T)\coloneqq \Kr(X,Y\miid T)^{\sm X}\coloneqq (\varepsilon_X\otimes\id_Y)\circ \Kr(X,Y\mid T) :\, \Tc \dto \Yc,
    \end{aligned}
    \]
    where $\id_X(X\in\cdot\mid X=x)=\delta_x(\cdot)$ and $\id_Y(Y\in \cdot\mid Y=y)=\delta_y(\cdot)$, and $\varepsilon_X:\Xc\dto \{*\}$ and $\varepsilon_Y:\Yc\dto \{*\}$ are the counit kernels, and we routinely identify 
    \[
    \Xc\times \{*\}\cong \Xc\cong \{*\}\times \Xc.
    \]

    \item There exists an essentially unique Markov kernel,\footnote{The existence and essential uniqueness are guaranteed by \cite[Lemma~2.23 and Theorem~2.24]{forre2021transitional} (see also \cite[Theorem~1.25]{kallenberg2017random} for a similar result). This generalizes the classical result of disintegration of probability distributions on standard measurable spaces to Markov kernels. This result can also be generalized to analytic measurable spaces \cite{bogachev20kant} and universal measurable spaces \cite{forre2021transitional}.}called a \textbf{conditional Markov kernel of $\Kr(X,Y\miid T)$ given $Y$}, $\tilde{\Kr}(X\miid Y,T):\, \Yc \times \Tc \dto \Xc$
    such that
    \[   \Kr(X,Y \miid T) = \tilde{\Kr}(X \miid Y,T) \otimes \Kr(Y\miid T),\]
    where $\Kr(Y\miid T)$ is the marginal Markov kernel of $\Kr(X,Y\miid T)$ over $Y$. We often denote $\tilde{\Kr}(X \miid Y,T)$ by $\Kr(X \mid Y\miid T)$ or $\Kr(X,Y\miid T)^{|Y}$. Here, essential uniqueness means that: if $\Qr(X\miid Y,T)$ is another Markov kernel, then we have 
    \[
    \Kr(X,Y\miid T) = \Qr(X\miid Y,T) \otimes \Kr(Y\miid T)
    \] 
    iff the set 
    \[
    \begin{aligned}
    N\coloneqq \big\{ (y,t) \in \Yc \times \Tc \mid  \exists &A \in \Sigma_\Xc \text{ s.t.\ }\,\\
    &\Qr(X \in A\miid Y=y,T=t) \neq \Kr(X \in A \mid Y=y \miid T=t)\big\}
    \end{aligned}
    \]
    is a measurable $\Kr(Y\miid T)$-null set in $\Yc \times \Tc$.\footnote{$N\subseteq \Yc \times \Tc$ is a measurable $\Kr(Y\miid T)$-null set in $\Yc \times \Tc$ if $\Kr(Y\in N_t\miid T=t)=0$ for all $t\in \Tc$ where $N_t=\{y\in \Yc\mid (y,t)\in N\}$.}

    \end{enumerate}
\end{definitionthm}

    \begin{definitionthm}[Absolute continuity and derivative {\protect\cite{forre2021transitional,forre2025mathematical}}]\label{defthm:derivative}
    Let $\Kr(W\miid T)$ and $\Qr(W\miid T)$ be two Markov kernels, and $\mu$ be a $\sigma$-finite measure on $(\Wc,\Sigma_\Wc)$.\footnote{WLOG, by a renormalization, we can take the $\sigma$-finite reference measure $\mu$ to be a probability measure.} We say that $\Kr(W\miid T)$ is \textbf{absolutely continuous} \wrt $\Qr(W\miid T)$ if for all $t\in \Tc$ and $D\in \Sigma_{\Wc}$
    \[
    \Qr(W\in D\miid T=t)=0 \quad  \Longrightarrow \quad  \Kr(W\in D\miid T=t)=0.
    \] 
    In symbols, we write $\Kr(W\miid T)\ll \Qr(W\miid T)$. The following two statements are equivalent:
    \begin{enumerate}
        \item  $\Kr(W\miid T)\ll \mu$.
        \item $\Kr(W\miid T)$ has a  derivative \wrt $\mu$, i.e., a \emph{joint measurable} map: $p:\Wc\times \Tc\to \Rb, (w,t)\mapsto p(w\miid t)$, such that for all $t\in \Tc$ and $D\in \Sigma_\Wc$:
        \[
            \Kr(W\in D\miid T=t)=\int_D p(w\miid t)\mu(\dr w).
        \]
        In this case, the derivative is essentially unique, i.e., for two such derivatives $p_1$ and $p_2$ we have $\mu(N_t)=0$ for all $t\in \Tc$ where 
        \[
            N\coloneqq \{(w,t)\in \Wc\times \Tc: p_1(w\miid t)\ne p_2(w\miid t)\}\in \Sigma_\Wc\otimes \Sigma_\Tc.
        \]
    \end{enumerate}
        Furthermore, $\Kr(W\miid T)$ has a strictly positive derivative \wrt $\mu$ iff 
        \[
        \mu\ll \Kr(W\miid T)\ll \mu
        \]
        (see, e.g., \cite[Corollary~2.3.20]{forre2025mathematical}).
    \end{definitionthm}

\begin{notation}[Equality of Markov kernels up to null sets]\label{notation:equality_markov_kernel} Let 
\[
\begin{aligned}
    &\Kr_1(X\miid T,U,W_1,W_2):\Tc\times \Uc\times \Wc_1\times \Wc_2\dto \Xc, \quad \Kr_2(X\miid U,W_1,W_2):\Uc\times \Wc_1\times \Wc_2\dto \Xc,\\
    &\Kr_3(X\miid T,U,W_1,W_2):\Tc\times \Uc\times \Wc_1\times \Wc_2\dto \Xc
\end{aligned}
\] 
be Markov kernels and $\mu_\Tc,\mu_\Uc$ be $\sigma$-finite reference measures on $\Tc$ and $\Uc$, respectively. We write
    \[ 
        \begin{aligned}
            \Kr_3(X\miid T,U,W_1,W_2)&\meq{\mu_{\Tc}\otimes\mu_{\Uc}}\Kr_1(X\miid T,U,W_1,W_2)\meq{\mu_{\Tc}\otimes\mu_{\Uc}}\Kr_2(X\miid U,W_1,W_2), \\
            \Kr_2(X\miid U,W_1,W_2)&\meq{\mu_{\Uc}}\Kr_2(X\miid U,\cancel{W_1},W_2)\\
            \Kr_1(X\miid T,U,W_1,W_2)&\meq{\mu_{\Uc}}\Kr_2(X\miid U,W_1,W_2)
        \end{aligned}
        \]
        to mean:
        \begin{enumerate}
            \item [(i)] equalities up to a measurable set $N\subseteq \Tc\times \Uc \times \Wc_1\times \Wc_2$ that is $\mu_{\Tc}\otimes\mu_{\Uc}$-null, i.e., $\mu_{\Tc}\otimes\mu_{\Uc}(N_{(w_1,w_2)})=0$ for all $(w_1,w_2)\in \Wc_1\times \Wc_2$;

            \item [(ii)] there exists a Markov kernel from $\Uc\times  \Wc_1\times\Wc_2$ to $\Xc$ independent on $W_1$ that is equal to $\Kr_2$ up to a measurable set $\tilde{N}\coloneqq N\times \Wc_1\subseteq (\Uc\times \Wc_2)\times \Wc_1$ such that $\tilde{N}$ is $\mu_{\Uc}$-null, i.e., $\mu_{\Uc}(N_{w_2})=0$ for all $w_2\in \Wc_2$;

            \item [(iii)] equality up to a measurable set $\tilde{N}\coloneqq \Tc\times N\subseteq \Tc\times (\Uc \times \Wc_1\times \Wc_2)$ that is $\mu_{\Uc}$-null, i.e., $\mu_{\Uc}(N_{(w_1,w_2)})=0$ for all $(w_1,w_2)\in \Wc_1\times \Wc_2$.
        \end{enumerate} 
\end{notation}

\begin{notation}[C-factor]\label{notation:cfactor}
Let $\Gf=(\Ic,\Vc,\Sc,\Ec)$ be an isADMG or $\Gf=(\Ic,\Vc,\Ec)$ be an iMAG/iSOPAG. Let $C\subseteq\Vc$. Assume that an s-iSCM $(\Mc,X_\Sc\in S)$ is the true underlying causal model and $\Gb(\Mc,X_\Sc\in S)\in[\Gf]_\Gs$.\footnote{If $\Gf$ is an isADMG, this simply reads $\Gb(\Mc,X_{\Sc}\in S)=\Gf$.} We introduce the following generic notation:
\[
    \Qc[C]\coloneqq \Prb_{\Mc}(X_C\mid X_\Sc\in S\miid  \Do(X_{\Vc\sm C}),X_{\Ic}).   
\]
\end{notation}

\subsection{Causal calculus}

From \cref{thm:sep_hsint_pag_mag,thm:sep_hsint_mag,defthm:causal_calculus,defthm:causal_calculus}, we obtain the following formal measure-theoretic causal calculus for iMAGs and iSOPAGs. 

\begin{theorem}[Main result III: Causal calculus for iMAG/iSOPAG]\label{thm:causal_calculus_mag_pag}
    Let $(\Mc,X_{\Sc}=\mathbf{1}_{|\Sc|})$ be an acyclic s-iSCM, where $\Mc=(\Ic,\Vc\dcup \Sc,\Wc,\Xc,\Prb,f)$ is an acyclic iSCM such that $\Prb_\Mc(X_\Sc=\mathbf{1}_{|\Sc|})>0$ where $\mathbf{1}_{|\Sc|}=(1,\ldots,1)\in \{0,1\}^{|\Sc|}$. Let $\Af\coloneqq \Gb(\Mc,X_{\Sc}=\mathbf{1}_{|\Sc|})$ be a causal isADMG and $\Gf$ be an iMAG or an iSOPAG such that $\Af\in[\Gf]_\Gs$. Let $A,B,C\subseteq \Vc$ and $D\subseteq \Ic\cup \Vc$ be pairwise disjoint. Write $D_1\coloneqq D\cap \Ic$ and $D_2\coloneqq D\cap \Vc$. Assume that there are $\sigma$-finite reference measures $\mu_v$ on $\Xc_v$ for each $v\in \Vc$ (write $\mu_F\coloneqq \bigotimes_{v\in F}\mu_v$ for $F\subseteq \Vc$). 
    
    \textbf{(1) Insertion/deletion of observations.}
        Suppose 
        \[A\sep{\mathsf{id}}{\Gf_{\Do(D)}} B\mid C\cup D.\] 
        Then there exists a Markov kernel $\Qr(X_A\miid X_C,X_D)$ unique up to a measurable set $\tilde{N}\coloneqq N\times \Xc_{\Ic\sm D_1}\subseteq \Xc_{C\cup D}\times \Xc_{\Ic\sm D_1}$, such that $\tilde{N}$ is $\Qc[D_2^c]^{\sm C^c}$-null and $\Qr$ is a version of $\big(\Qc[D_2^c]^{|B_1\cup C}\big)^{\sm A^c}$ for every $B_1\subseteq B$ simultaneously. If 
        \[\mu_{B\cup C}\ll \Prb_\Mc(X_B,X_C \mid X_\Sc=\mathbf{1}_{|\Sc|}\miid X_\Ic, \Do(X_{D_2}))\ll \mu_{B\cup C},\] then it follows that 
        \[ 
        \begin{aligned}
            &\Prb_\Mc(X_A\mid X_B,X_C,X_\Sc=\mathbf{1}_{|\Sc|} \miid X_\Ic,\Do(X_{D_2}))\\
            &\meq{\mu_{B\cup C}}\Prb_\Mc(X_A\mid X_C,X_\Sc=\mathbf{1}_{|\Sc|} \miid X_{\Ic} ,\Do(X_{D_2})) \\
            &\meq{\mu_{C}}\Prb_{\Mc}(X_A\mid X_C,X_\Sc=\mathbf{1}_{|\Sc|} \miid \cancel{X_{\Ic\sm D_1}}, X_{D_1} ,\Do(X_{D_2})).
        \end{aligned}
        \]

    \textbf{(2) Actions/observations exchange.}
        Suppose \[A\sep{\mathsf{id}}{\Gf_{\Do(I_B,D)}} I_B\mid B\cup C\cup D.\] Then there exists a  Markov kernel $\Qr(X_A\miid X_B,X_C,X_D)$ unique up to a measurable set $\tilde{N}\coloneqq N\times \Xc_{\Ic\sm D_1}\subseteq \Xc_{B\cup C\cup D}\times \Xc_{\Ic\sm D_1}$, such that $\tilde{N}$ is $\Qc[(B_2\cup D_2)^c]^{\sm(B_1\cup C)^c}$-null and $\Qr$ is a version of $\big(\Qc[(B_2\cup D_2)^c]^{|B_1\cup C}\big)^{\sm A^c}$ for every decomposition $B=B_1\dcup B_2$ simultaneously. If 
        \[
        \begin{aligned}
            \mu_{B\cup C}\ll \Prb_\Mc&(X_B,X_C\mid X_\Sc=\mathbf{1}_{|\Sc|} \miid X_\Ic, \Do(X_{D_2}))\ll \mu_{B\cup C} \quad \text{ and}\\ \mu_C\ll \Prb_\Mc&(X_C \mid X_\Sc=\mathbf{1}_{|\Sc|} \miid X_\Ic, \Do(X_B,X_{D_2}))\ll \mu_C,
        \end{aligned}
        \] 
        then it follows that 
        \[
            \begin{aligned}
                &\Prb_\Mc(X_A\mid X_C,X_\Sc=\mathbf{1}_{|\Sc|} \miid X_\Ic,\Do(X_{B},X_{D_2}))\\
                &\meq{\mu_{B\cup C}}\Prb_\Mc(X_A\mid X_B, X_C,X_\Sc=\mathbf{1}_{|\Sc|} \miid X_{\Ic},\Do(X_{D_2}))\\
                &\meq{\mu_{B\cup C}}\Prb_\Mc(X_A\mid X_B, X_C,X_\Sc=\mathbf{1}_{|\Sc|} \miid \cancel{X_{\Ic\sm D_1}}, X_{D_1},\Do(X_{D_2})).
            \end{aligned}
        \]

    \textbf{(3) Insertion/deletion of actions.}
    Suppose 
    \[A\sep{\mathsf{id}}{\Gf_{\Do(I_B,D)}} I_B \mid C\cup D.\] 
    Then there exists a Markov kernel $\Qr(X_A\miid X_C,X_D)$ unique up to a measurable set $\tilde{N}\coloneqq N\times \Xc_{\Ic\sm D_1}\subseteq \Xc_{C\cup D}\times \Xc_{\Ic\sm D_1}$, such that $\tilde{N}\times \Xc_{B_2}$ is $\Qc[(B_2\cup D_2)^c]^{\sm C^c}$-null and $\Qr$ is a version of $\big(\Qc[(B_2\cup D_2)^c]^{|C}\big)^{\sm A^c}$ for every $B_2\subseteq B$ simultaneously. If 
    \[
    \begin{aligned}
        \mu_{C}\ll \Prb_\Mc&(X_C \mid X_\Sc=\mathbf{1}_{|\Sc|}\miid X_{\Ic},\Do(X_B,X_{D_2}))\ll \mu_{C} \quad \text{ and} \\ 
        \mu_C\ll \Prb_\Mc&(X_C\mid X_\Sc=\mathbf{1}_{|\Sc|}\miid X_{\Ic}, \Do(X_{D_2}))\ll \mu_C,
    \end{aligned}
    \]
    then it follows that 
        \[
            \begin{aligned}
                &\Prb_\Mc(X_A\mid X_C,X_\Sc=\mathbf{1}_{|\Sc|} \miid X_\Ic,\Do(X_{B},X_{D_2}))\\
                &\meq{\mu_{C}} \Prb_\Mc(X_A\mid X_C,X_\Sc=\mathbf{1}_{|\Sc|} \miid X_\Ic,\Do(X_{D_2}))\\
                &\meq{\mu_{C}} \Prb_\Mc(X_A\mid X_C,X_\Sc=\mathbf{1}_{|\Sc|} \miid \cancel{X_{\Ic\sm D_1}}, X_{D_1},\Do(X_{D_2})).
            \end{aligned}
        \]
\end{theorem}

\begin{remark}
    Positivity conditions play an essential role in the soundness of the causal calculus. There are examples where identification results fail when positivity conditions are not met, even if the corresponding graphical criteria hold (see, e.g., \cite{forre2025mathematical,Yaroslav22revisite_gID,chen25notes}). There are various sufficient positivity conditions in the literature (see, e.g., \cite{shpitser06identification,forre2025mathematical,Yaroslav22revisite_gID,Hwang_2024_pos}). The problem of finding necessary positivity conditions remains open.
\end{remark}

\cref{thm:sep_hsint_pag_mag} yields the atomic completeness of the causal calculus developed in \cref{thm:causal_calculus_mag_pag}.

\begin{theorem}[Causal calculus is atomic complete]\label{thm:causal_calculus_atomic_complete}
    The causal calculus stated in \cref{thm:causal_calculus_mag_pag} is atomic complete for iMAGs and iCOPAGs: if a rule is not applicable in $\Gf$, then there must exist an s-iSCM $\Mc^\Sc \in \Mb^+(\Gf)$ for which the corresponding causal calculus rule is not applicable. 
\end{theorem}

Two simple corollaries of \cref{thm:sep_hsint_mag,thm:sep_hsint_pag_mag} are the iMAG/iSOPAG versions of \emph{invariance under intervention} \cite{spirtes2001causation,zhang2008causal} and of \emph{criteria for causal relationships} \cite{forre2025mathematical}. Briefly, the first gives a graphical characterization of $\Prb_\Mc(X_A\mid X_B,X_\Sc=\mathbf{1}_{|\Sc|})=\Prb_\Mc(X_A\mid X_B,X_\Sc=\mathbf{1}_{|\Sc|}\miid \Do(X_C))$ for an iMAG or iSOPAG, which is a corollary of the third rule of \cref{thm:causal_calculus_mag_pag}. The second characterizes when one may conclude that there is no (direct) causal effect or no confounding between two variables. See \cref{sec:criterion_relation} for details.

\subsection{Adjustment criterion and formula}

We formalize an adjustment criterion and the corresponding formula for iMAGs and iSOPAGs. Our formulation is inspired by \cite{forre2020causal,forre2021transitional} and generalizes the classic backdoor criterion \cite{Pearl93aAspects,pearl2009causality} and several of its variants and extensions \cite{Shpitser10validity,Pearl10confounding,Correa17causal,bareinboim2015recovering,shpitser06identification,Perkovic15complete}. 

\begin{theorem}[General adjustment criterion and formula]
    Assume the setting of \cref{thm:causal_calculus_mag_pag} and $\Ic\subseteq D$ WLOG. Let $J\coloneqq J_0\dcup J_1\subseteq \Vc$ and $H\subseteq \Vc$ be disjoint. Assume 
    \[
    \begin{aligned}
        \mu_{B\cup C\cup J\cup H}&\ll \Prb_\Mc(X_B,X_C,X_J,X_H\mid X_\Sc=\mathbf{1}_{|\Sc|}\miid \Do(X_D))\ll \mu_{B\cup C\cup J\cup H},\\
        \mu_{C\cup J\cup H}&\ll \Prb_\Mc(X_C,X_J,X_H\mid X_\Sc=\mathbf{1}_{|\Sc|}\miid \Do(X_B,X_D))\ll \mu_{ C\cup J\cup H}.
    \end{aligned}    
    \]
    Furthermore, assume
    \[
        \begin{aligned}
            J_0\cup H &\sep{\mathsf{id}}{\Gf_{\Do(I_B,D)}} I_B \mid C\cup D,\qquad 
            A \sep{\mathsf{id}}{\Gf_{\Do(I_B,D)}} J_1\cup I_B \mid B\cup C\cup D \cup J_0\cup H,\\
            H &\sep{\mathsf{id}}{\Gf_{\Do(I_B,D)}} B \mid I_B\cup C\cup D \cup J.
        \end{aligned}
    \]
    Then the adjustment formula holds true: 
    \[
    \begin{aligned}
        &\Prb_{\Mc}(X_A\mid X_C,X_\Sc=\mathbf{1}_{|\Sc|}\miid \Do(X_B,X_D)) \\
        &\meq{\mu_{B\cup C}}\Prb_{\Mc}(X_A\mid X_B,X_C,X_J,X_\Sc=\mathbf{1}_{|\Sc|}\miid \Do(X_D))\circ \Prb_{\Mc}(X_J\mid X_C,X_\Sc=\mathbf{1}_{|\Sc|}\miid \Do(X_D)). 
    \end{aligned}       
    \]
\end{theorem}

\begin{corollary}[Back-door adjustment criterion and formula]
    Assume the setting of \cref{thm:causal_calculus_mag_pag} and $\Ic\subseteq D$ WLOG. Assume 
    \[
        F \sep{\mathsf{id}}{\Gf_{\Do(I_B,D)}} I_B \mid  D,\quad \text{ and } \quad 
            A \sep{\mathsf{id}}{\Gf_{\Do(I_B,D)}}  I_B \mid B\cup F\cup D.
    \]
    Furthermore, assume
    \[  
        \begin{aligned}
            &\Prb_\Mc(X_F\mid X_{\Sc}=\1\miid \Do(X_D))\otimes \Prb_{\Mc}(X_B \mid X_\Sc=\1 \miid \Do(X_D)) \\
            &\ll \Prb_\Mc(X_F,X_B\mid X_\Sc=\1 \miid \Do(X_D)).
        \end{aligned}
    \]
    Then the adjustment formula holds true $\Prb_\Mc(X_B\mid X_\Sc=\1\miid \Do(X_D))\text{-a.s.}$:
    \[
        \begin{aligned}
            &\Prb_{\Mc}(X_A\mid X_\Sc=\mathbf{1}_{|\Sc|}\miid \Do(X_B,X_D))\\
            &=\Prb_{\Mc}(X_A\mid X_F,X_B,X_\Sc=\mathbf{1}_{|\Sc|}\miid \Do(X_D))\circ \Prb_{\Mc}(X_F\mid X_\Sc=\mathbf{1}_{|\Sc|}\miid \Do(X_D)).
        \end{aligned}
    \]
\end{corollary}

The above two results are derived from \cref{thm:sep_hsint_mag,thm:sep_hsint_pag_mag} and \cite[Theorem~5.2.3 and Corollary~5.2.5]{forre2025mathematical}.

\section{Identification algorithm for PAGs}\label{sec:IDalg}

In this section, we study the measure-theoretic identification algorithm for iSOPAGs under selection bias, termed \emph{sIDP algorithm}. We first introduce the formal definition of identifiability for (conditional) interventional kernels under selection bias, together with the graphical and measure-theoretic notions needed to formulate the algorithm. Then we show the soundness and completeness of the sIDP algorithm. 

\subsection{Causal identification}

Identifiability plays a fundamental role in statistical analysis. Causal effect identifiability \cite[Definition~3.2.4]{pearl2009causality} is subtle: small definitional changes may lead to errors (see, e.g., \cite{lee20gID,Yaroslav22revisite_gID}). Also note that in the presence of selection mechanisms, there are two types of causal identification: (i) s-recoverability \cite{bareinboim2015recovering}, and (ii) s-ID \cite{Abouei24sID,abouei24sIDlatent}. We focus on causal identificaiton of the s-ID type and formalize the identifiability and trackability of (conditional) interventional kernels under selection for isADMGs, iMAGs, and iSOPAGs.

\begin{definition}[Identifiability and trackability of (conditional) interventional kernels under selection]\label{def:iden}
Let $\Gf=(\Ic,\Vc,\Sc,\Ec)$ be an isADMG or $\Gf=(\Ic,\Vc,\Ec)$ be an iMAG/iSOPAG. Let $\Cc(\Gf)$ be a model class of s-iSCMs such that $\Gb(\Mc^\Sc)\in[\Gf]_\Gs$ for all $\Mc^\Sc=(\Mc,X_\Sc\in S)\in \Cc(\Gf)$.  Assume that the state spaces of models in $\Cc(\Gf)$ are such that $\Xc_i\subseteq [0,1]$ are standard Borel spaces, for all $i\in \Ic\dcup \Vc$.\footnote{Every standard Borel space is Borel-isomorphic to a Borel subset of $[0,1]$.}  
Define the kernel universe (including $D=\emptyset$)
\[
     \Pc(\Vc;\Ic)\coloneqq \bigcup_{D\subseteq \Vc}\big\{\Kr(X_{\Vc\sm D}\miid X_{D},X_{\Ic}): [0,1]^{|D\cup \Ic|}\dto [0,1]^{|\Vc\sm D|} \text{ Markov kernel} \big\}.
\] 
Define for $A,B,C\subseteq \Vc$ disjoint 
\[
\Phi_\Gf^{(A,B,C)}(\Mc^\Sc)\coloneqq \Prb_\Mc(X_A\mid X_C,X_{\Sc}\in S\miid \Do(X_B),X_\Ic)\in \Pc(\Vc;\Ic).\footnote{We identify $\Prb_\Mc(X_A\mid X_C,X_{\Sc}\in S\miid \Do(X_B),X_\Ic)$ with a kernel $\Kr(X_{\Vc\sm (B\cup C)}\miid X_{\Ic\cup B\cup C})\in \Pc(\Vc;\Ic)$ where
\[
    \begin{aligned}
        &\Kr(X_{\Vc\sm (B\cup C)}\miid X_{\Ic\cup B\cup C}=x_{\Ic\cup B\cup C})\\
    &=\begin{cases}
        \Prb_\Mc(X_A\mid X_C=x_C,X_{\Sc}\in S\miid \Do(X_B=x_B),X_\Ic=x_{\Ic})\otimes \delta_0^{\otimes |\Vc\sm (A\cup B\cup C)|}, \text{ if } x_{\Ic\cup B\cup C}\in \Xc_{\Ic\cup B\cup C}\\
        \delta_0^{\otimes |\Vc\sm (B\cup C)|}, \text{ otherwise.}
    \end{cases}
    \end{aligned}
\]
}
\]
We say that $\Prb_\Mc(X_A \mid X_C,X_\Sc\in S \miid  \Do(X_{B}), X_{\Ic})$ is \textbf{identifiable in $\Cc(\Gf)$} if there exists a mapping $\Psi_{\Gf}^{(A,B,C)}:\Phi_{\Gf}^{(\Vc,\emptyset,\emptyset)}(\Cc(\Gf))\to \Pc(\Vc;\Ic)$ such that  for all $\Mc^\Sc\in \Cc(\Gf)$
\[
\Phi_{\Gf}^{(A,B,C)}(\Mc^\Sc)=\Psi_\Gf^{(A,B,C)}\circ \Phi_{\Gf}^{(\Vc,\emptyset,\emptyset)}(\Mc^\Sc),
\]
where the equality holds up to versions of conditional kernels.

We say that $\Prb_\Mc(X_A \mid X_C,X_\Sc\in S \miid  \Do(X_{B}),X_{\Ic})$ is \textbf{trackable in $\Cc(\Gf)$ from $\Prb_\Mc(X_{\Vc}\mid X_\Sc\in S \miid X_{\Ic})$} (up to oracle choices of conditional kernels), if the mapping $\Psi_{\Gf}^{(A,B,C)}$ witnessing identifiability can be constructed as a composition of a finite sequence of operations of probability calculus and causal calculus in \cref{defthm:prob_calculus,defthm:causal_calculus,thm:causal_calculus_mag_pag} (with a model-dependent conditional kernel supplied at each conditioning step by an oracle) that depends only on $\Gf$ and $(A,B,C)$. 
\end{definition}

\begin{remark}
    At first glance, the existence of a mapping $\Psi$ witnessing causal identifiability need not imply trackability. However, completeness of the identification algorithm developed later (\cf \cref{thm:idp_complete}) shows that identifiability and trackability coincide in various cases.
\end{remark}

In \cref{def:iden}, the model class $\Cc(\Gf)$ may be, for example:
    \begin{enumerate}[label=(\roman*)]
        \item $\Mb(\Gf)$: all s-iSCMs $(\Mc,X_\Sc\in S)$ with $\Gb(\Mc,X_\Sc\in S)\in [\Gf]_\Gs$;
        \item $\Mb^+(\Gf)$: $(\Mc,X_\Sc\in S)\in \Mb(\Gf)$ for which there exist fixed $\sigma$-finite reference measures $\mu_v$ on $\Xc_v$ for all $v\in \Vc$ such that for every $D\subseteq \Vc$
        \[\mu_D\ll \Prb_\Mc(X_{D}\mid X_\Sc\in S \miid X_\Ic,\Do(X_{\Vc\sm D}))\ll \mu_D;\]
        \item $\Mb_c^+(\Gf)$: $(\Mc,X_\Sc\in S)\in \Mb^+(\Gf)$ such that there exists a causal Bayesian network
        \[
        \bigl(\Df=(\Ic,\Vc\dcup \Sc\dcup \Lc,\tilde{\Ec}),\{\Prb_v(X_v\miid X_{\pa_\Df(v)})\}_{v\in \Vc\cup \Sc\cup \Lc}\bigr)
        \]
        such that: it is interventionally equivalent to $\Mc$, its marginalized graph satisfies $\Df_{\sm \Lc}=\Gb(\Mc^\Sc)$, and for every $v\in \Vc\cup \Sc\cup \Lc$, the kernel $\Prb_v(X_v\miid X_{\pa_\Df(v)})$ is positive and continuous in the sense of \cref{def:pcmk}.

        \item $\Mb^+_{d}(\Gf)$: $(\Mc,X_\Sc\in S)\in \Mb(\Gf)$ for which $\Prb_\Mc(X_{\Vc\sm \Sc},X_{\Wc}\mid X_\Sc\in S \miid X_\Ic)$ has a positive probability mass function;\footnote{When endogenous variables of an SCM are discrete, assuming discrete exogenous variables with strictly positive probability mass function entails no loss of generality. See, e.g., \cite{Rosset_18_bound_card_hidden}.}
        \item $\Mb_{lg}^+(\Gf)$: $(\Mc,X_\Sc\in S)\in \Mb(\Gf)$ with linear causal mechanisms and Gaussian noise having a positive-definite covariance matrix.
         
    \end{enumerate}

    \begin{definition}[Positive and continuous Markov kernels \cite{Richard01causal_inference_continuous}]\label{def:pcmk}
    We say a Markov kernel $\Kr(X\miid Y):\Yc\dto \Xc$ is \textbf{positive and continuous} if
    \begin{enumerate}
        \item $\Xc$ and $\Yc$ are Polish spaces;

        \item (positivity) $\Kr(X\miid Y)$ is strictly positive on non-empty open subsets of $\Xc$;

        \item (Feller continuity) $\Kr(X\miid Y)$ is continuous as a map from $\Yc\to \Pc(\Xc)$ where $\Pc(\Xc)$ is the space of all probability measures on $\Xc$, equipped with the weak topology.
    \end{enumerate}
\end{definition}

See \cref{sec:pcmk} for several useful properties of positive and continuous Markov kernels.

\begin{remark}[On model class]\label{rem:model_class}
    \begin{enumerate}
        \item If $\Gf$ is an iADMG or an ADMG, interpret $\Cc(\Gf)$ as a class of iSCMs or SCMs, respectively.

        \item  For any isADMG/iMAG/iSOPAG $\Gf$,  
        \[
        \Mb^+_d(\Gf), \Mb^+_{lg}(\Gf)\subsetneq \Mb_{c}^+(\Gf) \subsetneq \Mb^+(\Gf)\subsetneq \Mb(\Gf).
        \]
        $\Mb(\Gf)$ is too broad for a reasonable causal identification result, since no positivity conditions are imposed (see e.g., \cite{Yaroslav22revisite_gID,forre2025mathematical,chen25notes}). $\Mb^+(\Gf)$ can yield almost-sure identification \wrt the reference measure $\mu$. For $\Mb_c^+(\Gf)$, one can obtain a pointwise identification result (provided one takes continuous version of conditional Markov kernels) and the same for $\Mb^+_d(\Gf)$ and $ \Mb^+_{lg}(\Gf)$. See \cref{thm:idp_sound}.

        \item If the density $p_{\Mc}(x_{v} \miid x_{\Ic},\Do(x_{\Vc\cup \Sc\sm \{v\}}))$ of $\Prb_{\Mc}(X_{v} \miid X_{\Ic},\Do(X_{\Vc\cup \Sc\sm \{v\}}))$, \wrt a $\sigma$-finite reference measure $\mu_{v}$ on $\Xc_{v}$, is strictly positive for all $v\in \Vc\cup \Sc$, $x_{\Vc\cup \Sc}\in \Xc_{\Vc\cup \Sc}$ and $x_{\Ic}\in \Xc_{\Ic}$, then $\Mc\in \Mb^+(\Gb(\Mc^\Sc))$ (see, e.g., \cite[Lemma~5.3.33]{forre2025mathematical}). $\Mc\in \Mb^+(\Gb(\Mc^\Sc))$ implies  $p_\Mc(x_{\Vc} \mid X_\Sc\in S \miid x_{\Ic})>0$ for all $x_\Vc\in \Xc_\Vc$ and $x_{\Ic}\in \Xc_{\Ic}$ but not conversely. 

    \end{enumerate}
\end{remark}

\subsection{Review: ID algorithm for ADMGs}

Before diving into the technical details of the sIDP algorithm, it is instructive to first review the basic idea behind the ID algorithm for ADMGs. Relevant references include \cite{tian02general,shpister06joint,shpitser2008complete,richardson2023nested,forre2025mathematical}. Let $\Mc\in \Mb^{+}_c(\Af)$ be an SCM with causal ADMG $\Af=(\Vc,\Ec)$. If nonempty sets $A,B\subseteq \Vc$ are disjoint, the ``one-line formulation'' of the ID algorithm, derived in \cite[Theorem~48]{richardson2023nested} is: if $\mathsf{Distr}(\Af_\Dc)\subseteq \mathsf{Intrin}(\Af)$ then
\begin{equation}\label{eqn:nested_mark}
 p_{\Mc}(x_A\miid \Do(x_B))=\sum_{x_{\Dc\sm A}}\prod_{D\in \mathsf{Distr}(\Af_\Dc)} \Qc[D] 
    = \sum_{x_{\Dc\sm A}}\prod_{D\in \mathsf{Distr}(\Af_\Dc)} \phi_{\Vc\sm D}(p_{\Mc}(x_{\Vc});\Af),
\end{equation}
where $\Dc=\anc_{\Af_{\Vc\sm B}}(A)$ and $\mathsf{Distr}(\Af_{\Dc})$ denotes the set of districts (i.e., c-components) of $\Af_{\Dc}$ and $\mathsf{Intrin}(\Af)$ denotes the set of intrinsic sets of $\Af$ \cite[Definition~33]{richardson2023nested}.\footnote{A general measure-theoretic formulation is:
    \[
        \Prb_\Mc(X_A\miid \Do(X_B))=\Big(\bigotimes_{D\in \mathsf{Distr}(\Af_\Dc)}^\succ \Qc[D]\Big)^{\sm (\Dc\sm A)},
    \]
    where the equality holds up to oracle choices of conditional kernels and the product of C-factors over districts is rigorously defined in \cite[Definition~5.3.16]{forre2025mathematical}.}   
    Every factor $\Qc[D]$ for $D\in \mathsf{Distr}(\Af_\Dc)\cap \mathsf{Intrin}(\Af)$ can be derived from $\Qc[\Vc]$ by applying the fixing operation \cite[Definition~19]{richardson2023nested} iteratively in an arbitrary order \cite[Theorem~31]{richardson2023nested}, which is defined as\footnote{Note that, conceptually, the fixing operation is different from hard intervention on graphs. We interpret $\phi_r(\Gf)\coloneqq \Gf_{\Do(r)}$ as a purely mathematical definition.} 
    \[
       \phi_r(\Gf)\coloneqq \Gf_{\Do(r)}, \quad \phi_r\big(q(x_{V}\miid x_W);\Gf\big)\coloneqq \frac{q(x_V\miid x_W)}{q(x_r\mid x_{\mathsf{Mb}_\Gf(r)\cap V} \miid x_{W})}
    \]
    for iADMG $\Gf=(W,V,\tilde{\Ec})$ and fixable node $r\in V$ in the sense that \cite[Definition~17]{richardson2023nested} \[\mathsf{Distr}_{\Gf}(r)\cap \de_{\Gf}(r)=\{r\}.\] This procedure is complete: if $\mathsf{Distr}(\Af_\Dc)\nsubseteq \mathsf{Intrin}(\Af)$ then the target interventional kernel is not identifiable \wrt $\Af$ \cite{shpitser2008complete,Huang2008OnTC,richardson2023nested,shpitser23doesidalgorithmfail}. 

    An algorithmic procedure of the above formulation could consist of three steps:
    \begin{enumerate}
        \item [(i)] Set $\Dc=\anc_{\Af_{\Vc\sm B}}(A)$.

        \item [(ii)] Decompose $\Af_{\Dc}$ into disjoint districts (c-components).

        \item [(iii)] For each $D\in \mathsf{Distr}(\Af_\Dc)$, check whether the fixing operation can be applied iteratively so as to obtain $\Af_{\Do(D^c)}$ graphically from $\Af$. If every $D\in \mathsf{Distr}(\Af_\Dc)$ passes the test, then $\Qc[D]=\phi_{D^c}(\Qc[\Vc];\Af)$ for all $D\in \mathsf{Distr}(\Af_\Dc)$; multiplying these factors and marginalizing yields the desired interventional kernel. If at least one $D\in \mathsf{Distr}(\Af_\Dc)$ fails the test, output $\textsc{Fail}$.
    \end{enumerate}
    

    The idea behind sIDP is to extend the above three-step procedure to iSOPAGs in a fully measure-theoretic setting. Conceptually, it suffices to identify the right iSOPAG counterparts of the key ADMG notions:

    \begin{enumerate}
    \item[(i)] a notion of an ``atomic unit'' in an iSOPAG (\cf \cref{def:bucket});
    \item [(ii)] a notion of ``ancestors'' (\cf \cref{def:podirected_path}) and a measure-theoretic rule of reducing the problem from $\Vc$ to ``ancestors'' of $A$ (cf. Rule~L0 in \cref{prop:sidp_rule});
    \item[(iii)] a notion of “district” in an iSOPAG (\cf \cref{def:pc-component,def:region}) and a measure-theoretic analogue of the “product over districts” factorization for iSOPAGs (\cf Rule~L1 in \cref{prop:sidp_rule});
    \item[(iv)] a notion of “fixable node,” together with a measure-theoretic fixing operation on kernels under iSOPAGs (\cf Rule~L2 in \cref{prop:sidp_rule}).
    \end{enumerate}

\subsection{sIDP: graphical notions}

We introduce the graphical notions needed to formulate the sIDP algorithm. The following definitions generalize those of \cite{jaber19idencom,jaber22causal}.

\begin{definition}[Bucket]\label{def:bucket}
    Let $\Gf=(\Ic,\Vc,\Ec)$ be an iPAG. Let $a,b\in \Vc$. We say that nodes $a$ and $b$ are in the same $\Gf$-\textbf{bucket} if there is a path $\pi$ in $\Gf$ from $a$ to $b$ such that there are no arrowheads on $\pi$. We write $a\in\bu_{\Gf}(b)$ and  $\bu_\Gf(B)\coloneqq\bigcup_{b\in B}\bu_{\Gf}(b)$ with $B\subseteq \Vc$. 
\end{definition}

\begin{definition}[pc-component]\label{def:pc-component}
Let $\Gf=(\Ic,\Vc,\Ec)$ be an iPAG. Let $a,b\in \Vc$. A path $\pi$ from $a$ to $b$ is called \textbf{pc-connecting} if it is of one of the following forms: 
\begin{enumerate}
    \item [(i)] $a\sus b$ not visible; or
    \item [(ii)] $a=v_0\suh v_{1} \huh \cdots \huh v_{n-1}\hus v_{n}=b$ for some $n>1$ and none of its edges are visible.
\end{enumerate}
We say that two nodes $a$ and $b$ are in the same \textbf{pc-component} in $\Gf$ if there is a pc-connecting path from node $a$ to node $b$ in $\Gf$. We write $a\in \pcc_{\Gf}(b)$ and  $\pcc_\Gf(B)\coloneqq\bigcup_{b\in B}\pcc_{\Gf}(b)$ with $B\subseteq \Vc$.
\end{definition}

\begin{definition}[Region]\label{def:region}
    Let $\Gf=(\Ic,\Vc,\Ec)$ be an iPAG. Let $a,b\in \Vc$. We say that node $a$ is in the \textbf{region} of $b$ in $\Gf$ if there exists $c\in \pcc_{\Gf}(b)$ such that $a\in \bu_{\Gf}(c)$. We write $a\in \re_{\Gf}(b)$ and $\re_{\Gf}(B)\coloneqq \bigcup_{b\in B}\re_{\Gf}(b)$ with $B\subseteq \Vc$.
\end{definition}

\subsection{sIDP: measure-theoretic operations}

We discuss measure-theoretic operations for sIDP in this subsection. Before that, we need to  first generalize the notion of topological order over nodes of an ADMG to buckets of an iSOPAG.

\begin{definition}[Topological order over buckets]
    Let $\Pf=(\Ic,\Vc,\Ec)$ be an iSOPAG. Let $D\subseteq \Ic\cup \Vc$. We call a partial order $\prec$ over buckets of $\Pf_D$ a \textbf{topological order of buckets of $\Pf_D$} if it satisfies the property:
    \[
      \Ab\subseteq \pan_{\Pf_D}(\Bb)\quad \Longrightarrow \quad   \Ab\prec \Bb. 
    \]
\end{definition}

There always exists a topological order over $\Pf_D$-buckets (\cf \cref{lem:topo_order}). Fix a topological order $\prec$ on the $\Pf_D$-buckets. Let $\Bf$ be the union of a set of $\Pf_D$-buckets. We write $\Bf^\prec$ and $\Bf^\succ$ as the union of all $\Pf_D$-buckets before and after all buckets in $\Bf$, respectively. Define $\Bf^\preceq\coloneqq \Bf^\prec\cup \Bf$ and $\Bf^\succeq\coloneqq \Bf^\succ\cup \Bf$.

\cref{thm:causal_calculus_mag_pag} yields the following three rules. These rules are measure-theoretic counterparts, for iSOPAGs, of the three elementary operations used in the ID algorithm for ADMGs.

\begin{proposition}\label{prop:sidp_rule}
    Let $\Pf=(\Ic,\Vc,\Ec)$ be an iSOPAG, and let $\Mc^\Sc=(\Mc,X_\Sc=\mathbf{1}_{|\Sc|})$ be an s-iSCM such that $\Gb(\Mc^\Sc)\in [\Pf]_\Gs$.  

    \begin{enumerate}[label=\textbf{Rule~L\arabic*:}, ref=Rule~L\arabic*, leftmargin=*, start=0]
        \item  Let $A,B\subseteq \Vc$ be disjoint with $A\ne \emptyset$. Define 
    \[
        \Dc\coloneqq \pant_{\Pf_{\Vc\sm B}}(A) \quad \text{ and } \quad  H\coloneqq (\Vc\sm (\Dc\cup B))\cup (\Ic\sm \tilde{\Dc}),
    \]
    where $\tilde{\Dc}\coloneqq \bigcup_{i \in \Ic} \pant_{\Pf_{(\Vc\sm B)\cup \{i\}}}(A)\cap \Ic$. Then $\Prb_{\Mc}(X_A\mid X_\Sc=\mathbf{1}_{|\Sc|} \miid \Do(X_B),X_{\Ic})$ is trackable from $\Qc[\Dc]$, and the following pointwise equalities hold:
    \[
        \begin{aligned}
             &\Prb_{\Mc}(X_A\mid X_{\Sc}=\1\miid \Do(X_B),X_{\Ic\sm H},\cancel{X_{\Ic\cap H}})\\
             &=\Prb_\Mc(X_A\mid X_{\Sc}=\1\miid \Do(X_B),X_{\Ic})\\
             &=\Prb_\Mc(X_A\mid X_{\Sc}=\1\miid \Do(X_{\Vc\sm \Dc}),X_{\Ic})=\Qc[\Dc]^{\sm(\Dc\sm A)}.
        \end{aligned}
    \]

    \item Let $D\subseteq \Vc$ and $A\subseteq D$. Set $R_1\coloneqq \re_{\Pf_D}(A)$ and $R_2\coloneqq \re_{\Pf_D}(D\sm R_1)$. Let $\Bb_1\prec \cdots \prec \Bb_n$ be a topological order of buckets of $\Pf_D$. Then $\Qc[D]$ is trackable from $\Qc[R_1]$ and $\Qc[R_2]$. If $\Mc^\Sc\in \Mb^+(\Pf)$, we have the pointwise equality 
    \begin{equation}\label{eq:sidp_rule1}
         \Qc[D]=\Qc[R_1]\boxtimes \Qc[R_2]=\Qc[R_2]\boxtimes \Qc[R_1],
    \end{equation}
    where 
    \[
        \Qc[R_1]\boxtimes \Qc[R_2]\coloneqq \bigotimes^>_{1\le i\le n} L_i
    \]
    and 
    \[  
        L_i\coloneqq \begin{cases}
            \Prb_{\Mc}(X_{\Bb_i}\mid X_{\Bb_i^\prec\cap R_1},X_{\Sc}=\1 \miid \Do(X_{R_1^c}),X_{\Ic}) \quad \text{ if }B_i\subseteq R_1\\
            \Prb_{\Mc}(X_{\Bb_i}\mid X_{\Bb_i^\prec\cap R_2},X_{\Sc}=\1 \miid \Do(X_{R_2^c}),X_{\Ic}) \quad \text{ if }B_i\subseteq R_2\sm R_1. 
        \end{cases}
    \] 

    \item  Let $D\subseteq \Vc$ and $\Ab\subseteq D$ be a bucket in $\Pf_D$. Write $D^+\coloneqq \pde_{\Pf_D}(\Ab)$ and $D^-\coloneqq (D\sm D^+)\cup \Ab$.  If $\pcc_{\Pf_D}(\Ab)\cap \pde_{\Pf_D}(\Ab)\subseteq \Ab$, then $\Qc[D\sm \Ab]$ is trackable from $\Qc[D]$ via
    \[
        \Qc[D\sm \Ab]=\Qc[D]^{|D^-}\otimes \Qc[D]^{\sm  D^+},
    \]
    where equality holds up to an oracle choice of the conditional kernel and holds $\mu_{\Ab}$-a.s.\ if $\Mc^\Sc \in \Mb^+(\Pf)$. If $\Mc^\Sc\in \Mb_c^+(\Pf)$, then, upon taking the continuous version of $\Qc[D]^{\mid D^-}\otimes \Qc[D]^{\setminus D^+}$, which exists, the equality holds pointwise.

    \end{enumerate} 
\end{proposition}

\begin{remark}
Let $\Af=(\Ic,\Vc,\Ec)$ be an iADMG and $\Mc$ an iSCM with $\Gb(\Mc)=\Af$. Let $A,B\subseteq \Vc$ be disjoint with $A\ne \emptyset$. 
    \begin{enumerate}
        \item Assume $\Ic=\emptyset$ and $\Mc\in \Mb^+_c(\Af)$. Define $\Dc\coloneqq \anc_{\Pf_{\Vc\sm B}}(A)$.  Then 
        \[
            p_\Mc(x_A\miid \Do(x_B))=p_\Mc(x_A\miid \Do(x_{\Vc\sm \Dc})).
        \]
        Rule~L0 in \cref{prop:sidp_rule} can be viewed as a generalization of this fact from ADMGs to iSOPAGs in full measure-theoretic generality.

        \item In Rule~L1 of \cref{prop:sidp_rule}, if $\Mc^\Sc\in \Mb^+(\Pf)$, then the density of $\Qc[R_1]\boxtimes\Qc[R_2]$ \wrt $\mu_{D}$ is 
        \begin{equation}\label{eq:markov_combination}
             \frac{q[R_1]\cdot q[R_2]}{q[R_1\cap R_2]}.
        \end{equation}
        where $q[R_1\cap R_2]$ denotes the density of $\Qc[R_1\cap R_2]$ \wrt $\mu_{R_1\cap R_2}$ and we use the fact that\footnote{This follows by the same argument as for Rule~1 in \cref{prop:sidp_rule}, noting that \cref{lem:region} and \cref{lem:IDP_sound_ci_2,lem:IDP_sound_ci_3} also apply to $R_1\cap R_2$. We therefore omit the proof.} 
        \[
            q[R_1\cap R_2]=\big(q[D]^{|(R_1\cap R_2)^{\prec}}\big)^{\sm (D\sm(R_1\cap R_2))} \quad \mu_{R_1\cap R_2}\text{-a.s.}
        \]    
        \cref{eq:sidp_rule1} extends the district product $\prod_{D\in \mathsf{Distr}(\Af_\Dc)} \Qc[D]$ in \cref{eqn:nested_mark} from ADMG models with discrete variables to iSOPAGs with general (possibly continuous) variables. Viewed purely as an operation on kernels, \cref{eq:markov_combination} is also a kernel-valued analogue of the \emph{Markov combination} of consistent probability distributions introduced by \cite{Dawid93Hyper_Markov}.\footnote{See, e.g., \cite{Byrne15Structural_Markov,massa2010combining,Goudie19Joining_Splitting} for further discussion of the Markov combination.}

        \item Consider the following operation for a fixable node $r\in \Vc$: 
        \[
            \begin{aligned}
                &\varphi_r\big(\Prb_{\Mc}(X_{\Vc}\miid X_{\Ic});\Af\big)\\
                &\coloneqq \Prb_{\Mc}(X_{\de_{\Af}(r)\sm \{r\}}\mid X_{\mathsf{NonDe}_{\Af_{\Vc}}(r)\cup \{r\}}\miid X_{\Ic}) \otimes \Prb_{\Mc}(X_{\mathsf{NonDe}_{\Af_{\Vc}}(r)}\miid X_{\Ic}),
            \end{aligned}
        \]
        where $\mathsf{NonDe}_{\Af_{\Vc}}(r)=\Vc\sm \de_{\Af}(r)$. If $\Mc\in \Mb^+_d(\Af)$, then
        \[
            \begin{aligned}
                &\varphi_r\big(p_{\Mc}(x_{\Vc}\miid x_{\Ic});\Af\big)\\
                &=p_{\Mc}(x_{\de_{\Af}(r)\sm \{r\}}\mid x_{\mathsf{NonDe}_{\Af_{\Vc}}(r)\cup \{r\}}\miid x_{\Ic})  p_{\Mc}(x_{\mathsf{NonDe}_{\Af_{\Vc}}(r)}\miid x_{\Ic})\\
                &=\frac{p_{\Mc}(x_{\de_{\Af}(r)\sm \{r\}}\mid x_{\mathsf{NonDe}_{\Af_{\Vc}}(r)\cup \{r\}}\miid x_{\Ic}) p_\Mc(x_{r}\mid x_{\mathsf{NonDe}_{\Af_{\Vc}}(r)}\miid x_\Ic)  p_{\Mc}(x_{\mathsf{NonDe}_{\Af_{\Vc}}(r)}\miid x_{\Ic})}{p_\Mc(x_{r}\mid x_{\mathsf{NonDe}_{\Af_{\Vc}}(r)}\miid x_\Ic)}\\
                &=\frac{p_\Mc(x_\Vc\miid x_{\Ic})}{p_\Mc(x_{r}\mid x_{\mathsf{NonDe}_{\Af_{\Vc}}(r)}\miid x_\Ic)}=\frac{p_\Mc(x_\Vc\miid x_{\Ic})}{p_\Mc(x_{r}\mid x_{\mathsf{Mb}_{\Af_{\Vc}}(r)}\miid x_\Ic)}=\phi_r\big(p_{\Mc}(x_{\Vc}\miid x_{\Ic});\Af\big),
            \end{aligned}
        \]
        where the fourth equality uses the fixability of $r$ (i.e., $\mathsf{Distr}_{\Af}(r)\cap \de_{\Af}(r)=\{r\}$) or \cite[Proposition~21]{richardson2023nested}. Therefore, the operation $\varphi_r(\cdot;\Af)$ provides a measure-theoretic generalization of the fixing operation in \cref{eqn:nested_mark}. Finally, Rule~L2 in \cref{prop:sidp_rule} generalizes $\varphi_\cdot(\cdot;\cdot)$ by formally applying the replacements:
        \[
            \big(\text{node}, \mathsf{Distr}_{\cdot}(\cdot),\de_{\cdot}(\cdot)\big) \curvearrowleft \big(\text{bucket}, \pcc_{\cdot}(\cdot), \pde_{\cdot}(\cdot)\big).
        \]
    \end{enumerate}
\end{remark}

\subsection{sIDP: ID algorithm for PAGs under
selection bias}

After the conceptual and technical preparations in the previous subsections, we are now ready to state the sIDP algorithm. To do so, it is useful to introduce a bookkeeping device that records how the final target interventional Markov kernel is assembled by repeated applications of Rule~L1 of \cref{prop:sidp_rule}. The motivation is as follows.

For the ID algorihm on ADMGs, the relevant decomposition is straightforward. One may start from an arbitrary node $v\in \anc_{\Af_{\Vc\sm B}}(A)$ and take its district $D$ in the induced subgraph $\Af_{\anc_{\Af_{\Vc\sm B}}(A)}$, remove $D$, and iterate on $\anc_{\Af_{\Vc\sm B}}(A)\sm D$ until all nodes are exhausted. This yields a disjoint collection of districts $\mathsf{Distr}(\Af_{\anc_{\Af_{\Vc\sm B}}(A)})=\{D_1,\ldots,D_n\}$. In contrast, for iSOPAGs the analogous procedure is slightly more delicate, and the ADMG-style “peel off districts” argument does not extend verbatim. The issue is that the region operator need not behave like a partition: there may exist bucket $\Bb\subseteq \re_{\Pf_{\Dc}}(\Ab)$, where $\Ab$ is a $\Pf_{\Dc}$-bucket, such that $\re_{\Pf_{\Dc}}(\Bb)\subsetneq \re_{\Pf_{\Dc}}(\Ab)$. Thus, naively iterating “take the region of a chosen bucket and remove it” may miss finer structure. A simple remedy is to apply the decomposition recursively: 
(i) pick an arbitrary bucket $\Ab$ and split $\Dc$ into two parts $R_1\coloneqq \re_{\Pf_{\Dc}}(\Ab)$ and $R_2\coloneqq\re_{\Pf_{\Dc}}(R_1\sm\Ab)$, which might overlap; (ii) repeat step~(i) for $R_1$ and $R_2$ until no further non-trivial decompositions are possible. This recursive splitting naturally motivates the following notion:

\begin{definition}[Assembly tree]\label{def:assembly_tree}
    An \textbf{assembly tree} for a non-empty set $\Dc$ is a finite rooted binary tree $\Ts=\Ts(\Dc)=(V,E,\cat{Lab})$ with:
    \begin{enumerate}
        \item  every node $v\in V$ has a non-empty label $\cat{Lab}(v)\subseteq \Dc$;
        \item if $r$ is the root node, $\cat{Lab}(r)=\Dc$;
        \item  if $v$ is internal with two children $v_1$ and $v_2$, then $\cat{Lab}(v)=\cat{Lab}(v_1)\cup \cat{Lab}(v_2)$.
    \end{enumerate}
\end{definition}

We denote by $\langle R\rangle$ an assembly tree consisting of a single node with label $R$, and by $\langle \emptyset \rangle$ an empty tree. Now we specify how to join two assembly trees and how to perform Rule~L1 of \cref{prop:sidp_rule} along an assembly tree.

\begin{definition}[Tree join]\label{def:tree_join}
    If $\cat{T}_1=\cat{T}_1(\Dc_1)$ and $\cat{T}_2=\cat{T}_2(\Dc_2)$, the \textbf{tree join} $\Ts=\Ts_1\Join \Ts_2$ is the assembly tree obtained by creating a new root $r$ with $\cat{Lab}(r)=\Dc_1\cup \Dc_2$, attaching the roots of $\Ts_1$ and $\Ts_2$ as children of $r$ while leaving the remaining parts of the two subtrees unchanged. Note that $\cat{T}\Join \langle \emptyset \rangle= \cat{T}$.
\end{definition}

\begin{definition}[Rule~L1 along an assembly tree]\label{def:mark_comb_tree}
    Assume the setting of \cref{prop:sidp_rule}. Let $\Ts=(V,E,\cat{Lab})$ be an assembly tree for set $\Dc$,  and let $\{\Qc[\cat{Lab}(v)]\}_{v\in V}$ be a collection of Markov kernels attached to every node $v\in V$. Assume that: if $v\in\Ts$ is an internal node with two children $v_1$ and $v_2$, then $\Qc[\cat{Lab}(v)]=\Qc[\cat{Lab}(v_1)]\boxtimes \Qc[\cat{Lab}(v_2)]$. Given this, we define the assembly product of kernels along the tree recursively:
    \begin{enumerate}
        \item  If $\Ts$ is an assembly tree consisting of a single node $v$ with label $\cat{Lab}(v)=R$, then set \[\bigotimes_\Ts\Qc[\bullet]=\bigotimes_{\langle R\rangle } \Qc[\bullet]\coloneqq \Qc[R].\]
        \item  If the root node $r$ of $\Ts$ is the root of two subtrees with $\Ts_1$ and $\Ts_2$, then set \[\bigotimes_\Ts\Qc[\bullet]=\bigotimes_{\Ts_1 \Join \Ts_2} \Qc[\bullet]\coloneqq \big(\bigotimes_{\Ts_1}\Qc[\bullet]\big)\boxtimes \big(\bigotimes_{\Ts_2}\Qc[\bullet]\big)=\big(\bigotimes_{\Ts_2}\Qc[\bullet]\big)\boxtimes \big(\bigotimes_{\Ts_1}\Qc[\bullet]\big),\]
        where the second equality holds up to order of the coordinates of the product space.
    \end{enumerate}
\end{definition}

After these preparations, we can now state the sIDP algorithm---\cref{alg:sIDP}. Conceptually, it consists of three steps:
\begin{enumerate}
    \item [(i)] we reduce the problem from $\Vc$ to $\Dc \coloneqq \pant_{\Pf_{\Vc \setminus B}}(A)$ by Rule~L0 in \cref{prop:sidp_rule};

    \item [(ii)] we decompose $\Dc$ into smaller pieces and construct an assembly tree $\Ts$ according to Rule~L1 in \cref{prop:sidp_rule};

    \item [(iii)] we check, for all leaf nodes $l_1,\ldots,l_n$ of $\Ts$ with label sets $C_1,\ldots,C_n$, whether we can track $\Qc[C_1],\ldots,\Qc[C_n]$ from $\Qc[\Vc]$ individually using Rule~L2 in \cref{prop:sidp_rule}.
\end{enumerate}
If $\Qc[C_1],\ldots,\Qc[C_n]$ are all trackable from $\Qc[\Vc]$, we combine them along the assembly tree $\Ts$ and obtain a proxy kernel $\hat{\Prb}(X_A\miid X_{\Vc\setminus \Dc},X_{\Ic})$ for the target interventional kernel.

\begin{algorithm}
\caption{$\operatorname{sIDP}(\Pf; A, B)$}
\label{alg:sIDP}
\begin{algorithmic}[1]
\State \textbf{Input:} iSOPAG $\Pf$ and disjoint sets $A, B \subseteq \Vc$ with $A\ne \emptyset$ and kernel $\Qc[\Vc]$
\State Put $\Dc=\Dc(\Pf;A,B) \coloneqq \pant_{\Pf_{\Vc \setminus B}}(A)$
\State $\Ts=(V,E,\cat{Lab}) \gets \cat{BuildTree}(\Dc,\Pf)$
\State $\{\Qr_v\}_{v\in V} \gets \cat{AttachKernel}\big(\Ts,\Vc,\Qc[\Vc],\Pf \big)$
\If{$\Qr_v\ne\textsc{Fail}$ for all $v$} 
\State $\hat{\Prb}(X_A\miid X_{\Vc\sm \Dc},X_{\Ic})\gets \big(\bigotimes_\Ts \Qr[\bullet]\big)^{\sm (\Dc\sm A)}$ 
\Return $\hat{\Prb}(X_A\miid X_{\Vc\sm \Dc},X_{\Ic})$ \Comment{Rule L0}
\Else 
\State \Return $\textsc{Fail}$
\EndIf 

\State \textbf{Output:} a proxy Markov kernel $\widehat{\Prb}(X_A \miid X_{\Vc\sm \Dc},X_{\Ic})$ for $\Prb_\Mc(X_A\mid X_S=\mathbf{1}_{|\Sc|} \miid \Do(X_B),X_{\Ic})$ or $\textsc{Fail}$

\Function{BuildTree}{$C, \Pf$}
    \If{$\exists\ \Pf_C\text{-bucket } \Bb\subsetneq C$ \st $\re_{\Pf_C}(\Bb)\subsetneq C$} \label{sIDP:L1} \Comment{Rule L1}
    \State pick one such $\Bb$ and set $C_1\coloneqq \re_{\Pf_C}(\Bb)$ and $C_2\coloneqq \re_{\Pf_C}(C\sm C_1)$
    \State $\Ts_1 \gets \cat{BuildTree}(C_1,\Pf)$, \quad $\Ts_2 \gets \cat{BuildTree}(C_2,\Pf)$
    \State \Return $\Ts_1\Join \Ts_2$
    \Else 
    \State \Return $\langle C \rangle$
    \EndIf
\EndFunction

\Function{AttachKernel}{$\Ts=(V,E,\cat{Lab}),T, \Qr,\Pf$}
    
   \State let $r$ be the root of $\Ts$ and set $R:=\cat{Lab}(r)$

  \If{$r$ has children $r_1,r_2$ with subtrees $\Ts_1,\Ts_2$}
    \State $\{\Qr_v\}_{v\in V_1}\gets \cat{AttachKernel}(\Ts_1,T,\Qr,\Pf)$, \quad $\{\Qr_v\}_{v\in V_2}\gets \cat{AttachKernel}(\Ts_2,T,\Qr,\Pf)$
    \If{$\exists \Qr_v=\textsc{Fail}$} $\Qr_r\gets \textsc{Fail}$
    \Else
    \State $\Qr_r \gets \Qr_{r_1}\boxtimes \Qr_{r_2}$ 
    \EndIf
     \State \Return $\{\Qr_v\}_{v\in V_1}\cup \{\Qr_v\}_{v\in V_2}\cup \{\Qr_r\}$
     \Else 
    \State $\tilde{T}\gets T$, \quad $\tilde{\Qr}\gets \Qr$
    \While{$\exists$ $\Pf_{\tilde{T}}$-bucket $\Bb\subseteq \tilde{T}\setminus R$ such that
           $\pcc_{\Pf_{\tilde{T}}}(\Bb)\cap\pde_{\Pf_{\tilde{T}}}(\Bb)\subseteq \Bb$} \label{sIDP:L2} \Comment{Rule L2}
      \State pick such $\Bb$,\quad $D^+\gets \pde_{\Pf_{\tilde{T}}}(\Bb)$,\quad $D^-\gets (\tilde{T}\setminus D^+)\cup \Bb$
      \State $\tilde{\Qr} \gets {\tilde{\Qr}}^{\mid D^-}\otimes {\tilde{\Qr}}^{\sm D^+}$,\quad $\tilde{T} \gets \tilde{T}\setminus \Bb$
    \EndWhile
    \If{$\tilde{T}=R$}\State \Return $\tilde{\Qr}$ \Else \State \Return $\textsc{Fail}$ \EndIf
    \EndIf
\EndFunction
\end{algorithmic}
\end{algorithm}

\subsection{sIDP: soundness and completeness}

The soundness of the sIDP algorithm follows from \cref{prop:sidp_rule}.

\begin{theorem}[Main result IV: sIDP is sound]\label{thm:idp_sound}
Let $\PAG$ be an iSOPAG and $A,B\subseteq \Vc$ be disjoint with $A\ne\emptyset$.
\begin{enumerate}
    \item \textbf{Soundness up to oracle choices.} If $\idp(\Pf;A,B)$ does not output $\textsc{Fail}$, then for every $(\Mc,X_S=\mathbf{1}_{|\Sc|})\in \Mb(\Gf)$ the kernel 
    \[
    \Prb_\Mc(X_A\mid X_S=\mathbf{1}_{|\Sc|}\miid \Do(X_B),X_\Ic) \quad  \text{ is trackable from } \quad \Prb_\Mc(X_\Vc\mid X_S=\mathbf{1}_{|\Sc|}\miid X_{\Ic})
    \] 
    up to oracle choices.

    \item \textbf{Soundness almost surely.} Let $(\Mc,X_S=\mathbf{1}_{|\Sc|})\in \Mb^+(\Pf)$. If $\idp(\Pf;A,B)$ does not output $\textsc{Fail}$, then 
    \[
        \widehat{\Prb}(X_A\miid X_{\Vc\sm \Dc},X_{\Ic})\meq{\mu_{\Vc\sm \Dc}} \Prb_\Mc(X_A\mid X_S=\mathbf{1}_{|\Sc|} \miid \Do(X_B),X_{\Ic}),
    \]

    \item \textbf{Soundness pointwise.} Let $(\Mc,X_S=\mathbf{1}_{|\Sc|})\in \Mb^+_c(\Pf)$. Assume that $\idp(\Pf;A,B)$ does not output $\textsc{Fail}$. Then every kernel produced by Rule~L2 in \cref{alg:sIDP} admits a continuous version. If, at each application of Rule~L2, we choose such a continouous version, then the following pointwise equality holds:
    \[
        \hat{\Prb}(X_A\miid X_{\Vc\sm \Dc},X_{\Ic})=\Prb_\Mc(X_A\mid X_S=\mathbf{1}_{|\Sc|} \miid \Do(X_B),X_{\Ic}).
    \]
\end{enumerate}
\end{theorem}

If at least one of $\Qc[C_1],\ldots,\Qc[C_n]$ is not trackable from $\Qc[\Vc]$ via Rule~L2, the algorithm outputs $\textsc{Fail}$. In that case, we can find a tuple $(C,T)$ with $\emptyset\neq C\subsetneq T\subseteq \Vc$ such that:
\begin{enumerate}
    \item [(i)] for all buckets $\Bb\subsetneq C$ of $\Pf_{\Dc}$ we have $\re_{\Pf_C}(\Bb)=C$, and

    \item [(ii)] for all buckets $\Bb\subseteq T\setminus C$ of $\Pf_{\Dc}$ it holds that $\pcc_{\Pf_T}(\Bb)\cap\pde_{\Pf_T}(\Bb)\nsubseteq \Bb$.
\end{enumerate}
 This allows us, for a COPAG $\Pf$, to construct a MAG $\Mf\in[\Pf]_\Ms$, and then an sADMG $\Af\in[\Mf]_\Gs$ such that the target interventional kernel is not identifiable \wrt $\Af$ (\cf \cref{sec:pf_prop_mag}). Hence, sIDP and the causal calculus in \cref{thm:causal_calculus_mag_pag} are complete.

\begin{theorem}[Main result V: sIDP and causal calculus are complete]\label{thm:idp_complete}
Let $\Pf=(\emptyset,\Vc,\Ec)$ be a COPAG and $A,B\subseteq \Vc$ be disjoint with $A\ne\emptyset$. If $\idp(\Pf;A,B)$ outputs $\textsc{Fail}$, then $\Prb_\Mc(X_A\mid X_S=\mathbf{1}_{|\Sc|}\miid \Do(X_B))$ is not identifiable in $\Mb^+_d(\Pf)$. 

The causal calculus (\cref{thm:causal_calculus_mag_pag}) for COPAGs is complete for identifying causal effects, i.e., any identifiable interventional Markov kernel can be achieved via a finite sequence of applications of the causal calculus rules and the probability calculus rules from the observational Markov kernel under a COPAG.
\end{theorem}

\section{Discussion}\label{sec:dis}

Inspired by and building on earlier foundational work in the literature, we establish a precise characterization of a restricted class of $id$-separations invariant across all isADMGs represented by  an iMAG or iCOPAG in terms of the corresponding $id$-separation in the iMAG or iCOPAG. For general iSOPAGs, we prove the corresponding soundness direction, which already suffices to formulate a measure-theoretic causal calculus and an identification algorithm under selection bias. For COPAGs, we show our calculus and algorithm are complete. Overall, this yields a complete method for causal identification from iMAGs and COPAGs under selection bias, and we conclude with several directions for further work.

We studied the problem of causal identification of unconditional intervential kernels under selection bias. To assess whether the \emph{selection-biased conditional interventional kernel} $\Prb_{\Mc}(X_A\mid X_C,X_S=\1\miid \Do(X_B),X_{\Ic})$ is trackable from $\Prb_\Mc(X_{\Vc}\mid X_S=\1\miid X_{\Ic})$ \wrt iSOPAG $\Pf$, one direct approach is to run $\idp(\Pf,A\cup C,B)$. If $\idp(\Pf,A\cup C,B)$ does not output $\textsc{Fail}$, then conditioning on $X_C$ immediately implies trackability. However, this procedure is not complete: there are cases where $\Prb_{\Mc}(X_A\mid X_C,X_S=\1\miid \Do(X_B),X_{\Ic})$ is identifiable while the above procedure fails. Based on the ideas in \cite{shpitser06identification} and \cite{jaber22causal}, we can first apply the second rule of the causal calculus (\cf \cref{thm:causal_calculus_mag_pag}) to exchange certain interventions and observations and then apply sIDP (\cf \cref{alg:sIDP}). This then yields the sCIDP algorithm given by \cref{alg:sCIDP} in \cref{sec:scidp}. Soundness follows directly from \cref{thm:causal_calculus_mag_pag,thm:idp_sound}, while we leave completeness to future work. It is worth mentioning that establishing completeness likely requires an appropriate conditional analogue of a so-called Hedge criterion for non-identifiability. Some existing formulations of the Hedge criterion in the literature, e.g., \cite[Theorem~3]{shpitser06identification} and \cite[Theorem~5]{jaber22causal}, appear to be imprecise, as argued by \cite{shpitser23doesidalgorithmfail}.

Even when a causal effect is not point-identified, one may still derive \emph{informative bounds} \cite{balke97bounds}. It would be of interest to develop a complete method for computing sharp bounds (for discrete variables) in MAGs and PAGs under selection bias. Recent progress in this direction without selection bias includes \cite{Bellot24bound_causal_effects_mark_equi}.

In the complete setting considered here, sIDP produces an identifying functional whenever the target effect is identifiable. The \emph{statistical properties} of estimators based on this functional remain to be studied—for instance, questions of efficiency and the construction of more efficient estimators when the naive plug-in is suboptimal. Related work in the no-selection setting (assuming a given PAG) is \cite{Jung21Estimate_causal_effec_on_markov_equiva}. A further step is to study the validity of “discover-then-estimate’’ pipelines with latent variables and selection bias. Recent post-selection inference results in nearby settings include \cite{Chang2025PostSelection,Gradu2025ValidInference}, though selection bias is not yet incorporated.

\emph{Validation of causal discovery} methods also remains challenging. Much of the literature evaluates learned graphs via structural distances such as structural Hamming distance, but this approach faces several obstacles: the ground-truth graph for real data is typically unknown; simulation benchmarks can be misleading \cite{reisach2021beware_simulated_dag}; learned graphs may be unstable to small perturbations; and, more fundamentally, not every data-generating process with a causal interpretation is well captured by standard graph models \cite{chen2025foundationsstructuralcausalmodels,Blom19beyond}. These considerations motivate effect-level validation, where performance is assessed through testable interventional predictions \cite{gentzel2019case_interventional_eval,dang2026effect_level_validation,Peters2015structural_int_distance}. Our identification results may provide useful tools for such validation in the presence of selection bias.

Another interesting direction for future work is \emph{early stopping} in the FCI–identification pipeline. In the present approach, causal calculus is applied only after FCI has terminated. In many cases, however, the partial graph available before termination may already suffice to identify the target causal query. Leveraging this could reduce computational cost and the number of conditional-independence tests.

Finally, \emph{adding input nodes} (regime indicators) offers a universal method for endowing graphical models with \emph{causal semantics} as discussed by Dawid in \cite[p.~348--351]{Lauritzen02chaingraph}. However, the construction is non-unique: different ways of adding input nodes generally lead to different causal interpretations. In this paper, we show that under the ``canonical'' causal interpretation of MAGs and PAGs as representations of causal ADMGs with selection variables, the ``canonical'' choice is the construction in \cref{def:int_mag,def:int_pag}, rather than merely adding an input node as a parent of the targeted observed node. This naturally raises a broader question: for other classes of graphical models, can one identify a principled “canonical’’ way to introduce input nodes that matches their intended ``canonical'' causal interpretation?

\acks{We thank Booking.com for support.}

\newpage

\appendix
\crefalias{section}{appendix}
\section{Preliminaries}\label{sec:preliminary}

\subsection{Structural Causal Model with inputs and selection mechanisms (s-iSCM)}

In this subsection, we present some basics of Structural Causal Model with inputs and selection mechanisms. Relevant references include \cite{forre2025mathematical,chen2025foundationsstructuralcausalmodels,bongers2021foundations}.

\begin{definition}[Structural Causal Model with inputs]\label{def:iSCM}
    A \textbf{Structural Causal Model with inputs (iSCM)} is a tuple $\Mc=(\Ic,\Vc,\Wc,\Xc,\Prb,f)$ such that
    \begin{enumerate} 
      \item   $\Ic,\Vc,\Wc$ are disjoint finite sets of labels for the \textbf{exogenous input variables}, \textbf{endogenous variables}, and the \textbf{latent exogenous random variables}, respectively;
      \item   the \textbf{state space} $\Xc = \prod_{i\in \Ic\dcup \Vc \dcup \Wc} \Xc_i$ is a product of standard measurable spaces $\Xc_i$;
      \item   the \textbf{exogenous distribution} $\Prb$ is a probability distribution on $\Xc_{\Wc}$ that factorizes as a product $\Prb = \bigotimes_{w\in \Wc} \Prb(X_w)$ of probability distributions $\Prb(X_w)$ on $\Xc_w$;
      \item   the \textbf{causal mechanism} is specified by the measurable mapping $f : \Xc \rightarrow \Xc_{\Vc}$.
\end{enumerate}
\end{definition}

\begin{definition}[Hard intervention]\label{def:scm_hard_intervention}
  Given an iSCM $\Mc$ and an intervention target $T \subseteq \Vc$, we define the \textbf{intervened iSCM} $$\Mc_{\Do(T)} \coloneqq  (\Ic\dcup T,\Vc\sm T,\Wc,\Xc,\Prb,f_{\Vc\setminus T}).$$
\end{definition}

\begin{definition}[iSCM with selection mechanism (s-iSCM)]\label{def:s-scm}
We call $\Mc^\Sc\coloneqq(\Mc,X_{\Sc}\in S)$ an \textbf{s-iSCM} or \textbf{iSCM with a selection mechanism}, where $\Mc=(\Ic,\Vc\dcup \Sc,\Wc,\Xc,\Prb,f)$ is an iSCM, and $S\subseteq \Xc_{\Sc}$ is a measurable subset such that $\Prb_{\Mc}(X_{\Sc}\in S\miid X_{\Ic}=x_{\Ic})>0$ for all $x_\Ic\in \Xc_{\Ic}$. The causal semantics of $\Mc^\Sc$ is as follows. 
\begin{enumerate}
    \item Observable Markov kernel:
    \[\Prb_{\Mc^\Sc}(X_{\Vc}\miid X_{\Ic})\coloneqq\Prb_\Mc(X_{\Vc}\mid X_\Sc\in S\miid X_{\Ic});\]

    \item Interventional Markov kernel: for $T\subseteq \Vc$ with $\Prb_{\Mc}(X_{\Sc}\in S\miid \Do(X_T=x_T),X_{\Ic}=x_{\Ic})>0$ for all $x_T\in \Xc_T$ and $x_\Ic\in \Xc_{\Ic}$, we define 
    \[
    \Prb_{\Mc^\Sc}(X_{\Vc\sm T}\miid \Do(X_T),X_{\Ic})\coloneqq\Prb_{\Mc}(X_{\Vc\sm T }\mid X_{\Sc}\in S\miid \Do(X_T),X_\Ic).
    \]
\end{enumerate}
\end{definition}

\begin{definition}[Parent according to iSCM]\label{def:parent}
    Let $\Mc=(\Ic,\Vc,\Wc,\Xc,\Prb,f)$ be an iSCM. For $i\in \Ic\dcup \Vc \dcup \Wc$ and $j\in \Vc$,
we say that $i$ is a \emph{parent} of $j$ according to $\Mc$ if there does not exist a
measurable function
\[
\tilde f_j:\ \mathcal{X}_{(\Ic\cup \Vc\cup \Wc)\setminus\{i\}}\to \mathcal{X}_j
\]
such that
\[
f_j(x)=\tilde f_j(x_{(\Ic\cup \Vc\cup \Wc)\setminus\{i\}})
\qquad \text{ for all $x\in \Xc\sm N$},
\]
where $N=\tilde{N}\times \Xc_\Vc$ with $\tilde{N}\subseteq \Xc_{\Ic}\times \Xc_\Wc$ being such that the section $\tilde{N}_{x_\Ic}$ is a $\Prb_{\Mc}(X_W)$-null set for each $x_\Ic\in \Xc_\Ic$.
\end{definition}

\begin{definition}[Graph of an s-iSCM]\label{def:graph_siscm}
Let $\Mc^\Sc\coloneqq(\Mc,X_{\Sc}\in S)$ be an s-iSCM with iSCM $\Mc=(\Ic,\Vc\dcup \Sc,\Wc,\Xc,\Prb,f)$. The isADMG $(\Ic,\Vc,\Sc,\Ec)$ has input nodes $\Ic$, output nodes $\Vc$, latent selection nodes $\Sc$, directed edges
\[
\Ec_1=\{\, i\tuh j:\ i\in \Ic\cup \Vc\cup \Sc,\ j\in \Vc\cup \Sc \text{ \st } i \text{ is parent of } j \text{ according to } \Mc \,\}
\]
and bidirected edges
\[
\begin{aligned}
    \Ec_2=\{\, j\huh k:\ & j\in \Vc\cup \Sc,\ k\in \Vc\cup \Sc,\ j\neq k \text{ \st }\\
    &j \text{ and } k \text{ share a common parent in } \Wc
\text{ according to } \Mc \,\},
\end{aligned}
\]
and $\Ec=\Ec_1\cup \Ec_2$.
\end{definition}

\subsection{Transitional probability theory, conditional independence and measure-theoretic causal calculus}

In this section, we present some basics of \emph{transitional probability theory}, which provides a convenient framework handling random variables and non-stochastic variables simultaneously, together with a corresponding notion of conditional independence for transitional random variables. We follow the setup introduced in \cite{forre2021transitional}.

\begin{definition}[Transitional probability space and random variable]
  Let $\Kr(W\miid T)$ be a Markov kernel from measurable space $(\Tc,\Sigma_\Tc)$ to $(\Wc,\Sigma_{\Wc})$. Then the tuple $(\Wc\times \Tc,\Kr(W\miid T))$ is called a \textbf{transitional probability space}. A measurable map
  $
    X: \mathcal{W} \times \mathcal{T} \rightarrow \Xc
  $
  is called a \textbf{transitional random variable}.
\end{definition}

\begin{definition}[Transitional conditional independence]\label{def:tran_ci}
Let $(\Wc \times \Tc, \mathrm{K}(W\miid T))$ be a transitional probability space. Consider transitional random variables:
\[ 
X: \Wc \times \Tc \rightarrow \Xc,\  Y:\,\Wc \times \Tc \rightarrow \mathcal{Y},\
Z:\Wc \times \Tc \rightarrow \mathcal{Z}.
\]
We say that \textbf{$X$ is independent of $Y$ conditioned on $Z$ w.r.t.\ $\Kr(W\miid T)$}, in symbols:
\[X \Ind{}{\mathrm{K}(W\miid T)} Y \mid Z, \]
if there exists a Markov kernel $\Qr(X\miid Z):\; \Zc \dashrightarrow \Xc,$ such that:
\[ \mathrm{K}(X,Y,Z\miid T) = \mathrm{Q}(X\miid Z) \otimes \mathrm{K}(Y,Z\miid T ),\]
where $\mathrm{K}(Y,Z\miid T)$ is the marginal of $\mathrm{K}(X,Y,Z\miid T)$.
As a special case, we define:
\[X  \Ind{}{\mathrm{K}(W\miid T)} Y \qquad :\iff \qquad X \Ind{}{\mathrm{K}(W\miid T)} Y \mid \ast. \]
\end{definition}

\begin{remark}[Essential uniqueness]
  The Markov kernel $\mathrm{Q}(X\miid Z)$ appearing in the conditional independence $X  \Ind{}{\mathrm{K}(W\miid T)} Y \mid Z$ in \cref{def:tran_ci} is then a version of a conditional Markov kernel $\mathrm{K}(X\mid Y,Z\miid T)$ and is thus essentially unique by \cref{defthm:prob_calculus}. We will use the following suggestive notation for it:
    \[ \mathrm{K}(X\mid \cancel{Y},Z\miid \cancel{T}) \coloneqq \mathrm{Q}(X\miid Z).\]
    So we have in case of $X  \Ind{}{\mathrm{K}(W\miid T)} Y \mid Z$:
    \[ \mathrm{K}(X,Y,Z\miid T) = \mathrm{K}(X\mid \cancel{Y},Z\miid \cancel{T})\otimes \mathrm{K}(Y,Z\mid T).\]
    Note that the conditional independence establishes that there is a version of the conditional Markov kernel $\mathrm{K}(X\mid Y,Z\miid T)$ that depends only on $z$ and does not depend on $y$ or $t$.
\end{remark}

We now present the measure-theoretic causal calculus for iADMGs. The rules below are corollaries of the following strong global Markov property for iADMGs. Relevant references include \cite{forre2021transitional,forre2025mathematical}.

\begin{theorem}[Strong global Markov property]\label{thm:strong_markov_property}
    Let $\Mc=(\Ic,\Vc,\Wc,\Xc,\Prb,f)$ be an iSCM whose causal graph is an iADMG $\Af\coloneqq \Gb(\Mc)=(\Ic,\Vc,\Ec)$. Then for all $A,B,C\subseteq \Ic\cup \Vc$ (not necessarily disjoint), the following implication holds:
    \[
        A\sep{\mathsf{id}}{\Af} B\mid C \quad \Longrightarrow \quad X_A\Ind{}{\Prb_{\Mc}(X_{\Vc}\miid X_{\Ic})} X_B\mid X_C.
    \]
\end{theorem}

\begin{theorem}[Causal calculus (iADMGs)]\label{defthm:causal_calculus}
Let $\SCM$ be an acyclic iSCM and $\Af\coloneqq \Gb(\Mc)$ be an iADMG of $\Mc$. Let $A,B,C\subseteq \Vc$ and $D\subseteq \Ic\cup \Vc$ be pairwise disjoint. Write $D_1\coloneqq D\cap \Ic$ and $D_2\coloneqq D\cap \Vc$. Assume that there are $\sigma$-finite reference measures $\mu_v$ on $\Xc_v$ for each $v\in \Vc$ (write $\mu_F\coloneqq \bigotimes_{v\in F}\mu_v$ for $F\subseteq \Vc$). 
     \begin{enumerate}
        \item Insertion/deletion of observations: Suppose 
        \[A\sep{\mathsf{id}}{\Af_{\Do(D)}} B\mid C\cup D.\] Then there exists Markov kernel $\Qr(X_A\miid X_C,X_D)$ unique up to a measurable set $\tilde{N}\coloneqq N\times \Xc_{\Ic\sm D_1}\subseteq \Xc_{C\cup D}\times \Xc_{\Ic\sm D_1}$, such that $\tilde{N}$ is $\Prb_\Mc(X_C \miid X_\Ic,\Do(X_{D_2}))$-null and $\Qr$ is a version of $\Prb_{\Mc}(X_A\mid X_{B_1},X_C \miid X_{\Ic},\Do(X_{D_2}))$ for every $B_1\subseteq B$ simultaneously. If 
        \[
        \mu_{B\cup C}\ll \Prb_\Mc(X_B,X_C \miid X_\Ic, \Do(X_{D_2}))\ll \mu_{B\cup C},
        \] 
        then we have
        \[ 
        \begin{aligned}
            \Prb_\Mc(X_A\mid X_B,X_C \miid X_\Ic,\Do(X_{D_2}))&\meq{\mu_{B\cup C}}\Prb_\Mc(X_A\mid X_C \miid X_{\Ic} ,\Do(X_{D_2})) \\
            &\meq{\mu_{C}}\Prb_{\Mc}(X_A\mid X_C \miid \cancel{X_{\Ic\sm D_1}}, X_{D_1} ,\Do(X_{D_2})).
        \end{aligned}
        \] 

        \item Actions/observations exchange: Suppose \[A\sep{\mathsf{id}}{\Af_{\Do(I_B,D)}} I_B\mid B\cup C\cup D.\] 
        Then there exists a Markov kernel $\Qr(X_A\miid X_B,X_C,X_D)$ unique up to a measurable set $\tilde{N}\coloneqq N\times \Xc_{\Ic\sm D_1}\subseteq \Xc_{B\cup C\cup D}\times \Xc_{\Ic\sm D_1}$, such that $\tilde{N}$ is $\Prb_\Mc(X_{B_1},X_C \miid X_\Ic,\Do(X_{B_2},X_{D_2}))$-null and $\Qr$ is a version of $\Prb_{\Mc}(X_A\mid X_{B_1},X_C \miid X_{\Ic},\Do(X_{B_2},X_{D_2}))$ for every decomposition $B=B_1\dcup B_2$ simultaneously. If 
        \[
        \begin{aligned}
            \mu_{B\cup C}\ll \Prb_\Mc&(X_B,X_C \miid X_\Ic, \Do(X_{D_2}))\ll \mu_{B\cup C} \quad \text{ and }\\
            \mu_C\ll \Prb_\Mc&(X_C \miid X_\Ic, \Do(X_B,X_{D_2}))\ll \mu_C,
        \end{aligned}
        \]
        then we have
        \[
            \begin{aligned}
                \Prb_\Mc(X_A\mid X_C \miid X_\Ic,\Do(X_{B},X_{D_2}))&\meq{\mu_{B\cup C}}\Prb_\Mc(X_A\mid X_B, X_C \miid X_{\Ic},\Do(X_{D_2}))\\
                &\meq{\mu_{B\cup C}}\Prb_\Mc(X_A\mid X_B, X_C \miid \cancel{X_{\Ic\sm D_1}}, X_{D_1},\Do(X_{D_2})).
            \end{aligned}
        \]

        \item Insertion/deletion of actions: Suppose 
        \[A\sep{\mathsf{id}}{\Af_{\Do(I_B,D)}} I_B \mid C\cup D.\] Then there exists a Markov kernel $\Qr(X_A\miid X_C,X_D)$ unique up to a measurable set $\tilde{N}\coloneqq N\times \Xc_{\Ic\sm D_1}\subseteq \Xc_{C\cup D}\times \Xc_{\Ic\sm D_1}$, such that $\tilde{N}\times \Xc_{B_2}$ is $\Prb_\Mc(X_C \miid X_\Ic,\Do(X_{B_2},X_{D_2}))$-null and $\Qr$ is a version of $\Prb_\Mc(X_A\mid X_C \miid X_{\Ic},\Do(X_{B_2},X_{D_2}))$ for every $B_2\subseteq B$ simultaneously. If 
        \[
            \begin{aligned}
                \mu_{C}\ll \Prb_\Mc&(X_C \miid X_{\Ic},\Do(X_B,X_{D_2}))\ll \mu_{C} \quad \text{ and }\\
                \mu_C\ll \Prb_\Mc&(X_C\miid X_{\Ic}, \Do(X_{D_2}))\ll \mu_C,
            \end{aligned}
        \]
        then we have
        \[
            \begin{aligned}
                 \Prb_\Mc(X_A\mid X_C \miid X_\Ic,\Do(X_{B},X_{D_2}))&\meq{\mu_{C}}\Prb_\Mc(X_A\mid X_C \miid X_\Ic,\Do(X_{D_2}))\\
                 &\meq{\mu_{C}}\Prb_\Mc(X_A\mid X_C \miid \cancel{X_{\Ic\sm D_1}}, X_{D_1},\Do(X_{D_2})).
            \end{aligned}    
        \]
    \end{enumerate}
\end{theorem}

\subsection{Some graphical notions}

To be self-contained, we recall some terminology and graphical notions used in the paper. We start with some basic terminology.

Let $\Gf=(\Vc,\Ec)$ be a  mixed graph with edges of the types $\Edges$  and $a,b\in \Vc$. Two nodes are called \emph{adjacent} if there is an edge between them.  We call $a$ a \emph{parent} of $b$ and $b$ a
\emph{child} of $a$ if $a \tuh b$ is in $\Gf$. A \emph{walk} in $\Gf$ is a sequence of nodes $( v_0,\ldots,v_n)$ such that $v_i$ and $v_{i+1}$ are distinct and adjacent for all $0\le i\le n-1$ in $\Gf$. A \emph{path} in $\Gf$ is a walk without repeating nodes. A \emph{directed path} from $v_0$ to $v_n$ in $\Gf$ is a path of the form $v_0 \tuh v_1 \tuh \cdots \tuh v_{n-1} \tuh v_n$ in $\Gf$. Node $a$ is called an \emph{ancestor} of $b$ and $b$ a \emph{descendant}
of $a$ if $a=b$ or there is a directed path from $a$ to $b$. We use $\pa_\Gf(\cdot),\ch_\Gf(\cdot),\anc_\Gf(\cdot),\de_\Gf(\cdot)$ to denote
the set of parents, children, ancestors, and descendants of a node in $\Gf$, respectively. We write $\pa_\Gf(A)=\bigcup_{a\in A}\pa_\Gf(a)$ for a subset $A\subseteq \Vc$, and analogously for the other notions. A
\emph{directed cycle} in $\Gf$ is a walk of the form $a\tuh \cdots \tuh b\tuh a$ in $\Gf$. An \emph{almost directed cycle} in $\Gf$ is a walk of the form $a\tuh \cdots \tuh b\huh a$ in $\Gf$. 
We use star $a\suh b$ to represent the possibilities of $a\tuh b$, or $a\huh b$, or $a\ouh b$. A \emph{collider path} is a path in which all the non-end nodes are colliders, i.e., $v_0\suh v_1\huh \cdots \huh v_{n-1} \hus v_n$. A \emph{bidirected path} is a path with only bidirected edges: $v_0\huh \cdots \huh v_n$. If there is a bidirected path between nodes $a$ and $b$, we say that $a$ and $b$ are in the same \emph{c-component} or \emph{district}. We call $\tilde{\pi}$ a \emph{subpath} of a path $\pi$ if $\tilde{\pi}$ consists of a subsequence of nodes of $\pi$. We denote by $\pi(v_i,v_j)$ a segment of $\pi$ starting from $v_i$ and ending at $v_j$. Given two paths $\pi_1:v_0^1\sus \cdots \sus v_n^1$ and $\pi_2:v_0^2\sus \cdots \sus v_m^2$ such that $v_n^1=v_0^2$, we denote by $\pi_1\oplus \pi_2: v_0^1\sus \cdots v_n^1=v_0^2\sus \cdots \sus v_m^2$ the \emph{concatenated path/walk of $\pi_1$ and $\pi_2$}. Edge of the form $a\suh b$ is called \emph{into} $b$ and edge of the form $a\sut b$ is called \emph{out of} $b$. We call $\tilde{\Gf}$ a \emph{subgraph} of $\Gf$ if $\tilde{\Gf}$ consists of a subset of nodes and edges of $\Gf$, and an \emph{induced subgraph} of $\Gf$ over $A$ with $A\subseteq \Vc$ if $\tilde{\Gf}$ consists of nodes $A$ and all the edges between nodes in $A$ of $\Gf$. We often denote by $\Gf_A$ the induced subgraph of $\Gf$ over $A$.

We now introduce additional graphical notions used in the appendix.

\begin{definition}[Discriminating path] \label{def:discriminating_path} In an iPAG, a path \[\pi: a\sus \cdots \sus y \sus z\] is called a \textbf{discriminating path for $y$} if
\begin{enumerate}
    \item $\pi$ includes at least three edges;
    \item node $a$ is not adjacent to node $z$ and every node between $a$ and $y$ is a collider on $\pi$ and is a parent of $z$.
\end{enumerate}
\end{definition}

See \cref{fig:discriminating_path} for an illustration of a discriminating path for $y$.

\begin{figure}\centering
  \begin{tikzpicture}
    \begin{scope}
      \node[ndout] (z) at (9,1.5) {$z$};
      \node[ndout] (y) at (9,0) {$y$};
      \node[ndout] (vn) at (7.5,0) {$v_n$};
      \node[ndout] (vn1) at (6,0) {$v_{n-1}$};
      \node (dots) at (4.5,0) {$\cdots$};
      \node[ndout] (v1) at (3,0) {$v_1$};
      \node[ndout] (a) at (1.5,0) {$a$};
      \draw[suh] (y) -- (vn);
      \draw[huh] (vn) -- (vn1);
      \draw[huh] (vn1) -- (dots);
      \draw[huh] (dots) -- (v1);
      \draw[tuh] (vn1) edge (z);
      \draw[tuh] (vn) edge (z);
      \draw[tuh] (v1) edge (z);
      \draw[suh]  (a) -- (v1);
      \draw[sus] (y) -- (z);
    \end{scope}
  \end{tikzpicture}
  \caption{Discriminating path for $y$. Only $y$ and $a$ can be input nodes; all other nodes have to be output nodes.}
  \label{fig:discriminating_path}
\end{figure}
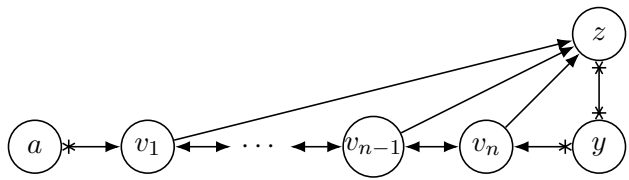

In the next definition, let $\Zb_{\ge 0}=\{0,1,2,\ldots\}$ denote the set of non-negative integers.

\begin{definition}[Proper/irreducible/tight path]\label{def:irreduciblepath}

Let $\PAG$ be an iPAG and $A,C,D\subseteq (\Ic\cup \Vc)\sm T$ and $B\subseteq \Ic\cup \Vc\cup \{I_d\}_{d\in D}$ with $T\subseteq \Ic\cup \Vc$. Let $\pi:v_0\sus v_1\sus\cdots \sus v_n$ be a path from $A$ to $B$ in $\Pf_{\Do(I_D,T)}$.

\begin{enumerate}
    \item  We call $\pi$ a \textbf{proper} `xyz' path in $\Pf_{\Do(I_D,T)}$ if $v_i\notin A\cup B$ for all $1<i<n$ (if any). Here, for example, `xyz' may be `open' or `potentially anterior'.

    \item Assume that $\pi$ is open given $C\cup T$. We call $\pi$ a \textbf{reducible} open path in $\Pf_{\Do(I_D,T)}$ if there exists an open path \[\tilde{\pi}:v_{i_0}\sus v_{i_1}\sus \cdots \sus v_{i_m}\] from $A$ to $B$ given $C\cup T$ in $\Pf_{\Do(I_D,T)}$, for some $0\leq m < n$ with $(v_{i_j})_{j=0}^m$ being a subsequence of $(v_i)_{i=0}^n$. We call an open path $\pi$ from $A$ to $B$ given $C\cup T$ in $\Pf_{\Do(I_D,T)}$ \textbf{irreducible} if it is not reducible. 

    \item Assume that $\pi$ is open given $C\cup T$. Define for $u\in \Ic\cup \Vc$
    \[
            \dist(u,C)\coloneqq \min\{\ell: \exists c\in C \text{ and a directed path } u=u_0\tuh \cdots \tuh u_\ell= c \}\in \Zb_{\ge 0} \cup \{\infty\},
    \]  
    with the convention $\dist(u,C)=\infty$ if no such directed paths exist. Define 
    \[
        \dist(\pi,C)\coloneqq \sum_{u\in \mathsf{Col}(\pi)}\dist(u,C),
    \]
    where $\mathsf{Col}(\pi)$ denotes the set of colliders on $\pi$. Define \[\|\pi\|\coloneqq \#\{\text{edges on } \pi\}=n\] to be the length of $\pi$. We say that $\pi$ is \textbf{$C$-slack} in $\Pf_{\Do(I_D,T)}$ if there exists an open path $\tilde{\pi}:u_{0}\sus u_{1}\sus \cdots \sus u_{m}$ from $A$ to $B$ given $C\cup T$ in $\Pf_{\Do(I_D,T)}$, such that \[(\|\tilde{\pi}\|,\dist(\tilde{\pi},C))<_{\mathsf{lex}}(\|\pi\|,\dist(\pi,C))\] where $<_{\mathsf{lex}}$ denotes the lexicographic order on $\mathbb N\times(\mathbb N\cup\{\infty\})$. If $\pi$ is not $C$-slack, we call it \textbf{$C$-tight}.  
\end{enumerate}
    
\end{definition}

If there exists an open path from $A$ to $B$ given $C\cup T$ in $\Pf_{\Do(I_D,T)}$, then there must exist an irreducible one and a $C$-tight one. A $C$-tight path is irreducible, and every irreducible open path is a proper open path.

\begin{definition}[C-forest]\label{def:cforest}
    Let $\Gf=(\Vc,\Ec)$ be an ADMG where $R$ is the root set, i.e., $R\subseteq \Vc$ is a set such that $\anc_{\Gf}(R)=\Vc$. Then graph $\Gf$ is called an $R$-rooted C-forest if 
    \begin{enumerate}
        \item all nodes in $\Gf$ lie in a single c-component, and 

        \item node $v$ has at most one child for all $v\in \Vc$. 
     \end{enumerate}
\end{definition}

\begin{definition}[Hedge]\label{def:hedge}
    Let $\Gf=(\Vc,\Ec)$ be an ADMG and $A,B\subseteq \Vc$ be disjoint. Let $\Hc,\Hc^\prime\subseteq \Vc$ be such that 
    \begin{enumerate}
        \item there are two subgraphs $\Gf^\Hc$ and $\Gf^{\Hc^\prime}$ of $\Gf$ over $\Hc$ and $\Hc^\prime$ (not necessarily induced subgraphs), respectively, such that they form two $R$-rooted C-forests with $R\subseteq \anc_{\Gf_{\Do(B)}}(A)$;
        \item $\Hc\cap B\ne \emptyset$, $\Hc^\prime \cap B=\emptyset$, and $\Gf^{\Hc^\prime}$ is a subgraph of $\Gf^{\Hc}$. 
    \end{enumerate}
    We call $(\Gf^\Hc,\Gf^{\Hc^\prime})$ a \textbf{hedge for $(A,B)$ in $\Gf$}. We also sometimes abuse the notation and call $(\Hc,\Hc^\prime)$ a hedge for $(A,B)$ in $\Gf$.
\end{definition}

\subsection{FCI with input nodes}
\begin{algorithm}
\begin{algorithmic}[1]
  \Input Input node set $\Ic$; output node set $\Vc$; independence model $M$ over $\Vc \mid \Ic$
  \Output mixed graph $\Gf$ with input nodes $\Ic$ and output nodes $\Vc$
  \State $\langle \Gf,\sepset \rangle \leftarrow \mathtt{FCIskeleton}(\Ic,\Vc,M)$\label{alg:fci:skeleton}
   \ForAll{edge $j \ouo v$ in $\Gf$ with $j \in \Ic$, $v \in \Vc$}\label{alg:fci:orientJ}\label{alg:fci:start_orient}
    \State orient $j \tuo v$ \Comment{input nodes cannot receive arrowheads}
  \EndFor

  \Repeat \label{alg:fci:R0}
    \begin{description}
      \item[$\Rc0$] if $i \sus j \sus k$ in $\Gf$ with $i \notin \Ic$, and $i$ and $k$ are not adjacent in $\Gf$, then orient $i \suh j \hus k$ if $j \notin \sepset(\{i,k\})$ 
    \end{description}
  \Until{this orientation rule is not applicable}
  \Repeat
    \begin{description}
      \item[$\Rc1$] if $i \suo j \hus k$ in $\Gf$ with $i \notin \Ic$, and $i$ and $k$ are not adjacent in $\Gf$, then orient $i \hut j$
      \item[$\Rc2$] if $i \tuh j \suh k$ or $i \suh j \tuh k$ in $\Gf$, and $i \suo k$ in $\Gf$, then orient $i \suh k$
      \item[$\Rc3$] if $i \suh j \hus k$ and $i \suo l \ous k$ and $l \suo j$ in $\Gf$ with $i \notin \Ic$, and $i$ and $k$ are not adjacent in $\Gf$, then orient $l \suh j$
      \item[$\Rc4$] if $\langle i, j, q_1, \dots, q_n, k \rangle$ is a discriminating path in $\Gf$ for $j$, and if $i \suo j$ in $\Gf$, then orient $i \hut j$ if $j \in \sepset(\{i,k\})$ and orient $i \huh j \huh q_1$ if $j \notin \sepset(\{i,k\})$ 
    \end{description}
  \Until{none of these orientation rules is applicable}
  \Repeat
    \begin{description}
      \item[$\Rc5$] if $i \ouo j$ in $\Gf$, and there is an uncovered circle path $i \ouo k \ouo \cdots \ouo l \ouo j$ in $\Gf$ such that $i$ is not adjacent to $l$ and $j$ is not adjacent to $k$, then orient $i \tut k \tut \cdots \tut l \tut j \tut i$
    \end{description}
  \Until{this orientation rule is not applicable}
  \Repeat
    \begin{description}
      \item[$\Rc6$] if $i \tut j \ous k$ in $\Gf$, then orient $j \tus k$
      \item[$\Rc7$] if $i \suo j \out k$ in $\Gf$ with $i \notin \Ic$, and $i$ and $k$ are not adjacent in $\Gf$, then orient $i \sut j$
    \end{description}
  \Until{none of these orientation rules is applicable}
  \Repeat
    \begin{description}
      \item[$\Rc8$] if $i \tuh j \tuh k$ or $i \tuo j \tuh k$ in $\Gf$, and $i \ouh k$ in $\Gf$, then orient $i \tuh k$
      \item[$\Rc9$] if $i \ouh k$, and $\pi = \langle i, j, \dots, k \rangle$ is an uncovered possibly directed path in $\Gf$ from $i$ to $k$ such that $j$ and $k$ are not adjacent in $\Gf$, then orient $i \tuh k$
      \item[$\Rc10$] if $i \ouh k$ in $\Gf$, $j \tuh k \hut l$ in $\Gf$, $\pi_1$ is a uncovered possibly directed path in $\Gf$ from $i$ to $j$, and $\pi_2$ is a uncovered possibly directed path in $\Gf$ from $i$ to $l$, then let $u_1$ be the node adjacent to $i$ on $\pi_1$ (possibly $u_1 = j$) and $u_2$ the node adjacent to $i$ on $\pi_2$ (possibly $u_2 = l$); if $u_1 \ne u_2$, and $u_1$ and $u_2$ are not adjacent in $\Gf$, then orient $i \tuh k$
    \end{description}
  \Until{none of these orientation rules is applicable}\label{alg:fci:end_orient}
  \end{algorithmic}
  \caption{Extended FCI Algorithm}\label{alg:FCI}
\end{algorithm}

\begin{proposition}[FCI is sound {\protect\cite[Theorem~12.8.1]{forre2025mathematical}}]\label{prop:FCI_sound}
   The Extended FCI algorithm (\cref{alg:FCI}) is \emph{sound}: if its input consists of the independence model $\IM(\Af \mid \Sc)$ of an isADMG $\Af$ given $\Sc$, then its output is a valid iPAG $\Pf$ representing $\Af$ given $\Sc$.
\end{proposition}

See \cite[Section~12.7]{forre2025mathematical} for the details of the skeleton search function $\mathtt{FCIskeleton}$.

\section{Additional discussion}\label{sec:additional_discussion}

We present additional discussion on positive and continuous Markov kernels, criteria for causal relationships in MAGs and PAGs, causal identification from undirected graphs, and a causal-identification algorithm for conditional causal effects in iSOPAGs under selection bias (sCIDP).

\subsection{Continuous and positive Markov kernels}\label{sec:pcmk}

The appeal of the class of positive and continuous Markov kernels is twofold:
(i) it is closed under marginalization, product, and composition of Markov kernels;
(ii) it yields a canonical conditioning operation provided that the conditional kernel can be taken to be continuous.

\begin{lemma}[Properties of positive and continuous Markov kernels {\protect\cite{chen25notes,Richard01causal_inference_continuous}}]\label{lem:pcmk}
   Let 
   \[
   \begin{aligned}
       &\Kr(X,Y\miid T):\,\Tc \dto  \Xc \times \Yc,\quad   \Kr_1(Z\miid U,X,T):\, \Uc \times \Xc\times \Tc \dto \Zc, \\
       \text{ and } &\Kr_2(X,Y\miid T,W):\, \Tc \times \Wc \dto \Xc \times \Yc
   \end{aligned}  
   \] 
    be positive and continuous Markov kernels. Then we have:
    \begin{enumerate}
            \item The marginalized Markov kernels $\Kr(X \miid T)$ and $\Kr(Y \miid T)$ are positive and continuous.
            \item The product Markov kernel 
            $
            \Kr_1(Z\miid U,X,T) \otimes \Kr_2(X,Y\miid T,W)
            $
            is positive and continuous.
            \item Suppose that the conditional Markov kernel $\Kr(X \mid Y\miid T)$ of $\Kr(X,Y\miid T)$ given $Y$ can be chosen to be continuous. Then it is pointwise unique among continuous versions of the conditional kernel, and moreover for all $y\in \Yc,t\in \Tc$ 
            \[
            \Kr(X \mid Y=y\miid T=t)=\lim_{\delta \downarrow 0}\Kr(X\mid Y\in B(y,\delta)\miid T=t),
            \]
            where $B(y,\delta)$ denotes a ball centered at $y$ with radius $\delta$ and the limit is taken in $\Pc(\Xc)$ equipped with the weak topology. Note that $\Kr(X\mid Y\in B(y,\delta)\miid T=t)$ is well-defined by positivity of $\Kr(X,Y\miid T)$.
    \end{enumerate}
\end{lemma}

\begin{remark}[Sufficient conditions for positive and continuous Markov kernels]
Let $\Kr(X\miid Y)$ be a Markov kernel from a Polish space $\Yc$ to a Polish space $\Xc$, and suppose it admits a $\mu$-a.s.\ positive density $k(\cdot\miid\cdot)$ \wrt a $\sigma$-finite reference measure $\mu$ that is strictly positive on non-empty open subsets of $\Xc$. If for $\mu$-a.e.\ $x\in\Xc$, the map $y\mapsto k(x\miid y)$ is continuous, and there exists an integrable function $g\in L^1(\mu)$ such that $k(x\miid y)\le g(x)$ for all $x\in\Xc$ and $y\in\Yc$, then $\Kr(X\miid Y)$ is positive and continuous. If there exists $L\in L^1(\mu)$ such that
$
|k(x\miid y_1)-k(x\miid y_2)|\le L(x)\,d_\Yc(y_1,y_2)
$
for all $y_1,y_2$ in a neighborhood of each $y$ and for $\mu$-a.e.\ $x\in\Xc$, then $\Kr(X\miid Y)$ is positive and continuous.
\end{remark}

\begin{remark}[Continuous version of conditioning via density]
Let $\Kr(X,Y\miid Z)$ be a Markov kernel from a Polish space $\Zc$ to a Polish space $\Xc\times \Yc$ that admits a strictly positive jointly continuous density $k(\cdot,\cdot\miid\cdot)$ \wrt a $\sigma$-finite reference measure $\mu_\Xc\otimes \mu_\Yc$ on $\Xc\times \Yc$, where $\mu_\Xc$ and $\mu_\Yc$ are strictly positive on nonempty open subsets of $\Xc$ and $\Yc$, respectively. Assume that there exists $g\in L^1(\mu_\Xc)$ such that
\[
k(x,y\miid z)\le g(x)
\qquad
\text{for all }(x,y,z)\in \Xc\times \Yc\times \Zc.
\]
Then
\[
k(y\miid z)\coloneqq \int_\Xc k(x,y\miid z)\,\mu_\Xc(\dr x)
\]
is finite, continuous, and strictly positive. Hence
\[
k(x\mid y\miid z)\coloneqq \frac{k(x,y\miid z)}{k(y\miid z)}
\]
is well-defined and continuous. If moreover $k(x\mid y\miid z)$ is dominated by an integrable function of $x$, then
\[
k(x\mid y\miid z)\,\mu_\Xc(\dr x)
\]
induces a positive and continuous Markov kernel from $\Yc\times \Zc$ to $\Xc$.
\end{remark}

\subsection{Criteria for causal relationships in MAGs and PAGs}\label{sec:criterion_relation}

For an ADMG $\Af=(\Vc,\Ec)$ and $a,b\in \Vc$, we say that, according to $\Af$: 
\begin{enumerate}
    \item [(i)] variable $X_a$ does not have a direct causal effect on $X_b$ if there is no directed edge $a\tuh b$ in $\Af$, which is equivalent to $b\sep{\mathsf{id}}{\Af_{\Do(I_a,\Vc\sm \{a,b\})}} I_a \mid \Vc\sm \{a,b\}$;
    
    \item [(ii)] variable $X_a$ does not have a causal effect on $X_b$ if there are no directed paths from $a$ to $b$ in $\Af$, which is equivalent to $b\sep{\mathsf{id}}{\Af_{\Do(I_a)}} I_a$;

     \item [(iii)] there is no confounding between $X_a$ and $X_b$ if there is no bidirected edge between $a$ and $b$ in $\Af_{\sm \{a,b\}^c}$, which is equivalent to $b\sep{\mathsf{id}}{\Af_{\Do(I_a)}} I_a \mid a$.
\end{enumerate}

By \cref{thm:sep_hsint_mag,thm:sep_hsint_pag_mag,lem:anc_selec_I,lem:anc_selec_II}, we have the following result.

\begin{corollary}[Criterion for causal relationships in MAGs and PAGs]\label{cor:causal_relation}
Under the setting of \cref{thm:causal_calculus_mag_pag},  let $\{a,b\}\subseteq \Vc$. 
\begin{enumerate}
    \item \textbf{Direct causal effect}: If
    \[
        b\sep{\mathsf{id}}{\Gf_{\Do(I_a,\Vc\sm \{a,b\})}} I_a \mid \Vc\sm \{a,b\},
    \]
    then variable $X_a$ does not have a direct causal effect on $X_b$ according to $\Af$ for every $\Af\in [\Gf]_\Gs$. Otherwise, there exists $\Af \in [\Gf]_{\Gs}$ such that $X_a$ does have a direct causal effect on $X_b$ according to $\Af$.

    \item \textbf{Causal effect}: If
    \[
        b\sep{\mathsf{id}}{\Gf_{\Do(I_a)}} I_a,
    \]
    then variable $X_a$ does not have a causal effect on $X_b$ according to $\Af$ for every $\Af\in [\Gf]_\Gs$. Otherwise, there exists $\Af \in [\Gf]_{\Gs}$ such that $X_a$ does have a causal effect on $X_b$ according to $\Af$.

    \item \textbf{confounding}: If 
    \[
        b\sep{\mathsf{id}}{\Gf_{\Do(I_a)}} I_a \mid a,
    \]
    then there is no confounding between variables $X_a$ and $X_b$ according to $\Af$ for every $\Af\in [\Gf]_\Gs$. Otherwise, there exists $\Af \in [\Gf]_{\Gs}$ such that there is confounding between variables $X_a$ and $X_b$ according to $\Af$.

    \item \textbf{Ancestors of selection nodes}: If there is an arrowhead towards node $a$ in $\Gf$ then variable $X_a$ is not an ancestor of a selection variable according to $\Af$ for every $\Af\in [\Gf]_\Gs$. Otherwise, there exists $\Af \in [\Gf]_{\Gs}$ such that $X_a$ is an ancestor of a selection variable.
\end{enumerate}
\end{corollary}

We give one simple application of \cref{cor:causal_relation}(4). In general, it is impossible to study the s-recoverability problem under ancestral graphs in general when selection variables are conditioned upon and therefore we can gain no information about conditional independence of the type $A\sep{}{} \Sc\mid B$. However, using \cref{cor:causal_relation}(4) and the connection between s-recoverability and s-ID derived in \cite[Theorem~6.1]{abouei24sIDlatent} we can show the following result under certain assumptions on the structure of the selection mechanisms:
\begin{corollary}[s-recoverability]\label{cor:s-recover}
    Let $\Gf=(\Vc,\Ec)$ be a MAG or SOPAG. Let $[\Gf]_\Gs^s$ denote the set of sADMGs represented by $\Gf$ such that $\ch_{\Af}(\Sc_\Af)\sm \Sc_{\Af}=\emptyset$ and for every $s\in \Sc_\Af$ it holds $\pa_{\Af}(s)\cap \Vc\ne\emptyset$. Assume $(\Mc,X_{\Sc}=\1)\in \Mb^+_d(\Gf)$ with $\Gb(\Mc,X_{\Sc}=\1)\in [\Gf]_{\Gs}^s$ is the true underlying causal model. Let $A,B\subseteq \Vc$ be disjoint. If $\Prb_{\Mc}(X_A \mid X_\Sc=\1 \miid \Do(X_B))$ is identifiable from $\Prb_\Mc(X_\Vc\mid X_\Sc=\1)$ and $A\sep{\mathsf{id}}{\Gf_{\Do(B)}} D\mid B$ where $D=\{v\in \Vc\mid \nexists u\suh v \text{ in }\Gf\}$, then $\Prb_{\Mc}(X_A \miid \Do(X_B))$ can be identified from $\Prb_\Mc(X_\Vc\mid X_\Sc=\1)$.
\end{corollary}

\begin{proof}
    From \cite[Theorem~6.1]{abouei24sIDlatent}, it suffices to show $A\sep{\mathsf{id}}{\Af_{\Do(B)}} \Sc \mid B$ where $\Af=\Gb(\Mc,X_{\Sc}=\1)$. Assume, for contradiction, that this is not true. Let 
    \[
    \pi: A\ni v_0\sus \cdots \sus v_n\in \Sc
    \]
    be an irreducible open path given $B$ in $\Af_{\Do(B)}$. Note that $\pi$ cannot contain any colliders. Since $\ch_{\Af}(\Sc)\sm \Sc=\emptyset$, path $\pi$ must be of the form
    \[
            v_0\sus v_1\sus \cdots \sus v_{n-1}\suh v_n.
    \]
    Then by the assumption that $\pa_{\Af}(s)\cap \Vc\ne\emptyset$ for all $s\in \Sc$ we have open path from $A$ to $d\in \Vc$ given $\Sc\cup B$ in $\Af_{\Do(B)}$:
    \[
        v_0\sus v_1\sus \cdots \sus v_{n-1}\suh v_n\hut d.
    \]
    \cref{cor:causal_relation}(4) implies that $d\in D$. Therefore, by
    \cref{thm:sep_hsint_mag,thm:sep_hsint_pag_mag}, we have $A\nsep{\mathsf{id}}{\Gf_{\Do(B)}} D\mid B$. This leads to a contradiction to our assumption and therefore we finish the proof.
    
\end{proof}

\subsection{Causal identification from undirected graphs}\label{sec:iden_undirected_graph}

At first glance, causal analysis under undirected graphs may seem impossible, since undirected edges carry no directionality. However, by interpreting an undirected graph as a special case of a MAG, one can still study causal identification in this setting. We illustrate this with the following example.

\begin{example}[Identification from an undirected graph]\label{ex:identification_undirected_graph}
Consider a MAG $\Mf$ shown in \cref{fig:noniden_mag}. Assume $(\Mc,X_\Sc=\mathrm{1}_{|\Sc|})\in \Mb^+_d(\Mf)$ is the true underlying s-SCM. \cref{lem:noniden} tells us that $\Prb_{\Mc}(X_a\mid X_{\Sc}=\1\miid \Do(X_b))$ is non-identifiable in $\Mb_d^+(\Mf)$. We can also see this by constructing an isADMG $\Af\in [\Mf]_\Gs$ shown in \cref{fig:noniden_mag}. Indeed, 
\[
\Hc\coloneqq \{b,c_2\} \quad  \text{ and } \quad  \Hc^\prime\coloneqq \{c_2\}
\]
form a hedge for $(\{a\}\cup \Sc_\Af,\{b\})$, and the set $D\subseteq \Sc_\Af$ in \cref{prop:hedge} is empty, since \[
a\nsep{\mathsf{id}}{\Af_{\Do(I_D,b)}} I_D\mid \{b\}\cup \Sc_\Af
\]
for all $D\subseteq \Sc_\Af$. One can also check that $\idp(\{a\},\{b\},\Mf)$ outputs $\textsc{Fail}$, witnessed by 
\[
C=\Dc=\{a\}\cup \{c_1,c_2\} \quad  \text{ and } \quad T=\Vc=\{a,b\}\cup \{c_1,c_2\}.
\]
In contrast, since
\[
    a \sep{\mathsf{id}}{\Mf_{\Do(I_b)}} I_b \mid \{b,c_1,c_2\} \quad \text{ and } \quad a \sep{\mathsf{id}}{\Mf_{\Do(I_b)}} I_b \mid \{c_1,c_2\},
\]
we can conclude from the causal calculus (\cref{thm:causal_calculus_mag_pag})
\[
    \begin{aligned}
        \Prb_\Mc(X_a\mid X_{\{c_1,c_2\}}, X_\Sc=\mathrm{1}_{|\Sc|} \miid \Do(X_b))&=\Prb_\Mc(X_a\mid X_b,X_{\{c_1,c_2\}}, X_\Sc=\mathrm{1}_{|\Sc|})\\
        &=\Prb_\Mc(X_a\mid X_{\{c_1,c_2\}},X_\Sc=\mathrm{1}_{|\Sc|}).
    \end{aligned}
\]

\begin{figure}[ht]
\centering
\begin{tikzpicture}[scale=0.8, transform shape]
\begin{scope}[xshift=0]
    \node[ndout] (b) at (0,3) {$b$};
    \node[ndout] (c1) at (-1,1.5) {$c_1$};
    \node[ndout] (c2) at (1,1.5) {$c_2$};
    \node[ndout] (a) at (0,0) {$a$};
    \draw[tut] (b) to (c1);
    \draw[tut] (b) to (c2);
    \draw[tut] (c1) to (a);
    \draw[tut] (c2) to (a);
    \node at (0,-1) {$\Mf$};
\end{scope}

\begin{scope}[xshift=5cm]
    \node[ndout] (b) at (0,3) {$b$};
    \node[ndout] (c1) at (-1,1.5) {$c_1$};
    \node[ndout] (c2) at (1,1.5) {$c_2$};
    \node[ndout] (a) at (0,0) {$a$};
    \node[ndsel] (s1) at (1.25,2.75) {$s_1$};
    \node[ndsel] (s2) at (1.25,0.25) {$s_2$};
    \node[ndsel] (s3) at (-1.25,2.75) {$s_3$};
    \node[ndsel] (s4) at (-1.25,0.25) {$s_4$};
    \draw[arout] (b) to (c2);
    \draw[arout] (b) to (s1);
    \draw[arout] (c2) to (s1);
    \draw[arout] (c2) to (s2);
    \draw[arout] (a) to (s2);
    \draw[arout] (b) to (s3);
    \draw[arout] (c1) to (s3);
    \draw[arout] (c1) to (s4);
    \draw[arout] (a) to (s4);
    \draw[arout] (c2) to (a);
    \draw[arlat, bend right] (b) to (c2);
    \node at (0,-1) {$\Af$};
\end{scope}
\end{tikzpicture}
\caption{A MAG $\Mf$ with only undirected edges and an isADMG $\Af\in [\Mf]_\Gs$ in \cref{ex:identification_undirected_graph}.}
\label{fig:noniden_mag}
\end{figure}
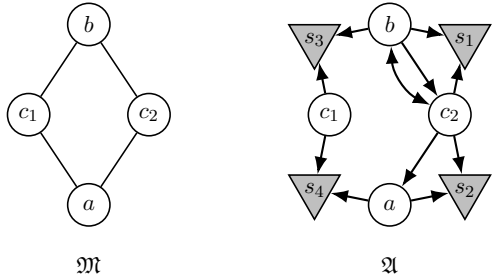
\end{example}

\begin{remark}[Causal interpretation of concentration graphs and $\mathrm{MTP}_2$ distributions]\label{rem:causal_concentration_MTP2}
    Concentration graphs (aka.\ undirected graphs) are not uncommon in the literature on graphical models. A distribution $\Prb(X)$ on $\Xc\subseteq \Rb^d$ is called \emph{multivariate totally positive of order 2 ($\mathrm{MTP}_2$)} if its density $f$ (\wrt some reference measure) satisfies\footnote{For $x=(x_1,\ldots,x_d)$ and $y=(y_1,\ldots,y_d)$, we define $x\wedge y \coloneqq (\min(x_1,y_1),\ldots,\min(x_d,y_d))$ and $x\vee y \coloneqq (\max(x_1,y_1),\ldots,\max(x_d,y_d))$.} 
    \[
    f(x)f(y)\le f(x\wedge y)f(x\vee y)
    \]
    for all $x,y\in \Xc$, which forms a rich class with strong structural and inferential properties; see \cite{bartolucci2000_mtp2_lrt,fallat2017_mtp2_markov_structures, lauritzen2019_mle_gaussian_mtp2,lauritzen2021_mtp2_expfam_binary} and references therein for nice properties and applications of $\mathrm{MTP}_2$ distributions. In particular, if an $\mathrm{MTP}_2$ distribution has a continuous density and coordinatewise connected support, then its conditional independence model is faithfully represented by its concentration graph \cite[p.\ 1167 and Theorem~6.1]{fallat2017_mtp2_markov_structures}. Observe that, in general, a distribution faithful to a concentration graph need not admit an ADMG representation (consider, e.g., the undirected graph shown in \cref{fig:noniden_mag}). Therefore, one might conclude that, in general, graphical causal analysis is impossible for $\mathrm{MTP}_2$ distributions. In contrast, as \cref{ex:identification_undirected_graph} illustrates, if we interpret the undirected graph as a MAG,\footnote{For example, in \cite[Example~4.1 and 4.2]{fallat2017_mtp2_markov_structures}, we see no strong evidence against such an interpretation and against the presence of latent selection bias.} then certain causal conclusions can be derived by sound and complete rules (\cref{thm:causal_calculus_mag_pag,thm:causal_calculus_atomic_complete}). Further empirical validation of this interpretation needs to be explored.
\end{remark}

\subsection{sCIDP: ID algorithm for conditional causal effects from PAGs under selection bias}\label{sec:scidp}

We present the sCIDP algorithm in this subsection.

\begin{algorithm}
\caption{$\operatorname{sCIDP}(\Pf; A, B,C)$}
\label{alg:sCIDP}
\begin{algorithmic}[1]
\State \textbf{Input:} iSOPAG $\Pf$ and three disjoint sets $A, B,C \subset \Vc$ with $A\ne \emptyset$
\State \textbf{Output:} a proxy Markov kernel $\widehat{\Kr}$ for $\Prb_\Mc(X_A\mid X_C, X_S=\mathbf{1}_{|\Sc|} \miid \Do(X_B),X_{\Ic})$, or $\textsc{Fail}$

\State Let $D \coloneqq \pant_{\Pf_{\Vc \setminus B}}(A\cup C)$
\State Let $\Bb_1,\ldots \Bb_n$ be buckets in $\Pf$

\State $\breve{B} \gets B$; $\breve{C}\gets C$

\While{$\exists \Bb_i$ such that $\Bb_i\cap D\neq \emptyset$, $\Bb_i\nsubseteq D$, and $\Bb_i\cap \breve B\neq \emptyset$}
    \State $\tilde{B} \gets \Bb_i\cap \breve{B}$
    \If{$A\sep{\mathsf{id}}{\Pf_{\Do(I_{\tilde{B}},\,\breve{B}\setminus \tilde{B})}} I_{\tilde{B}} \mid \breve{B}\cup \breve{C}$}
        \State $\breve{B} \gets \breve{B}\setminus \tilde{B}$; $\breve{C}\gets \breve{C}\cup \tilde{B}$
        \State $D\gets \pant_{\Pf_{\Vc \setminus \breve{B}}}(A\cup \breve{C})$
    \Else
        \State \textbf{return} $\textsc{Fail}$
    \EndIf
\EndWhile

\While{$\exists\, i \text{ such that } C_i\coloneqq \breve C\cap \Bb_i \neq \emptyset
\text{ and } A\sep{\mathsf{id}}{\Pf_{\Do(I_{C_i},\,\breve B)}} I_{C_i} \mid \breve B\cup \breve C$}
    \State $\breve{B} \gets \breve{B}\cup C_i$; $\breve{C}\gets \breve{C}\setminus C_i$
\EndWhile
            \State $\hat{\Kr}\gets \textbf{sIDP}(\Pf,A\cup \breve{C},\breve{B})^{|\breve{C}}$
          \State \textbf{return} $\hat{\Kr}$ 
\end{algorithmic}
\end{algorithm}

\section{Auxiliary results}\label{sec:pf_auxiliary}

We present and prove auxiliary results used in the proofs of the results in \cref{sec:causal_mag_pag,sec:IDalg}.

\begin{lemma}[Property of graph representation]\label{lem:inducing_path}
  Let $\PAG$ be an iPAG that represents isADMG $\Af=(\Ic,\Oc,\Sc,\tilde{\Ec})$. Then for any two nodes $a,b$ in $\Pf$:
  \begin{enumerate}
    \item  $a\in \anc_{\Pf}(b)$ implies that $a\in \anc_{\Af}(\{b\}\cup \Sc)$;

    \item  if $a\suh b$ is in $\Pf$, then there exists an $\Sc$-inducing path from $a$ to $b$ in $\Af$ that is into $b$;

    \item if $a\huh b$ is in $\Pf$, then there exists an $\Sc$-inducing path from $a$ to $b$ in $\Af$ that goes both into $a$ and into $b$.
  \end{enumerate}
\end{lemma}

\begin{proof}[Proof of \cref{lem:inducing_path}]
    See \cite[Lemma~12.3.6]{forre2025mathematical}.
\end{proof}

The following result shows that, for our purposes, inducing paths and inducing walks are equivalent. Some arguments are more naturally phrased in terms of paths, whereas others are cleaner in terms of walks; we will therefore use the two notions interchangeably.

\begin{lemma}
    Let $\isADMG$ be an isADMG and $a\in \Oc$ and $b\in \Ic\cup \Oc$ be distinct nodes. Then the following are equivalent:
\begin{enumerate}
    \item there is an $\Sc$-inducing path from $a$ to $b$ in $\Af$;

    \item there is an $\Sc$-inducing walk from $a$ to $b$ in $\Af$;

    \item $a\nsep{\mathsf{id}}{\Af} b\mid \Sc\cup Z$ for all $Z\subseteq (\Ic\cup \Oc)\sm \{a,b\}$;

    \item $a\nsep{\mathsf{id}}{\Af} b\mid \Sc\cup Z$ for $Z=\left(\Ic\cup \anc_{\Af}(\{a,b\}\cup \Sc)\right)\sm \{a,b\}$;
\end{enumerate}
\end{lemma}

\begin{proof}
    The proof is similar to that of \cite[Theorem~4.2]{richardson2002ancestral}.
\end{proof}

The following two lemmas show that the visibility of a directed edge $a\tuh b$ in an iMAG exactly characterizes whether some isADMG represented by that iMAG can contain the bidirected edge $a\huh b$. This generalizes the result in \cite{zhang2008causal} to the case where we have exogenous input nodes and latent selection nodes.

\begin{lemma}[Visible edge I]\label{lem:visible_edge_I}
  Let $\isADMG$ be an isADMG and $\Mf=(\Ic,\Vc,\tilde{\Ec})$ be an iMAG that represents $\Af$. Let $a\in \Ic\cup \Vc$ and $b\in \Vc$ be such that $a\tuh b$ is in $\Mf$. If $a\tuh b$ is visible, then there exists \emph{no} $\Sc$-inducing walk from $a$ to $b$ in $\Af$ that is
  into $a$. In particular, the bidirected edge $a\huh b$ is not in $\Af$.
\end{lemma}

\begin{proof}[Proof of \cref{lem:visible_edge_I}]
  If $a\in \Ic$, then there are no edges into $a$ in $\Af$, and therefore there exists no $\Sc$-inducing walk from $a$ to $b$ in $\Af$ that is into $a$. Now assume $a\in \Vc$. We argue by contradiction. Assume the contrary that there exists an $\Sc$-inducing walk $\pi$ from $a$ to $b$ that is into $a$. Let $c\in \Vc\sm \{a,b\}$ be a node that witnesses the visibility of $a\tuh b$. Then we have two cases:
  \begin{itemize}
    \item[(i)] If the edge $c\suh a$ is in $\Mf$, then there is an $\Sc$-inducing walk from $c$ to $a$ that is into $a$ in $\Af$ by Lemma~\ref{lem:inducing_path}. Concatenating this with the $\Sc$-inducing walk from $a$ to $b$ that is into $a$ (exists by assumption) gives an $\Sc$-inducing walk from $c$ to $b$ in $\Af$ as $a$ is a collider in this walk and $a\in \anc_{\Af}(\{b\}\cup \Sc)$, because $a\tuh b$ is in $\Mf$ and all the other colliders are in $\anc_{\Af}(\{a,b\}\cup \Sc)$.

    \item[(ii)] Assume that there is a (definite collider) path $\pf: c\suh v_1\huh \cdots \huh v_{n-1} \huh a$ in $\Mf$ for some $n\geq 2$ and $v_i\in \pa_{\Mf}(b)$ for all $i\in \{1,\ldots,n-1\}$. For each pair of subsequent nodes $(v_i,v_{i+1})$ on $\pf$ with $i\in\{1,\ldots,n-1\}$, defining $v_{n}\coloneqq a$, there exists an $\Sc$-inducing walk $\pi_i$ from $v_i$ to $v_{i+1}$ that is into $v_{i+1}$ and into $v_i$ by Lemma~\ref{lem:inducing_path}. \cref{lem:inducing_path} also implies that there is an inducing walk $\pi_0$ from $c$ to $v_1$ that is into $v_1$. Since $v_i\in \anc_{\Af}(\{b\}\cup \Sc)$ for $i\in \{0,\ldots,n-1\}$ (since $v_i\in \pa_{\Mf}(b)$) and the walk $\pi$ is into $a$, we can concatenate these walks $\pi_i$ with $\pi$ to get an $\Sc$-inducing walk from $c$ to $b$.
  \end{itemize}

  Hence, in both cases above, $c$ and $b$ are adjacent in $\Mf$. Since this is true for all such $c$, we can conclude that the directed edge $a\tuh b$ in $\Mf$ must be invisible. This yields a contradiction, so there are no $\Sc$-inducing walks from $a$ to $b$ that are into $a$.
\end{proof}

\begin{lemma}[Visible edge II]\label{lem:visible_edge_II}
  Let $\Af=(\Ic,\Oc,\Sc,\Ec)$ be an isADMG and $\Mf$ the iMAG that represents $\Af$. Let $a\tuh b$ be a directed edge in $\Mf$ with $a,b\in \Oc$. If $a\tuh b$ is invisible, then there exists an isADMG $\tilde{\Af}=(\Ic,\Oc,\tilde{\Sc},\tilde{\Ec})$ that is represented by $\Mf$ such that $a\huh b$ is in $\tilde{\Af}$.
\end{lemma}

\begin{proof}[Proof of \cref{lem:visible_edge_II}]
  We define 
  $\tilde{\Af}=(\Ic,\Oc,\tilde{\Sc},\tilde{\Ec})$ as follows: for each undirected edge $u\tut v$ in $\Mf$, introduce a selection node $s_{uv}$ (with $s_{uv}\sim s_{vu}$) and edges $u\tuh s_{uv} \hut v$; delete all undirected edges and keep all other edges in $\Mf$; and in addition add the bidirected $a\huh b$. It is easy to see that $\tilde{\Af}$ is an isADMG. We shall next show that $\cat{MAG}(\tilde{\Af})=\Mf$.

  First, let $x$ and $y$ be non-adjacent in $\Mf$. The goal is to show that there is no $\tilde{\Sc}$-inducing walk from $x$ to $y$ in $\tilde{\Af}$. Assume the contrary that there is one such walk $\pi$ with minimal length. We know that it must be a collider walk where all colliders are in $\anc_{\tilde{\Af}}(\{x,y\}\cup \tilde{\Sc})$. This walk cannot contain any nodes of $\tilde{\Sc}$. If it does, then it must be of the form 
  \[
  x\sus\cdots\sus u\tuh s_{uv} \hut v\sus \cdots \sus y
  \] 
  with $u\tut v$ in $\Mf$. Since $\pi$ is a collider walk, we must have that $u$ and $v$ are endnodes of the walk. Therefore, $\pi$ must be $x\tuh s_{xy} \hut y$ with $x\tut y$ in $\Mf$. It contradicts the fact that $x$ and $y$ are non-adjacent in $\Mf$. Hence, $\pi$ can only consist of nodes in $\Oc$ but not nodes in $\tilde{\Sc}$.

  We assume that the walk $\pi$ does not contain $a\huh b$, which implies that $\pi$ is also present in $\Mf$ since we just showed that $\pi$ does not contian nodes in $\tilde{\Sc}$. Let $z$ denote a collider on the walk $\pi$. We know that $z\in \anc_{\tilde{\Af}}(\{x,y\}\cup \tilde{\Sc})$. Assume $z\in \anc_{\tilde{\Af}}(\tilde{\Sc})$. Every directed path from $z$ to $\tilde{\Sc}$ must be of the form $z\tuh \cdots\tuh u\tuh s_{uv}$ in $\tilde{\Af}$ for some $u\tut v$ in $\Mf$. By the construction of $\tilde{\Af}$, we know that the path $z\tuh \cdots\tuh u\tut v$ is present in $\Mf$, which is impossible. So we have that $z\in \anc_{\tilde{\Af}}(\{x,y\})$, which implies that $z\in \anc_{\Af}(\{x,y\}\cup \Sc)$. Overall, there exists an $\Sc$-inducing walk from $x$ to $y$ in $\Af$, which contradicts the fact that $\Mf$ is a MAG representing $\Af$ and $x$ is not adjacent to $y$ in $\Mf$.

  Now assume that $a\huh b$ is on the walk $\pi$, i.e., $\pi$ is of the form 
  \[
  x\sus \cdots\sus a\huh b\sus \cdots\sus y.
  \] 
  If $x=a$, then the walk $x=a\tuh b\sus\cdots\sus y$ is present in $\Mf$. Therefore, we can find an $\Sc$-inducing walk in $\Af$ from $x$ to $y$ similarly to what we did above. This leads to a contradiction. So, we can assume that $x\ne a$ in the following. Note that the subwalk $x=v_0\sus v_1\sus \cdots \sus v_{n-1}\sus a$ of $\pi$ from $x$ to $a$ is a collider walk. We shall show by induction that $v_i\in \pa_{\Mf}(b)$ for $0\leq i\leq n-1$.

  The node $v_{n-1}$ must be adjacent to $b$ in $\Mf$, otherwise $a\tuh b$ would be visible in $\Mf$. The edge between $v_{n-1}$ and $b$ cannot be $v_{n-1}\hut b$, otherwise we have $a\tuh b\tuh v_{n-1}\suh a$ in $\Mf$, which contradicts the fact that $\Mf$ is a MAG. Also, it cannot be $v_{n-1}\huh b$, otherwise we can find an $\tilde{\Sc}$-inducing walk 
  \[
  x\sus v_1\sus \cdots \sus v_{n-1} \huh b\sus \cdots\sus y
  \] 
  in $\tilde{\Af}$, which is shorter than $\pi$ and therefore contradicts the assumption of $\pi$ being shortest. Besides, the edge between $v_{n-1}$ and $b$ cannot be $v_{n-1}\tut b$. Otherwise, there would be the configuration $a\tuh b\tut v_{n-1}$ in $\Mf$, which never occurs in a MAG. Hence, we can conclude that we must have $v_{n-1}\tuh b$.

  We now assume that for all $k< i\leq n-1$ with some $0\leq k<n-1$, we have that $v_i\tuh b$ is in $\Mf$. We next prove that $v_k\tuh b$ is in $\Mf$. First note that $v_k$ must be adjacent to $b$ in $\Mf$, otherwise $a\tuh b$ would be visible in $\Mf$. As in the previous part, we can see that the edge between $v_k$ and $b$ cannot be $v_k\hut b$, $v_k\huh b$ or $v_k\tut b$ for otherwise we would have almost cycle $b\tuh v_k\huh v_{k+1}\tuh b$, shorter $\tilde{\Sc}$-inducing walk $x\sus\cdots\sus v_k\huh b\sus \cdots \sus y$ in $\tilde{\Af}$ or impossible configuration $v_k\tut b\hut a$ in a MAG, respectively. Therefore, we have $v_k\tuh b$.

  By induction, $v_i\in \pa_{\Mf}(b)$ for $0\leq i\leq n-1$. Since we have $x\tuh b$, then we have a shorter $\tilde{\Sc}$-inducing walk $x\tuh b\sus \cdots \sus y$ in $\tilde{\Af}$, which contradicts the choice of $\pi$. Hence, there cannot be an $\tilde{\Sc}$-inducing walk from $x$ to $y$ in $\tilde{\Af}$ and therefore $x$ and $y$ cannot be adjacent in $\cat{MAG}(\tilde{\Af})$.

  Overall, if $x$ and $y$ are not adjacent in $\Mf$, then $x$ and $y$ are not adjacent in $\cat{MAG}(\tilde{\Af})$.

  We now consider the case that $x$ and $y$ are adjacent in $\Mf$. Assume that $x\tuh y$ is in $\Mf$. Then $x\tuh y$ is in $\tilde{\Af}$ by construction. It implies that $y\notin \anc_{\Af}(\Sc)$ and therefore $y\notin \anc_{\tilde{\Af}}(\tilde{\Sc})$. Therefore $x\tuh y$ is in $\cat{MAG}(\tilde{\Af})$ by construction. The case where the edge $x\hut y$ is in $\Mf$ is similar. If $x\huh y$ is in $\Mf$, then $x\huh y$ is in $\tilde{\Af}$ and $x,y\notin \anc_{\tilde{\Af}}(\tilde{\Sc})$. Therefore, we have edge $x\huh y$ in $\cat{MAG}(\tilde{\Af})$. If $x\tut y$ is in $\Mf$, then by construction we have $x\tuh s_{xy}\hut y$ in $\tilde{\Af}$ and therefore $x\tut y$ is in $\cat{MAG}(\tilde{\Af})$. We are done.

\end{proof}

\begin{remark}
    Let $\Wc\coloneqq \{a\tuh b: a\tuh b \text{ is invisible in } \Mf\}$. In general, there need not exist an isADMG $\Af$ represented by $\Mf$ such that $a\huh b$ is in $\Af$ for every $a\tuh b $ in $\Wc$. See \cref{fig:countex_invis} for an example.
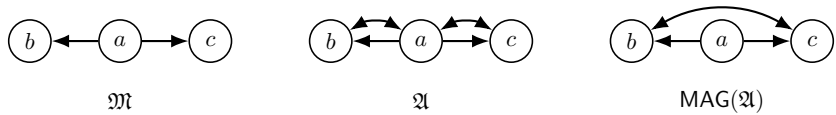
\begin{figure}[ht]
\centering
\begin{tikzpicture}[scale=0.8, transform shape]
\begin{scope}[xshift=0]
    \node[ndout] (a) at (0,0) {$a$};
    \node[ndout] (b) at (-1.5,0) {$b$};
    \node[ndout] (c) at (1.5,0) {$c$};
    \draw[arout] (a) to (b);
    \draw[arout] (a) to (c);
    \node at (0,-1) {$\Mf$};
\end{scope}

\begin{scope}[xshift=5cm]
    \node[ndout] (a) at (0,0) {$a$};
    \node[ndout] (b) at (-1.5,0) {$b$};
    \node[ndout] (c) at (1.5,0) {$c$};
    \draw[arout] (a) to (b);
    \draw[arout] (a) to (c);
    \draw[arlat, bend left] (b) to (a);
    \draw[arlat, bend left] (a) to (c);
    \node at (0,-1) {$\Af$};
\end{scope}
\begin{scope}[xshift=10cm]
    \node[ndout] (a) at (0,0) {$a$};
    \node[ndout] (b) at (-1.5,0) {$b$};
    \node[ndout] (c) at (1.5,0) {$c$};
    \draw[arout] (a) to (b);
    \draw[arout] (a) to (c);
    \draw[arlat, bend left] (b) to (c);
    \node at (0,-1) {$\cat{MAG}(\Af)$};
\end{scope}

\end{tikzpicture}
\caption{$\Mf$ is a MAG. $\Af$ is an isADMG constructed from $\Mf$ by adding bidirected edges to all the invisible directed edges of $\Mf$.  $\cat{MAG}(\Af)\ne \Mf$.}
\label{fig:countex_invis}
\end{figure}
\end{remark}

\begin{lemma}[Property of subgraphs]\label{lem:head_circle}
    Let $\PAG$ be an iSOPAG and $\Pf_A$ be its induced subgraph over $A\subseteq \Vc$. For any nodes $a,b,c\in A$, if $a\suh b\ous c$ is in $\Pf_A$, then $a\suh c$ must be in $\Pf_A$. Furthermore, if $a\tuh b\ous c$ is in $\Pf_A$, then $a\tuh c$ or $a\ouh c$ is in $\Pf_A$, i.e., the edge between nodes $a$ and $c$ is not $a\huh c$.  
\end{lemma}

\begin{proof}
The first claim is direct from Property~\ref{sopag:p3} in \cref{def:SOPAG} and the definition of full subgraphs. One can derive the second claim using $\fci{2}$.
\end{proof}

\begin{corollary}\label{cor:prop_subgraph_II}
    Assume the setting of \cref{lem:head_circle} and furthermore $b\ouo c$. If $a\huh b$, then  $a\huh c$. If $a\ouh b$ or $a\tuh b$, then $a\ouh c$ or $a\tuh c$.
\end{corollary}

\begin{proof}
    If $a\suh c$ is $a\ouh c$, then by applying \cref{lem:head_circle} to the segment $b\huh a\ouh c$ we would conclude that the edge between nodes $b$ and $c$ is $b\suh c$, a contradiction. If we have $a\tuh c$, then applying \cref{lem:head_circle} to the segment $a\tuh c\ouo b$ shows that the edge between nodes $a$ and $b$ is not bidirected, which is again a contradiction. Overall, the edge between nodes $a$ and $c$ must be $a\huh c$. The second statement is derived similarly.
\end{proof}

\begin{lemma}[Property of buckets]\label{lem:prop_bucket}
Let $\Gf=(\Ic,\Vc,\Ec)$ be a full subgraph of an iSOPAG $\Pf$. Let nodes $a$ and $b$ be in the same bucket in $\Gf$. If we have $c\suh a$ in $\Gf$, then we have $c\suh b$ in $\Gf$. If we have $c\huh a$ in $\Gf$, then we have $c\huh b$ in $\Gf$. If we have $c\ouh a$ or $c\tuh a$, then we have $c\ouh b$ or $c\tuh b$.
\end{lemma}

\begin{proof}
     Let \[\pi:a\sus v_1\sus \cdots \sus v_{n-1}\sus b\] be a path from $a$ to $b$ in the bucket. Since $c\suh a$ is in $\Gf$, we cannot have $a\tus v_1$ or $a\out v_1$ in $\Gf$ by \cref{def:SOPAG}. So we have $a\ouo v_1$. This implies $c\suh v_1$ by \cref{lem:head_circle}. Repeat the argument. Then we can conclude that $c\suh b$ is in $\Gf$. The second and third claims follow from \cref{cor:prop_subgraph_II}.
\end{proof}

The following lemma extends \cite[Lemma~7.5]{Maathuis15generalized} and \cite[Lemma~B.4]{Zhang08complete} to iSOPAGs.

\begin{lemma}\label{lem:podpath}
    Let $\Pf$ be an iSOPAG and $a,b$ be two distinct nodes in $\Pf$. If there is a potentially directed path from $a$ to $b$, then it is impossible to have $b\suh a$ in $\Pf$.
\end{lemma}

\begin{proof}
    We first show that, if there is a potentially directed path from $a$ to $b$ in $\Pf$, then there is a shortest potentially directed path \[\pf:a=v_0\sus \cdots \sus v_n=b\] such that if $v_{i-1}\suh v_{i}$ for some $i\in\{2,\ldots,n\}$ then $v_{j-1}\tuh v_j$ for all $j\in \{i+1,\ldots,n\}$. The case with $n\leq 2$ is trivial. Note that $\pf$ can contain the following edges $v_{i-1}\ouo v_i$, $v_{i-1}\ouh v_i$, $v_{i-1}\tuh v_i$, and $v_{i-1}\tuo v_i$ with $i=2,\ldots,n$. The pattern $v_{i-1}\suh v_i\tuo v_{i+1}$ cannot occur in an iSOPAG by \cref{def:SOPAG}. By the same argument in \cite[Lemma~7.2]{Maathuis15generalized}, we know that $v_{i-1}\suh v_i\ouh v_{i+1}$ and $v_{i-1}\suh v_i\ouo v_{i+1}$ cannot occur on $\pf$ provided that $\pf$ is a shortest potentially directed path. This shows the claim.

    Then by \cref{lem:head_circle}, the same argument in the second paragraph of the proof of \cite[Lemma~B.4]{Zhang08complete} completes the proof.

\end{proof}

\begin{lemma}\label{lem:property_visible}
    Let $\PAG$ be an iSOPAG and $A\subseteq \Vc$. For any $a,b,c\in A$, the following hold in $\Pf_A$:
\begin{enumerate}
\item If $a\suh b\tuh c$ or $a\suh b\ouh c$ or $a\suh b\ouo c$ and both edges are not visible directed in $\Pf_A$,
then $a\suh c$ is present in $\Pf_A$ and is not visible directed.
\item If $a\ous b \ous c$ and $a\suh c$ are present, then $a\suh c$ is neither
a visible directed edge nor a bidirected edge $a\huh c$.
\end{enumerate}
\end{lemma}

\begin{proof}
\textbf{Step~0: preparatory work}. First note that if $a\ous b$ is $a\out b$, then by \cref{def:SOPAG}, $a\suh c$ cannot be visible. Since $a\suh c$, it follows that that $b\ous c$ cannot be $b\out c$. Therefore, we have three cases: 

\begin{caselist}
  \item $a\huh b \ouo c$;

  \item $a\huh b$ combined with $b\tuh c$ invisible or $b\ouh c$;

  \item $a\tuh b$ invisible, $a\ouo b$, or $a\ouh b$ combined with $b\tuh c$ invisible, $b\ouo c$, or $b\ouh c$.
\end{caselist}

\textbf{Step~1: show Case~1}. By \cref{cor:prop_subgraph_II}, we know that we must have $a\huh c$ in this case, which is not visible directed.

\textbf{Step~2: show Case~2}. We only need to consider the case where $a\tuh c$. Assume for contradiction that there exists a node $d\in A$ non-adjacent to $c$ such that: $d\suh a \tuh c$ or $d\suh v_1\huh \cdots \huh v_{n-1} \huh a\tuh c$ with $n>1$ and $v_i\in \pa_{\Pf_A}(c)$ for all $1\le i\le n-1$. Since $a\huh b$ and $d$ is non-adjacent to $c$, we have a discriminating path $(d,a,b,c)$ or $(d,v_1,\ldots, v_{n-1},a,b,c)$ for $b$. This contradicts either the invisibility of $b\tuh c$ or the circle in $b\ouh c$ by $\fci{4}$. So, $a\tuh c$ must be invisible

\textbf{Step~3: show Case~3}. We only need to consider the case where $a\tuh c$. Assume for contradiction that there exists a node $d\in A$ non-adjacent to $c$ such that: $d\suh a $ or $d\suh v_1\huh \cdots \huh v_{n-1} \huh a$ with $n>1$ and $v_i\in \pa_{\Pf_A}(c)$ for all $1\le i\le n-1$. In the first subcase, by \cref{lem:head_circle} or invisibility of $a\tuh b$, there must be $d\suh b$. Again, by \cref{lem:head_circle} or the invisibility of $b\tuh c$, node $d$ is adjacent to $c$, which is a contradiction. We now consider the second subcase. First note that if there exists $v_i$ such that $v_i \huh b$, then $(d, v_1, \ldots ,v_i, b, c)$ forms a discriminating path for $b$, which contradicts the invisibility of $b\tuh c$ or the circle in $b\ous c$. Therefore, we can exclude this case. Note that we can also exclude the case where $a\ouo b$ similarly. If $v_{n-1}\tuh b$, then by the argument in Step~2, we can conclude that $v_{n-1}\tuh b$ is invisible. So, we could have either $v_{n-1}\tuh b$ invisible or $v_{n-1}\ouh b$. In either of the two cases, we have $v_{n-2}\tuh b$ or $v_{n-2}\ouh b$. If we have $v_{n-2}\tuh b$, then since $v_{n-2}\huh v_{n-1}$ with $v_{n-1}\ouh b$ or $v_{n-1}\tuh b$ invisible, by the argument in Step~2, the edge $v_{n-2}\tuh b$ must be invisible. Repeating the above argument, we can eventually reach $v_1$ and show that $v_1\tuh b$ invisible or $v_1\ouh b$. This implies that node $d$ is adjacent to $c$, which is a contradiction. Therefore, $a\tuh c$ must be invisible.

This finishes the proof.

\end{proof}

The following lemma extends \cite[Lemma~4]{jaber19idencom} to iSOPAGs.

\begin{lemma}[Property of regions]\label{lem:region}
    Let $\Pf=(\Ic,\Vc,\Ec)$ be an iSOPAG. Let $A\subsetneq B \subseteq \Vc$. Then for $c\in \re_{\Pf_B}(A)$ there does not exist a node $b\in B\sm \re_{\Pf_B}(A)$ such that $b\ouh c$, such that $b\tuh c$ $\Pf_B$-invisible, or such that $b\huh c$. In other words, for $c\in \re_{\Pf_B}(A)$ and $b\in B$, if we have $b\ouh c$, $\Pf_B$-invisible directed edge $b\tuh c$, or $b\huh c$, then $b\in \re_{\Pf_B}(A)$.
\end{lemma}

\begin{proof}
    We argue by contradiction and assume on the contrary that such node $b\in B$ exists, i.e., there exists $b\in B\sm \re_{\Pf_B}(A)$ such that non-visible edge $b\suh c$ is in $\Pf_{B}$. Since $c\in \re_{\Pf_B}(A)$, we can WLOG assume that (if there exists $\tilde{c}$ in the bucket of $c$ satisfying the assumption, then we replace $c$ with $\tilde{c}$ and continue with the argument) there exists $a\in A$ and a pc-connecting path $\pi$ from $a$ to $c$ that is of the form $a\sus c$ not visible or $a=v_0\suh v_1 \huh \cdots \huh v_{n-1}\hus v_{n}=c$ where $n>1$ and none of the edges are visible. In the first case, it is easy to check that this case cannot happen. Indeed, by the property of iSOPAGs, we know that $a\sus c$ cannot be $a\tut c$, $a\tuo c$, or $a\out c$, since we have $b\suh c$ in $\Pf_B$. If we have $a\suh c$, then $a\suh c\hus b$ forms a pc-connecting path from $a$ to $b$, which implies that $b\in \re_{\Pf_B}(A)$. If we have $a\huo c$, $a\ouo c$, or $a\hut c$ invisible, then node $b$ must be adjacent to node $a$ (by \cref{lem:head_circle} or invisibility) and the edge between them must be non-visible by \cref{lem:property_visible,lem:property_visible}. Therefore, we only need to consider the second case. Then the same argument in \cite[Lemma~4]{jaber19idencom} applies by noting that arrowheads will not meet edges $v\tut u$, $v\tuo u$, or $v\out u$, and \cite[Lemma~7.5]{Maathuis15generalized} can be generalized to iSOPAGs (\cf \cref{lem:podpath}).

\end{proof}

The following lemma extends \cite[Lemma~48]{perkovic18complete} to the case under selection bias.

\begin{lemma}\label{lem:orientation_invisible}
    Let $\Pf=(\Ic,\Vc,\Ec)$ be an iCOPAG. Let $\{c_i\}_{i=1}^n\subseteq \Vc$ be 
    nodes of $\Pf$ such that no two of them lie in the same circle component of $\Pf$. Let $\Mf\in [\Pf]_\Ms$ be an iMAG constructed via the following orientation scheme: 
    \begin{enumerate}
        \item orient all $\ouh$ edges into directed edges $\tuh$;

        \item orient all $\out$ edges into undirected edges $\tut$;

        \item for any $a\ouo b$, if there is no arrowhead into $a$ or $b$, then orient it into an undirected edge $a\tut b$;

        \item further orient the circle components of the graph into a DAG with no unshielded colliders such that no new arrowheads are introduced into any $c_i$. 
    \end{enumerate}
    Then every edge that is either $c_i\ouo d$, $c_i\ouh d$ or invisible $c_i\tuh d$ in $\Pf$ is not a visible directed edge in $\Mf$.
\end{lemma}

\begin{proof}
    The same proof of \cite[Lemma~48]{perkovic18complete} applies using \cref{lem:0,cor:prop_subgraph_II}.
\end{proof}

Given an s-SCM $(\Mc,X_{\Sc}=\1)\in \Mb^+$ with sADMG $\Af$ and a target interventional Markov kernel $\Prb_{\Mc}(X_A\mid X_{\Sc}=\1\miid \Do(X_B))$ with $A,B\subseteq \Vc$ disjoint, the conditional version of the hedge criterion on ADMGs derived in \cite{shpitser2008complete} gives a convenient graphical criterion for the non-identifiability of $\Prb_{\Mc}(X_A\mid X_{\Sc}=\1\miid \Do(X_B))$.

\begin{proposition}[Hedge criterion]\label{prop:hedge}
Let $\Gf=(\Vc,\Sc,\Ec)$ be an sADMG and let $A,B\subseteq \Vc$ be disjoint. Let $D\subseteq \Sc$ be the unique maximal set such that 
    \[
    A\sep{\mathsf{id}}{\Gf_{\Do(I_D,B)}} I_D\mid B\cup \Sc,
    \] 
    whose existence and uniqueness are guaranteed by \cite[Theorem~20]{shpitser2008complete}. If there exists a hedge $(\Hc,\Hc^\prime)$ for $(A\cup (\Sc\sm D),B\cup D)$ in $\Gf$, then $\Prb_{\Mc}(X_A\mid X_{\Sc}=x_{\Sc}\miid \Do(X_B))$, for some $x_\Sc\in \Xc_{\Sc}$, is not identifiable in $\Mb_d^+(\Gf)$.
\end{proposition}

\begin{proof}
    By \cite[Theorems~17,20, and 21]{shpitser2008complete}, there exist $\Mc_1,\Mc_2\in \Mb_d^+(\Gf)$ such that 
    \[
        \Prb_{\Mc_1}(X_{\Vc\cup \Sc})=\Prb_{\Mc_2}(X_{\Vc\cup \Sc}), \quad \text{ and }\quad  \Prb_{\Mc_1}(X_{A}\mid X_{\Sc}\miid \Do(X_B) )\ne \Prb_{\Mc_2}(X_{A}\mid X_{\Sc}\miid \Do(X_B)).
    \]
    This then implies for some $x_\Sc\in \Xc_\Sc$
    \[
        \begin{aligned}
            &\Prb_{\Mc_1}(X_{\Vc}\mid X_{\Sc}=x_{\Sc})=\Prb_{\Mc_2}(X_{\Vc}\mid X_{\Sc}=x_{\Sc}), \text{ and }\\  
            &\Prb_{\Mc_1}(X_{A}\mid X_{\Sc}=x_{\Sc}\miid \Do(X_B) )\ne \Prb_{\Mc_2}(X_{A}\mid X_{\Sc}=x_{\Sc}\miid \Do(X_B)).
        \end{aligned}
    \]
    This shows the result.
\end{proof}

\begin{example}[Adapted from \cite{shpitser23doesidalgorithmfail}]\label{ex:hedge}
    Consider an isADMG $\Af$ shown in \cref{fig:hedge} with s-SCM $(\Mc,X_s=1)\in \Mb^+_d(\Af)$. The subgraphs $\Gf^1$ and $\Gf^2$ are two $\{a,c\}$-rooted C-forests. Set $\Hc\coloneqq \{b_2,a,c\}$ and $\Hc^\prime\coloneqq \{a,c\}$. We have 
    \[
    \{a,c\}\subseteq \anc_{\Af_{\Do(b_2)}}(\{a,s\}),\quad  \Hc\cap \{b_2\}=\{b_2\}\ne \emptyset,\quad  \text{ and } \quad\Hc^\prime\cap \{b_2\}=\emptyset.
    \]
    So $(\Hc,\Hc^\prime)$ forms a hedge for $(\{a,s\},\{b_2\})$ in $\Af$. Since 
    \[
    a\nsep{\mathsf{id}}{\Af_{\Do(I_s,b_2)}} I_s\mid \{b_2,s\},
    \]
    the interventional Markov kernel $\Prb_{\Mc}(X_a\mid X_s=1\miid \Do(X_{b_2}))$ is not identifiable in $\Mb^+_d(\Af)$ by \cref{prop:hedge}. One can check that  $(\Hc,\Hc^\prime)$ also forms a hedge for $(\{a,s\},\{b_1,b_2\})$ in $\Af$. Since 
    \[
    a\nsep{\mathsf{id}}{\Af_{\Do(I_s,\{b_1,b_2\})}} I_s \mid \{b_1,b_2,s\},
    \]
    \cref{prop:hedge} establishes that $\Prb_{\Mc}(X_{a}\mid X_s=1\miid \Do(X_{b_1},X_{b_2}))$ is not identifiable in $\Mb^+_d(\Af)$. On the other hand, $(\Hc,\Hc^\prime)$ is not a hedge for $(\{a\},\{b_1,b_2\})$ in $\tilde{\Af}$, since 
    \[
    \{a,c\}\nsubseteq \anc_{\tilde{\Af}_{\Do(\{b_1,b_2\})}}(\{a\}).
    \]
    Indeed, $\Prb_{\tilde{\Mc}}(X_a\miid \Do(X_{b_1},X_{b_2}))$ is identifiable in $\Mb^+_d(\tilde{\Af})$; see \cite[Section~3]{shpitser23doesidalgorithmfail} for an identification formula.

\begin{figure}[ht]
\centering
\begin{tikzpicture}[scale=0.8, transform shape]
\begin{scope}[xshift=0]
    \node[ndout] (a) at (0,0) {$a$};
    \node[ndout] (b1) at (-1.5,0) {$b_1$};
    \node[ndout] (b2) at (0,1.5) {$b_2$};
    \node[ndout] (c) at (-1.5,1.5) {$c$};
    \node[ndsel] (s) at (-0.75,2.5) {$s$};
    \draw[arout] (b1) to (a);
    \draw[arout] (c) to (b1);
    \draw[arout] (b2) to (a);
    \draw[arout] (c) to (s);
    \draw[arout] (b2) to (s);
    \draw[huh] (c) to (a);
    \draw[huh] (c) to (b2);
    \node at (-0.75,-1) {$\Af$};
\end{scope}

\begin{scope}[xshift=5cm]
    \node[ndout] (a) at (0,0) {$a$};
    \node[ndout] (b1) at (-1.5,0) {$b_1$};
    \node[ndint] (b2) at (0,1.5) {$b_2$};
    \node[ndout] (c) at (-1.5,1.5) {$c$};
    \node[ndsel] (s) at (-0.75,2.5) {$s$};
    \draw[arout] (b1) to (a);
    \draw[arout] (c) to (b1);
    \draw[arout] (b2) to (a);
    \draw[arout] (c) to (s);
    \draw[arout] (b2) to (s);
    \draw[huh] (c) to (a);
    \node at (-0.75,-1) {$\Af_{\Do(b_2)}$};
\end{scope}

\begin{scope}[xshift=10cm]
    \node[ndout] (a) at (0,0) {$a$};
    \node[ndint] (b1) at (-1.5,0) {$b_1$};
    \node[ndint] (b2) at (0,1.5) {$b_2$};
    \node[ndout] (c) at (-1.5,1.5) {$c$};
    \node[ndsel] (s) at (-0.75,2.5) {$s$};
    \draw[arout] (b1) to (a);
    \draw[arout] (b2) to (a);
    \draw[arout] (c) to (s);
    \draw[arout] (b2) to (s);
    \draw[huh] (c) to (a);
    \node at (-0.75,-1) {$\Af_{\Do(\{b_1,b_2\})}$};
\end{scope}

\begin{scope}[xshift=0cm,yshift=-5cm]
    \node[ndout] (a) at (0,0) {$a$};
    \node[ndout] (b2) at (0,1.5) {$b_2$};
    \node[ndout] (c) at (-1.5,1.5) {$c$};
    \draw[arout] (b2) to (a);
    \draw[huh] (c) to (a);
    \draw[huh] (c) to (b2);
    \node at (-0.75,-1) {$\Gf^1$};
\end{scope}

\begin{scope}[xshift=5cm,yshift=-5cm]
    \node[ndout] (a) at (0,0) {$a$};
    \node[ndout] (c) at (-1.5,1.5) {$c$};
    \draw[huh] (c) to (a);
    \node at (-0.75,-1) {$\Gf^2$};
\end{scope}

\begin{scope}[xshift=10cm,yshift=-5cm]
    \node[ndout] (a) at (0,0) {$a$};
    \node[ndout] (b1) at (-1.5,0) {$b_1$};
    \node[ndout] (b2) at (0,1.5) {$b_2$};
    \node[ndout] (c) at (-1.5,1.5) {$c$};
    \draw[arout] (b1) to (a);
    \draw[arout] (c) to (b1);
    \draw[arout] (b2) to (a);
    \draw[huh] (c) to (a);
    \draw[huh] (c) to (b2);
    \node at (-0.75,-1) {$\tilde{\Af}$};
\end{scope}

\end{tikzpicture}
\caption{An isADMG $\Af$, its hard-manipulated graphs $\Af_{\Do(b_2)}$ and $\Af_{\Do(\{b_1,b_2\})}$, and its two subgraphs $\Gf^1$ and $\Gf^2$ from \cref{ex:hedge}; also shown is the isADMG $\tilde{\Af}$ without a selection node from \cref{ex:hedge}.}
\label{fig:hedge}
\end{figure}
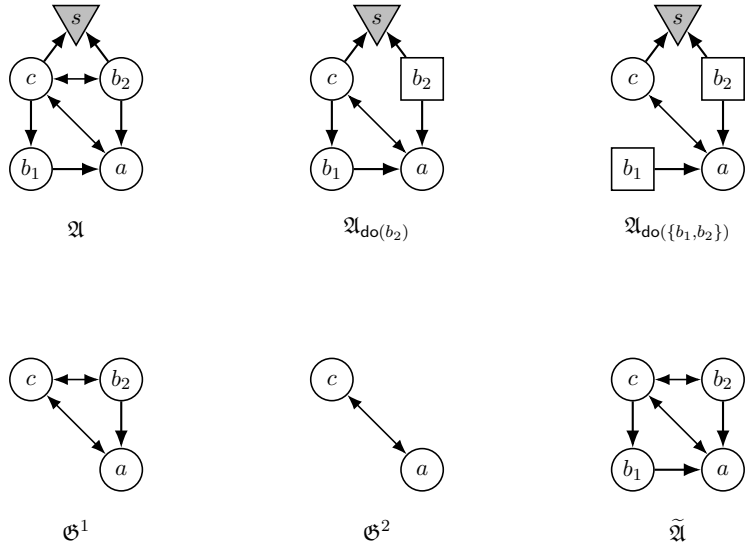
\end{example}

\section[Proofs for Section~\ref{sec:causal_mag_pag}]{Proofs for \cref{sec:causal_mag_pag}}\label{sec:pf_1}

\subsection[Proof of Proposition~\ref{prop:mag_rep}]{Proof of \cref{prop:mag_rep}}
\begin{proof}[Proof of \cref{prop:mag_rep}]
  We will construct a mixed graph $\Mf$ with input nodes $\Ic$ and output nodes $\Vc$ as follows. Let two nodes $a, b \in \Ic \cup \Vc$ be adjacent in $\Mf$ if and only if (i) $a \neq b$, (ii) $\{a, b\} \nsubseteq \Ic$, and (iii) there is an $(\Lc,\Sc)$-inducing path between $a, b$ in $\Af$. In that case, orient the edge between $a$ and $b$ in $\Mf$ as follows:
  $$
  \begin{cases}a \tut b & \text { if } a \in \anc_\Af(\{b\} \cup \Sc)   \text { and } b \in \anc_\Af(\{a\} \cup \Sc) , \\ a \tuh b & \text { if } a \in \anc_\Af(\{b\} \cup \Sc)   \text { and } b \notin \anc_\Af(\{a\} \cup \Sc) , \\ a \hut b & \text { if } a \notin \anc_\Af(\{b\} \cup \Sc)  \text { and } b \in \anc_\Af(\{a\} \cup \Sc) , \\ a \huh b & \text { if } a \notin \anc_\Af(\{b\} \cup \Sc)  \text { and } b \notin \anc_\Af(\{a\} \cup \Sc) . \end{cases}
  $$
  Note that if there is an inducing path from an input node $a$ to an output node $b$, then $a\in \anc_{\Af}(\{b\} \cup \Sc)$. Therefore, there are no arrowheads towards input nodes. 
  Hence, $\Mf$ is the unique MAG that $(\Lc,\Sc)$-represents $\Af$.
\end{proof}

\subsection[Proof of Proposition~\ref{prop:sep_mag}]{Proof of \cref{prop:sep_mag}}

\begin{proof}[Proof of \cref{prop:sep_mag}]
    Let $A\subseteq \Vc$ and $B,C\subseteq \Ic\cup \Vc$. Let $\Af\coloneqq\cat{isADMG}(\Mf)\in [\Mf]_\Gs$. Let \[\pi:v_0\sus \cdots \sus v_n\] be a $C$-open path from $A$ to $B\cup \Ic$ in $\Mf$, which does not contain non-endnodes from $\Ic$. Replacing all the undirected edges $v_i\tut v_{i+1}$ (if any) on $\pi$ in $\Mf$ with $v_{i}\tuh s_{v_iv_{i+1}} \hut v_{i+1}$ gives a path $\tilde{\pi}$ in $\Af$. It is easy to see that $\tilde{\pi}$ is an open path from $A$ to $B\cup \Ic$ given $C\cup \Sc_\Af$ in $\Af$. Therefore, we have shown that 
    \[
        A\nsep{\mathsf{id}}{\Mf} B\mid C \quad \Longrightarrow \quad \exists \Af\in [\Mf]_\Gs,\ A\nsep{\mathsf{id}}{\Af} B\mid C\cup \Sc_\Af.
    \]

    Let $\Af=(\Ic,\Oc,\Sc,\tilde{\Ec})\in [\Mf]_\Gs$ and $\pi$ be an $(C\cup \Sc)$-open walk from $A$ to $B\cup \Ic$ in $\Af$, which does not contain non-endnodes from $\Ic$. Define $\Af^*\coloneqq (\emptyset,\Oc\dcup \Ic,\Sc,\tilde{\Ec})$. Note that $\Mf_\Oc=\cat{MAG}(\Af^*)_\Oc$ and for $a\in \Ic\cup \Oc$ and $b\in \Oc$, edge $a\tus b$ is in $\Mf$ iff the edge $a\tus b$ is in $\cat{MAG}(\Af^*)$, since for $a\in \Ic\cup \Oc$ and $b\in \Oc$ there is an inducing path from node $a$ to $b$ in $\Af$ iff there is an inducing path from node $a$ to $b$ in $\Af^*$, and 
    $
    a\in \anc_{\Af}(\{b\}\cup \Sc) $ iff $ a\in \anc_{\Af^*}(\{b\}\cup \Sc).
    $
    Then there must be an $C$-open walk $\tilde{\pi}$ from $A$ to $B\cup \Ic$ in $\tilde{\Mf}\coloneqq \cat{MAG}(\Af^*)$ such that every collider on $\tilde{\pi}$ is in $C$ and $\tilde{\pi}$ does not have any non-endnodes in $\Ic$ by \cite[Section~3.4.2]{richardson2002ancestral} and \cite[Theorem~4.18]{richardson2002ancestral}. The walk $\tilde{\pi}$ is also present in $\Mf$ and since set $C$ contains all the colliders and  does not contain non-colliders on $\tilde{\pi}$, the walk $\tilde{\pi}$ is open given $C$ in $\Mf$. It follows that $A\nsep{\mathsf{id}}{\Mf} B\mid C$. This means that
    \[
        \Big(\exists \Af\in [\Mf]_\Gs,\ A\nsep{\mathsf{id}}{\Af} B\mid C\cup \Sc_{\Af}\Big) \quad \Longrightarrow \quad A\nsep{\mathsf{id}}{\Mf} B\mid C.
    \]
    Overall, we have shown that for $A, B,C\subseteq \Ic\cup \Vc$
    \[
        A\sep{\mathsf{id}}{\Mf} B\mid C \quad \Longleftrightarrow \quad \forall \Af\in [\Mf]_\Gs,\ A\sep{\mathsf{id}}{\Af} B\mid C\cup \Sc_{\Af}.
    \]

\end{proof}

\subsection[Proof of Proposition~\ref{prop:int_com}]{Proof of \cref{prop:int_com}}

\begin{proof}[Proof of \cref{prop:int_com}]
  Clause 1 is immediate from the definition. Note that the definitions of $\Ec_i$ for $i=1,2,3,5$ are invariant if one replaces $\Mf$ with $\Mf_{\Do(I_A)}$ or $\Mf_{\Do(I_B)}$. The same holds for $\Ec_4$ by Lemma~\ref{lem:vis_sint}. The second clause then follows.
\end{proof}

The following lemma tells us that soft manipulation preserves visibility of directed edges.

\begin{lemma}[Visible edge and soft manipulation]\label{lem:vis_sint}
    Let $\MAG$ be an iMAG and $D\subseteq \Ic\cup \Vc$. Then a directed edge $a\tuh b$ is visible in $\Mf$ for $a,b\in \Vc$ iff it is visible in $\Mf_{\Do(I_D)}$.
\end{lemma}

\begin{proof}[Proof of \cref{lem:vis_sint}]
    By the definition of $\Mf_{\Do(I_D)}$, we have that a node $c\in \Ic\cup \Vc$  satisfies item (ii) of Definition~\ref{def:visible_edge} in $\Mf$ iff the same node $c$ satisfies the same condition in $\Mf_{\Do(I_D)}$. This shows that if the directed edge $a\tuh b$ is visible in $\Mf$, then it is also visible in $\Mf_{\Do(I_D)}$.

    To finish the proof, it suffices to show that if $a\tuh b$ is visible in $\Mf_{\Do(I_D)}$ then $a\tuh b$ must be visible in $\Mf$. Let $\Af\in [\Mf]_\Gs$. Then there exists $\tilde{\Af}\in [\Mf_{\Do(I_D)}]_\Gs$ such that $\Af_{\Do(I_D)}$ is a subgraph of $\tilde{\Af}$ and $(\Af_{\Do(I_D)})_{\Vc}=\tilde{\Af}_{\Vc}$. Since $a\tuh b$ is visible in $\Mf_{\Do(I_D)}$, by \cref{lem:visible_edge_I}, for all $\breve{\Af}\in [\Mf_{\Do(I_D)}]_\Gs$ there is no bidirected edge $a\huh b$ in $\breve{\Af}$. This implies that there is no bidirected edge $a\huh b$ in $\Af$ for arbitrary $\Af\in [\Mf]_\Gs$, which by \cref{lem:visible_edge_II} implies that $a\tuh b$ is visible in $\Mf$.

\end{proof}

\subsection[Proof of Theorem~\ref{thm:sep_hsint_mag}]{Proof of \cref{thm:sep_hsint_mag}}

We present the proof of \cref{thm:sep_hsint_mag}. See \cref{fig:pf_structure_sep_hsint_mag} for an overall structure of the proof.

\begin{proof}[Proof of \cref{thm:sep_hsint_mag}]
    Let $\Af\in [\Mf]_\Gs$ be arbitrary.
    By Lemma~\ref{lem:sint_mag_III}, there exists $\tilde{\Mf}\in [\Mf_{\Do(I_D)}]_\Ms$ such that $\cat{MAG}(\Af_{\Do(I_D)})$ is a subgraph of $\tilde{\Mf}$ and therefore $\cat{MAG}(\Af_{\Do(I_D)})_{\Do(T)}$ is a subgraph of $\tilde{\Mf}_{\Do(T)}$.
    We then have 
    \[
    \begin{aligned}
        A \sep{\mathsf{id}}{\Mf_{\Do(I_D,T)}} B\mid C\cup T \quad &\Longrightarrow \quad A \sep{\mathsf{id}}{\tilde{\Mf}_{\Do(T)}} B\mid C\cup T  \\
        &\Longrightarrow \quad A \sep{\mathsf{id}}{\cat{MAG}(\Af_{\Do(I_D)})_{\Do(T)}} B\mid C\cup T \\
        &\stackrel{\ref{lem:sep_hint_mag}}{\Longrightarrow} \quad A \sep{\mathsf{id}}{\Af_{\Do(I_D,T)}} B\mid C\cup T \cup \Sc_\Af.
    \end{aligned}       
    \]
     Hence, 
    \[
        A \sep{\mathsf{id}}{\Mf_{\Do(I_D,T)}} B\mid C\cup T \quad \Longrightarrow \quad \forall \Af \in[\Mf]_\Gs, \  A \sep{\mathsf{id}}{\Af_{\Do(I_D,T)}} B\mid C\cup T\cup \Sc_{\Af}.
    \]
    
    For the other direction, it suffices to show 
    \[
        A \nsep{\mathsf{id}}{\Mf_{\Do(I_D,T)}} B\mid C\cup T \quad \Longrightarrow \quad \exists \Af \in[\Mf]_\Gs, \  A \nsep{\mathsf{id}}{\Af_{\Do(I_D,T)}} B\mid C\cup T\cup \Sc_{\Af}.
    \]
    Assume that $A \nsep{\mathsf{id}}{\Mf_{\Do(I_D,T)}} B\mid C\cup T$. Then there exists a shortest open path \[
    \pi:a \sus v_1\sus \cdots \sus v_{n-1}\sus b
    \]
    from $A$ to $B\cup \Ic\cup \{I_d\}_{d\in D}\cup T$ given $C\cup T$ in $\Mf_{\Do(I_D,T)}$ such that $v_i\notin B\cup \Ic\cup \{I_d\}_{d\in D}\cup T$ for $i=0,\ldots,n-1$. Note that $b\notin T$. We shall show that there is an open path $\breve{\pi}$ from $A$ to $B\cup \Ic\cup \{I_d\}_{d\in D}\cup T$ given $C\cup T\cup \Sc_{\Af}$ in $\Af_{\Do(I_D,T)}$ for some $\Af\in [\Mf]_\Gs$. 
    
    Now, we show the existence of such $\breve{\pi}$. We first consider the case where the path $\pi$ only contains edges from $\Mf$. Then $\pi$ is an open  path in $\Mf$ from $A$ to $B\cup \Ic$ given $C$, and it does not intersect $T$. Applying the first paragraph of proof of \cref{prop:sep_mag}, we can find an open  path $\breve{\pi}$ from $A$ to $B\cup \Ic$ given $C\cup \Sc_\Af$ in $\Af$ not intersecting $T$. This path $\breve{\pi}$ satisfies the required properties.
    
    In the following, we can assume that the path $\pi$ contains nodes from $\{I_d\}_{d\in D}$ and edges not in $\Mf$. Since $\pi$ is shortest, there are no non-endnodes of $\pi$ in $\{I_d\}_{d\in D}$.
Then we have the following five cases:
\begin{enumerate}[label=\textbf{Case~\arabic*:}, ref=Case~\arabic*, leftmargin=*]
    \item   $\cdots \sus v_{n-1} \hut I_d $ on $\pi$ with $v_{n-1}\ne d$;
    \item   $\cdots \sut v_{n-1} \tut I_d$ on $\pi$ with $v_{n-1}\ne d$;
    \item  $\cdots \sus d=v_{n-1} \hut I_d $ on $\pi$;
    \item $\cdots \sut d=v_{n-1} \tut I_d $ on $\pi$;
    \item  $\cdots \sut d=v_{n-1} \out I_d $ on $\pi$.
\end{enumerate}

For Case~1, note that $d\tuh v_{n-1}$ is present and invisible in $\Mf$ by \cref{def:int_mag}. We construct an isADMG $\Af\in [\Mf]_\Gs$ from $\Mf$ by adding bidirected edge $v_{n-1}\huh d$ (\cf proof of \cref{lem:visible_edge_II}) and replacing all undirected edges $v\tut u$ with $v\tuh s_{vu} \hut u$. We define $\breve{\pi}$ in $\Af_{\Do(I_D)}$ by replacing all the undirected edges $v_i\tut v_{i+1}$ on $\pi(v_0,v_{n-1})$ with $v_i\tuh s_{v_iv_{i+1}}\hut v_{i+1}$ and adding $v_{n-1}\hut d\hut I_d$ if $d\notin C$ or adding $v_{n-1}\huh d \hut I_d$ if $d\in C$. This $\breve{\pi}$ satisfies the requirements.

For Case~2, first note that $d\tut v_{n-1}$ is in $\Mf$. We consider $\Af\in [\Mf]_\Gs$ constructed from $\Mf$ by replacing all undirected edges $v\tut u$ with $v\tuh s_{vu} \hut u$ and adding bidirected edge $s_{v_{n-1}d}\huh d$. We can construct $\breve{\pi}$ similarly to the previous case for the part involving $v_0,\ldots, v_{n-1}$ but adding $v_{n-1}\tuh s_{v_{n-1}d}\hut d\hut I_d$ if $d\notin C$ or adding $v_{n-1}\tuh s_{v_{n-1}d}\huh d\hut I_d$ if $d\in C$, which satisfies the requirements.

For Cases~3 and 4, we consider $\Af=\cat{isADMG}(\Mf)$ and construct $\breve{\pi}$ by replacing all undirected edges $v\tut u$ on $\pi$ with $v\tuh s_{vu} \hut u$ and keeping other parts of $\pi$ intact. Then $\breve{\pi}$ satisfies the requirement. 

For Case~5 where we have an circle edge $d\out I_d$, the only possibility is that $\sut d \out I_d$.  Therefore, we can reduce this case to Case 3 and Case 4. 

This finishes the proof. 

\end{proof}

\begin{lemma}[Ancestors of selection nodes I]\label{lem:anc_selec_I}
Let $\MAG$ be an iMAG and $a,b\in \Vc$. There exists an edge $a\hus b$ in $\Mf$ iff $a\notin \anc_{\Af}(\Sc)$ for all $\Af=(\Ic,\Vc,\Sc,\tilde{\Ec})\in [\Mf]_\Gs$.
\end{lemma}

\begin{proof}[Proof of \cref{lem:anc_selec_I}]
    If there exists an edge $a\hus b$ in $\Mf$, then $a\notin \anc_{\Af}(\{b\}\cup \Sc)$ and therefore $a\notin \anc_{\Af}(\Sc)$ for all $\Af=(\Ic,\Vc,\Sc,\tilde{\Ec})\in [\Mf]_\Gs$.

    The goal now is to show that if there does not exist an edge $a\hus b$ in $\Mf$ then there exists an isADMG $\Af=(\Ic,\Vc,\Sc,\tilde{\Ec})\in [\Mf]_\Gs$ such that $a\in \anc_{\Af}(\Sc)$. We can construct an isADMG $\Af$ by adding a node $s_a$ and an edge $a \tuh s_a$ to $\cat{isADMG}(\Mf)$. We next show that $\cat{MAG}(\Af)=\Mf$. Since there are no arrowheads towards node $a$, we have that for any $u,v\in \Vc\cup \Ic$ with $\{u,v\}\nsubseteq \Ic$: 
    \begin{enumerate}
        \item [(i)] there is an $\Sc$-inducing path from $u$ to $v$ in $\Af$ iff $u$ and $v$ are adjacent in $\Mf$,

        \item [(ii)] $u\notin \anc_{\Af}(\{v\}\cup \Sc)$ iff $u\hus v$ in $\Mf$, and

        \item [(iii)] $u\in \anc_{\Af}(\{v\}\cup \Sc)$ iff $u\tus v$ in $\Mf$.
    \end{enumerate}
    Hence, $\cat{MAG}(\Af)=\Mf$.
\end{proof}

This lemma implies that if there are no arrowheads towards node $a$ in an iMAG $\Mf$, then there must exist one isADMG $\Af$ represented by $\Mf$ such that node $a$ is an ancestor of selection nodes in $\Af$ (though $a$ may also be a non-ancestor of selection nodes in some isADMG $\tilde{\Af}$ represented by $\Mf$).

If there is an undirected edge connecting to node $a$ in an iMAG $\Mf$, then node $a$ must be an ancestor of a selection node in every isADMG $\Af$ represented by $\Mf$. 

\begin{lemma}[Ancestors of selection nodes II]\label{lem:anc_selec_II}
Let $\MAG$ be an iMAG and $a,b\in \Vc$. There exists an edge $a\tut b$ in $\Mf$ iff $\{a,b\}\subseteq  \anc_{\Af}(\Sc)$  for all $\Af=(\Ic,\Vc,\Sc,\tilde{\Ec})\in [\Mf]_\Gs$.
\end{lemma}

\begin{proof}[Proof of \cref{lem:anc_selec_II}]
    This immediate from the definition. Indeed, we have $a\in \anc_{\Af}(\{b\}\cup \Sc_{\Af})$ and $b\in \anc_{\Af}(\{a\}\cup \Sc_{\Af})$ for all $\Af\in [\Mf]_\Gs$. Since there are no cycles in $\Af$, we must have $a,b\in  \anc_{\Af}( \Sc_{\Af})$.
\end{proof}

The next several lemmas explore the relation between a soft-manipulated iMAG $\Mf_{\Do(I_D)}$ and a soft-manipulated isADMG $\Af_{\Do(I_D)}$ where $\Af\in [\Mf]_\Gs$. 

\begin{lemma}[Soft manipulation and MAG representation I]\label{lem:sint_mag_I}
  Let $\isADMG$ be an isADMG. Denote $\Mf\coloneqq \cat{MAG}(\Af)$ and $\tilde{\Mf}\coloneqq \cat{MAG}(\Af_{\Do(I_a)})$ for $a\in \Oc$. Then the induced subgraph $\tilde{\Mf}_{\Ic\cup\Oc}$ is equal to $\Mf$ and the
  edge $I_a\tus a$ is in $\tilde{\Mf}$. Furthermore, if $a\notin \anc_{\Af}(\Sc)$ then we have $I_a\tuh a$ in $\tilde{\Mf}$, and if $a\in \anc_{\Af}(\Sc)$ then we have $I_a\tut a$ in $\tilde{\Mf}$.
\end{lemma}

\begin{proof}[Proof of \cref{lem:sint_mag_I}]
  For every $c\in \Oc$ and $d\in \Oc\cup \Ic$, a path from $c$ to $d$ in $\Af$ is $\Sc$-inducing if and only if it is $\Sc$-inducing in $\Af_{\Do(I_a)}$. Also, note that
  \[
    \begin{aligned}
        &c\in \anc_{\Af}(\{d\}\cup \Sc) \quad \Longleftrightarrow \quad c\in\anc_{\Af_{\Do(I_a)}}(\{d\}\cup \Sc)\\
        & d\in \anc_{\Af}(\{c\}\cup \Sc) \quad \Longleftrightarrow \quad d\in\anc_{\Af_{\Do(I_a)}}(\{c\}\cup \Sc).
    \end{aligned}
  \]
  Therefore, we have $\tilde{\Mf}_{\Ic\cup\Oc}=\Mf$.
  By the construction of $\tilde{\Mf}$, it is easy to see that $I_a\tuh a$ is in $\tilde{\Mf}$ if $a\notin \anc_{\Af}(\Sc)$, and $I_a\tut a$ is in $\tilde{\Mf}$ if $a\in \anc_{\Af}(\Sc)$.
\end{proof}

\begin{remark}
  This means that if two isADMGs $\Af^1=(\Ic,\Oc,\Sc^1,\Ec^1)$ and $\Af^2=(\Ic,\Oc,\Sc^2,\Ec^2)$ have the same MAG representation over $\Oc$, then the MAG representations $\Mf^1$ and $\Mf^2$ of
  $\Af^1_{\Do(I_a)}$ and $\Af^2_{\Do(I_a)}$ on $\Oc$ for $a\in \Oc$ respectively have the same subgraphs on the output nodes. Also, the edges from input nodes $\Ic$ to output nodes $\Oc$ are the same in $\Mf^1$ and $\Mf^2$. The only possible difference between $\Mf^1$ and $\Mf^2$ is in the edges from input node $I_a$ to the output nodes $\Oc$.

\end{remark}

\begin{lemma}[Soft manipulation and MAG representation II]\label{lem:sint_mag_II}
  Let $\MAG$ be an iMAG and $a\in \Vc$.  Then we have that for all $b\in \Vc\sm \{a\}$
    \begin{enumerate}
      \item  the edge $I_a\tuh b$ is in $\cat{MAG}(\Af_{\Do(I_a)})$ for some $\Af\in [\Mf]_\Gs$ iff $a\tuh b$ is invisible in $\Mf$;
      \item the edge $I_a\tut b$ is in $\cat{MAG}(\Af_{\Do(I_a)})$ for some $\Af\in [\Mf]_\Gs$ iff $a\tut b$ is in $\Mf$; and
      \item the nodes $I_a$ and $b$ are adjacent in $\cat{MAG}(\Af_{\Do(I_a)})$ for some $\Af\in [\Mf]_\Gs$ iff $a\tuh b$ is invisible in $\Mf$ or $a\tut b$ is in $\Mf$.
    \end{enumerate}
\end{lemma}

\begin{proof}[Proof of \cref{lem:sint_mag_II}]
  \textbf{Step 0: preparatory work}. Let $b\in \Vc\sm \{a\}$ and suppose that there exists $\Af=(\Ic,\Oc,\Sc,\tilde{\Ec})\in [\Mf]_\Gs$ such that $I_a \sus b$ is in $\cat{MAG}(\Af_{\Do(I_a)})$. Then there is an $\Sc$-inducing path $\pi$ from $I_a$ to $b$ in $\Af_{\Do(I_a)}$. On the path $\pi$, all the colliders are in $\anc_{\Af_{\Do(I_a)}}(\{b\}\cup \Sc)$ (this implies that $a\in \anc_{\Af_{\Do(I_a)}}(\{b\}\cup \Sc)$), since
 $\anc_{\Af_{\Do(I_a)}}(I_a)=\{I_a\}$. Removing the first edge of the path $\pi$ gives an $\Sc$-inducing path from $a$ to $b$ in $\Af$. This implies that $a$ and $b$ are adjacent in $\Mf$. Assume that $a\notin \anc_{\Af}(\Sc)$.
 Then $a\in \anc_{\Af}(b)$. It implies that $b\notin \anc_{\Af}(\{a\}\cup \Sc)$, otherwise it violates the ancestral property of a MAG. Therefore, we have $a\tuh b$ in $\Mf$ and $I_a\tuh b$ in $\cat{MAG}(\Af_{\Do(I_a)})$. Assume now that $a\in \anc_{\Af}(\Sc)$. Then we have $a\tus b$ in $\Mf$ and $I_a\tus b$ in $\cat{MAG}(\Af_{\Do(I_a)})$. In this case, if we have $a\tuh b$ in $\Mf$ then we have $I_a\tuh b$ in $\cat{MAG}(\Af_{\Do(I_a)})$, but if we have $a\tut b$ in $\Mf$ then we have $I_a\tut b$ in $\cat{MAG}(\Af_{\Do(I_a)})$. Overall, this shows that if $I_a$ and $b$ are adjacent in $\cat{MAG}(\Af_{\Do(I_a)})$ for some $\Af\in [\Mf]_\Gs$, then the edge can only be of the type $I_a\tuh b$ or $I_a\tut b$. Furthermore, if $I_a$ and $b$ are adjacent in $\cat{MAG}(\Af_{\Do(I_a)})$, then nodes $a$ and $b$ must be adjacent in $\Mf$, and $I_a\tuh b$ in $\cat{MAG}(\Af_{\Do(I_a)})$ implies $a\tuh b$ in $\Mf$.

 \textbf{Step 1: show ``$\Leftarrow$'' of (1)}. We now show that if $a\tuh b$ is invisible in $\Mf$, then there exists $\Af=(\Ic,\Oc,\Sc,\tilde{\Ec})\in [\Mf]_\Gs$ such that $I_a\tuh b$ in $\cat{MAG}(\Af_{\Do(I_a)})$. Since $a\tuh b$ is invisible in $\Mf$, there exists an isADMG $\Af=(\Ic,\Oc,\Sc,\tilde{\Ec})$ that is represented by $\Mf$ such that $a\huh b$ is in $\Af$ by Lemma~\ref{lem:visible_edge_II}. Since $a\tuh b$ is in $\Mf$ and $\Mf$ represents $\Af$, we have that $a\in \anc_{\Af}(\{b\}\cup \Sc)$ and $b\notin \anc_{\Af}(\{a\}\cup \Sc)$. This also implies that the path $I_a\tuh a\huh b$ is an $\Sc$-inducing path from $I_a$ to $b$ in $\Af_{\Do(I_a)}$.
 Therefore, there must be a directed edge $I_a\tuh b$ in $\cat{MAG}(\Af_{\Do(I_a)})$.

 \textbf{Step 2: show ``$\Rightarrow$'' of (1)}. We next show that if there exists $\Af\in[\Mf]_\Gs$ such that
 $I_a\tuh b$ is in $\cat{MAG}(\Af_{\Do(I_a)})$, then the edge $a\tuh b$ must be present and invisible in $\Mf$. First note that by the first paragraph of the proof, we know that the edge $a\tuh b$ must be in $\Mf$.
 Let $\Af=(\Ic,\Oc,\Sc,\tilde{\Ec})\in[\Mf]_\Gs$ be such that $I_a\tuh b$ is in $\cat{MAG}(\Af_{\Do(I_a)})$.  If there exists a node $c$ such that $c\suh a\tuh b$ is in $\Mf$, then there are inducing paths 
 \[
 \begin{aligned}
     &\pi_1: c\suh v_1^1\sus \cdots \sus v_{n-1}^1 \suh a\\
     &\pi_2: I_a\tuh a\hus v_1^2\sus \cdots \sus v_{m-1}^2 \suh b.
 \end{aligned}
 \]
in $\Af_{\Do(I_a)}$ ($I_a\tuh b$ is in $\cat{MAG}(\Af_{\Do(I_a)})$), by Lemma~\ref{lem:inducing_path}. Deleting $I_a$ from $\pi_2$ and concatenating these two walks give a walk 
\[
\pi: c\suh v_1^1\sus \cdots \sus v_{n-1}^1 \suh a \hus v_1^2\sus \cdots \sus v_{m-1}^2 \suh b
\]
Since $\pi$ does not contain $I_a$, the walk $\pi$ is also in $\Af$. On $\pi$, all colliders are in $\anc_{\Af}(\{a,b,c\}\cup \Sc)$. Note that $a\in \anc_{\Af}(\{b\}\cup \Sc)$. Therefore, all colliders on $\pi$ are in $\anc_{\Af}(\{b,c\}\cup \Sc)$ and $\pi$ is an $\Sc$-inducing walk from $c$ to $b$ in $\Af$. This means that $c$ must be adjacent to $b$ in $\Mf$. Now assume that there exists a node $c$ such that the path $c\suh v_1\huh \cdots \huh v_{n-1}\huh a$ is in $\Mf$ with $v_1,\ldots,v_{n-1}\in \pa_{\Mf}(b)$ for some integer $n\geq 2$. Similarly as before, we have inducing paths 
\[
\begin{aligned}
    &\pi_1:c\suh v_1^1\sus \cdots \sus v_{n_1-1}^1 \suh v_1,\\ 
    &\pi_2:v_1\huh v_1^2\sus \cdots \sus v_{n_2-1}^2 \suh a, \text{ and }\\ 
    &\pi_3:I_a\tuh a\hus v_1^3\sus \cdots \sus v_{n_3-1}^3 \suh b
\end{aligned}
\]
in $\Af$. Concatenating these three paths and deleting $I_a$ gives a walk 
\[
\pi: c\suh v_1^1\sus \cdots \sus v_{n_1-1}^1 \suh v_1 \huh v_1^2\sus \cdots \sus v_{n_2-1}^2 \suh a \hus v_1^3\sus \cdots \sus v_{n_3-1}^3 \suh b.
\]
All colliders on $\pi$ are in $\anc_{\Af}(\{a,b,c,v_1\}\cup \Sc)$. Since $v_1\in \pa_{\Mf}(b)$, we know that $v_1\in \anc_{\Af}(\{b\}\cup \Sc)$. Also recall that $a\in \anc_{\Af}(\{b\}\cup \Sc)$. Thus, $\anc_{\Af}(\{a,b,c,v_1\}\cup \Sc)\subseteq \anc_{\Af}(\{b,c\}\cup \Sc)$. Hence, the walk $\pi$ is an $\Sc$-inducing walk from $c$ to $b$, which means that $c$ must be adjacent to $b$. Overall, the edge $a\tuh b$ must be invisible in $\Mf$.

 \textbf{Step 3: show ``$\Leftarrow$'' of (2)}. We show that if the edge $a\tut b$ is in $\Mf$ then there exists $\Af=(\Ic,\Oc,\Sc,\tilde{\Ec})\in [\Mf]_\Gs$ such that $I_a\tut b$ is in $\cat{MAG}(\Af_{\Do(I_a)})$. We can construct an isADMG $\Af=(\Ic,\Oc,\Sc,\tilde{\Ec})$ by
 \begin{itemize}
     \item replacing $a\tut b$ with $a\tuh b$ and $a\huh s_{ab}\hut b$ in $\Mf$,
     \item replacing all the other undirected edges $c\tut d$ by $c\tuh s_{cd}\hut d$ in $\Mf$,
     \item defining $\Sc\coloneqq \{s_{yz}: y\tut z \text{ in } \Mf \}/\sim $, where $s_{yz}\sim s_{zy}$.
 \end{itemize}
 Then $\Af\in [\Mf]_\Gs$, and the edge $I_a\tut b$ is in $\cat{MAG}(\Af_{\Do(I_a)})$, since 
 \[
 I_a\tuh a \huh s_{ab} \hut b
 \]
 is an $\Sc$-inducing path ($a,s_{ab}$ are colliders in $\anc_\Af(\Sc)$) and $I_a,b\in \anc_\Af(\Sc)$.

 \textbf{Step 4: show ``$\Rightarrow$'' of (2)}. Finally, we prove that if there exists $\Af\in[\Mf]_\Gs$ such that the edge $I_a\tut b$ is in $\cat{MAG}(\Af_{\Do(I_a)})$ then the edge $a\tut b$ is in $\Mf$. From the first paragraph of the proof, we know that we must have $a\tus b$ in $\Mf$, since we have $I_a\tut b$ in $\cat{MAG}(\Af_{\Do(I_a)})$. Note that $a\tus b$ must be $a\tut b$, otherwise we would have $I_a\tuh b$ in $\Mf$ by Step~0 of the proof, which is a contradiction.

 \textbf{Step 5: show (3)}. By Step~0 of the proof, we know that there can only be two types of edges $I_a\tuh b$ and $I_a\tut b$ between nodes $I_a$ and $b$. Given the first and second items of \cref{lem:sint_mag_II}, we can conclude that there exists $\Af\in [\Mf]_\Gs$ such that the nodes $I_a$ and $b$ are adjacent in $\cat{MAG}(\Af_{\Do(I_a)})$ iff $a\tuh b$ is invisible in $\Mf$ or $a\tut b$ is in $\Mf$.

\end{proof}

\begin{remark}
  \cref{lem:anc_selec_I,lem:sint_mag_I,lem:sint_mag_II} can be easily generalized to the case where the node $a$ is replaced by a subset of nodes $A\subseteq \Vc$.
\end{remark}

\begin{lemma}[Soft manipulation and MAG representation III]\label{lem:sint_mag_III}
    Let $\MAG$ be an iMAG. Then for every $\Af\in [\Mf]_\Gs$ there exists an iMAG $\tilde{\Mf}\in [\Mf_{\Do(I_D)}]_\Ms$ such that $\cat{MAG}(\Af_{\Do(I_D)})$ is a subgraph of $ \tilde{\Mf}$.
\end{lemma}

\begin{proof}[Proof of \cref{lem:sint_mag_III}]
    Let $\Af\in [\Mf]_\Gs$. We orient all the circles of $I_d\tuo d$ for $d\in D$ in $\Mf_{\Do(I_D)}$ by tails if $d\in \anc_{\Af}(\Sc_{\Af})$ and by arrowheads if $d\notin \anc_{\Af}(\Sc_{\Af})$. This gives us an iMAG $\tilde{\Mf}\in [\Mf_{\Do(I_D)}]_\Ms$. 
    By Lemma~\ref{lem:sint_mag_I}, the subgraph $\tilde{\Mf}_{\Ic\cup \Vc}$ is equal to the subgraph $\cat{MAG}(\Af_{\Do(I_D)})_{\Ic\cup \Vc}$. By \cref{def:int_mag,lem:sint_mag_I,lem:sint_mag_II}, for $d\in D$ and $b\in \Vc$ we have that if $I_d\tuh b$ is in $\cat{MAG}(\Af_{\Do(I_d)})$ then $I_d\tuh b $ is in $\tilde{\Mf}$ and if $I_d\tut b$ is in $\cat{MAG}(\Af_{\Do(I_d)})$ then $I_d\tut b $ is in $\tilde{\Mf}$. Overall, this implies that $\cat{MAG}(\Af_{\Do(I_D)})$ is a subgraph of $\tilde{\Mf}$.
\end{proof}

\begin{lemma}[Separations in hard-manipulated MAGs]\label{lem:sep_hint_mag}
    Let $\MAG$ be an iMAG and $T\subseteq \Ic\cup \Vc$. Let $A\subseteq \Vc\sm T$ and $B,C\subseteq (\Ic\cup \Vc)\sm T $ be pairwise disjoint. Then we have
  \[
      A \sep{\mathsf{id}}{\Mf_{\Do(T)}} B\mid C\cup T \quad \Longleftrightarrow \quad \forall \Af\in[\Mf]_\Gs: \  A \sep{\mathsf{id}}{\Af_{\Do(T)}} B\mid C\cup T \cup \Sc_{\Af}.
  \]
\end{lemma}

\begin{proof}[Proof of \cref{lem:sep_hint_mag}]

    Define an isADMG $\Af\coloneqq \cat{isADMG}(\Mf)\in [\Mf]_\Gs$. Let 
    \[
    \pi:A \ni v_0\sus \cdots \sus v_{n}\in B\cup \Ic\cup T
    \] 
    be a shortest open path from $A$ to $B\cup \Ic\cup T$ given $C\cup T$ in $\Mf_{\Do(T)}$ and $\pi_*$ be the corresponding path in $\Af$ where we replace all undirected edges $v_i\tut v_{i+1}$ on $\pi$ in $\Mf_{\Do(T)}$ with $v_i\tuh s_{v_{i}v_{i+1}} \hut v_{i+1}$  in $\Af$. Note that $\pi_*$ is indeed well-defined in $\Af_{\Do(T)}$, since $\pi$ does not contain nodes from $T$. Also, it is easy to see that $\pi_*$ is open given $C\cup T\cup \Af_\Sc$ in $\Af_{\Do(T)}$. This implies 
     \[
        A \nsep{\mathsf{id}}{\Mf_{\Do(T)}} B\mid C\cup T \quad \Longrightarrow \quad \exists \Af\in[\Mf]_\Gs, \  A \nsep{\mathsf{id}}{\Af_{\Do(T)}} B\mid C\cup T\cup \Sc_{\Af}.
    \]

    For the other direction, let $\Af\in [\Mf]_\Gs$ and let $\pi$ be an open path from $A$ to $B\cup \Ic\cup T$ given $C\cup T\cup \Sc_\Af$ in $\Af_{\Do(T)}$. The goal is to find an open path from $A$ to $B\cup \Ic\cup T$ given $C\cup T$ in $\Mf_{\Do(T)}$.

    \textbf{Step 0: preparatory work}. Without loss of generality, we can assume that for every $t\in T$, there is at least one arrowhead pointing to node $t$ in $\Mf$. To see this, first note that $\pi$ is also open given $C\cup T\cup \Sc_\Af$ in $\Af_{\Do(T\sm \{t\})}$. Assume that node $t$ does not have any arrowheads towards it in $\Mf$. Furthermore, assume that we can find an open path $\ol{\pi}$ from $A$ to $B\cup \Ic\cup {T\sm \{t\}}$ given $C\cup T\cup \Sc_\Af$ in $\Mf_{\Do(T\sm \{t\})}$ provided that $\pi$ is open given $C\cup\{t\}\cup (T\sm \{t\})\cup \Sc_\Af$ in $\Af_{\Do(T\sm\{t\})}$. Since no arrowheads towards $t$, path $\ol{\pi}$ cannot contain $t$ and there are no colliders on $\ol{\pi}$ having a directed path to $C\cup T\cup \Sc_{\Af}$ across node $t$. Therefore, path $\ol{\pi}$ is still present and open from $A$ to $B\cup \Ic\cup T$ given $C\cup T\cup \Sc_{\Af}$ in $\Mf_{\Do(T)}$. 

    By Proposition~\ref{prop:sep_mag}, there is a shortest open path 
    \[
    \tilde{\pi}:A \ni v_0\sus \cdots \sus v_{n}\in B\cup \Ic\cup T
    \] 
    from $A$ to $B\cup \Ic\cup T$ given $C\cup T$ in $\tilde{\Mf}\coloneqq\cat{MAG}(\Af_{\Do(T)})$. The idea is to find the desired open path in $\Mf_{\Do(T)}$ via the path $\tilde{\pi}$ in $\tilde{\Mf}$. For that purpose, we define 
    \[
    \pi^{\Do(T)}:v_0\sus \cdots \sus v_{n}
    \] 
    to be the path in $\Mf_{\Do(T)}$ consisting of the same sequence of nodes as $\tilde{\pi}$ does. The path $\pi^{\Do(T)}$ is well-defined. Indeed, since an inducing path in $\Af_{\Do(T)}$ must be present in $\Af$ and $\tilde{\pi}$ does not contain nodes in $T$, nodes $v_i$ and $v_{i+1}$ are adjacent in $\Mf_{\Do(T)}$. Note that the path $\pi^{\Do(T)}$ need not have the desired properties in $\Mf_{\Do(T)}$. We will therefore perform a suitable “surgery’’ on it and show that the resulting path achieves our goal. Before doing so, it is helpful to make the connection between $\Mf_{\Do(T)}$ and $\tilde{\Mf}$ precise.
    
    \textbf{Step 1: connect $\Mf_{\Do(T)}$ to $\tilde{\Mf}$}. For $a,b\in \Ic\cup \Vc$, if $a\in \anc_{\Af_{\Do(T)}}(\{b\}\cup \Sc)$, then $a\in \anc_{\Af}(\{b\}\cup \Sc)$. Therefore, if nodes $a$ and $b$ are adjacent in $\tilde{\Mf}$ with a tail on node $a$, then nodes $a$ and $b$ are adjacent in $\Mf_{\Do(T)}$ with a tail on node $a$.  
    
    We observe that if we have $a\tuh b$ or $a\huh b$ in $\tilde{\Mf}$ then it is impossible to have $a\tut b$ in $\Mf_{\Do(T)}$. Assume on the contrary that we have $a\tut b$ in $\Mf_{\Do(T)}$ and therefore in $\Mf$. This implies that there is a directed path from $b$ to $\{a\}\cup \Sc_\Af$ in $\Af$ and all such directed paths must intersect $T$, i.e., they must be of the form with $t\in T$
    \[
    \pf:b\tuh \cdots\tuh t \tuh \cdots \tuh d\in\{a\}\cup \Sc_\Af.
    \] 
    Since for every $t\in T$ there is at least one arrowhead toward $t$ in $\Mf$ by our assumption, Lemma~\ref{lem:anc_selec_I} implies that $T\cap \anc_{\Af}(\Sc_\Af)=\emptyset$ and therefore $b\notin \anc_{\Af}(\Sc_\Af)$. However, since we have $a\tut b$ in $\Mf$, Lemma~\ref{lem:anc_selec_II} implies that $b\in \anc_{\Af}(\Sc_\Af)$. This causes a contradiction and therefore the starting claim is correct. For convenience, if we have $a\tuh b$ in $\Mf_{\Do(T)}$ and $a\huh b$ in $\tilde{\Mf}$ then we call such directed edge $a\tuh b$ in $\Mf_{\Do(T)}$ \emph{ghost} \wrt $\tilde{\Mf}$.

    We now show that the ghost directed edge between $a$ and $b$ must be invisible in $\Mf_{\Do(T)}$. For a proof, we assume without loss of generality that $a\huh b$ is in $\tilde{\Mf}$ and $a\tuh b$ is in $\Mf_{\Do(T)}$. Assume the contrary that $a\tuh b$ is visible in $\Mf_{\Do(T)}$. Then there exists a node $d$ non-adjacent to node $b$ such that 
    \[
    d\suh a \tuh b \quad \text{ or } \quad 
    d\suh u_1\huh \cdots \huh u_{m-1} \huh a\tuh b
    \]
    with $u_1,\ldots, u_{m-1}\in \pa_{\Mf_{\Do(T)}}(b)$ in $\Mf_{\Do(T)}$. It is easy to see that these patterns are present in $\Mf$ if they are present in $\Mf_{\Do(T)}$. Also note that if node $d$ is non-adjacent to $b$ in $\Mf_{\Do(T)}$ then it is non-adjacent to $b$ in $\Mf$. Therefore, the directed edge $a\tuh b$ is visible in $\Mf$. On the other hand, since $a\huh b$ is in $\tilde{\Mf}$, there is an inducing path from $a$ to $b$ that is into $a$ in $\Af_{\Do(T)}$ by Lemma~\ref{lem:visible_edge_I}. This inducing path is also present in $\Af$, which implies that the directed edge $a\tuh b$ is invisible in $\Mf$. This is a contradiction, so the initial hypothesis is wrong and we have proven the claim.

    Given the above discussions, we summarize the connections between $\tilde{\Mf}$ and $\Mf_{\Do(T)}$: 
    \begin{enumerate}[label=(\roman*)]
        \item the skeleton of $\tilde{\Mf}$ is contained in that of $\Mf_{\Do(T)}$;
        \item $a\tut b$ in $\tilde{\Mf}$ implies $a\tut b$ in $\Mf_{\Do(T)}$;
        \item $a\tuh b$ in $\tilde{\Mf}$ implies $a\tuh b$ in $\Mf_{\Do(T)}$;
        \item $a\huh b$ in $\tilde{\Mf}$ implies either $a\huh b$, or invisible $a\tuh b$, or invisible $a\hut b$ in $\Mf_{\Do(T)}$.
    \end{enumerate}

    \textbf{Step 2: decompose $\pi^{\Do(T)}$ into ``good'' part and ``bad'' part}. Observe that we can decompose the path $\tilde{\pi}$ into several subpaths $\tilde{\pi}_1,\ldots, \tilde{\pi}_\ell$ such that 
    \[
    \tilde{\pi}=\tilde{\pi}_1\oplus\cdots\oplus \tilde{\pi}_\ell
    \]
    where $\tilde{\pi}_k: v^k_0\sus \cdots \sus v^k_{n_k}$ for $1\leq k\leq \ell$ are the maximal connected components of $\tilde{\pi}$ consisting of either only directed and bidirected edges or only undirected edges, i.e., either 
    \begin{enumerate}[label=(\roman*)]
        \item (good part) $v^k_0\tut v_{1}^k\tut \cdots \tut v_{n_k-1}^k\tut v^k_{n_k}$ and $v^{k-1}_{n_{k-1}-1}\hut v^k_0$ and $v_{n_k}^k\tuh  v^{k+1}_{1}$ (if any), or

        \item (bad part)  $v^k_0\sus \cdots \sus  v^k_{n_k}$ does not contain any undirected edges $v_{i}^k\tut v_{i+1}^k$ but we have $v^{k-1}_{n_{k-1}-1}\tut v^k_0$ and $v^k_{n_k} \tut v^{k+1}_{1}$ (if any).
    \end{enumerate}  
    From item (ii) of the connections between $\tilde{\Mf}$ and $\Mf_{\Do(T)}$, we know that if $\tilde{\pi}_k$ consists of only undirected edges in $\tilde{\Mf}$, i.e., belongs to the good part, then the corresponding subpath $\pi^{\Do(T)}_k$ of $\pi^{\Do(T)}$ is intact in $\Mf_{\Do(T)}$. On the contrary, if $\tilde{\pi}_k$ consists of directed and bidirected edges in $\tilde{\Mf}$, i.e., belongs to the bad part, then the corresponding path $\pi^{\Do(T)}_k$ may not be open from $v^k_{0}$ to $v^k_{n_k}$ given $C\cup T$ in $\Mf_{\Do(T)}$ since some collider of the form $v^k_{i-1}\huh v^k_{i} \huh v^k_{i+1}$ on $\tilde{\pi}_k$ could become a non-collider on $\pi^{\Do(T)}_k$, for example $v^k_{i-1}\hut v^k_{i} \huh v^k_{i+1}$ on $\pi^{\Do(T)}_k$. Observe that if we can find an open path $\breve{\pi}_k$ from $v^k_{0}$ to $v^k_{n_k}$ in $\Mf_{\Do(T)}$ then the walk $\pi^{\Do(T)}_{k-1}\oplus \breve{\pi}_k\oplus \pi^{\Do(T)}_{k+1}$ ($\pi^{\Do(T)}_{k-1}$ or $\pi^{\Do(T)}_{k+1}$ could be empty) is open from $v^{k-1}_{0}$ to $v^{k+1}_{n_{k+1}}$ given $C\cup T$ in $\Mf_{\Do(T)}$, since $v^k_{0}$ and $v^k_{n_k}$ are both non-colliders on $\tilde{\pi}$ in $\tilde{\Mf}$ and on $\pi^{\Do(T)}$ in $\Mf_{\Do(T)}$. 

    \textbf{Step 3: perform surgery to the bad part}. Let $\{\tilde{\pi}_{k_i}\}_{i=1}^{l}$ be the subcollection of subpaths $\{\tilde{\pi}_j\}_{j=1}^\ell$ such that all consist of only directed and bidirected edges, i.e., the collection of the bad parts. Define\footnote{Note that even in the case without exogenous input nodes, $\mathrm{Und}_{\tilde{\Mf}}$ is different from the  undirected component of $\tilde{\Mf}$ defined in \cite[Section~3.2]{richardson2002ancestral} in general.}
    \[
    \begin{aligned}
        \mathrm{Und}_{\tilde{\Mf}}\coloneqq &\ \{a\in \Vc: \exists b\in \Ic\cup \Vc \text{ s.t.\ } a\tut b \text{ is in } \tilde{\Mf}\}, \text{ and } \\ H\coloneqq &\ (\Vc\sm \mathrm{Und}_{\tilde{\Mf}})\cup \{v_0^{k_1},v_{n_{k_1}}^{k_1},\ldots, v_0^{k_l},v_{n_{k_l}}^{k_l}\}.
    \end{aligned}
    \]
    Note that there are no undirected edges between nodes in $\{v_0^{k_1},v_{n_{k_1}}^{k_1},\ldots, v_0^{k_l},v_{n_{k_l}}^{k_l}\}$. Otherwise it contradicts the assumption that $\tilde{\pi}$ is a shortest open path. Then the endogenized subgraph $\hat{\Mf}\coloneqq (\tilde{\Mf}_{H})^*$ is a MAG without undirected edges and exogenous input nodes \cite[p.985]{richardson2002ancestral}, and $\tilde{\pi}_{k_i}$, for all $i=1,\ldots,l$, are open given $C\cup T$ in $\hat{\Mf}$. Define 
    \[
    \breve{\Mf}\coloneqq (\Mf_{\Do(T)})_H^*\sm \{a\tut b: a,b\in H, a\tut b \text{ in } \Mf_{\Do(T)}\}.
    \]  
    By \cite[Lemma~12]{zhang2008causal}, there exist open paths $\breve{\pi}_{k_i}$ from $v_0^{k_i}$ to $v_{n_{k_i}}^{k_i}$ given $C\cup T$ in $\breve{\Mf}$ for all $i=1,\ldots, l$, and therefore in $\Mf_{\Do(T)}$. By the previous observation, we can concatenate all the paths $\{\tilde{\pi}_j\}_{j=1}^\ell\sm \{\tilde{\pi}_{k_i}\}_{i=1}^l$ and $\{\breve{\pi}_{k_i}\}_{i=1}^l$ to get an open walk from $A$ to $B\cup \Ic\cup T$ given $C\cup T$ in $\Mf_{\Do(T)}$.  Hence, we have shown 
     \[
        \exists \Af\in[\Mf]_\Gs, \  A \nsep{\mathsf{id}}{\Af_{\Do(T)}} B\mid C\cup T\cup \Sc_{\Af} \quad \Longrightarrow \quad A \nsep{\mathsf{id}}{\Mf_{\Do(T)}} B\mid C\cup T .
    \]
   
   Overall, by contraposition we have 
    \[
        A \sep{\mathsf{id}}{\Mf_{\Do(T)}} B\mid C\cup T \quad \Longleftrightarrow \quad \forall \Af\in[\Mf]_\Gs, \  A \sep{\mathsf{id}}{\Af_{\Do(T)}} B\mid C\cup T\cup  \Sc_{\Af}.
    \]
\end{proof}

\subsection[Proof of Theorem~\ref{thm:sep_hsint_pag_mag}]{Proof of \cref{thm:sep_hsint_pag_mag}}

We present the proof of \cref{thm:sep_hsint_pag_mag}. See \cref{fig:pf_structure_sep_hsint_pag_mag} for an overall structure of the proof.

\begin{proof}[Proof of \cref{thm:sep_hsint_pag_mag}]
From \cref{lem:4}, we have
\[
    \Big(\exists \Mf\in[\Pf]_\Ms, \  A \nsep{\mathsf{id}}{\Mf_{\Do(I_D,T)}} B\mid C \cup T\Big) \quad \Longrightarrow \quad A \nsep{\mathsf{id}}{\Pf_{\Do(I_D,T)}} B\mid C \cup T,
\]
which, by contraposition, implies 
\[
    A \sep{\mathsf{id}}{\Pf_{\Do(I_D,T)}} B\mid C\cup T \quad \Longrightarrow \quad \forall \Mf\in[\Pf]_\Ms, \  A \sep{\mathsf{id}}{\Mf_{\Do(I_D,T)}} B\mid C\cup T.
\]

We now prove the converse under the additional assumption that $\Pf$ is an iCOPAG. It suffices to show
 \[
     A \nsep{\mathsf{id}}{\Pf_{\Do(I_D,T)}} B\mid C\cup T \quad \Longrightarrow \quad \exists \Mf\in[\Pf]_\Ms, \  A \nsep{\mathsf{id}}{\Mf_{\Do(I_D,T)}} B\mid C\cup T.
  \]
For that, let \[\pi:v_0\sus v_1\sus \cdots \sus v_n\] be an irreducible open path from $A$ to $B\cup \Ic \cup  \{I_d\}_{d\in D}\cup T$ given $C\cup T$ with $n\geq 1$ (the case where $n=0$ is trivial) in $\Pf_{\Do(I_D,T)}$.  Note that $\pi$ does not contain nodes from $T$. There are four cases in total:

\begin{enumerate}[label=\textbf{Case~\arabic*:}, ref=Case~\arabic*, leftmargin=*]
    \item $v_n\in (B\cap \Vc)\cup \Ic$;
    \item $v_{n-1}\tut v_n=I_d$;
    \item $v_{n-1}\hut v_n=I_d$;
    \item $v_{n-1}\out v_n=I_d$.
\end{enumerate}

\textbf{Case~1}. If $v_n\in (B\cap \Vc)\cup \Ic$, then it is obvious that $\pi_*$ is an open path from $A$ to $B\cup \Ic \cup \{I_d\}_{d\in D}\cup T$ given $C\cup T$ in $\Mf_{\Do(I_D,T)}$ for every $\Mf\in [\Pf]_\Ms$, where $\pi_*$ is the corresponding path in $\Mf_{\Do(I_D,T)}$ consisting of the same sequence of nodes as $\pi$ does in $\Pf_{\Do(I_D,T)}$ ($\pi_*$ is obviously well-defined in $\Mf_{\Do(I_D,T)}$). 

\textbf{Case~2}. Now we assume that $v_n\in \{I_d\}_{d\in D}$ and $v_n=I_d$ for some $d\in D$. If we have $v_{n-1}\tut I_d$ in $\Pf_{\Do(I_D,T)}$, then the same edge must be in $\Mf_{\Do(I_D,T)}$ for every $\Mf\in [\Pf]_\Ms$.  So,  $\pi_*$ is well-defined and open given $C\cup T$ in $\Mf_{\Do(I_D,T)}$ for every $\Mf\in [\Pf]_\Ms$. 

\textbf{Case~3}. We can divide Case~3 into two subcases: 
\begin{enumerate}[label=\textbf{Case~3.\arabic*:}, ref=Case~\arabic*, leftmargin=*]
    \item $d=v_{n-1}$, and
    \item $d\ne v_{n-1}$. 
\end{enumerate}

\textbf{Case~3.1}. If we have $d=v_{n-1}\hut I_d$, then we have $v_{n-1}\hut I_d$ for every $\Mf \in [\Pf]_\Ms$. If $v_{n-1}$ is a non-collider on $\pi$, then we are done. If $v_{n-1}$ is a collider, then note that we can always find an iMAG $\Mf\in [\Pf]$ such that the potentially directed path from $v_{n-1}$ to $C$ in $\Pf_{\Do(I_D,T)}$ becomes directed in $\Mf_{\Do(I_D,T)}$. 

\textbf{Case~3.2}. We consider the case where $v_{n-1}\hut I_d$ is in $\Pf_{\Do(I_D,T)}$ with $v_{n-1}\ne d$. \cref{lem:orientation_invisible} allows us to construct an iMAG $\Mf \in [\Pf]_\Ms$ such that $d\tuh v_{n-1}$ is present and invisible in $\Mf$ and furthermore if there is a potentially directed path from $v_{n-1}$ to $C$ in $\Pf_{\Do(I_D,T)}$ then it becomes a directed path in $\Mf_{\Do(I_D,T)}$. So, $\pi_*$ is well-defined and open given $C\cup T$ in $\Mf_{\Do(I_D,T)}$. 

\textbf{Case~4}. Finally, we are left with the case where $v_{n-1}\out v_n=I_d$ in $\Pf_{\Do(I_D,T)}$. Note that the patterns $v_{n-2} \suo d=v_{n-1}\out v_n=I_d$ (if $n\ge 2$) with $v_{n-2}$ non-adjacent to $v_n$, and $v_{n-2} \suh d=v_{n-1}\out v_n=I_d$ cannot occur by \cref{def:int_pag}. If $v_{n-2}\sut d=v_{n-1}\out v_n$ then $\pi_*$ is obviously well-defined and open given $C\cup T$ in $\Mf_{\Do(I_D,T)}$ for every $\Mf\in [\Pf]_\Ms$. Therefore, it suffices to consider the case where $d\ne v_{n-1}$. The remaining possibilities are $d \out v_{n-1}$ or $d\ouo v_{n-1}$ or $d\tuo v_{n-1}$ if $v_{n-1}\ne d$ and $v_{n-1}\out v_n=I_d$ in $\Pf_{\Do(I_D,T)}$. We therefore have two subcases: 
\begin{enumerate}[label=\textbf{Case~4.\arabic*:}, ref=Case~\arabic*, leftmargin=*]
    \item $d \out v_{n-1}$ or $d\tuo v_{n-1}$, and
    \item  $d\ouo v_{n-1}$. 
\end{enumerate}

\textbf{Case~4.1}.  If we have $d\out v_{n-1}$ or $d\tuo v_{n-1}$, then by \ref{sopag:p2} there are no arrowheads towards $v_{n-1}$ and therefore $v_{n-1}$ must be a non-collider on $\pi$ in $\Pf_{\Do(I_D,T)}$. By \ref{sopag:p4}, we can construct an iMAG $\Mf\in [\Pf]_\Ms$ such that $d\tut v_{n-1}$ is present in $\Mf_{\Do(I_D,T)}$. Hence, $\pi_*$ is  well-defined and open given $C\cup T$ in $\Mf_{\Do(I_D,T)}$. 

\textbf{Case~4.2}. In the next, we assume $d\ouo v_{n-1}$. The case where $n\le 1$ is easy to see. We therefore assume $n\ge 2$. We can divide this case into two subcases: 
\begin{enumerate}[label=\textbf{Case~4.2.\arabic*:}, ref=Case~\arabic*, leftmargin=*]
    \item $v_{n-1}$ is a non-collider on $\pi$ in $\Pf_{\Do(I_D,T)}$, and
    \item  $v_{n-1}$ is a collider on $\pi$ in $\Pf_{\Do(I_D,T)}$.
\end{enumerate}

\textbf{Case~4.2.1}. Assume $v_{n-1}$ is a non-collider on $\pi$ in $\Pf_{\Do(I_D,T)}$. If there are no arrowheads towards $v_{n-1}$ in $\Pf$, then by \ref{sopag:p4} there exists $\Mf \in [\Pf]_\Ms$ in which $v_{n-1}\tut d$ is present, and $\pi_*$ is well-defined and open given $C\cup T$ in $\Mf_{\Do(I_D,T)}$. Assume now that there are arrowheads towards $v_{n-1}$ in $\Pf$.  If $v_{n-2} \sut v_{n-1}$ on $\pi$, then by \ref{sopag:p4} and \cref{lem:orientation_invisible} there exists $\Mf \in [\Pf]_\Ms$ in which $d\tuh v_{n-1}$ is present and invisible and therefore we are done. If $v_{n-2}\suo v_{n-1}\out v_n=I_d$ with $v_{n-2}$ non-adjacent to $v_n$, the only possibilities are $v_{n-2} \huo v_{n-1}$ or $v_{n-2} \ouo v_{n-1}$ on $\pi$ since $v_{n-2} \tuo v_{n-1}$ is excluded by the fact that we have arrowheads towards $v_{n-1}$. Again by \ref{sopag:p4} and \cref{lem:orientation_invisible}, there exists $\Mf\in [\Pf]_\Ms$  such that $d\tuh v_{n-1}$ is invisible and $v_{n-2}\hut v_{n-1}$ is in $\Mf_{\Do(I_D,T)}$. This establishes the openness of $\pi_*$ given $C\cup T$ in $\Mf_{\Do(I_D,T)}$. 

\textbf{Case~4.2.2}. The final case is that $v_{n-1}$ is a collider and $v_{n-2}\suh v_{n-1}\out v_n=I_d$ on $\pi$. If $v_{n-1}\in C$, then by \ref{sopag:p4} and \cref{lem:orientation_invisible}, there is $\Mf\in [\Pf]_\Ms$ such that $d\tuh v_{n-1}$ is invisible in $\Mf$ and $\pi_*$ is well-defined and open given $C\cup T$ in $\Mf_{\Do(I_D,T)}$. In the following, we assume that $v_{n-1}\notin C$. Since $\pi$ is open in $\Pf_{\Do(I_D,T)}$ given $C\cup T$, there is a shortest non-trivial potentially directed path $\pf$ from $v_{n-1}$ to $C$. If $d$ is not on $\pf$, then by \ref{sopag:p4} and \cref{lem:orientation_invisible}, we can find one $\Mf\in [\Pf]_\Ms$ such that $d\tuh v_{n-1}$ invisible and the corresponding path $\pf_*$ of $\pf$ becomes a directed path in $\Mf_{\Do(I_D,T)}$. Then $\pi_*$ is well-defined and open given $C\cup T$ in $\Mf_{\Do(I_D,T)}$. If $d$ lies on $\pf$, choose $\Mf \in [\Pf]_\Ms$ so that the corresponding path $\pf_*$ becomes directed in $\Mf_{\Do(I_D,T)}$. In particular, this yields $v_{n-1}\tuh d$.  We define in $\Mf_{\Do(I_D,T)}$
\[
\tilde{\pi}\coloneqq\begin{cases}
     \pi_*(v_0,d)\oplus (d,I_d) &\quad \text{ if $d$ is on $\pi$; }\\
     \pi_*(v_0,v_{n-1})\oplus (v_{n-1},d)\oplus (d,I_d) &\quad \text{ otherwise}.
\end{cases}
\] 
Consider the case where $d\in C$. If $d$ is on $\pi$, then $v_{i-1}\suh d \hus v_{i+1}$ on $\pi$ and therefore $v_{i-1}\suh d\hut I_d$ on $\tilde{\pi}$. If $d$ is not on $\pi$, then $v_{n-2}\suh v_{n-1}\tuh d  \hut I_d$ on $\tilde{\pi}$ (note that $v_{n-1}\notin C$ and $d\in C$). Hence, in both of the two cases, $\tilde{\pi}$ is open given $C\cup T$ in $\Mf_{\Do(I_D,T)}$. Now we assume that $d\notin C$. Since $d$ is an ancestor of $C$ and $d\notin C$ in $\Mf_{\Do(I_D,T)}$, node $d$ must be open on $\tilde{\pi}$ no matter whether $d$ is a collider on $\tilde{\pi}$ or not. So, $\tilde{\pi}$ is open given $C\cup T$ in $\Mf_{\Do(I_D,T)}$.

This finishes the proof together with \cref{prop:sep_mag}.
\end{proof}

\begin{lemma}[Properties of soft-manipulated iSOPAGs]\label{lem:soft_mani_SOPAG}
    Let $\PAG$ be an iSOPAG representing an isADMG $\Af$. Let $A\subseteq  \Vc$. Then $\Pf_{\Do(I_{A})}$ is closed under $\fci{1}$ and $\fci{4}$.
   
\end{lemma}

\begin{proof}[Proof of \cref{lem:soft_mani_SOPAG}]

\textbf{Check $\fci{1}$}. Assume $i\suo j$. If $j\in A$ and $k=I_j$, then by \cref{def:int_pag} node $k$ is adjacent to $i$ and therefore $\fci{1}$ is not applicable. We now consider the case where $k=I_a$ for some $a\in A$ and $j\ne a$. By \cref{def:int_pag}, the only possibilities for having $k\tuh j$ are $a\ouh j$ or $a\tuh j$ invisible. \cref{lem:head_circle} yields $a\ouh i$ or $a\tuh i$. \cref{lem:property_visible} implies that $a\tuh i$ is invisible if it is present. Therefore, by \cref{def:int_pag} node $k$ is adjacent to $i$ and therefore $\fci{1}$ is not applicable.

\textbf{Check $\fci{4}$}. By \cref{def:discriminating_path}, we only need to consider the case where $k=I_a$ for some $a \in A$. We want to show that $q_n \tuh i$ must be invisible. If this is proved, then by \cref{def:int_pag} node $k$ must be adjacent to $i$ and therefore $\fci{4}$ is not applicable. Assume on the contrary that $q_n\tuh i$ is visible. Then we have 
\[
c\suh q_n \tuh i \quad\text{ or }\quad c\suh v_1\huh \cdots \huh v_{m-1}\huh q_n \tuh i
\] 
with $c$ non-adjacent to $i$ and all $v_\ell$ parents of $i$. In the above two cases, we have 
\[
c\suh q_n \huh \cdots \huh q_1\hus j\ous i \text{ or } c\suh v_1 \huh \cdots \huh v_{m-1} \huh q_n \huh \cdots \huh q_1\hus j\ous i,
\]
both of which form a discriminating path for $j$. Since $\Pf$ is already sufficiently oriented, $\fci{4}$ would have oriented the circle mark on the edge $j\ous i$. This shows that $q_n\tuh i$ is invisible. Therefore, $\fci{4}$ cannot be applied further in $\Pf_{\Do(I_A)}$.

\end{proof}

\Cref{lem:0,lem:1,lem:2,lem:3} are generalizations of corresponding results in \cite[Lemma~5.1.7]{zhang2006causal} in the sense that we work with soft-manipulated iSOPAGs rather than CPAGs without selection bias. With some essential modifications, the original idea still works in this more general setting,  given the appropriate definitions of soft manipulation (\cref{def:int_pag}) and graph separations (\cref{def:sep_PAG}) and the observation in \cref{lem:soft_mani_SOPAG}.

\begin{lemma}[Discriminating path in MAGs and PAGs]\label{lem:0}
Let $\MAG$ be an iMAG and $\Pf$ be an iSOPAG such that $\Mf\in [\Pf]_\Ms$. Let $D\subseteq \Vc$. Assume $a\in \Ic\cup \Vc$ or $a\in \{I_d\}_{d\in D}$ is non-adjacent to $z$ in $\Pf_{\Do(I_D)}$. If a path 
\[
a\sus v_1 \sus \cdots \sus v_n\sus y\sus z
\] 
is a discriminating path for $y$ in $\Mf_{\Do(I_D)}$ and the corresponding subpath from node $a$ to node $y$ in $\Pf_{\Do(I_D)}$ is a well-defined collider path, then the path is also a discriminating path for $y$ in $\Pf_{\Do(I_D)}$.
\end{lemma}

\begin{proof}[Proof of \cref{lem:0}]
    Under the assumption of \cref{lem:0}, node $a$ is non-adjacent to $z$. Note that it is assumed that the subpath from $a$ to $y$ in $\Pf_{\Do(I_D)}$ is a collider path. Therefore, to show the target result, it suffices to show $v_1,\ldots,v_n\in\pa_{\Pf_{\Do(I_D)}}(z)$. We argue by induction. 
    
    We first show the base case, i.e., $v_1\in \pa_{\Pf_{\Do(I_D)}}(z)$. Since $\Mf\in[\Pf]_\Ms$, node $v_1$ must be adjacent to $z$ and we cannot have  $v_1 \hus z$ in $\Pf_{\Do(I_D)}$. The pattern $a\suh v_1 \tut z$ cannot occur in $\Mf_{\Do(I_D)}$ by definition. Also, the pattern $a\suh v_1 \tuo z$ cannot occur in $\Pf_{\Do(I_D)}$ by \cref{def:SOPAG}.  The case $a\suh v_1 \ous z$ cannot happen, since it contradicts $\fci{1}$ by \cref{lem:soft_mani_SOPAG} and the fact that 
    $a$ and $z$ are not adjacent in $\Pf_{\Do(I_D)}$. Thus, it must be $a\suh v_1 \tuh z$.

    Now we assume that $v_{1},\ldots,v_{k-1}\in \pa_{\Pf_{\Do(I_D)}}(z)$. Then there is a discriminating path from node $a$ to $z$ for $v_k$ in $\Pf_{\Do(I_D)}$. Note that it is impossible to have $v_k \huh z$ in $\Pf$. Therefore, we have $v_k\tuh z$ by \cref{lem:soft_mani_SOPAG} and $\fci{4}$. This finishes the proof.
\end{proof}

\begin{remark}
    In \cref{def:int_pag}, one might be tempted to replace $I_a \tuo b$ with $I_a\tuh b$ when $a\ouo b$ is present and there exists $c\suh b$. However, this definition does not guarantee the validity of \cref{lem:soft_mani_SOPAG}. Therefore, \cref{lem:0} may not hold in this case.
\end{remark}

\begin{lemma}\label{lem:1}
Let $\MAG$ be an iMAG and let $\Pf$ be an iSOPAG such that  $\Mf\in [\Pf]_\Ms$. Let $D\subseteq \Vc$. Let $A\subseteq\Vc,\ B\subseteq \Ic\cup \Vc$ and $C\subseteq \Ic\cup \{I_d\}_{d\in D}\cup \Vc$ be pairwise disjoint.
Let \[\pi:a=v_0\sus v_1\sus \cdots \sus v_{n-1} \sus v_{n}=b\] be an irreducible non-trivial open path from $A$ to $B\cup \Ic\cup \{I_d\}_{d\in D}$ given $C$ in $\Mf_{\Do(I_D)}$. Let $\pi^*$ be the corresponding path in $\Pf_{\Do(I_D)}$ consisting of the same sequence of nodes as $\pi$. It holds that for every $1\leq i\leq n-1$, if $v_i$ is not of a definite status on $\pi^*$ and $v_{i+1}\notin \{I_d\}_{d\in D}$, then $v_{i+1}\in \pa_{\Mf_{\Do(I_D)}}(v_{i-1})$ and $v_{i+1}$ is a collider on $\pi$ in $\Mf_{\Do(I_D)}$.
\end{lemma}

\begin{proof}[Proof of \cref{lem:1}]
    First, note that $\pi^*$ is well-defined by \cref{def:int_mag,def:int_pag}. Second, note that none of the non-endnodes $v_i$ for $1\leq i\leq n-1$ can be in $\Ic\cup \{I_d\}_{d\in D}$. Indeed, if $v_{i}\in \Ic\cup \{I_d\}_{d\in D}$, then the subpath $\pi(a,v_{i})$ of $\pi$ is open from $A$ to $B\cup \Ic \cup \{I_d\}_{d\in D}$ given $C$ in $\Mf_{\Do(I_D)}$, which contradicts the choice of $\pi$.   If every non-endnode on $\pi^*$ is of a definite status, then the lemma trivially holds. Now we assume that $v_i$ is not of a definite status for some $1\leq i\leq n-1$. We argue by induction. 

    \textbf{Step~1: base case}. We first show the base case. Let $v_{i_1}$ be the first node that is not of a definite status on $\pi^*$ starting from node $a$. Since $v_{i_1}$ is not of a definite status, we have $v_{i_1-1}\suo v_{i_1}\hus v_{i_1+1}$ or $v_{i_1-1}\suo v_{i_1}\ous v_{i_1+1}$ with $v_{i_1-1}$ adjacent to $v_{i_1+1}$ or $v_{i_1-1}\suh v_{i_1}\ous v_{i_1+1}$. By $\fci{1}$, nodes $v_{i_1-1}$ and $v_{i_1+1}$ are adjacent in $\Pf$ and therefore are adjacent in $\Mf$, since $v_{i_1-1},v_{i_1},v_{i_1+1}\notin \{I_d\}_{d\in D}$. Let $\tilde{\pi}$ be a path in $\Mf_{\Do(I_D)}$ constructed from $\pi$ by replacing $v_{i_1-1} \sus v_{i_1} \sus v_{i_1+1}$ with $v_{i_1-1}\sus v_{i_1+1}$. Then $\tilde{\pi}$ cannot be open from $A$ to $B\cup \Ic \cup \{I_d\}_{d\in D}$ given $C$ in $\Mf_{\Do(I_D)}$, since $\pi$ is an irreducible open path. Since the local configurations of all the nodes on $\tilde{\pi}$ except $v_{i_1-1}$ and $v_{i_1+1}$ are the same as the ones on $\pi$, the path $\tilde{\pi}$ can only be blocked at $v_{i_1-1}$ or $v_{i_1+1}$. We have four cases regarding whether $v_{i_1-1}$ and $v_{i_1+1}$ are colliders or not on $\pi$ and $\tilde{\pi}$ respectively. 

    \textbf{Step~1.1: base case 1}. The node $v_{i_1-1}$ is a non-collider on $\pi$ but a collider on $\tilde{\pi}$ and $v_{i_1-1}\notin \anc_{\Mf_{\Do(I_D)}}(C)$. Note that $v_{i_1-1}$ cannot be an endnode (by the assumption on $v_{i_1-1}$) and the pattern $v_{i_1-2}\suh v_{i_1-1}\tut v_{i_1}$ is impossible in a MAG. Then we have $v_{i_1-2}\suh v_{i_1-1}\tuh v_{i_1}$ and $v_{i_1-1}\hus v_{i_1+1}$ in $\Mf$. This implies that we have $v_{i_1}\hus v_{i_1+1}$ in $\Mf$, otherwise we would have an (almost) cycle $v_{i_1-1}\tuh v_{i_1} \tuh v_{i_1+1}\suh v_{i_1-1}$. Therefore, $v_{i_1}$ must be a collider on $\pi$ in $\Mf$. Since $\pi$ is open given $C$ in $\Mf_{\Do(I_D)}$, we have that $v_{i_1}\in \anc_{\Mf_{\Do(I_D)}}(C)$. Hence, $v_{i_1-1}\in \anc_{\Mf_{\Do(I_D)}}(C)$, which contradicts the assumption that $v_{i_1-1}\notin \anc_{\Mf_{\Do(I_D)}}(C)$. It shows that this case cannot happen.

    \textbf{Step~1.2: base case 2}. The node $v_{i_1+1}$ is a non-collider on $\pi$ but a collider on $\tilde{\pi}$ and $v_{i_1+1}\notin \anc_{\Mf_{\Do(I_D)}}(C)$. This case is similar to Case 1 and one can argue similarly that this case cannot happen by noting that when $v_{i_1+2}=I_d$ for some $d\in D$ it is impossible to have $v_{i_{1}+2}\tuh  v_{i_1+1} \tut v_{i_1}$ or $v_{i_{1}+2}\tuo v_{i_1+1} \tut v_{i_1}$ in $\Mf_{\Do(I_D)}$.

    \textbf{Step~1.3: base case 3}. The node $v_{i_1-1}$ is a collider on $\pi$ but a non-collider on $\tilde{\pi}$ and $v_{i_1-1}\in C$. In this case, we first note that $v_{i_1-1}\ne a$, since it is a collider on $\pi$. Since $v_{i_1-1}$ is a collider on $\pi$ but a non-collider on $\tilde{\pi}$, we have $v_{i_1-2}\suh v_{i_1-1}\hus v_{i_1}$ and $v_{i_1-1}\tuh v_{i_1+1}$ in $\Mf_{\Do(I_D)}$ (not $v_{i_1-1}\tut v_{i_1+1}$ because $v_{i_1-2}\suh v_{i_1-1}\tut v_{i_1+1}$ cannot be in a MAG). Since $v_{i_1}$ is the first node not of definite status on $\pi^*$ in $\Pf_{\Do(I_D)}$  starting from node $a$, the node $v_{i_1-1}$ must be of definite status on $\pi^*$. Therefore, $v_{i_1-2}\suh v_{i_1-1} \hus v_{i_1}$ must be in $\Pf_{\Do(I_D)}$. We claim that $v_{i_1-2}$ is adjacent to $v_{i_1+1}$. Indeed, assume that this is not the case. Then $v_{i_1-2}\suh v_{i_1-1}\hus v_{i_1} \sus v_{i_1+1}$ is a discriminating path for $v_{i_1}$ in $\Mf_{\Do(I_D)}$. Lemma~\ref{lem:0} tells us that it also forms a discriminating path for $v_{i_1}$ in $\Pf_{\Do(I_D)}$. By $\fci{4}$, we have that $v_{i_1}$ is of definite status on $\pi^*$ in $\Pf_{\Do(I_D)}$. This contradicts the assumption that $v_{i_1}$ is not of definite status on $\pi^*$ in $\Pf_{\Do(I_D)}$. Besides, the edge between $v_{i_1-2}$ and $v_{i_1+1}$ in $\Mf_{\Do(I_D)}$ must be $v_{i_1-2}\suh v_{i_1+1}$, since $v_{i_1-2}\suh v_{i_1-1}\tuh v_{i_1+1}$ is present in $\Mf_{\Do(I_D)}$. Since $v_{i_1}\suh v_{i_1-1}\tuh v_{i_1+1}$ is in $\Mf_{\Do(I_D)}$, we have $v_{i_1}\suh v_{i_1+1}$ in $\Mf_{\Do(I_D)}$. We construct a path $\breve{\pi}$ from $\pi$ by replacing $v_{i_1-2}\suh v_{i_1-1}\hus v_{i_1}\suh v_{i_1+1}$ with $v_{i_1-2}\suh v_{i_1+1}$. The path $\breve{\pi}$ cannot be open given $C$ in $\Mf_{\Do(I_D)}$ by the construction of $\pi$. Since the local configurations of $v_{i_1+1}$ on $\pi$ and $\breve{\pi}$ are the same, $\breve{\pi}$ must be blocked at $v_{i_1-2}$. If we have $v_{i_1-2}\tuh v_{i_1-1}$ in $\Mf_{\Do(I_D)}$, then we have $v_{i_1-2} \tuh v_{i_1+1}$. In this case, $\breve{\pi}$ cannot be blocked at $v_{i_1-2}$. So we must have $v_{i_1-3}\suh v_{i_1-2}\huh v_{i_1-1}$ and $v_{i_1-2}\tuh v_{i_1+1}$ in $\Mf_{\Do(I_D)}$. By induction, we can show that every node between $a$ and $v_{i_1}$ on $\pi$ is a collider on $\pi$ in both $\Mf_{\Do(I_D)}$ and $\Pf_{\Do(I_D)}$, and is a parent of $v_{i_1+1}$ in $\Mf_{\Do(I_D)}$. It follows that node $a$ is adjacent to $v_{i_1+1}$. Otherwise by Lemma~\ref{lem:0}, $a\sus \cdots \sus v_{i_1+1}$ forms a discriminating path for $v_{i_1}$ in $\Pf_{\Do(I_D)}$, so node $v_{i_1}$ should be of definite status on $\pi^*$, which is a contradiction. The edge between $a$ and $v_{i_1+1}$ cannot be out of $v_{i_1+1}$. That is because otherwise we have an (almost) directed cycle $v_{1}\tuh v_{i_1+1}\tuh a\suh v_1$. So we must have $a\suh v_{i_1+1}$. We can then replace the subpath $\pi(a,v_{i_1+1})$ with $a\suh v_{i_1+1}$ to get an open subpath of $\pi$ from $A$ to $B\cup \Ic\cup \{I_d\}_{d\in D}$ given $C$ in $\Mf_{\Do(I_D)}$. This contradicts the assumption that the path $\pi$ is an irreducible one. So Case 3 is impossible.

    \textbf{Step~1.4: finish base case.} Overall, Cases 1, 2, and 3 have been ruled out. Then the only possibility is that node $v_{i_1+1}$ is a collider on $\pi$ but a non-collider on $\tilde{\pi}$ and $v_{i_1+1}\in C$. In the setting of Case~4, we have $v_{i_1-1}\sut v_{i_1+1}$ in $\Mf$. Since $v_{i_1+1}$ is a collider on $\pi$, we must have $v_{i_1+1}\in \pa_{\Mf_{\Do(I_D)}}(v_{i_1-1})$. This finishes the proof of the base case. 

    \textbf{Step~2: induction step.} We now proceed with the induction step. To achieve that, we first show the following statement via induction.  

    \textbf{Step~2.1: show \cref{ind:1}}.
    \begin{statement}\label{ind:1}
        For the induction step, assume that for $1\leq r \leq i_{k-1}$, if $v_{r}$ is not of a definite status on $\pi^*$, then $v_{r+1}\in \pa_{\Mf_{\Do(I_D)}}(v_{r-1})$ and $v_{r+1}$ is a collider on $\pi$ in $\Mf_{\Do(I_D)}$. Let $v_{i_k}$ be the next node after $v_{i_{k-1}}$ on $\pi^*$ that is not of definite status on $\pi^*$. We now show:
        \begin{enumerate}
            \item[(i)] we have $v_{i_{k-1}-1}\hus v_{i_{k-1}}$ in $\Mf_{\Do(I_D)}$, and

            \item[(ii)] for every $i_{k-1}\le r \leq i_k$, we have $v_r\in \pa_{\Mf_{\Do(I_D)}} (v_{i_{k-1}-1})$ and $v_r$ is a collider on $\pi$ in $\Mf_{\Do(I_D)}$.
        \end{enumerate}
    \end{statement}
    \begin{proof}[Proof of \cref{ind:1}]
        \textbf{Show (i)}.  We argue by contradiction and assume that item~(i) does not hold. Then we have $v_{i_{k-1}-1}\tuh v_{i_{k-1}}$ in $\Mf_{\Do(I_D)}$. That is because
     we have $v_{i_{k-1}+1}\in \pa_{\Mf_{\Do(I_D)}}(v_{i_{k-1}-1})$ by the induction hypothesis and we cannot  have $v_{i_{k-1}+1}\tuh v_{i_{k-1}-1}\tut v_{i_{k-1}}$ in $\Mf_{\Do(I_D)}$. Since we have $v_{i_{k-1}+1}\in \anc_{\Mf_{\Do(I_D)}}(v_{i_{k-1}})$ and the edge $v_{i_{k-1}-1}\tuh v_{i_{k-1}}$, we have that the edge between nodes $v_{i_{k-1}}$ and $v_{i_{k-1}+1}$ must be $v_{i_{k-1}}\hut v_{i_{k-1}+1}$. Then we can replace $v_{i_{k-1}-1}\tuh v_{i_{k-1}} \hut v_{i_{k-1}+1}$ with $v_{i_{k-1}-1}\hut v_{i_{k-1}+1}$ and get a shorter path $\tilde{\pi}$. The local configuration of $v_{i_{k-1}+1}$ on $\pi$ is the same as the one on $\tilde{\pi}$. Since $\pi$ is not blocked at $v_{i_{k-1}+1}$, the path $\tilde{\pi}$ cannot be blocked at $v_{i_{k-1}+1}$. Node $v_{i_{k-1}-1}$ is a parent of $v_{i_{k-1}}$. Since $v_{i_{k-1}}$ is a collider on $\pi$ and $\pi$ is open, we know that $v_{i_{k-1}-1}$ must be an ancestor of $C$. Also note that $v_{i_{k-1}-1}\notin C$, since $\pi$ is open. Therefore, the path $\tilde{\pi}$ is an open subpath of  $\pi$.  This is a contradiction.  

    \textbf{Show (ii)}. We show item~(ii) by induction. The base case for $r=i_{k-1}$ holds by the assumption of \cref{ind:1}. Now assume that for all $i_{k-1}\leq r<i_{k-1}+l<i_k$ with some $l>0$, we have $v_r\in \pa_{\Mf_{\Do(I_D)}}(v_{i_{k-1}-1})$ and $v_r$ is a collider on $\pi$ in $\Mf_{\Do(I_D)}$. Nodes $v_r$ are also colliders on $\pi^*$ in $\Pf_{\Do(I_D)}$ by assumption. So $v_{i_{k-1}+l}$ must be adjacent to $v_{i_{k-1}-1}$ in $\Mf_{\Do(I_D)}$. Otherwise the subpath $\pi^*(v_{i_{k-1}-1},v_{i_{k-1}+l})$ would be a discriminating path for $v_{i_{k-1}}$ in $\Pf_{\Do(I_D)}$. In this case, $v_{i_{k-1}}$ must be of definite status, which is a contradiction. We must have $v_{i_{k-1}-1}\hus v_{i_{k-1}+l}$ in $\Mf_{\Do(I_D)}$, since we have $v_{i_{k-1}-1}\hut v_{i_{k-1}+l-1}\hus v_{i_{k-1}+l}$. We construct a path $\breve{\pi}$ from $\pi$ by replacing the subpath $\pi(v_{i_{k-1}-1},v_{i_{k-1}+l})$ with $v_{i_{k-1}-1}\hus v_{i_{k-1}+l}$. It cannot be open from $A$ to $B\cup \Ic \cup \{I_d\}_{d\in D}$ given $C$ in $\Mf_{\Do(I_D)}$, since it is a subpath of $\pi$. By item~(i), we know that the local configuration of $v_{i_{k-1}-1}$ is invariant, so the path $\breve{\pi}$ can only be blocked at $v_{i_{k-1}+l}$. Similar to previous arguments, we obtain that $v_{i_{k-1}+l}$ is a collider on $\pi$ and a non-collider on $\breve{\pi}$. It means that we have $v_{i_{k-1}-1}\hut v_{i_{k-1}+l}$ in $\Mf_{\Do(I_D)}$ and therefore $v_{i_{k-1}+l}\in \pa_{\Mf_{\Do(I_D)}}(v_{i_{k-1}-1})$. This shows item~(ii). 

    Overall, we complete the proof of \cref{ind:1}.
    \end{proof}

    \textbf{Step~2.2: finish induction step}. We now go back to the argument of the outer induction. First note that the nodes $v_{i_k-1}$ and $v_{i_k+1}$ are adjacent, since $v_{i_k}$ is not of definite status and they are not from $\{I_d\}_{d\in D}$. We can construct a path $\tilde{\pi}$ from $\pi$ by replacing $v_{i_k-1}\sus v_{i_k} \sus v_{i_k+1}$ with $v_{i_k-1}\sus v_{i_k+1}$. Similar to the base case, there are four cases to consider. Cases 1 and 2 can be ruled out similarly. For Case 3, we need the following argument. Assume that $v_{i_k-1}$ is a collider on $\pi$ but a non-collider on $\tilde{\pi}$. This implies that $v_{i_k-1}\in \pa_{\Mf_{\Do(I_D)}}(v_{i_{k}+1})$. Similar to \cref{ind:1}(i), we have $v_{i_k}\suh v_{i_k+1}$ in $\Mf_{\Do(I_D)}$. Similar to \cref{ind:1}(ii), for every $0\leq j\leq i_{k}-i_{k-1}-1$, we have that $v_{i_{k-1}+j}\in \pa_{\Mf_{\Do(I_D)}}(v_{i_k+1})$ and $v_{i_{k-1}+j}$ is a collider on $\pi$ in $\Mf_{\Do(I_D)}$. Also by (ii), we have that the nodes $v_{i_{k-1}},\ldots, v_{i_k}$ are colliders on $\pi$ and are ancestors of $\{v_{i_{k-1}-1}, v_{i_k+1}\}$. This implies that the subpath $\pi(v_{i_{k-1}-1},v_{i_k+1})$ forms an inducing path in $\Mf$. Since $\Mf$ is maximal, we have that $v_{i_{k-1}-1}$ is adjacent to $v_{i_k+1}$ in $\Mf$ and therefore in $\Mf_{\Do(I_D)}$. This edge cannot be $v_{i_{k-1}-1}\tut v_{i_k+1}$, $v_{i_{k-1}-1}\tuh v_{i_k+1}$, or $v_{i_{k-1}-1}\hut v_{i_k+1}$, because there are arrowheads towards $v_{i_{k-1}-1}$ and $v_{i_k+1}$ and almost directed cycles ($v_{i_k}\tuh v_{i_{k-1}-1} \tuh v_{i_k+1}\huh v_{i_k}$ and $v_{i_{k-1}}\tuh v_{i_{k}+1} \tuh v_{i_{k-1}-1}\huh v_{i_{k-1}}$) are not allowed. So we have $v_{i_{k-1}-1}\huh v_{i_k+1}$. We can then replace the subpath $\pi(v_{i_{k-1}-1},v_{i_k+1})$ with $v_{i_{k-1}-1}\huh v_{i_k+1}$ to get an open subpath of $\pi$ given $C$ in $\Mf_{\Do(I_D)}$, which is a contradiction. Therefore, the only possibility is that $v_{i_k+1}$ is a collider on $\pi$ but a non-collider on $\tilde{\pi}$. It implies that $v_{i_k+1}\in \pa_{\Mf_{\Do(I_D)}}(v_{i_k-1})$. By induction, we finish the proof of the lemma.
   
\end{proof}

\begin{lemma}\label{lem:2}
    Let $\MAG$ be an iMAG and let $\Pf$ be an iSOPAG such that $\Mf\in [\Pf]_\Ms$. Let $D\subseteq \Vc$. Let $A\subseteq\Vc,\ B\subseteq \Ic\cup \Vc$ and $C\subseteq \Ic\cup \{I_d\}_{d\in D}\cup \Vc$ be pairwise disjoint. Let 
    \[
    \pi: a=v_0\sus v_1\sus \cdots \sus v_{n-1} \sus v_{n}=b
    \] 
    be an irreducible non-trivial open path from $A$ to $B\cup \Ic$ given $C$ in $\Mf_{\Do(I_D)}$. Let $\pi^*$ be the corresponding path in $\Pf_{\Do(I_D)}$ consisting of the same sequence of nodes as $\pi$. It holds that  for every $1\leq i\leq n-1$, if $v_i$ is not of a definite status on $\pi^*$ and $v_{i+1}\notin \{I_d\}_{d\in D}$, then $v_{i-1}\in \pa_{\Mf_{\Do(I_D)}}(v_{i+1})$ and $v_{i-1}$ is a collider on $\pi$ in $\Mf_{\Do(I_D)}$.
\end{lemma}

\begin{proof}[Proof of \cref{lem:2}]
 Note that \cref{lem:1} implies that $v_{i+1}$ is a collider on $\pi$ and therefore $v_{i+1}\notin \Ic$. The proof is in parallel to that of \cref{lem:1}. One can start from the first node $v_{i_1}$ that is not of definite status when traversing the path from $b$ toward $A$, while exchanging the roles of $v_{i_1+1}$ and $v_{i_1-1}$ in the proof of \cref{lem:1}. The only difference occurs in the counterpart of Step~1.3 when $v_{i_1+2}=I_d$ (which plays the role of $v_{i_1-2}$ in the proof of \cref{lem:1}) and the other parts of the proof are symmetric to that of \cref{lem:1}. Therefore, we make some comments on Step~1.3.  
 
 Recall that the goal is to show that $v_{i_1+2}$ is adjacent to $v_{i_1-1}$ (which plays the role of $v_{i_1+1}$ in the proof of \cref{lem:1}). We prove by contradiction and assume that $v_{i_1+2}$ and $v_{i_1-1}$ are non-adjacent. If $v_{i_1+1}=d$ (which plays the role of $v_{i_1-1}$ in the proof of \cref{lem:1}), then $v_{i_1+2}\tuh v_{i_1+1}\hus v_{i_1}\sus v_{i_1-1}$ forms a discriminating path for $v_{i_1}$ in $\Pf_{\Do(I_D)}$. Then by \cref{lem:soft_mani_SOPAG} the rest of the argument is the same as \cref{lem:1}. If $v_{i_1+1}\ne d$, then we may have $v_{i_1+2}\tuh v_{i_1+1}$ or $v_{i_1+2}\tuo v_{i_1+1}$ in $\Pf_{\Do(I_D)}$. The first case can be argued exactly as in the case where $v_{i_1+1}=d$. For the second case, by \cref{def:int_pag} we know that to have $v_{i_1+2}\tuo v_{i_1+1}$ we must have $d \ouo v_{i_1+1}$ in $\Pf_{\Do(I_D)}$ provided the presence of $v_{i_1+1}\hus v_{i_1}$. We also have $v_{i_1+2}\tuh d$ and $d \hus v_{i_1}$ in $\Pf_{\Do(I_D)}$ by \cref{def:int_pag,lem:head_circle}. Therefore, if $v_{i_1+2}$ is non-adjacent to $v_{i_1-1}$, then $v_{i_1+2}\tuh d \hus v_{i_1}\sus v_{i_{1}-1}$ forms a discriminating path for $v_{i_1}$ in $\Pf_{\Do(I_D)}$. By \cref{lem:soft_mani_SOPAG}, we either have $v_{i_1}\tuh v_{i_1-1}$ or $d\huh v_{i_1} \huh v_{i_1-1}$ in $\Pf_{\Do(I_D)}$. In the first case, we can immediately see that $v_{i_1}$ is of definite status on $\pi^*$. In the second case, since $v_{i_1}\huh v_{i_1-1}$ and $v_{i_1}$ is not of a definite status on $\pi^*$ by assumption,  we have $v_{i_1+1}\huo v_{i_1}$ or $v_{i_1+1}\ouo v_{i_1}$. Therefore, from \cref{lem:head_circle} and $d\huh v_{i_1}$, we must have $d\suh v_{i_1+1}$, which contradicts the fact that we have $d\ouo v_{i_1+1}$. Overall, we derive a contradiction when assuming that $v_{i_1+2}$ and $v_{i_1-1}$ are non-adjacent. This implies that $v_{i_1+2}$ is adjacent to $v_{i_1-1}$. 
\end{proof}

\begin{lemma}\label{lem:3}
    Let $\MAG$ be an iMAG and let $\Pf$ be an iSOPAG such that $\Mf\in [\Pf]_\Ms$. Let $D\subseteq \Vc$. Let $A\subseteq\Vc,\ B\subseteq \Ic\cup \Vc$ and $C\subseteq \Ic\cup \{I_d\}_{d\in D}\cup \Vc$ be pairwise disjoint.
    \begin{enumerate}
        \item Let $\pi: a=v_0\sus v_1\sus \cdots \sus v_{n-1} \sus v_{n}=b$ be an irreducible non-trivial open path from $A$ to $B\cup \Ic$ given $C$ in $\Mf_{\Do(I_D)}$. Let $\pi^*$ be the corresponding path consisting of the same sequence of nodes in $\Pf_{\Do(I_D)}$. Then for every $1\leq i\leq n-1$ (if any), node $v_i$ must be of definite status.

        \item Let $\varpi:a=u_0\sus u_1\sus \cdots \sus u_{m-1} \sus u_{m}=I_d\in \{I_d\}_{d\in D}$ be an irreducible non-trivial open path from $A$ to $\{I_d\}_{d\in D}$ given $C$ in $\Mf_{\Do(I_D)}$. Then there exists a path $\tilde{\varpi}$ from $A$ to $\{I_d\}_{d\in D}$ in $\Pf_{\Do(I_D)}$ consisting of a subsequence $(u_{i_j})_{j=0}^\ell$ of $(u_{i})_{i=0}^m$ such that every node on $\tilde{\varpi}$ is of the same definite (non-)collider status as it is on $\varpi$.
    \end{enumerate}
\end{lemma}

\begin{proof}[Proof of \cref{lem:3}]
    \textbf{Step~1: show (1)}. If $v_i$ is not of definite status,  then $v_{i-1}\in \pa_\Mf(v_{i+1})$ and $v_{i+1}\in \pa_{\Mf}(v_{i-1})$ by \cref{lem:1,lem:2}. This is impossible. Hence, $v_i$ must be of definite status.

    \textbf{Step~2: show (2)}. By the same argument as the first part, it holds that $u_{i}$, for all $1\le i \le m-2$ (when $m\ge 3$), are of definite status. If $u_{m-1}$ is of definite status, then we are done. Therefore, there are two cases:
    \begin{enumerate}[label=\textbf{Case~\arabic*:}, ref=Case~\arabic*, leftmargin=*]
        \item $u_{m-2}\suo u_{m-1} \out u_m=I_d$ with $u_{m-2}$ being adjacent to $u_m$;
        \item $u_{m-2}\suo u_{m-1} \hut u_m=I_d$.
    \end{enumerate}

    \textbf{Case~1}. We consider the case where $u_{m-2}\suo u_{m-1} \out u_m=I_d$ with $u_{m-2}$ being adjacent to $u_m$. Note that if $m\le 2$, then we are done. So in the following, we can assume that $m\ge 3$. Since $u_{m-2}$ is adjacent to $u_m=I_d$, \cref{def:SOPAG} leaves two possibilities: 
    \begin{enumerate}[label=\textbf{Case~1.\arabic*:}, ref=Case~\arabic*, leftmargin=*]
        \item $d\tuh u_{m-2}$ invisible or $d\ouh u_{m-2}$, and
        \item $d\tuo u_{m-2}$, $d\ouo u_{m-2}$, or $d\out u_{m-2}$.
    \end{enumerate}

    \textbf{Case~1.1}. First assume that $d\ne u_{m-1}$. If we have $u_{m-2} \hut d$ invisible or $u_{m-2} \huo d$, then we have $u_{m-2}\huo u_{m-1}$ since $u_{m-2}\tuo u_{m-1}$ is excluded by \cref{def:SOPAG}, and $u_{m-2}\ouo u_{m-1}$ would imply $d\suh u_{m-1}$ by \cref{lem:head_circle}, which contradicts $u_{m-1}\out u_m$. Then, $\tilde{\varpi}\coloneqq \varpi^*(u_0,u_{m-2})\oplus (u_{m-2},u_m)$ satisfies our goal. The case where $u_{m-2}\huo u_{m-1}=d$ can be argued similarly.
    
    \textbf{Case~1.2}. Now assume $u_{m-2} \ouo d$ or $u_{m-2} \out d$ or $u_{m-2} \tuo d$.  Note that it is impossible to have $u_{m-2} \huo u_{m-1}$ in this case. Indeed, if $d\ne u_{m-1}$, the patterns $u_{m-1}\ouh u_{m-2} \out d$ and $u_{m-1}\ouh u_{m-2} \tuo d$ are excluded by \cref{def:SOPAG}, and $u_{m-1}\ouh u_{m-2}\ouo d$ implies $u_{m-1}\suh d$ by \cref{lem:head_circle}, which contradicts $u_{m-1}\out u_m$. This means that we have either $u_{m-2} \ouo u_{m-1}$ or $u_{m-2} \tuo u_{m-1}$. Then we cannot have $u_{m-3}\suh u_{m-2}$ by \cref{def:SOPAG} and the fact that $u_{m-2}$ is of definite status on $\varpi$. Therefore, we could have either $u_{m-3}\sut u_{m-2}$ or $u_{m-3}\suo u_{m-2}$. If we have $u_{m-3}\sut u_{m-2}$ or $u_{m-3}\suo u_{m-2}\out u_{m}$ with $u_{m}$ non-adjacent to $u_{m-3}$, then $\tilde{\varpi}\coloneqq \varpi^*(u_0,u_{m-2})\oplus (u_{m-2},u_m)$ will do the job. We consider the remaining case where $u_{m-3}\suo u_{m-2} \out u_m$ with $u_m$ adjacent to $u_{m-3}$. We repeat the above argument with $u_{m-2}\suo u_{m-1} \out u_m$ replaced by $u_{m-3}\suo u_{m-2} \out u_m$. Since the length of the path is finite, the argument will eventually terminate and give our target $\tilde{\varpi}$.

    \textbf{Case~2}. We are left with the case where $u_{m-2}\suo u_{m-1} \hut u_m=I_d$. If $u_{m-1}=d$, then we have $d\ouh u_{m-2}$ or $d\ouo u_{m-2}$ and therefore $u_m$ is adjacent to $u_{m-2}$. If $u_{m-1}\ne d$, then we could have $d\tuh u_{m-1}$ invisible or $d\ouh u_{m-1}$. If we have $d\tuh u_{m-1}$ invisible, then \cref{lem:head_circle,lem:property_visible} imply that $d\tuh u_{m-2}$ invisible or $d\ouh u_{m-2}$. If we have $d\ouh u_{m-1}$, then it is impossible to have $d\huh u_{m-2}$, since otherwise we would have $u_{m-2}\suh u_{m-1}$ by \cref{lem:head_circle}, which causes a contradiction. Therefore, we could have $d\tuh u_{m-2}$ or $d\ouh u_{m-2}$. In addition, \cref{lem:property_visible} yields that the directed edge $d\tuh u_{m-2}$ must be invisible if present. In summary, by \cref{def:SOPAG}, $u_{m-2}$ and $u_m$ are adjacent and we have $d\ouo u_{m-2}$, or $d\ouh u_{m-2}$ or $d\tuh u_{m-2}$ invisible. Therefore, a similar argument to that for Case~1 concludes the proof.
\end{proof}

\begin{lemma}\label{lem:definite_hint}
    Let $\MAG$ be an iMAG and let $\Pf$ be an iSOPAG such that $\Mf\in [\Pf]_\Ms$. Let $D,T\subseteq \Vc$ be disjoint. Let $A\subseteq\Vc\sm T,\ B\subseteq (\Ic\cup \Vc)\sm D$ and $C\subseteq \Ic\cup \{I_d\}_{d\in D}\cup \Vc$ be pairwise disjoint. Let
    \[
    \pi: a=v_0\sus v_1\sus \cdots \sus v_{n-1} \sus v_{n}=b
    \] 
    be an irreducible non-trivial  open path from $A$ to $B\cup \Ic \cup \{I_d\}_{d\in D} \cup T$ given $C\cup T$ in $\Mf_{\Do(I_D,T)}$. Let $\pi^*$ be the corresponding path in $\Pf_{\Do(I_D,T)}$ consisting of the same sequence of nodes as $\pi$. If $b\notin \{I_d\}_{d\in D}$, then for every $1\leq i\leq n-1$ (if any), node $v_i$ must be of definite status on $\pi^*$ in $\Pf_{\Do(I_D,T)}$. If $b\in \{I_d\}_{d\in D}$, then there exists a subseqence $(v_{i_j})_{j=0}^\ell$ of $(v_{i})_{i=0}^n$ such that it forms a path in $\Pf_{\Do(I_D,T)}$   on which every intermediate node is of the same definite (non-)collider status as it is on $\pi$ in $\Mf_{\Do(I_D,T)}$.
\end{lemma}
 
\begin{proof}[Proof of \cref{lem:definite_hint}]

The proof is inspired by the argument in \cite[Lemma~I of Lemma~14]{jaber22causal}. 

First note that $\pi^*$ is a well-defined path in $\Pf_{\Do(I_D,T)}$ and that the path $\pi$ is also open from $A$ to $B\cup \Ic\cup \{I_d\}_{d\in D}$ given $C$ in $\Mf_{\Do(I_D)}$.  
    
    We show that the path $\pi$ is irreducible in $\Mf_{\Do(I_D)}$. Assume for contradiction that this is not the case, i.e., there is a path 
    \[
    \tilde{\pi}:v_{i_0}\sus v_{i_1}\sus \cdots \sus v_{i_{m-1}}\sus v_{i_m}
    \] 
    with $1\leq m < n$ where $(v_{i_j})_{j=0}^{m}$ is a subsequence of $(v_{i})_{i=0}^{n}$ such that $\tilde{\pi}$ is open from $A$ to $B\cup \Ic \cup \{I_d\}_{d\in D}$ given $C$ in $\Mf_{\Do(I_D)}$ but is not open in $\Mf_{\Do(I_D,T)}$. Note that path $\pi$ does not contain nodes from $T$ and therefore $\tilde{\pi}$ does not either. This implies that the path $\tilde{\pi}$ is also present in $\Mf_{\Do(I_D,T)}$.

    Assume that the path $\tilde{\pi}$ is blocked at $v_{i_j}$ by $C\cup T$ in $\Mf_{\Do(I_D,T)}$. Since $\pi$ is open at $v_{i_j}$ in $\Mf_{\Do(I_D,T)}$ given $C$, we know that node $v_{i_j}$ must have different local configurations on $\pi$ and on $\tilde{\pi}$ respectively in $\Mf_{\Do(I_D,T)}$. If $v_{i_j}$ is a non-collider on $\tilde{\pi}$, then  
    $v_{i_j}\notin C$, since $\tilde{\pi}$ is open given $C$ in $\Mf_{\Do(I_D)}$. This implies that $\tilde{\pi}$ is open at $v_{i_j}$ given $C\cup T$ in $\Mf_{\Do(I_D,T)}$, which is a contradiction. 

    Now we assume that node $v_{i_j}$ is a collider on $\tilde{\pi}$ and a non-collider on $\pi$ in $\Mf_{\Do(I_D,T)}$. In this case, we have $v_{i_{j-1}}\suh v_{i_j} \hus v_{i_{j+1}}$ on $\tilde{\pi}$ and assume $v_{i_j-1}\hut v_{i_j}\sus v_{i_j+1}$ on $\pi$ in $\Mf_{\Do(I_D,T)}$ (not $v_{i_j-1}\tut v_{i_j}$ since $v_{i_{j-1}}\suh v_{i_j}$). Note that $v_{i_{j-1}}\ne v_{i_j-1}$ and the subpath $\pi(v_{i_{j-1}},v_{i_j-1})$ cannot be a directed path $v_{i_{j-1}}\hut \cdots \hut v_{i_{j}-1}$ in $\Mf_{\Do(I_D,T)}$, otherwise we would have an (almost) directed cycle $v_{i_j}\hus v_{i_{j-1}}\hut \cdots \hut v_{i_j-1}\hut v_{i_j}$. Since arrowheads cannot meet undirected edges in a MAG, there must be a collider on $\pi(v_{i_{j-1}},v_{i_j})$. Let node $u$ be the first collider on $\pi(v_{i_{j-1}},v_{i_j})$ after node $v_{i_j}$. The node $u$ does not block the path $\pi$ given $C\cup T$ in $\Mf_{\Do(I_D,T)}$, so we have $u\in \anc_{\Mf_{\Do(I_D,T)}}(C)$. Since we have $v_{i_j}\tuh v_{i_j-1}\tuh \cdots \tuh u$ in $\Mf_{\Do(I_D,T)}$, we have $v_{i_j}\in \anc_{\Mf_{\Do(I_D,T)}}(C)$. This contradicts the fact that $\tilde{\pi}$ is blocked at $v_{i_j}$ by $C\cup T$ in $\Mf_{\Do(I_D,T)}$. Thus, this case cannot happen either. The case where $v_{i_j-1}\sus v_{i_j}\tuh v_{i_j+1}$ on $\pi$ in $\Mf_{\Do(I_D,T)}$ can be proved similarly.
    
    Hence, $\pi$ is an irreducible open path from $A$ to $B\cup \Ic \cup \{I_d\}_{d\in D}$ given $C$ in $\Mf_{\Do(I_D)}$. By Lemma~\ref{lem:3}, all the non-endnodes on $\pi^*$ must be of definite status or there exists a subseqence $(v_{i_j})_{j=0}^\ell$ of $(v_{i})_{i=0}^n$ such that it forms a path in $\Pf_{\Do(I_D,T)}$   on which every intermediate node is of definite status as it is on $\pi$ in $\Mf_{\Do(I_D,T)}$.
\end{proof}

\begin{lemma}\label{lem:4}
    Let $\MAG$ be an iMAG and let $\Pf$ be an iSOPAG such that $\Mf\in [\Pf]_\Ms$. Let $D,T\subseteq \Vc$ be disjoint. Let $A\subseteq\Vc\sm T,\ B\subseteq (\Ic\cup \Vc)\sm D$ and $C\subseteq \Ic\cup \{I_d\}_{d\in D}\cup \Vc$ be pairwise disjoint.
    Let 
    \[
    \pi: a=v_0\sus v_1\sus \cdots \sus v_{n-1} \sus v_{n}=b
    \] 
    be a non-trivial $C$-tight open path from $A$ to $B\cup \Ic\cup \{I_d\}_{d\in D}\cup T$ given $C\cup T$ in $\Mf_{\Do(I_D,T)}$. Then there exists a definite open path from $A$ to $B\cup \Ic\cup \{I_d\}_{d\in D}\cup T$ given $C\cup T$ in $\Pf_{\Do(I_D,T)}$.
\end{lemma}

\begin{proof}[Proof of \cref{lem:4}]
There are two cases: \textbf{Case~1}: $b\in B\cup \Ic\cup T$ and \textbf{Case~2}: $b\in \{I_d\}_{d\in D}$.

\textbf{Case~1}. The proof of Case~1 is inspired by the argument in \cite[Lemma~2 in p.213]{zhang2006causal} and can be divided into four steps.

\textbf{Step~0: preparatory work}. \cref{lem:definite_hint} implies that every non-endnode on $\pi^*$ (if any) is of a definite status. Since $\Pf$ is an iSOPAG of $\Mf$, every definite non-collider on $\pi^*$ corresponds to a non-collider on $\pi$ and therefore is not in $C$. Similarly, for any definite collider $v_i$ on $\pi^*$, the node $v_i$ is also a collider on $\pi$. The goal is to show that $\pi^*$ is not blocked at $v_i$ by C in $\Pf_{\Do(I_D,T)}$. Let 
\[
\pf:v_i=u_0\tuh u_1\tuh \cdots \tuh u_{m-1}\tuh u_{m}=c
\]
be a shortest directed path from $v_i$ to $c\in C$ in $\Mf_{\Do(I_D,T)}$ (possibly of length zero). None of the nodes on $\pi$ and $\pf$ are in $T$. The case that $v_i=c$ is trivial so we assume in the following that $v_i\ne c$. Let $\pf^*$ be the corresponding path in $\Pf_{\Do(I_D,T)}$. The goal is to show that $\pf^*$ is definite directed, i.e., there are no circles on it. We argue by contradiction and assume that there is a circle on $\pf^*$. 

\textbf{Step~1: show $v_i\ous u_1$}. We want to show that $v_i\ous u_1$ is in $\Pf_{\Do(I_D,T)}$. First, note that we can only have four possible types of edges on $\pf^*$: $u_{j}\ouh u_{j+1}$, $u_{j}\ouo u_{j+1}$, $u_{j}\tuh u_{j+1}$ and $u_{j}\tuo u_{j+1}$ for $0\leq j\leq m-1$ in $\Pf_{\Do(I_D,T)}$.  It is impossible to have $v_i\tuo u_1$ in $\Pf_{\Do(I_D,T)}$ by \cref{def:SOPAG} and the fact that $v_i$ is a collider on $\pi^*$. We show that having $v_i\tuh u_1$ in $\Pf_{\Do(I_D,T)}$ is impossible. Assume on the contrary that we have $v_i\tuh u_1$ in $\Pf_{\Do(I_D,T)}$, then we have either 
\[
v_i\tuh u_1 \tuh \cdots \tuh u_{j-1} \tuh u_j \tuo u_{j+1} \quad \text{ or } \quad v_i\tuh u_1 \tuh \cdots \tuh u_{j-1} \tuh u_j \ous u_{j+1}
\]
for some $j$. The first case cannot happen by \cref{def:SOPAG}. For the second case, we have $u_{j-1}\tuh u_{j+1}$ or $u_{j-1}\ouh u_{j+1}$ in $\Pf$ and therefore in $\Pf_{\Do(I_D,T)}$ by \cref{lem:head_circle}. Both of these cases contradict the choice of $\pf$. Hence, $v_i \ous u_1$ must be present in $\Pf_{\Do(I_D,T)}$. 

\textbf{Step~2: construct path $\tilde{\pi}$ contradicting $C$-tightness of $\pi$ when $u_1$ not on $\pi^*$}. In the next, we want to construct a path $\tilde{\pi}$ such that $(\|\tilde{\pi}\|,\dist(\tilde{\pi},C))<_{\mathsf{lex}}(\|\pi\|,\dist(\pi,C))$ under the assumption that node $u_1$ is not on path $\pi^*$. If this is achieved, then since $\pi$ is $C$-tight we derive a contradiction. By \cref{lem:head_circle}, we have $v_{i-1}\suh u_1$ and $v_{i+1}\suh u_1$ in $\Pf_{\Do(I_D,T)}$. Now we show the following statement: 
\begin{statement}\label{state:2}
    There exists $v_j$ with $j<i$ such that 
    \begin{enumerate}
        \item [(i)]  $v_j\suh u_1$ in $\Mf_{\Do(I_D,T)}$, and

        \item [(ii)] the (non-)collider status of node $v_j$ on path $\pi$ is the same as the (non-)collider status of $v_j$ on $a\sus v_1 \sus \cdots \sus v_j \suh u_1$ in $\Mf_{\Do(I_D,T)}$.
    \end{enumerate}
\end{statement}

\begin{proof}[Proof of \cref{state:2}]
    It suffices to show that if no nodes between nodes $a$ and $v_i$ on path $\pi$ satisfy the two conditions, then node $a$ must satisfy them. Assume that there are no nodes between nodes $a$ and $v_i$ on path $\pi$ satisfying the two conditions. If $v_{i-1}=a$, then node $a$ satisfies item (i) and item (ii) trivially. So we assume that $v_{i-1}\ne a$. We then prove by induction that every node between nodes $a$ and $v_i$ is a collider on $\pi$ and is a parent of $u_1$ in $\Mf_{\Do(I_D,T)}$.

    \textbf{Base case}. Since we have already shown that $v_{i-1}\suh u_1$ is in $\Mf_{\Do(I_D,T)}$, we have that item (ii) does not hold for $v_{i-1}$ by our assumption. It implies that either $v_{i-1}$ is a non-collider on $\pi$ but a collider on $a\sus v_1 \sus \cdots \sus v_{i-1} \suh u_1$, or the other way around. The former case implies that we have $v_{i-1}\tuh v_i \tuh u_1$ and $v_{i-1}\huh u_1$, which forms an almost directed cycle in $\Mf_{\Do(I_D,T)}$ and cannot appear in a MAG. So only the latter case is possible, where we have $v_{i-1}\tuh u_1$ and node $v_{i-1}$ is a collider on $\pi$.

    \textbf{Induction step}. For induction, assume that every node $v_j$ is a collider on path $\pi$ and $v_j\in \pa_{\Mf_{\Do(I_D,T)}}(u_1)$ for $k<j\leq i-1$ where $1\leq k< i-1$ is fixed. The goal is to show that node $v_k$ is a collider on $\pi$ and $v_k\in \pa_{\Mf_{\Do(I_D,T)}}(u_1)$. First note that node $v_k$ is adjacent to node $u_1$, otherwise $v_k\suh \cdots \huh v_{i-1}\hus v_i \tuh u_1$ forms a discriminating path for $v_i$ in $\Mf_{\Do(I_D,T)}$ and in $\Mf$. Lemma~\ref{lem:0} gives that it is also a discriminating path for $v_i$ in $\Pf$ and in $\Pf_{\Do(I_D,T)}$, since every non-endnode on $\pi^*$ is of a definite status. $\fci{4}$ should have oriented the circle of $v_i\ous u_1$ at $v_i$ in $\Pf$, which contradicts the fact that we have $v_i\ous u_1$ in $\Pf_{\Do(I_D,T)}$. Since we have $v_k\suh v_{k+1}\tuh  u_1$ and that $v_k$ is adjacent to $u_1$, we have $v_k\suh u_1$ in $\Mf_{\Do(I_D,T)}$ (otherwise there would be a (almost) directed cycle), i.e., item (i) holds. Since we assume that node $v_k$ does not satisfy the two conditions, item (ii) does not hold for $v_k$. Similar argument to the last paragraph
    implies that node $v_k$ is a collider on path $\pi$ and that $v_k\in \pa_{\Mf_{\Do(I_D,T)}}(u_1)$.

    Then we know that node $a$ must be adjacent to node $u_1$, otherwise the circle at $v_i$ on $v_i \ous u_1$ would have been oriented by $\fci{4}$ in $\Pf$. The edge between $a$ and $u_1$ must be of the form $a\suh u_1$; otherwise $\Mf$ would fail to be ancestral. So item (i) and item (ii) hold for node $a$. This completes the proof of \cref{state:2}.
\end{proof}

    By symmetry, we have the following statement:
   \begin{statement}\label{state:3}
       There exists a node $v_l$ with $l>i$ such that 
    \begin{enumerate}
        \item [(i)] $v_l\suh u_1$ in $\Mf_{\Do(I_D,T)}$, and

        \item [(ii)] the (non-)collider status of node $v_l$ on path $\pi$ is the same as the (non-)collider status of $v_l$ on $u_1\hus v_l \sus \cdots \sus v_n$ in $\Mf_{\Do(I_D,T)}$.
    \end{enumerate}
   \end{statement} 
   Then the path 
   \[
   \tilde{\pi}\coloneqq \pi(a,v_j)\oplus (v_j,u_1,v_l)\oplus \pi(v_l,b)
   \] 
   is open from $A$ to $B\cup \Ic$ given $C\cup T$ in $\Mf_{\Do(I_D,T)}$. Note that 
   \[
   (\|\tilde{\pi}\|,\dist(\tilde{\pi},C))<_{\mathsf{lex}}(\|\pi\|,\dist(\pi,C)),
   \] 
   which contradicts the choice of $\pi$. This finishes Step~2.

\textbf{Step~3: construct path $\tilde{\pi}$ contradicting $C$-tightness of $\pi$ when $u_1$ on $\pi^*$}. Finally, if node $u_1$ is on path $\pi^*$, equivalently on $\pi$, then it either lies on subpath $\pi(a,v_{i-1})$ or $\pi(v_{i+1},b)$ of path $\pi$. Without loss of generality, assume that it is on $\pi(v_{i+1},b)$ and $u_1=v_l$ for some $l>i$. Similar argument gives a node $v_j$ with $j<i$ such that 
    \begin{enumerate}
        \item [(i)] there exists $v_j\suh u_1$, and
        \item [(ii)] the (non-)collider status of node $v_j$ on path $\pi$ is the same as the (non-)collider status of $v_j$ on $a\sus v_1 \sus \cdots \sus v_j \suh u_1$.
    \end{enumerate}
    Then the path 
    \[
    \tilde{\pi}:a\sus v_1 \sus \cdots \sus v_j \suh u_1=v_l \sus v_{l+1} \sus  \cdots \sus v_{n} \sus b
    \]
    is open from $A$ to $B\cup \Ic$ given $C\cup T$ in $\Mf_{\Do(I_D,T)}$ but satisfies the property that 
    \[
    (\|\tilde{\pi}\|,\dist(\tilde{\pi},C))<_{\mathsf{lex}}(\|\pi\|,\dist(\pi,C)).
    \] 
    It is open because local configurations of nodes $v_i$ on $\tilde{\pi}$ for $i=1,\ldots, j, l,\ldots,n$ are the same as the ones on $\pi$ and $v_l\in \anc_{\Mf_{\Do(I_D,T)}}(C)$.   This is a contradiction.

    In summary, the hypothesis that the path $\pf^*$ has circles in $\Pf_{\Do(I_D,T)}$ is false. It means that $\pf^*$ is a definite directed path in $\Pf_{\Do(I_D,T)}$ and $v_i\in \anc_{\Pf_{\Do(I_D,T)}}(C)$. Hence, $\pi^*$ is an (definite) open path from $A$ to $B\cup \Ic\cup T$ given $C\cup T$ in $\Pf_{\Do(I_D,T)}$. This finishes the proof of Case~1.

    \textbf{Case~2}. Let 
    \[
    \ol{\pi}:a=w_0\sus \cdots \sus w_\ell=b
    \]
    be a path in $\Pf_{\Do(I_D,T)}$ that consists of a subsequence of nodes from $\pi$ and every node on $\ol{\pi}$ has the same definite (non-)collider status as it does on $\pi$ in $\Mf_{\Do(I_D,T)}$ from \cref{lem:definite_hint}. Let $w_i$ be a collider on $\ol{\pi}$. If $w_{i+1}\in \{I_d\}_{d\in D}$, then since $\pf^*$ is potentially directed in $\Pf_{\Do(I_D,T)}$, path $\pi^*$ is not blocked at $w_{i}$. Therefore, we consider the case where $w_{i}$ is a collider and $w_{i+1}\notin \{I_d\}_{d\in D}$. The proof in Case~1 carries over with certain modifications. First, the same argument in Step~1 works. For Step~2, the proof for \cref{state:2} applies without any changes. In contrast, essential changes need to be made when it comes to \cref{state:3} because $w_\ell \in \{I_d\}_{d\in D}$ and the symmetry breaks. To be more precise, instead of \cref{state:3}, we have the following statement:
    \begin{statementprime}\label{state:4}
        There exists a node $w_l$ with $l>i$ such that 
    \begin{enumerate}
        \item [(i)] $w_l\suh u_1$ in $\Mf_{\Do(I_D,T)}$ if $l<\ell$, and $w_l\tuh u_1$ or $w_l\tuo u_1$ in $\Pf_{\Do(I_D,T)}$ if $l=\ell$, and

        \item [(ii)] the (non-)collider status of node $w_l$ on path $\ol{\pi}$ is the same as the (non-)collider status of $w_l$ on $u_1\hus w_l \sus \cdots \sus w_\ell$ in $\Mf_{\Do(I_D,T)}$ if $l<\ell$.
    \end{enumerate}
    \end{statementprime}
    We now prove \cref{state:4}. By the discussion at the beginning of Case~2, we can assume that $w_{i+1}\ne w_\ell$. We then prove by induction that every node between nodes $w_{\ell}$ and $w_i$ is a collider on $\ol{\pi}$ and is a parent of $u_1$ in $\Mf_{\Do(I_D,T)}$, as we did in Case~1, under the assumption that the node $w_l$ stated in \cref{state:4} does not exist. Then the rest of the argument for \cref{state:2} applies directly. Therefore,  by \cref{lem:soft_mani_SOPAG} we can conclude that if no nodes between $w_{\ell}$ and $w_{i}$ on path $\ol{\pi}$ satisfy the two conditions in \cref{state:4}, then $w_\ell$ satisfies the endnode alternative. This then establishes \cref{state:4}.

    Let $w_{i_k}$ be the $k$-th collider on $\ol{\pi}$ after $a$. Let $w_{j_{k}}$ and $w_{l_{k}}$ be the nodes from applying \cref{state:2,state:4} to $w_{i_k}$, respectively. Define $r\coloneqq \min\{k: l_k=\ell\}$ if $\{k: l_k=\ell\}\ne \emptyset$. Let $u_1^{i_k}$ be the second node on the non-trivial potentially directed path from $w_{i_k}$ to $C$.
    
    If $\{k: l_k=\ell\}= \emptyset$, then for all $i_k$ such that $w_{i_{k}+1}\notin \{I_d\}_{d\in D}$, the path 
    \[
    \tilde{\pi}\coloneqq \pi(a,w_{j_k})\oplus (w_{j_k},u_1^{i_k},w_{l_k})\oplus \pi(w_{l_k},b)
    \]
    is open from $A$ to $B\cup \Ic\cup \{I_d\}_{d\in D}$ given $C\cup T$ in $\Mf_{\Do(I_D,T)}$. Indeed,  $w_{j_k}$ and $w_{l_k}$ have the same (non-)collider status on $\pi$ and on $\tilde{\pi}$ by \cref{lem:3}, \cref{state:2}(ii) and \cref{state:4}(ii). Note that $(\|\tilde{\pi}\|,\dist(\tilde{\pi},C))<_{\mathsf{lex}}(\|\pi\|,\dist(\pi,C))$. This contradicts the $C$-tightness of $\pi$ and therefore implies that $\ol{\pi}$ is open given $C\cup T$ at its collider $w_{i_k}$ in $\Pf_{\Do(I_D,T)}$. Overall, in this case path $\ol{\pi}$ is open given $C\cup T$ in $\Pf_{\Do(I_D,T)}$.

    If $\{k: l_k=\ell\}\ne \emptyset$, then 
    \[
    \tilde{\pi}\coloneqq \ol{\pi}(a,w_{j_r})\oplus (w_{j_r},u_1^{i_r},w_{l_r})
    \] 
    is open given $C\cup T$ in $\Pf_{\Do(I_D,T)}$. Indeed, $\ol{\pi}(a,w_{j_r})\oplus (w_{j_r},u_1^{i_r})$ is open since all the colliders before $w_{i_r}$ can be shown to be open given $C\cup T$ using the argument in the last paragraph. Besides, $w_{l_r}\in \{I_d\}_{d\in D}$ and $u_1^{i_r}$ is a definite collider on $\tilde{\pi}$ such that $u_1^{i_r}$ is a possible ancestor of $C$.

    Overall, this finishes the proof.
\end{proof}

\section[Proofs for Section~\ref{sec:IDalg}]{Proofs for \cref{sec:IDalg}}\label{sec:pf_2}

\subsection[Proof of Proposition~\ref{prop:sidp_rule} and Theorem~\ref{thm:idp_sound}]{Proof of \cref{prop:sidp_rule,thm:idp_sound}}

We present the proofs of \cref{prop:sidp_rule,thm:idp_sound}. See \cref{fig:pf_structure_idp_sound} for an overview of the proof structure. We introduce some additional notation used in this subsection. For $H\subseteq \Ic\cup \Vc$, we write $I_H\coloneqq (H\cap \Ic) \cup \{I_h\}_{h\in H\cap \Vc}$.

\begin{proof}[Proof of \cref{prop:sidp_rule,thm:idp_sound}]
    These two results follow from \cref{lem:IDP_sound_I,lem:IDP_sound_III,lem:IDP_sound_II}. Also note that \cref{lem:pcmk} implies that if $\Mc\in \Mb^+_c$ then any interventional kernel induced by $\Mc$ is positive and continuous.
\end{proof}

\begin{lemma}[Rule~L0]\label{lem:IDP_sound_I}
    Let $\Pf=(\Ic,\Vc,\Ec)$ be an iSOPAG, $A,B\subseteq \Vc$ be disjoint with $A\ne \emptyset$ and $(\Mc,X_\Sc=\mathbf{1}_{|\Sc|})$ be an s-iSCM such that $\Gb(\Mc,X_\Sc=\mathbf{1}_{|\Sc|})\in [\Pf]_\Gs$. Let $\Dc$ and $H$ be defined as in \cref{prop:sidp_rule}. Then $\Prb_{\Mc}(X_A\mid X_\Sc=\mathbf{1}_{|\Sc|} \miid \Do(X_B),X_{\Ic})$ is trackable from $\Qc[\Dc]$ and 
    \[
        \begin{aligned}
             &\Prb_{\Mc}(X_A\mid X_{\Sc}=\1\miid \Do(X_B),X_{\Ic\sm H},\cancel{X_{\Ic\cap H}})\\
             &=\Prb_\Mc(X_A\mid X_{\Sc}=\1\miid \Do(X_B),X_{\Ic})=\Prb_\Mc(X_A\mid X_{\Sc}=\1\miid \Do(X_{\Vc\sm \Dc}),X_{\Ic})=\Qc[\Dc]^{\sm(\Dc\sm A)}.
        \end{aligned}
    \]
\end{lemma}

\begin{proof}[Proof of \cref{lem:IDP_sound_I}]
    This follows from \cref{lem:IDP_sound_ci_I,thm:causal_calculus_mag_pag}. 
\end{proof}

\begin{lemma}[Rule~L1]\label{lem:IDP_sound_II}
    Let $\Pf=(\Ic,\Vc,\Ec)$ be an iSOPAG, $D\subseteq \Vc$ , and $A\subseteq D$. Set $R_1\coloneqq \re_{\Pf_D}(A)$ and $R_2\coloneqq \re_{\Pf_D}(D\sm R_1)$. Then the Markov kernel $\Qc[D]$ is trackable from $\Qc[R_1]$ and $\Qc[R_2]$. If $\Mc^\Sc\in \Mb^+(\Pf)$ then the following pointwise equality holds
    \[
        \Qc[D]=\Qc[R_1]\boxtimes \Qc[R_2]=\Qc[R_2]\boxtimes \Qc[R_1].
    \]
\end{lemma}

\begin{proof}[Proof of \cref{lem:IDP_sound_II}]
    Let $\Bb_1\prec \cdots \prec \Bb_m$ be a topological order of $\Pf_D$-buckets. To simplify notation, we omit $X_\Ic$. Note that
    \[
        \Qc[D]=\Prb_{\Mc^\Sc}(X_D\mid \Do(X_{D^c}))=\bigotimes_{\Bb_i\subseteq D}^\succ\Prb_{\Mc^\Sc}(X_{\Bb_i}\mid X_{\Bb_i^\prec}\miid \Do(X_{D^c})).
    \]
    Since $\Mc^\Sc\in \Mb^+$, by \cref{lem:IDP_sound_ci_2,lem:IDP_sound_ci_3,thm:causal_calculus_mag_pag}, for $\Bb_i\subseteq R_j$ with $j=1,2$
    \[
        \begin{aligned}
            \Prb_{\Mc^\Sc}(X_{\Bb_i}\mid X_{\Bb_i^\prec}\miid \Do(X_{D^c}))&\meq{\mu_{\Bb_i^\prec}}\Prb_{\Mc^\Sc}(X_{\Bb_i}\mid X_{\Bb_i^\prec}\miid \Do(X_{(R_{j}\cup \Bb_i^\preceq)^c}))&\\
            &\meq{\mu_{\Bb_i^\prec}} \Prb_{\Mc^\Sc}(X_{\Bb_i}\mid X_{\Bb_i^\prec\cap R_{j}} \miid \Do(X_{R_{j}^c})).
        \end{aligned}
    \]
    Hence,
    \[
        \Qc[D]=\Qc[R_1]\boxtimes \Qc[R_2]=\Qc[R_2]\boxtimes \Qc[R_1].
    \]
\end{proof}

\begin{remark}
    We provide some intuition for the two $\mu_{\Bb_i^\prec}$-equalities in the proof of \cref{lem:IDP_sound_II}. Consider the case of ADMGs and SCMs with discrete variables. Set $D\coloneqq \mathsf{Distr}_{\Af}(A)$ and $B\coloneqq \Vc\sm D$. Then
    \[
        \begin{aligned}
            p_\Mc(x_A)&=p_\Mc(x_A\miid \Do(x_{A^\succ})) \quad \text{ and } \\
            p_\Mc(x_A\mid x_{A^\prec \cap D}, x_{A^\prec \cap B})&=p_\Mc(x_A\mid x_{A^\prec \cap D} \miid \Do(x_{A^\prec\cap B})).
        \end{aligned}
    \]
    In words: a topological order on nodes of an ADMG can be viewed as a time order in which variables are measured. Interventions on variables that are after $X_A$ do not influence $X_A$. After conditioning on the relevant district history $X_{A^\prec \cap D}$, the remaining past variables outside the district $X_{A^\prec\cap B}$ carry no further latent confounding with $X_A$, so intervening on non-confounded history $X_{A^\prec\cap B}$ is equivalent to conditioning on them.
\end{remark}

\begin{remark}[On the condition $\Mc^\Sc\in \Mb^+(\Pf)$]
In \cref{lem:IDP_sound_II}, the pointwise equality holds under the condition $\Mc^\Sc\in \Mb^+(\Pf)$. Without that condition, the pointwise equality may not hold in general since it is unclear how to control the null sets. More precisely, by \cref{lem:IDP_sound_ci_2,thm:causal_calculus_mag_pag} and the essential uniqueness of conditional kernels (\cref{defthm:prob_calculus}), there exists a kernel $\Qr(X_{\Bb_{k}}\miid X_{\Bb_{k}^\prec},X_{D^c})$ such that $N_1$ is a $\Prb_{\Mc^\Sc}\big(X_{\Bb_{k}^\prec}\miid \Do(X_{D^c})\big)$-null set and $N_2$ is a $\Prb_{\Mc^\Sc}\big(X_{\Bb_{k}^\prec}\miid \Do(X_{(R_{j}\cup \Bb_i^\preceq)^c})\big)$-null set where
    \[
        \begin{aligned}
            N_1\coloneqq \Big\{x_{\Bb_{k}^\prec \cup (R_{j}\cup \Bb_{k}^\preceq)^c}\in \Xc_{\Bb_{k}^\prec \cup (R_{j}\cup \Bb_{k}^\preceq)^c}: \Qr(X_{\Bb_{k}}\miid X_{\Bb_{k}^\prec}=x_{\Bb_{k}^\prec},X_{D^c}=x_{D^c})\\ \ne \Prb_{\Mc^\Sc}(X_{\Bb_{k}}\mid X_{\Bb_{k}^\prec}=x_{\Bb_{k}^\prec}\miid \Do(X_{D^c}=x_{D^c})) \Big\}\\
            N_2\coloneqq \Big\{x_{\Bb_{k}^\prec \cup (R_{j}\cup \Bb_{k}^\preceq)^c}\in \Xc_{\Bb_{k}^\prec \cup (R_{j}\cup \Bb_{k}^\preceq)^c}: \Qr(X_{\Bb_{k}}\miid X_{\Bb_{k}^\prec}=x_{\Bb_{k}^\prec},X_{D^c}=x_{D^c})\\ \ne \Prb_{\Mc^\Sc}\Big(X_{\Bb_{k}}\mid X_{\Bb_{k}^\prec}=x_{\Bb_{k}^\prec}\miid \Do\big(X_{(R_{j}\cup \Bb_i^\preceq)^c}=x_{(R_{j}\cup \Bb_i^\preceq)^c}\big)\Big) \Big\}.
        \end{aligned}
    \]
    
    Similar to the above argument, there exists a kernel 
    \[
    \tilde{\Qr}\big(X_{\Bb_{k}}\miid X_{\Bb_{k}^\prec}, X_{(R_{j}\cup \Bb_i^\preceq)^c}\big)
    \]
    such that  
    \[
    \begin{aligned}
        &\text{$N_3$ is a } \Prb_{\Mc^\Sc}(X_{\Bb_{k}^\prec} \miid \Do(X_{(R_{j}\cup \Bb_i^\preceq)^c}\big) \text{-null set in } \Xc_{\Bb_{k}^\prec \cup (R_{j}\cup \Bb_{k}^\preceq)^c}\ \text{ and} \\
        &\text{$N_4$ is a } \Prb_{\Mc^\Sc}\big( X_{\Bb_{k}^\prec\cap R_{j}}\miid \Do\big(X_{R_{j}^c}\big)\big)\text{-null set in } \Xc_{\Bb_{k}^\prec \cup (R_{j}\cup \Bb_{k}^\preceq)^c},
    \end{aligned}
    \]
    where
    \[
        \begin{aligned}
            N_3\coloneqq \Big\{ &x_{\Bb_{k}^\prec \cup (R_{j}\cup \Bb_{k}^\preceq)^c}\in \Xc_{\Bb_{k}^\prec \cup (R_{j}\cup \Bb_{k}^\preceq)^c}:\\ &\tilde{\Qr}\big(X_{\Bb_{k}}\miid X_{\Bb_{k}^\prec}=x_{\Bb_{k}^\prec},X_{(R_{j}\cup \Bb_i^\preceq)^c}=x_{(R_{j}\cup \Bb_i^\preceq)^c}\big)\\
            &\ne \Prb_{\Mc^\Sc}\Big(X_{\Bb_{k}}\mid X_{\Bb_{k}^\prec}=x_{\Bb_{k}^\prec}\miid \Do(X_{(R_{j}\cup \Bb_i^\preceq)^c}=x_{(R_{j}\cup \Bb_i^\preceq)^c}\big)\Big) \Big\}\\
            N_4\coloneqq \Big\{& x_{\Bb_{k}^\prec \cup (R_{j}\cup \Bb_{k}^\preceq)^c}\in \Xc_{\Bb_{k}^\prec \cup (R_{j}\cup \Bb_{k}^\preceq)^c}: \\
            &\tilde{\Qr}\big(X_{\Bb_{k}}\miid X_{\Bb_{k}^\prec}=x_{\Bb_{k}^\prec}, X_{(R_{j}\cup \Bb_i^\preceq)^c}=x_{(R_{j}\cup \Bb_i^\preceq)^c}\big) \\
            &\ne \Prb_{\Mc^\Sc}\Big(X_{\Bb_{k}}\mid X_{\Bb_{k}^\prec\cap R_{j}}=x_{\Bb_{k}^\prec\cap R_{j}}\miid \Do\big(X_{R_{j}^c}=x_{R_{j}^c}\big)\Big) \Big\}.
        \end{aligned}
    \]
    This means that in general we \emph{cannot} conclude the following equality:
    \[
        \Prb_{\Mc^\Sc}(X_{\Bb_i}\mid X_{\Bb_i^\prec}\miid \Do(X_{D^c}))=\Prb_{\Mc^\Sc}(X_{\Bb_i}\mid X_{\Bb_i^\prec\cap R_{j}} \miid \Do(X_{R_{j}^c}))
    \]
    up to a null set $N$ that is simultaneously a
    \[
    \Prb_{\Mc^\Sc}( X_{\Bb_i^\prec}\miid \Do(X_{D^c}))\text{-null set and a } 
    \Prb_{\Mc^\Sc}( X_{\Bb_i^\prec\cap R_{j}} \miid \Do(X_{R_{j}^c}))\text{-null set};
    \]
    and therefore cannot obtain the pointwise equality in \cref{lem:IDP_sound_III}.

\end{remark}

\begin{lemma}[Rule~L2]\label{lem:IDP_sound_III}
    Let $\Pf=(\Ic,\Vc,\Ec)$ be an iSOPAG. Let $D\subseteq \Vc$, and $\Ab\subseteq D$ be a $\Pf_D$-bucket. Write $D^+\coloneqq \pde_{\Pf_D}(\Ab)$ and $D^-\coloneqq (D\sm D^+)\cup \Ab$. Let $(\Mc,X_\Sc=\mathbf{1}_{|\Sc|})$ be an s-iSCM such that $\Gb(\Mc,X_\Sc=\mathbf{1}_{|\Sc|})\in [\Pf]_\Gs$.  If $\pcc_{\Pf_D}(\Ab)\cap \pde_{\Pf_D}(\Ab)\subseteq \Ab$, then $\Qc[D\sm \Ab]$ is trackable from $\Qc[D]$ by
    \[
        \Qc[D\sm \Ab]=\Qc[D]^{|D^-}\otimes \Qc[D]^{\sm  D^+},
    \]
    where equality holds up to an oracle choice of conditional kernel and $\mu_{\Ab}$-a.s.\ if $\Mc^\Sc \in \Mb^+(\Pf)$. If $\Mc^\Sc \in \Mb^{+}_c$, then $\Qc[D]^{|D^-}\otimes \Qc[D]^{\sm  D^+}$ admits a continuous version and the equality holds pointwise provided that the continuous version is taken.
\end{lemma}

\begin{proof}[Proof of \cref{lem:IDP_sound_III}]
To simplify notation, we omit $X_\Ic$. Then by \cref{lem:IDP_sound_ci_IV} and \cref{thm:causal_calculus_mag_pag}
\[
    \begin{aligned}
        \Qc[D\sm \Ab]&=\Prb_{\Mc^\Sc}(X_{D\sm \Ab}\miid \Do(X_{\Ab\cup D^c}))\\
        &=\Prb_{\Mc^\Sc}(X_{D^+\sm \Ab}\mid X_{D^-\sm \Ab}\miid \Do(X_{\Ab\cup D^c}))\otimes \Prb_{\Mc^\Sc}(X_{D^-\sm \Ab}\miid \Do(X_{\Ab\cup D^c}))\\
        &=\Prb_{\Mc^\Sc}(X_{D^+\sm \Ab}\mid X_{D^-}\miid \Do(X_{D^c}))\otimes \Prb_{\Mc^\Sc}(X_{D^-\sm \Ab}\miid \Do(X_{D^c}))\\
        &=\Qc[D]^{|D^-}\otimes \Qc[D]^{\sm  D^+},
    \end{aligned}
\]
where the equality holds up to an oracle choice of conditional kernel and $\mu_{\Ab}$-a.s.\ if $\Mc^\Sc \in \Mb^+(\Pf)$ and pointwise if $\Mc^\Sc \in \Mb^{+}_c$ and the conditional kernels are taken to be continuous. Note that if $\Mc^\Sc \in \Mb^{+}_c$ then $\Qc[D\sm \Ab]$ is a positive and continuous Markov kernel by \cref{lem:pcmk}. Therefore, we can always modify $\Qc[D]^{|D^-}$ on a $\mu_\Ab$-null set so as to obtain a continuous version; such a modification remains a version of the original conditional kernel.
\end{proof}

\begin{remark}\label{rem:idp_sound_III}
    We provide some intuition for the third equality in the proof of \cref{lem:IDP_sound_III}. Consider an ADMG $\Af=(\Vc,\Ec)$ and an SCM $\Mc$ with discrete variables such that $\Gb(\Mc)=\Af$. Let $a,b,c\in \Vc$ be distinct nodes. Assume $\de_{\Af}(a)\cap \mathsf{Distr}_{\Af}(a)=\{a\}$. If $b\in \de_{\Af}(a)$ and $c\in \Vc\sm \de_{\Af}(a)$, then 
    \[
        p_\Mc(x_b\miid \Do(x_a))=p_\Mc(x_b\mid x_a) \quad \text{ and }\quad p_\Mc(x_c\miid \Do(x_a))=p_\Mc(x_c).
    \]
    In words: for descendants of fixable node $a$, intervening on $X_a$ acts like conditioning on $X_a$; for non-descendants of $a$, intervening on $X_a$ has no effect. The measure-theoretic causal calculus in \cref{thm:causal_calculus_mag_pag}, together with \cref{lem:IDP_sound_ci_2,lem:IDP_sound_ci_3}, makes the corresponding principle rigorous in our iSOPAG setting with formal replacement:
    $
        \big(\text{node},\de_{\cdot}(\cdot)\big) \curvearrowleft \big(\text{bucket}, \pde_{\cdot}(\cdot)\big).
    $
\end{remark}

\begin{lemma}[Possible descendant and posterior]\label{lem:pde_vs_pdet}
    Let $\Pf_D$ be an induced subgraph of iSOPAG $\Pf$ over $D$ and $\Bb$ a $\Pf_D$-bucket. Then we have $\pde_{\Pf_D}(\Bb)=\pdet_{\Pf_D}(\Bb)$. 
\end{lemma}

\begin{proof}[Proof of \cref{lem:pde_vs_pdet}]
    First, note that $\Bb\subseteq \pde_{\Pf_D}(\Bb)$ and $\Bb\subseteq \pdet_{\Pf_D}(\Bb)$. It is easy to see $\pde_{\Pf_D}(\Bb)\subseteq \pdet_{\Pf_D}(\Bb)$. If $\pdet_{\Pf_D}(\Bb)\subseteq \Bb$, then we are done. Assume now that $\Bb\subsetneq \pdet_{\Pf_D}(\Bb)$. Let $a\in \pdet_{\Pf_D}(\Bb)\sm \Bb$. Then there is a potentially anterior path from $b$ to $a$ for some $b\in \Bb$. Then there must exist $\tilde{b}\in \Bb$ such that $a\in \pde_{\Pf_D}(\{\tilde{b}\})\subseteq \pde_{\Pf_D}(\Bb)$.  To see this, we can argue similarly to \cref{lem:podpath}. Since $v_{i-1}\suh v_{i} \out v_{i+1}$ and $v_{i-1}\suh v_{i} \tuo v_{i+1}$ cannot occur, we can find a potentially anterior path from $b$ to $a$ such that there exists a node $\tilde{b}$ on $\pi$ such that all the edges without arrowheads are before it. This implies $\pdet_{\Pf_D}(\Bb)\subseteq \pde_{\Pf_D}(\Bb)$.
\end{proof}

\begin{lemma}[Topological order of buckets]\label{lem:topo_order}
    Let $\Pf=(\Ic,\Vc,\Ec)$ be an iSOPAG. Let $D\subseteq \Ic\cup \Vc$. There exists a partial order $\prec$ on buckets of $\Pf_D$ such that $\Ab\prec \Bb$ for every distinct $\Pf_D$-buckets $\Ab$ and $\Bb$ such that $\Ab\subseteq \pan_{\Pf_D}(\Bb)$. 
\end{lemma}

\begin{proof}[Proof of \cref{lem:topo_order}]
    It suffices to show that if there is a potentially directed path $\pi$ from $a\in \Ab$ to $b\in \Bb$ then it is impossible to have a potentially directed path from $\tilde{b} \in \Bb$ to $\tilde{a}\in \Ab$. Assume on the contrary that there is a shortest potentially directed path 
    \[\tilde{\pi}:\tilde{b}=u_0\sus \cdots \sus u_m=\tilde{a}\] 
    from $\tilde{b} \in \Bb$ to $\tilde{a}\in \Ab$. Since $\tilde{\pi}$ is potentially directed and $\Ab$ and $\Bb$ are distinct buckets, by the first claim in the proof of \cref{lem:podpath} we must have $u_{m-1}\tuh u_m$ or $u_{m-1}\ouh u_m$. \cref{lem:prop_bucket} implies that $u_{m-1}\tuh a$ or $u_{m-1}\ouh a$. Thus, there is a potentially directed path from $\tilde{b}$ to $b$ with arrowhead at $b$. By \cref{lem:prop_bucket}, there is a non-trivial potentially directed path from $\tilde{b}$ to itself. \cref{lem:podpath} tells us that this is impossible. This finishes the proof.
\end{proof}

\begin{lemma}[Conditional independence~I]\label{lem:IDP_sound_ci_I}
    Let $\Pf=(\Ic,\Vc,\Ec)$ be an iSOPAG and let $A,B\subseteq \Vc$ be disjoint with $A\ne \emptyset$. Write 
    \[
        \Dc\coloneqq \pant_{\Pf_{\Vc\sm B}}(A) \quad \text{ and } \quad  H\coloneqq (\Vc\sm (\Dc\cup B))\cup (\Ic\sm \tilde{\Dc}),
    \]
    where
    \[
        \tilde{\Dc}=\{i\in \Ic: \exists \text{ anterior path } i\sus v_1\sus \cdots \sus v_{n}\in A \text{ with $v_j \in \Vc$ for $1\le j< n$ (if any)}\}.
    \]
    Then we have
    \[
        A\sep{\mathsf{id}}{\Pf_{\Do(I_H,B)}} I_H \mid B\cup \tilde{\Dc}.
    \]
\end{lemma}

\begin{proof}[Proof of \cref{lem:IDP_sound_ci_I}]
    Assume on the contrary that there is an irreducible open path 
    \[\pi:A\ni a\sus v_1\sus \cdots \sus v_{n-1}\sut I_h \in I_H\]
    from $A$ to $I_H$ given $B\cup \tilde{\Dc}$ in $\Pf_{\Do(I_H,B)}$ (not to $B\cup (\Ic\sm H)$ since we condition on it). Note that path $\pi$ cannot contain colliders since otherwise $\pi$ is blocked at the colliders.  Since $\pi$ is definitely open and does not contain a collider and \cref{def:SOPAG} excludes the case $v_1\suh u_1\tuo w \out u_2 \hus v_2$, we know that $\pi$ is a potentially anterior path from $I_H$ to $A$ or from $A$ to $I_H$ not intersecting $B$. We first assume that $\pi$ is a potentially anterior path from $I_H$ to $A$. Since $H\cap \Dc=\emptyset$, we have that $\pi$ does not intersect $H$. Then we have $v_{n-1}\sut I_h$ with $v_{n-1}\notin H$ and there must be an edge between the nodes $v_{n-1}$ and $h\in \Vc\cap H$ such that it is not into $h$ by the definition of $\Pf_{\Do(I_H,B)}$. This implies $H\cap \Dc\ne \emptyset$, which contradicts the fact that $H\cap \Dc=\emptyset$. We now assume that $\pi$ is a potentially anterior path from $A$ to $I_H$ but not from $I_H$ to $A$. This implies that there exists $v_i\suh v_{i+1}$ on $\pi$ for some $i$. Since $\pi$ is definitely open without colliders, we must have $v_{n-2}\suh v_{n-1}$ and $v_{n-1}\tut I_h$, which is an impossible pattern.  Hence, we can conclude \[A\sep{\mathsf{id}}{\Pf_{\Do(I_H,B)}} I_H\mid B\cup \tilde{\Dc}.\]
\end{proof}

\begin{lemma}[Conditional independence~II]\label{lem:IDP_sound_ci_2}
    Let $\Pf=(\Ic,\Vc,\Ec)$ be an iSOPAG, $D\subseteq \Vc$ , and $A\subseteq D$. Let $\Bb_1\prec \cdots \prec \Bb_m$ be a topological order of $\Pf_D$-buckets. Set $R_1\coloneqq \re_{\Pf_D}(A)$ and $R_2\coloneqq \re_{\Pf_D}(D\sm R_1)$. Fix an arbitrary $\Bb_i$ and write $H_j\coloneqq D\sm (\Bb_i^\preceq\cup R_j)$ for $j=1,2$. Then:
    \begin{enumerate}[itemsep=0pt,label=(\roman*)]
        \item if $\Bb_i\subseteq R_1$, it holds $\Bb_i\sep{\mathsf{id}}{\Pf_{\Do(I_{H_1},D^c)}} I_{H_1}\mid \Bb_i^\prec\cup D^c\cup \Ic$;
        \item if $\Bb_i\subseteq R_2$, it holds $\Bb_i\sep{\mathsf{id}}{\Pf_{\Do(I_{H_2},D^c)}} I_{H_2}\mid \Bb_i^\prec\cup D^c\cup \Ic$.
    \end{enumerate}
\end{lemma}

\begin{proof}[Proof of \cref{lem:IDP_sound_ci_2}]
    We prove the first statement and omit the second since the proof is similar. The case $H_1=\emptyset$ is trivial, so assume $H_1\ne \emptyset$. We assume that $\Ic=\emptyset$ because $\Ic$ is always conditioned on. Assume for contradiction that 
    \[
    \pi:\Bb_i\ni v_0\sus v_1\sus \cdots \sus v_{n-1}\sut I_h\in I_{H_1}
    \]
    is an irreducible open path from $\Bb_i$ to $I_{H_1}$ given $\Bb_i^\prec\cup D^c$ in $\Pf_{\Do(I_{H_1},D^c)}$. We claim that $n\geq 2$. Indeed, because $v_0\in \Bb_i$ and $h\in H_1$, by the definition of $\Pf_{\Do(I_{H_1},D^c)}$ there cannot be an edge between $v_0$ and $I_h$. Therefore, the case where $n=1$ is excluded. There are three cases: 
    \begin{enumerate}[label=\textbf{Case~\arabic*:}, ref=Case~\arabic*, leftmargin=*]
        \item $v_{n-1}\hus v_{n-2}$,
        \item $v_{n-1}\ous v_{n-2}$, and
        \item $v_{n-1}\tus v_{n-2}$.
    \end{enumerate}
    
     \textbf{Case 1}. We consider the case in which $v_{n-1}\hus v_{n-2}$. First, we assume that $v_{n-1}=h$. Then $h$ is a collider on $\pi$, i.e., $I_h\tuh h\hus v_{n-2}$. Since $h\in H_1$, we have \[h\notin \pan_{\Pf_{\Do(I_{H_1},D^c)}}(\Bb_i^\prec\cup D^c).\] Therefore, the path $\pi$ is blocked by $\Bb_i^\prec\cup D^c$ at node $h$ in $\Pf_{\Do(I_{H_1},D^c)}$. So, this case cannot occur. Second, we assume that $v_{n-1}\ne h$. In this case, since $v_{n-2}\suh v_{n-1}$ and $v_{n-1}$ is adjacent to $I_h$, we have $h\tuh v_{n-1}$ or $h\ouh v_{n-1}$ or $h\ouo v_{n-1}$ by \cref{def:int_pag}. This implies that node $v_{n-1}$ is a collider on $\pi$ after $\Bb_i^\prec$. Note that there cannot be a potentially directed path from $v_{n-1}$ to $\Bb_i^\prec$ and therefore path $\pi$ is blocked by $\Bb_i^\prec\cup D^c$ at $v_{n-1}$ in $\Pf_{\Do(I_{H_1},D^c)}$.

    \textbf{Case 2}. We consider the case in which $v_{n-1}\ous v_{n-2}$. We first assume that $v_{n-1}=h$. Since $v_{n-1}$ is of definite status, we have either $I_h\tut v_{n-1}\ous v_{n-2}$ or $I_h\tuo v_{n-1}\ous v_{n-2}$ with $I_h$ non-adjacent to $v_{n-2}$. The first case contradicts $\fci{6}$ by \cref{def:int_pag} and the second case cannot happen by \cref{def:int_pag}. Hence, we can exclude the case where $v_{n-1}=h$. Next assume that $v_{n-1}\ne h$. This implies that we have $h\tuo v_{n-1}$ or $h\ouo v_{n-1}$ or $h \out v_{n-1}$ with no undirected edges connecting to $v_{n-1}$. Since path $\pi$ is open given $\Bb_i^\prec\cup D^c$ in $\Pf_{\Do(I_{H_1},D^c)}$ and $h\in H_1$, we have $v_{n-1}\in \Bb_j$ and $v_{n-2}\in \Bb_k$ for some $j,k>i$.  Therefore, there must be a collider on $\pi(v_0,v_{n-1})$. 
    
    To prove this claim, we assume for contradiction that there are no colliders on $\pi(v_0,v_{n-1})$. If $v_{n-1}\ouh v_{n-2}$, then since $v_{n-2}$ is of definite status we have $v_{n-2}\tuh v_{n-3}$ (note that $n\ge 3$ because $v_{n-2}\notin \Bb_i$). Repeat the argument until we reach $v_0$. This then implies that $v_{n-1}$ is a possible ancestor of $v_0$, which contradicts the topological order. For the case where $v_{n-1} \ouo v_{n-2}$ or $v_{n-1}\out v_{n-2}$, consider $v_r$ with \[r\coloneqq \min\{t\mid v_t\in \Bb_j \text{ and } v_s\in \Bb_j \ \forall s: t\le s\le n-1\}.\] Then we can have $v_{r+1} \out v_r$ or $v_{r+1}\tuo v_r$ or $v_{r+1}\ouo v_r$ or $v_{r+1} \tut v_{r}$. Therefore, we have $v_r\tuh v_{r-1}$ or $v_r \ouh v_{r-1}$ since $v_{r}$ is of definite status. Applying the previous argument for the case where $v_{n-1}\ouh v_{n-2}$, we arrives at a contradiction as well. This finishes the proof of the claim. 
    
    Let $v_l$ be the first collider on $\pi$ after $v_{n-1}$. Then $v_l\notin \anc_{\Pf_{\Do(I_{H_1},D^c)}}(\Bb_i^\prec\cup D^c)$. Otherwise, we would have $v_{n-1}\in \pant_{\Pf_{\Do(I_{H_1},D^c)}}(\Bb_i^\prec\cup D^c)$, which is a contradiction. So the path $\pi$ is blocked by $\Bb_i^\prec\cup D^c$ in $\Pf_{\Do(I_{H_1},D^c)}$.

    \textbf{Case 3}. We finally consider the case in which $v_{n-1}\tus v_{n-2}$. First, we assume that $v_{n-1}=h$. If $h\tuo v_{n-2}$ or $h\tut v_{n-2}$, then we would have $I_h\tuo v_{n-2}$ or $I_h\tut v_{n-2}$ correspondingly, which contradicts the fact that the path $\pi$ is irreducible. For the case where $h\tuh v_{n-2}$, there must be colliders on $\pi$ by a similar argument to the claim in Case~2. A similar argument to before gives that the path $\pi$ is blocked by $\Bb_i^\prec\cup D^c$ at $v_k$ in $\Pf_{\Do(I_{H_1},D^c)}$ where $v_k$ is the first collider on $\pi$ after $h$. Therefore, we can exclude this case. Second, we assume that $v_{n-1}\ne h$. Note that node $h$ must be adjacent to node $v_{n-1}$ and the edge between them cannot have an arrowhead on $h$. Therefore, $v_{n-1}\in \Bb_j$ for some $j>i$ and similar to before there must be a blocked collider on $\pi(v_0,v_{n-1})$ given $\Bb_i^\prec\cup D^c$ in $\Pf_{\Do(I_{H_1},D^c)}$. So, the path $\pi$ cannot be open given $\Bb_i^\prec\cup D^c$ in $\Pf_{\Do(I_{H_1},D^c)}$. 
    
    This completes the proof.
\end{proof}

\begin{lemma}[Conditional independence~III]\label{lem:IDP_sound_ci_3}
    Let $\Pf=(\Ic,\Vc,\Ec)$ be an iSOPAG, $D\subseteq \Vc$ , and $A\subseteq D$. Let $\Bb_1\prec \cdots \prec \Bb_m$ be a topological order of $\Pf_D$-buckets. Set $R_1\coloneqq \re_{\Pf_D}(A)$ and $R_2\coloneqq \re_{\Pf_D}(D\sm R_1)$. Fix an arbitrary $\Bb_i$ and write $H_j\coloneqq \Bb_i^\prec\sm R_j$ and $T_j\coloneqq (R_j\cup \Bb_i^\preceq)^c$ for $j=1,2$. Then:
    \begin{enumerate}[itemsep=0pt,label=(\roman*)]
        \item if $\Bb_i\subseteq R_1$, it holds $\Bb_i\sep{\mathsf{id}}{\Pf_{\Do(I_{H_1},T_1)}} I_{H_1}\mid \Bb_i^\prec\cup T_1 \cup \Ic$;
        \item if $\Bb_i\subseteq R_2$, it holds $\Bb_i\sep{\mathsf{id}}{\Pf_{\Do(I_{H_2},T_2)}} I_{H_2}\mid \Bb_i^\prec\cup T_2\cup \Ic$.
    \end{enumerate}
\end{lemma}

\begin{proof}[Proof of \cref{lem:IDP_sound_ci_3}]
    We show the first statement. The second can be argued similarly. The case where $H_1=\emptyset$ is trivial so we assume $H_1\ne \emptyset$. We assume $\Ic=\emptyset$ because $\Ic$ is always conditioned on. Assume for contradiction that there exists an irreducible open path \[\pi:v_0\sus v_1\cdots \sus v_{n-1}\sut I_h\] from $\Bb_i$ to $I_{H_1}$ given $\Bb_i^\prec\cup T_1$ in $\Pf_{\Do(I_{H_1},T_1)}$. There are two cases: 
    \begin{enumerate}[label=\textbf{Case~\arabic*:}, ref=Case~\arabic*, leftmargin=*]
        \item $v_{n-1}=h$ and
        \item $v_{n-1}\ne h$.
    \end{enumerate}

    \textbf{Case 1}. In this case, we have $I_h\tus h\sus v_{n-2}$. First note that we cannot have $h\ous v_{n-2}$, $h\tut v_{n-2}$, or $h\tuo v_{n-2}$ by an argument similar to the proof of Cases~2 and 3 of \cref{lem:IDP_sound_ci_2}. It is impossible to have $h\tuh v_{n-2}$, since $h\in H_1$ and therefore in this case the path $\pi$ is blocked by $\Bb_i^\prec\cup T_1$. We now assume that we have $h\hus v_{n-2}$ in $\Pf_{\Do(I_{H_1},T_1)}$. If we have $h\hut v_{n-2}$ or $h\huo v_{n-2}$ on $\pi$ in $\Pf_{\Do(I_{H_1},T_1)}$, then $v_{n-2}\in \Bb_i^\prec$ and $v_{n-2}$ is a non-collider on $\pi$ in $\Pf_{\Do(I_{H_1},T_1)}$. So, this case can be excluded. We are left with the case where we have $h\huh v_{n-2}$ in $\Pf_{\Do(I_{H_1},T_1)}$. If $v_{n-2}\in \Bb_i$, then we have $h\in R_1$, which contradicts the fact that $h\in H_1$. Therefore, we have $v_{n-2}\notin \Bb_i$. There must be colliders on $\pi(v_0,v_{n-1})$. Indeed, if $v_{n-2}\in \Bb_i^\prec$, then $v_{n-2}$ must be a collider on $\pi$ so that path $\pi$ is open given $\Bb_i^\prec\cup T_1$ in $\Pf_{\Do(I_{H_1},T_1)}$. If $v_{n-2}\in \Bb_j$ for some $j>i$, then $\pi(v_0,v_{n-2})$ cannot be a potentially anterior path from $v_{n-2}$ to $v_0$, since $v_0\in \Bb_i$. This also implies that there must be colliders on $\pi(v_0,v_{n-1})$.  Let $v_i$ be the first collider on $\pi$ after node $v_0$. Having $v_i\in \anc_{\Pf_{\Do(I_{H_1},T_1)}}(\Bb_i^\prec)$ implies that $\Bb_i\ni v_0\in \pant_{\Pf_{\Do(I_{H_1},T_1)}}(\Bb_i^\prec)$, which causes a contradiction. In summary, we have excluded the case where $v_{n-1}=h$.

    \textbf{Case 2}. Since $h\in \Bb_i^\prec\sm R_1$ and $v_0\in \Bb_i$, \cref{def:int_pag,lem:region} imply $n\geq 2$. We have two subcases: 
    \begin{enumerate}[label=\textbf{Case~2.\arabic*:}, ref=Case~\arabic*, leftmargin=*]
        \item  $v_{n-1}$ is a non-collider on $\pi$;
        \item $v_{n-1}$ is a collider on $\pi$. 
    \end{enumerate}
    
    \textbf{Case 2.1}. In this case, we first consider the subcase where $I_h\tut v_{n-1}$ is in $\Pf_{\Do(I_{H_1},T_1)}$. By the definition of $\Pf_{\Do(I_{H_1},T_1)}$, there must be an undirected edge $h\tut v_{n-1}$ in $\Pf_{\Do(I_{H_1},T_1)}$. Therefore, $v_{n-1}\in \Bb_i^\prec$. It means that the path $\pi$ cannot be open given $\Bb_i^\prec\cup T_1$ in $\Pf_{\Do(I_{H_1},T_1)}$, so this case cannot occur. If we have $I_h\tuo v_{n-1}$, then by the definition of $\Pf_{\Do(I_{H_1},T_1)}$ there must be $h\tuo v_{n-1}$ or $h\ouo v_{n-1}$ or $h\out v_{n-1}$ in $\Pf_{\Do(I_{H_1},T_1)}$. This implies that $v_{n-1}\in \Bb_i^\prec$ and therefore the path $\pi$ cannot be open given $\Bb_i^\prec\cup T_1$ in $\Pf_{\Do(I_{H_1},T_1)}$. So this case cannot occur. We then assume that $I_h\tuh v_{n-1}$ is in $\Pf_{\Do(I_{H_1},T_1)}$. Then there must be $\Pf$-invisible $h\tuh v_{n-1}$ or $h\ouh v_{n-1}$ in $\Pf_{\Do(I_{H_1},T_1)}$. Since $h\notin R_1$, we have that $v_{n-1}\notin R_1$ by Lemma~\ref{lem:region}. Note that the cases $v_{n-1}\tuo v_{n-2}$ and $v_{n-1}\tut v_{n-2}$ are excluded by the fact that the edge between $h$ and $v_{n-1}$ has an arrowhead on $v_{n-1}$. Since node $v_{n-1}$ is a non-collider on path $\pi$, we have $v_{n-1}\tuh v_{n-2}$.  We have $v_{n-1}\notin \Bb_i^\preceq$, since otherwise $v_{n-1}\in \Bb_i^\prec$ implies that the path $\pi$ is blocked by $\Bb_i^\prec\cup T_1$ in $\Pf_{\Do(I_{H_1},T_1)}$. Also note that $v_{n-2}\ne v_0\in \Bb_i$, since $v_{n-1}\in \Bb_i^\succ$ and $v_{n-1}\tuh v_{n-2}$ is present. Then there exists a collider on the subpath $\pi(v_0,v_{n-1})$. Similar to previous argument, we can conclude that this would lead to a contradiction and therefore this case cannot occur.

    \textbf{Case 2.2}. In the following we assume that $v_{n-1}$ is a collider on $\pi$. This implies $v_{n-2}\suh v_{n-1}$. If the edge between $v_{n-1}$ and $v_{n-2}$ is not bidirected, then $v_{n-2}\in \Bb_i^\prec$ since $v_{n-1}\in \pan(\Bb_i^\prec)$ and $v_{n-2}\in \pan(v_{n-1})$. In this case, node $v_{n-2}$ can only be a non-collider on $\pi$, which is blocked by $\Bb_i^\prec\cup T_1$ in $\Pf_{\Do(I_{H_1},T_1)}$. So we can exclude this case as well. Assume we have $v_{n-1}\huh v_{n-2}$. If we have $h\tuh v_{n-1}$ $\Pf$-invisible or $h\ouh v_{n-1}$, then by \cref{lem:region}, $v_{n-1}\notin R_1$ and therefore $v_{n-2}\notin R_1$. If we have $h\ouo v_{n-1}$, then $v_{n-1}\notin R_1$ since $v_{n-1}$ is in the same bucket with $h$ and $h\notin R_1$. Again, \cref{lem:region} implies $v_{n-2}\notin R_1$. Note that $\Bb_i\subseteq R_1$ and therefore $v_{n-2}\notin \Bb_i$. If $v_{n-2}\in \Bb_i^{\succ}$, then there is a collider on the subpath $\pi(v_0,v_{n-1})$. Similar to previous argument, we can exclude this case. If $v_{n-2}\in \Bb_i^\prec$, then we know that node $v_{n-2}$ has to be a collider on path $\pi$, since $\pi$ is open given $\Bb_i^\prec\cup T_1$ in $\Pf_{\Do(I_{H_1},T_1)}$. Then a similar argument to the above part of Case~2.2 can be applied to $v_{n-3}\suh v_{n-2}\huh v_{n-1}$. It implies that $v_{n-4}\huh v_{n-3}$ and $v_{n-4}\notin R_1$ and $v_{n-4}$ is a collider on $\pi$. Repeating the argument implies that $\pi$ is infinite, which is a contradiction to the fact that our path $\pi$ is finite.
    
    We have excluded all possible cases, and therefore can conclude  
    \[
        \Bb_i \sep{\mathsf{id}}{\Pf_{\Do(I_{H_1},T_1)}} I_{H_1}\mid \Bb_i^\prec \cup T_1.
    \]
\end{proof}

\begin{lemma}[Conditional independence~IV]\label{lem:IDP_sound_ci_IV}
Let $\Pf=(\Ic,\Vc,\Ec)$ be an iSOPAG. Let $D\subseteq \Vc$ and $\Ab\subseteq D$ be a $\Pf_D$-bucket. Write $D^c\coloneqq \Vc\sm D$ and $\pde_{\Pf_D}(\Ab)^c\coloneqq \Vc\sm \pde_{\Pf_D}(\Ab)$. If $\pcc_{\Pf_D}(\Ab)\cap \pde_{\Pf_D}(\Ab)\subseteq \Ab$, then we have 
\[
     D\sm \pde_{\Pf_D}(\Ab)\sep{\mathsf{id}}{\Pf_{\Do(I_\Ab,D^c)}} I_{\Ab}\mid D^c\cup \Ic, \quad \text{and}\quad \pde_{\Pf_D}(\Ab)\sm \Ab\sep{\mathsf{id}}{\Pf_{\Do(I_\Ab,D^c)}} I_{\Ab}\mid \pde_{\Pf_D}(\Ab)^c\cup \Ab\cup \Ic.
\]
\end{lemma}

\begin{proof}[Proof of \cref{lem:IDP_sound_ci_IV}]
First note that WLOG we can assume that $\Ic=\emptyset$. Let 
\[
\underbrace{\Bb_1\prec \cdots \prec \Bb_{l-1}}_{\Ab^\prec} \prec \Ab \prec \underbrace{\Bb_{l+1} \prec \cdots \prec \Bb_\ell}_{\Ab^\succ}
\]
be a topological order of $\Pf_D$-buckets such that $\bigcup_{i\ge l}\Bb_i= \pde_{\Pf_D}(\Ab)$. Write $\Bb^{(i)}\coloneqq \bigcup_{j\leq i}\Bb_j\sm \Ab$. The desired separations are equivalent to: 
\[
     \Ab^{\prec}\sep{\mathsf{id}}{\Pf_{\Do(I_\Ab,D^c)}} I_{\Ab}\mid D^c\cup \Ic, \quad \text{ and }\quad \Ab^{\succ}\sep{\mathsf{id}}{\Pf_{\Do(I_\Ab,D^c)}} I_{\Ab}\mid \Ab^{\preceq}\cup D^c\cup \Ic.
\]

\textbf{Step~1: show the first separation}. We show the first separation. Assume for contradiction that there is an irreducible open path $\pi$ from $\Bb_i\subseteq \Ab^{\prec}$ to $I_\Ab$ given $D^c$ in $\Pf_{\Do(I_\Ab,D^c)}$. Note that the path $\pi$ must contain a collider. Otherwise, similar to the argument in \cref{lem:IDP_sound_ci_I}, we can assume WLOG that $\pi$ is a potentially anterior path from $I_\Ab$ to $\Bb_i$ since there are only tails to input nodes. From \cref{def:int_pag,lem:pde_vs_pdet}, this contradicts the fact that $\Bb_i$ does not contain any possible descendants of $\Ab$. Colliders on $\pi$ cannot have potentially or definitely directed paths to $D^c$, since $D^c$ is hard-manipulated in $\Pf_{\Do(I_{\Ab},D^c)}$. Therefore, $\pi$ must be blocked by $D^c$. This shows the first separation.

\textbf{Step~2: show the second separation}. We show the second separation. By the Left Contraction rule of $id$-separation \cite[Theorem~5.11]{forre2021transitional}, to derive 
\[
    \Ab^{\succ}\sep{\mathsf{id}}{\Pf_{\Do(I_\Ab,D^c)}} I_{\Ab}\mid \Ab^{\preceq}\cup D^c,
\] 
it suffices to show for every $\Bb_i\subseteq \Ab^{\succ}$
\[
    \Bb_i \sep{\mathsf{id}}{\Pf_{\Do(I_\Ab,D^c)}} I_\Ab \mid \Bb^{(i-1)}\cup \Ab\cup D^c.
\]

Assume for contradiction that there is a shortest open path \[\pi:b\sus v_1\sus \cdots \sus v_{n-1}\sut I_a\] from $\Bb_i\subseteq \Ab^{\succ}$ to $I_\Ab$ given $\Bb^{(i-1)}\cup D^c \cup \Ab$ in $\Pf_{\Do(I_\Ab,D^c)}$. 

\textbf{Step~2.1: show $n\ge 2$ and $v_{n-2}\notin \Ab$}. Since $\pcc_{\Pf_D}(\Ab)\cap \pde_{\Pf_D}(\Ab)\subseteq \Ab$, if we have that node $a$ is adjacent to node $b$ for $a\in \Ab$ and $b\in \Vc\sm \Ab$ then we have $a\hus b$ or visible edge $a\tuh b$ in $\Pf$. This implies that  we have $n\geq 2$, i.e., node $I_a$ cannot be adjacent to node $b$, by the fact that $b\notin \Ab$ and the definition of $\Pf_{\Do(I_\Ab,D^c)}$. We show $v_{n-2}\notin \Ab$. For that, assume on the contrary that $v_{n-2}\in \Ab$. Since $\pi$ is open given $\Bb^{(i-1)}\cup D^c \cup \Ab$ in $\Pf_{\Do(I_\Ab,D^c)}$, node $v_{n-2}$ must be a collider on $\pi$. Therefore, the edge between nodes $I_{v_{n-2}}$ and $v_{n-2}$ must be $I_{v_{n-2}}\tuh v_{n-2}$ and $v_{n-2}$ is a collider on the path $\pi(b,v_{n-2})\oplus (v_{n-2},I_{v_{n-2}})$ and therefore it is open given $\Bb^{(i-1)}\cup D^c \cup \Ab$ in $\Pf_{\Do(I_\Ab,D^c)}$ and shorter than $\pi$. This contradicts the construction of $\pi$ being shortest and therefore $v_{n-2}\notin \Ab$.

\textbf{Step~2.2: show $n>2$ and $v_{n-2}\in \Ab^{\prec}$}. Since $\pcc_{\Pf_D}(\Ab)\cap \pde_{\Pf_D}(\Ab)\subseteq \Ab$, we have $v_{n-1}\in \Ab$ by \cref{def:int_pag}. Note that $v_{n-1}$ cannot be a non-collider since we condition on $\Bb^{(i-1)}\cup D^c \cup \Ab$. So we have $v_{n-1}\hus v_{n-2}$. If we have $v_{n-1}\hut v_{n-2}$ or $v_{n-1}\huo v_{n-2}$, then $v_{n-2}\in \Ab^{\prec}$ and therefore $v_{n-2}\notin \Bb_i$ since $\Bb_i\subseteq \Ab^{\succ}$. If we have $v_{n-1}\huh v_{n-2}$, then $v_{n-2}\in \pcc_{\Pf_D}(\Ab)$. It implies  $v_{n-2}\notin \pde_{\Pf_D}(\Ab)$, since $v_{n-2}\notin \Ab$ and $\pcc_{\Pf_D}(\Ab)\cap \pde_{\Pf_D}(\Ab)\subseteq \Ab$. Therefore, $v_{n-2}\in \Ab^{\prec}$ and $v_{n-2}\notin \Bb_i$. This implies that the path $\pi$ has non-endnodes between nodes $v_{n-1}$ and $b$, i.e., $n> 2$. 

\textbf{Step~2.3: show $\pi$ blocked}. We first consider the case where all the nodes between $b$ and $v_{n-1}$ are colliders, i.e., $b \suh v_1\huh \cdots \huh v_{n-2}\hus v_{n-1}$. Note that in this case, since $v_{n-2}\in \Ab^{\prec}$ and $v_{n-1}\in \Ab$, the edge between $v_{n-2}$ and $v_{n-1}$ must be $v_{n-2}\huh v_{n-1}$. If $b\huh v_1$, then $b\in \pcc_{\Pf_D}(\Ab)$ since we have $b\huh \cdots \huh v_{n-1}$ in $\Pf_D$ and $v_{n-1}\in \Ab$. Recall that $b\in \Bb_i \subseteq \Ab^{\succ}=\pde_{\Pf_D}(\Ab)$. This means that $b\in \pcc_{\Pf_D}(\Ab) \cap \pde_{\Pf_D}(\Ab)$ while $b\notin \Ab$, which contradicts $\pcc_{\Pf_D}(\Ab) \cap \pde_{\Pf_D}(\Ab)\subseteq \Ab$. If we have $b\tuh v_1$ or $b\ouh v_1$, then $v_1\in \Ab^{\succ}$ and $v_1\in \pde_{\Pf_D}(\Ab)$. So, $v_1\in \pcc_{\Pf_D}(\Ab) \cap \pde_{\Pf_D}(\Ab)$ but $v_1\notin \Ab$. This again contradicts $\pcc_{\Pf_D}(\Ab) \cap \pde_{\Pf_D}(\Ab)\subseteq \Ab$. Hence, $\pi$ contains a non-collider $v_j$ for some $1\leq j \leq n-2$. Since $\pi$ is open at $v_j$ given $\Bb^{(i-1)}\cup D^c \cup \Ab$, we have $v_j\in \Bb_k$ for some $k>i$. We shall show that there must be a collider $u$ on $\pi(b,v_{n-1})$ such that $u\in \Bb_t$ for some $t>i$. If this is proved, then we can conclude that the path $\pi$ is blocked at $u$ by $\Bb^{(i-1)}\cup D^c \cup \Ab$ in $\Pf_{\Do(I_\Ab,D^c)}$ since there cannot be a directed path from node $u$ to $\Bb^{(i-1)}\cup D^c \cup \Ab$. 

Indeed, let $v_{r}$ and $v_s$ be the most left and right nodes in buckets after $\Bb_i$ on $\pi(b,v_{n-1})$ respectively, i.e., 
\[r\coloneqq \min\{t\mid \exists \Bb_m \text{ with }  m>i \text{ s.t.\ } v_t\in \Bb_m\} \text{ and } s\coloneqq \max\{t\mid \exists \Bb_m \text{ with }  m>i \text{ s.t.\ } v_t\in \Bb_m\}.\] Nodes $v_r$ and $v_s$ are well-defined by the existence of $v_j$ and finiteness of $\pi$ (it may be $r=s$). Assume that the target result does not hold, i.e., there are no colliders on $\pi(b,v_{n-1})$ such that $u\in \Bb_t$ for some $t>i$. By the definition of $v_s$, it holds $v_s\hus v_{s+1}$. Since $v_{s}$ is a definite non-collier on $\pi(b,v_{n-1})$, we have $v_{s-1}\hut v_s$. We can repeat the argument to $v_{s-q}$ for $1\le q \le s-r-1$ (if $s>r+1$) and deduce that $v_r\hus v_{r+1}$. Note that we must have $v_{r-1}\suh v_{r}$ by the definition of $v_r$. Then $v_r$ is a collider such that $v_r\in \Bb_t$ for some $t>i$. This contradicts the starting hypothesis, and hence proves the claim.

Overall, there cannot be an open path from $\Bb_i$ to $I_\Ab$ given $\Bb^{(i-1)}\cup D^c \cup \Ab$ in $\Pf_{\Do(I_\Ab,D^c)}$ and we can conclude for every $\Bb_i\subseteq \Ab^{\succ}$
\[
    \Bb_i \sep{\mathsf{id}}{\Pf_{\Do(I_\Ab,D^c)}} I_\Ab \mid \Bb^{(i-1)}\cup \Ab\cup D^c.
\]
This finishes the proof.

\end{proof}

\subsection[Proof of Theorem~\ref{thm:idp_complete}]{Proof of \cref{thm:idp_complete}}

We prove \cref{thm:idp_complete}, building on ideas from \cite{jaber19idencom,jaber22causal}. See \cref{fig:pf_structure_idp_complete} for an overall structure of the proof. We start with the following lemma, which provides a sufficient graphical criterion for non-identifiability in the PAG setting—analogous in form to \cite[Theorem~3]{tian02general}. 

\begin{lemma}\label{lem:noniden}
    Let $\Pf=(\emptyset,\Vc,\Ec)$ be a COPAG and $A,B\subseteq \Vc$ be disjoint subsets. If there exists a proper potentially anterior path \[\pi:B\ni v_0\sus \cdots \sus v_n\in A\] from $B$ to $A$ where $v_0\sus v_1$ is not a visible directed edge, then $\Prb_{\Mc}(X_A\mid X_{\Sc}=\1 \miid \Do(X_B))$ is not identifiable in $\Mb^{+}(\Pf)$.
\end{lemma}

\begin{proof}[Proof of \cref{thm:idp_complete}]
Since $\idp(\Pf;A,B)$ outputs $\textsc{Fail}$, there exists $(C,T)$ with $\emptyset\neq C\subsetneq T\subseteq \Vc$ such that
\begin{enumerate}
    \item [(i)] for all buckets $\Bb\subsetneq C$ of $\Pf_{\Dc}$ we have $\re_{\Pf_C}(\Bb)=C$, and

    \item [(ii)] for all buckets $\Bb\subseteq T\setminus C$ of $\Pf_{\Dc}$ it holds that $\pcc_{\Pf_T}(\Bb)\cap\pde_{\Pf_T}(\Bb)\nsubseteq \Bb$.
\end{enumerate}
Note that if $\idp(\Pf;A,B)$ fails and there exists an irreducible potentially anterior path from $B$ to $A$ in $\Pf$ that starts with an edge that is not visible directed, then by \cref{lem:noniden} we have that $\Prb_{\Mc}(X_A\mid X_{\Sc}=\1 \miid \Do(X_B))$ is not identifiable and therefore we are done. Hence, in the following we can assume that \cref{cond:complete_I} holds.   

\begin{condition}\label{cond:complete_I}
    $\idp(\Pf;A,B)$ fails and every proper potentially anterior path from $B$ to $A$ in $\Pf$ starts with a visible directed edge.
\end{condition}


If $\idp(\Pf;A,B)$ fails but \cref{cond:complete_I} does not hold, then by contradiction one can prove that every $\Pf$-bucket is either entirely contained in $\Dc$ or entirely contained in $\Vc\sm \Dc$.\footnote{In the procedure of sIDP, the only operation that can split buckets is Rule~L0. Rule~L1 and Rule~L2 do not split buckets. Essentially, \cref{cond:complete_I} guarantees that the step of taking $\Dc=\pant_{\Pf_{\Vc\sm B}}(A)$ in sIDP will not split buckets.} To see this, we assume that this is not the case. Let $d\in \Dc$ and $c\in \Vc\sm \Dc$ be in the same $\Pf$-bucket. Let $\pi$ be a path connecting $d$ and $c$ and consisting of nodes in the same bucket of $d$ and $c$. If $\pi$ does not intersect $B$, then this would imply that $c\in \pant_{\Pf_{\Vc\sm B}}(A)=\Dc$, which contradicts the fact that $c\notin \Dc$. If $\pi$ intersects $B$, then we have a proper potentially anterior path from $B$ to $A$ in $\Pf$ that starts with an edge that is not visible directed. This contradicts \cref{cond:complete_I}. Hence, we can conclude: 

\begin{statement}[No buckets split]\label{state:no_bucket_split}
Assume \cref{cond:complete_I} holds.  Then Rule~L0 of sIDP does not split any buckets in $\Pf$. Note that Rule~L1 and Rule~L2 do not split buckets. Therefore, no buckets of $\Pf$ are split during the whole process of the sIDP, i.e., every $\Pf$-bucket contained in $T$ is either entirely in $C$ or entirely in $T\sm C$.  
\end{statement}   

This no-bucket-split observation allows us to construct a MAG $\Mf\in [\Pf]_\Ms$ by the following orientation scheme.

\begin{lemma}\label{lem:construct_MAG}

    Suppose the setting of \cref{thm:idp_complete} and \cref{cond:complete_I}. We can construct a MAG $\Mf\in[\Pf]_\Ms$ by the following procedures:
   
    \begin{enumerate}
        \item Orient $v\ouh u$ as $v\tuh u$ and $v\tuo u$ as $v\tut u$. Orient $v\ouo u$ as $v\tut u$ if there is no arrowhead into $v$ or $u$. Denote by $\Pf^{tag}$ the resulting graph.
        \item For circle component $\Cb\subseteq \Pf_{\Vc\sm T}^{tag}$, orient it into a DAG without unshielded colliders.
        \item For each $\Pf$-bucket $\Bb$ contained in $T\sm C$, pick one $t\in \Bb$ such that $t\in \pcc_{\Pf_T}(z)$ where $z\in \pch_{\Pf_T}(t)$ and $z\notin \Bb$. Then orient the circle component in $\Pf_{T\sm C}^{tag}$ contained in $\Bb$ into a DAG without unshielded colliders such that for any $t\ouo b$ (if any) where $b\in \Bb$, orient $t\ouo b$ as $t\tuh b$. Denote by $\tilde{T}$ the set of these nodes $t$.
        \item Let $\Bb_1\prec \cdots \prec \Bb_m$ be a topological order of $\Pf$-buckets contained in $C$:
        \begin{enumerate}
            \item In $\Bb_m$, pick arbitrarily a node $c^*\in \Bb_m$. Then orient the circle component in $\Pf_{C}^{tag}$ contained in $\Bb_m$ into a DAG without unshielded colliders such that for any $c^*\ouo d$ (if any), orient $c^*\ouo d$ as $c^*\tuh d$.
            \item For every $\Bb_i$ with $1\leq i<m$, arbitrarily choose $c\in \pcc_{\Pf_C}(c^*)\cap \Bb_i$. Then orient the circle component in $\Pf_{C}^{tag}$ contained in $\Bb_i$ into a DAG without unshielded colliders such that all edges $c\ouo d$ (if any) are oriented as $c\tuh d$. 
        \end{enumerate}  
        Denote by $\hat{C}$ the set of these nodes $c$ and $c^*$, and define $\hat{T}\coloneqq \hat{C}\dcup \tilde{T}$.
    \end{enumerate}
    Furthermore, \cref{cond:complete_I} guarantees $\hat{C}\ne \emptyset$, $\hat{T}\ne \emptyset$, and $\hat{T}\sm \hat{C}\ne \emptyset$.
\end{lemma}

This MAG $\Mf$ satisfies the following properties:
\begin{lemma}\label{lem:prop_construct_MAG_III}
    Let $\Mf$ be the MAG constructed in \cref{lem:construct_MAG}. Let $B$ be as in \cref{thm:idp_complete}. Then there exist subsets $\breve{T}\subseteq \hat{T}$ and $\breve{C}\subseteq \hat{C}$ such that $\emptyset\ne \breve{C}\subsetneq \breve{T}$ and:
    \begin{enumerate}
        \item $\breve{C}$ is a single c-component in $\Mf_{\breve{C}}$; 
        
        \item for every $t\in \breve{T}\sm \breve{C}$, there exists $u\in \ch_{\Mf_{\breve{T}}}(t)$ such that $t$ and $u$ are in the same c-component in $\Mf_{\breve{T}}$;

        \item $B\cap (\breve{T}\sm \breve{C})\ne \emptyset$; and

        \item $\breve{C}\subseteq \anc_{\Mf_{\Vc\sm B}}(A)$, which also implies that $\breve{C}\cap B=\emptyset$.
    \end{enumerate}
\end{lemma}

We consider isADMG \[\Af\coloneqq \cat{isADMG}(\Mf)\in [\Mf]_\Gs.\] In $\Af_{\breve{T}}$, the subgraph 
$\Af_{\breve{C}}$ is a single c-component and every $t\in \breve{T}\sm \breve{C}$ has a child $u$ such that $t$ and $u$ are in the same c-component. It follows that $\breve{T}=\anc_{\Af_{\breve{T}}}(\breve{C})$ and $\breve{T}$ forms a single c-component. Let $R$ be the root set of $\breve{C}$ (which always exists since we can take $R=\breve{C}$) in $\Af_{\breve{C}}$. Then $R$ is also the root set of $\breve{T}$ in $\Af_{\breve{T}}$. We can remove directed edges from $\Af_{\breve{T}}$ so that subgraphs over $\breve{C}$ and $\breve{T}$ form $R$-rooted C-forests. We want to apply \cref{prop:hedge} with $(\Gf,\Sc,D)\curvearrowleft (\Af,\Sc_{\Af},D)$. Note that $(\breve{T}\sm \breve{C})\cap B\ne \emptyset$ and $\breve{C}\cap B=\emptyset$. We have \[R\subseteq \breve{C}\subseteq \anc_{\Af_{\Vc\sm B}}(A)\subseteq \anc_{\Af_{\Do(B\cup D)}}(A\cup (\Sc_\Af\sm D)),\] since $D\subseteq \Sc_{\Af}$ does not have any children. Hence, ($\breve{T},\breve{C}$) forms a hedge for $(A\cup (\Sc_\Af\sm D),B\cup D)$ in $\Af$ and the interventional kernel $\Prb_{\Mc}(X_A\mid X_{\Sc}=\1 \miid \Do(X_B))$ is not identifiable in $\Mb^{+}_c(\Pf)$.
\end{proof}

We first show \cref{lem:noniden}. The proofs of \cref{lem:construct_MAG,lem:prop_construct_MAG_III} are lengthy and are given in \cref{sec:pf_lem_cons_mag,sec:pf_prop_mag}, respectively.

\begin{proof}[Proof of \cref{lem:noniden}]
    We can construct a MAG $\Mf\in [\Pf]_\Ms$ with an irreducible anterior path from $B$ to $A$ that starts with an invisible directed edge $v_0\tuh v_1$ or undirected edge $v_0\tut v_1$ by \cref{sopag:p4,lem:orientation_invisible}.
    
    There exists an isADMG $\Af\in [\Mf]_\Gs$ such that:
    \begin{enumerate}
        \item [(i)] $v_0\tuh v_1$ and $v_0\huh v_1$ are in $\Af$,
        \item [(ii)] if $v_{i}\tut v_{i+1}$ on $\pi$ (if any) then in $\Af$ we have \[v_i\tuh v_{i+1} \quad \text{ and }\quad  v_i\tuh s_{v_iv_{i+1}}\hut v_{i+1}, \text{ and }\] 
        \item [(iii)] undirected edges $v\tut u$ (if any) that are not on $\pi$, are replaced by $v\tuh s_{vu}\hut u$.
    \end{enumerate}
    We want to apply \cref{prop:hedge} with $(\Gf,\Sc,D)\curvearrowleft (\Af,\Sc_{\Af},D)$. In the next, we show
    \[
    \Hc\coloneqq\{v_0,v_1\} \quad  \text{ and } \quad \Hc^\prime\coloneqq \{v_1\}
    \] 
    form a hedge $(\Hc,\Hc^\prime)$ for $(A\cup (\Sc_\Af\sm D),B\cup D)$ in $\Af$. In fact, consider 
    \[
    \Gf^\Hc\coloneqq \{v_0\tuh v_1,v_0\huh v_1\} \quad \text{ and }\quad \Gf^{\Hc^\prime}\coloneqq \{v_1\}.
    \] 
    Then $\Gf^\Hc$ and $\Gf^{\Hc^\prime}$ are $\{v_1\}$-rooted C-forests, and $\{v_1\}\subseteq \anc_{\Af_{\Do(B\cup D)}}(A\cup (\Sc_\Af\sm D))$, since $\pi$ is irreducible and $D\subseteq \Sc_\Af$ does not have any children. It is easy to see that $\Hc\cap B\ne \emptyset$, $\Hc^\prime \cap B=\emptyset$, and $\Gf^{\Hc^\prime}$ is a subgraph of $\Gf^\Hc$. Hence, by \cref{prop:hedge}, $\Prb_{\Mc}(X_A\mid X_{\Sc}=\1 \miid \Do(X_B))$ is not identifiable in $\Mb^{+}_c(\Af)$ and therefore not identifiable in $\Mb^{+}_c(\Pf)$.
\end{proof}

\subsubsection[Proof of Lemma~\ref{lem:construct_MAG}]{Proof of \cref{lem:construct_MAG}}\label{sec:pf_lem_cons_mag}

\begin{proof}[Proof of \cref{lem:construct_MAG}]
    We first observe that there always exists such a node $t$ in Step~(3). In fact, \cref{cond:complete_I} guarantees that there exists a $\Pf$-bucket $\Bb$  contained in $T\sm C$ (which is also a $\Pf_T$-bucket), and \cref{lem:idp_version} implies that for every $\Pf_T$-bucket $\Bb$ contained in $T\sm C$, there exists $t\in \Bb$ such that $\pcc_{\Pf_T}(t)\cap \pch_{\Pf_T}(t)\sm \Bb \ne \emptyset$. 
    
    We show that for every $\Bb_i$ with $1\leq i<m$, there exists $c\in \pcc_{\Pf_C}(c^*)\cap \Bb_i$; in particular, $\hat{C}\ne \emptyset$. In fact, we have $\re_{\Pf_C}(\Bb)=C$ for every $\Pf_T$-bucket $\Bb$ (which is also $\Pf$-bucket by \cref{state:no_bucket_split}) with $\Bb\subseteq C$ since sIDP outputs $\textsc{Fail}$. This implies that $\re_{\Pf_C}(\Bb_m)=C$ in particular. Fix a bucket $\Bb_i$ with $i\ne m$. Then there exists a node $c\in \Bb_i$ such that there is a pc-connecting path from $c$ to a node $d\in \Bb_m$. Note that $\Bb_m$ is the last bucket according to a topological order and $c\notin \Bb_m$. Therefore, we have four cases: 
    \begin{enumerate}
        \item [(i)] $d\hut c$ invisible,
        \item [(ii)] $d\huo c$,
        \item [(iii)] $d\huh c$, and
        \item [(iv)] $d\huh v_1\huh \cdots \huh v_{n-1}\hus c$.
    \end{enumerate}
    By \cref{lem:prop_bucket,lem:property_visible}, in all cases, there exists a pc-connecting path from $c$ to $c^*$.
    
    By \cite[Theorem~2]{Zhang08complete} (or \cite[Section~4.3.1]{zhang2006causal}) together with \cref{def:SOPAG}, we have $\Mf\in [\Pf]_\Ms$. This finishes the proof.
\end{proof}

\begin{lemma}\label{lem:idp_version}
    Suppose the setting of \cref{thm:idp_complete} and \cref{cond:complete_I}. Then for all $\Pf$-bucket $\Bb\subseteq T\sm C$, there exists $t\in \Bb$ such that $\pcc_{\Pf_T}(t)\cap \pch_{\Pf_T}(t)\sm \Bb \ne \emptyset$.
\end{lemma}

\begin{proof}[Proof of \cref{lem:idp_version}]

Each $\Pf_T$-bucket $\Bb$ is not split between $C$ and $T\sm C$, so it is either in $\Bb\subseteq C$ or in $\Bb\subseteq T\sm C$.   The sIDP fails, so we have $\pcc_{\Pf_T}(\Bb)\cap \pde_{\Pf_T}(\Bb)\nsubseteq \Bb$ for all bucket $\Bb\subseteq T\sm C$ and $\re_{\Pf_C}(\Bb)=C$ for all bucket $\Bb\subseteq C$. Let $\Bb_1\prec \cdots \prec \Bb_r$ be a topological order of buckets in $\Pf_C$ and pick $c^\dagger\in \Bb_r$.

\textbf{Step~1}. We show that for every bucket $\Bb\subseteq C$, there exists $c\in \Bb$ such that there is a pc-connecting path from $c$ to $c^\dagger$ that is into $c^\dagger$. For every bucket $\Bb\subseteq C$, since $\re_{\Pf_C}(\Bb)=C$, there exist $c\in \Bb$ and $\tilde{c}\in \Bb_r$ such that $c\in \pcc_{\Pf_C}(\tilde{c})$. If $\Bb=\Bb_r$, we just pick $c=\tilde{c}=c^\dagger$. We now assume $\Bb\ne \Bb_r$. Since $\Bb_r$ is the last element according to a topological order and $\Bb\ne \Bb_r$, the pc-connecting path from $c$ to $\tilde{c}$ must be into $\tilde{c}$.  By \cref{lem:prop_bucket,lem:property_visible}, there must be a pc-connecting path from $c$ to $c^\dagger$ that is into $c^\dagger$ (similar to the argument of \cref{lem:construct_MAG}).

\textbf{Step~2}. We show that for every bucket $\Ab\subseteq T\sm C$ in $\Pf_T$, we can find a pc-connecting path from $a\in \Ab$ to $u\in \pde_{\Pf_T}(\Ab)\sm \Ab$ that is into $u$. Let 
\[
d\in(\pcc_{\Pf_T}(\Ab)\cap \pde_{\Pf_T}(\Ab))\sm \Ab.
\] 
Then the pc-connecting path between $d$ and $\Ab$ cannot be $d\tut v_n$ with $v_n\in \Ab$ or a single edge with an arrowhead on $v_n$ by \cref{lem:podpath}. If a pc-connecting path from $d$ to $v_n$ consists of a single edge, then it must be invisible $d\hut v_n$ or $d\huo v_n$, which shows the claim. Now we consider the case where a pc-connecting path between $d$ and $\Ab$ consists of more than one edges. Let 
\[
\pi:d\suh v_1\huh \cdots \huh v_{n-1}\hus v_n\in \Ab
\] 
be a pc-connecting path from $d$ to $\Ab$ with $n>1$ ($v_1\notin \Ab$). If we have $d\huh v_1$, then we are done. So, in the following, we assume that $d\suh v_1$ is $d\ouh v_1$ or $d\tuh v_1$ invisible. Since $d\in \pde_{\Pf_T}(\Ab)$, by \cite[Lemmas~B1 and B2]{Zhang08complete},\footnote{Lemmas~B1 and B2 in \cite{Zhang08complete} are stated for $\Pc_{AFCI}$ but they also hold for COPAGs.} there is an uncovered potentially directed path from $\Ab$ to $d$ that is into $d$
\[
\tilde{\pi}:\Ab \ni u_0\sus \cdots \sus u_{m-1}\suh d.
\] 
Note that $u_{m-1}\suh d$ cannot be $u_{m-1}\huh d$, since otherwise $\tilde{\pi}$ is not a potentially directed path from $\Ab$ to $d$. If we have $d\ouh v_1$, then by \cref{lem:head_circle} we have $u_{m-1}\suh v_1$. Note that it is impossible to have $u_{m-1}\huh v_1$, since it contradicts \cref{lem:head_circle} if we have $u_{m-1}\ouh d$ or contradicts $\fci{2}$ if we have $u_{m-1}\tuh d$. If we have $d\tuh v_1$, then since $d\tuh v_1$ is invisible we have $u_{m-1}\suh v_1$. We can exclude the case $u_{m-1}\huh v_1$ by $\fci{2}$ if we have $u_{m-1}\ouh d$ or by the fact that PAGs do not have almost cycles if we have $u_{m-1}\tuh d$. Therefore, we have $u_{m-1}\tuh v_1$ or $u_{m-1}\ouh v_1$. It follows that $v_1\in \pde_{\Pf_T}(\Ab)\sm \Ab$ such that $v_1\in \pcc_{\Pf_T}(\Ab)$ and there is a pc-connecting path into $v_1$.  

\textbf{Step~3}. The goal is to show that for every bucket $\Ab$ in $T\sm C$, there is $w_1\in \Ab$ such that $w_1$ is in the same pc-component with $c^\dagger$ and the corresponding pc-connecting path is into $c^\dagger$. By the above two steps, for every bucket $\Ab$ in $\Pf_T$, we can find a sequence of nodes $\{w_i\}_{i=1}^{\ell}$ where $w_1\in\Ab$ and $w_{\ell}=c^\dagger$ such that every pair $(w_i,w_{i+1})$ is connected with a pc-connecting path into $w_{i+1}$ (the pc-connecting paths cannot consist of single undirected edge).  We argue by induction on the number $\ell$ to show that $w_1\in \pcc_{\Pf_T}(c^\dagger)$. The case where $\ell=1$ trivially holds. Now assume that the claim holds for $\ell=k\geq 2$. The goal is to show that it holds when $\ell=k+1$. Let 
\[
\begin{aligned}
    &\pi^1:w_1 =v_{0}^1\suh v_{1}^1\huh \cdots \huh v_{n_1-1}^1\huh v_{n_1}^1= w_2 \quad \text{ and }\\
    &\pi^2:w_2=v_{0}^2\suh v_{1}^2\huh \cdots \huh v_{n_2-1}^2\huh v_{n_2}^2= w_3
\end{aligned}
\]
be pc-connecting paths from $w_1$ to $w_2$ and from $w_2$ to $w_3$, respectively. If $\pi^2$ is into $w_2$, then there is a pc-connecting path from $w_1$ to $w_3$ that is into $w_3$ and we obtain a sequence of nodes $\{w_1,w_3,\ldots,w_{k+1}\}$ whose lengthen is $k$, and it proves the claim. We consider the case where the pc-connecting path between $w_2$ and $w_3$ is not into $w_2$. If we have $v_i^1\huh v_1^2$ for some $i=0,\ldots,n_1-1$ or $v_0^1\tuh v_1^2$ invisible or $v_0^1\ouh v_1^2$, then we find a pc-connecting path from $w_1$ to $w_3$ that is into $w_3$. Therefore, we can obtain a shorter sequence of nodes by deleting $w_2$ and show the claim. We now assume that this is not the case. In this case, by \cref{lem:head_circle,lem:property_visible}, we have invisible edge $v_{n_1-1}^1\tuh v_0^2$ or $v_{n_1-1}^1\ouh v_0^2$. Then for all $i=0,\ldots,n_1-1$, we have $v_i^1\tuh v_1^2$ invisible or $v_i^1\ouh v_1^2$.  Since we have invisible $w_1\tuh v_1^2$ or $w_1\ouh v_1^2$, we can obtain a shorter sequence by deleting $w_2$. This proves the desired result. 

\textbf{Step~4}. Every bucket $\Ab$ in $T\sm C$ has a possible child in another bucket $\Bb$, since otherwise all the possible descendants of $\Ab$ would be in $\Ab$, which contradicts the failure of sIDP. Fix an arbitrary bucket $\Ab\subseteq T\sm C$. Let $a^*\in \Ab$ be such that there exists a possible child $b^*$ in another bucket $\Bb$, i.e., $a^*\tuh b^*$ or $a^*\ouh b^*$. By previous steps there exist $a\in \Ab$ and $b\in \Bb$ such that there are pc-connecting paths from $a$ to $c^\dagger$ and from $b$ to $c^\dagger$ that are both into $c^\dagger$. Note that we also have $a^* \tuh b$ or $a^*\ouh b$ by \cref{lem:prop_bucket}. We have two cases: $a\suh c^\dagger$ or $a\suh v_1\huh \cdots \huh v_{n-1}\huh c^\dagger$ with $n>1$.  We consider the first case. If $a\ouh c^\dagger$ or $a\tuh c^\dagger$ invisible, then $c^\dagger \in (\pcc_{\Pf_T}(a)\cap \pch_{\Pf_T}(a))\sm \Ab$. If $a\huh c^\dagger$, then \cref{lem:prop_bucket} implies $a^* \huh c^\dagger$. Concatenating $a^* \huh c^\dagger$ with the pc-connecting path from $b$ to $c^\dagger$ gives a pc-connecting path from $a^*$ to $b$. Therefore, $b \in (\pcc_{\Pf_T}(a^*)\cap \pch_{\Pf_T}(a^*))\sm \Ab$. We now consider the second case. If $a\ouh v_1$ or $a\tuh v_1$ invisible, then $v_1 \in (\pcc_{\Pf_T}(a)\cap \pch_{\Pf_T}(a))\sm \Ab$. If $a\huh v_1$, then  \cref{lem:prop_bucket} implies $a^* \huh v_1\huh \cdots \huh v_{n-1}\huh c^\dagger$. Concatenating $a^* \huh v_1\huh \cdots \huh v_{n-1}\huh c^\dagger$ with the pc-connecting path from $b$ to $c^\dagger$ gives a pc-connecting path from $a^*$ to $b$. Hence, $b \in (\pcc_{\Pf_T}(a^*)\cap \pch_{\Pf_T}(a^*))\sm \Ab$.

This finishes the proof.
    
\end{proof}

\subsubsection[Proof of Lemma~\ref{lem:prop_construct_MAG_III}]{Proof of \cref{lem:prop_construct_MAG_III}}\label{sec:pf_prop_mag}

We introduce some additional notation. If there is a bidirected path from $a$ to $b$ in $\Mf_D$ with MAG $\Mf=(\Vc,\Ec)$ and $a,b\in D\subseteq \Vc $, then we write $a\in \dcc_{\Mf_D}(b)$. 

\begin{proof}[Proof of \cref{lem:prop_construct_MAG_III}]
   \textbf{Step~1: show (1)}. Let $b\in B\cap (\hat{T}\sm \hat{C})$ be from Part~(3) of \cref{lem:prop_construct_MAG} and $c\in \hat{C}$ be such that there are no edges out of $c$ in $\Mf_{\hat{C}}$. This choice of $c\in \hat{C}$ is possible because there are no cycles in a MAG and $\Mf_{\hat{C}}$ is not a purely undirected graph by Part~(1) of \cref{lem:prop_construct_MAG}. Define 
   \[
   \begin{aligned}
       H_1&\coloneqq ( \hat{T}\sm \hat{C} )\cap ( \pcc_{\Mf_{\hat{T}}}(b)\sm \dcc_{\Mf_{\hat{T}}}(b) ), \quad &H_2\coloneqq \pcc_{\Mf_{\hat{C}}}(c)\sm \dcc_{\Mf_{\hat{C}}}(c)\\
       \breve{T}&\coloneqq\hat{T}\sm (H_1\cup H_2), \quad \text{ and } &\breve{C}\coloneqq \hat{C}\sm H_2.
   \end{aligned}
   \] 
   By Part~(2) of \cref{lem:prop_construct_MAG}, $\pcc_{\Mf_{\hat{C}}}(c)=\hat{C}$ and therefore $\breve{C}=\dcc_{\Mf_{\hat{C}}}(c)$ is a single c-component in $\Mf_{\breve{C}}$. This proves Part~(1). Also note that $\emptyset\ne \breve{C}\subsetneq \breve{T}$.

   \textbf{Step~2}. We shall in the following show that Parts~(1)\&(2) of \cref{lem:prop_construct_MAG} still hold with $(\hat{T},\hat{T}\sm\hat{C})\curvearrowleft (\hat{T}\sm H_1,(\hat{T}\sm\hat{C})\sm H_1)$. 

    \textbf{Step~2.1: show \cref{lem:prop_construct_MAG_II} holds with $\hat{T}\curvearrowleft \hat{T}\sm \{u\}$ for $u\in H_1$}. We show that \cref{lem:prop_construct_MAG_II} still holds with $\hat{T}\curvearrowleft \hat{T}\sm \{u\}$ for $u\in H_1\ne \emptyset$. Note that $u\ne b$ by the definition of $H_1$. By \cref{lem:prop_construct_MAG_II}, we have $\pcc_{\Mf_{\hat{T}}}(b)=\hat{T}$. 
    
    We show that there is a pc-connecting path from $u$ to $b$ of the form: a collider path $\pi:u\tuh u_1\huh \cdots \huh u_{n-1}\huh u_n=b$ when $n\ge 2$, and  $\pi:u\tuh u_1=b$ invisible when $n=1$ (not $u\hut b$ or $u\huh b$ or $u\tut b$) in $\Mf_{\hat{T}}$.  Since $u\in H_1$, there is a pc-connecting path 
   \[
   \pi:u\suh u_1\huh \cdots \huh u_{n-1}\hus u_n=b
   \] 
   from $u$ to $b$ in $\Mf_{\hat{T}}$. We show that it is impossible to have invisible $b\tuh u_{n-1}$ or $b\tut u$ in $\Mf_{\hat{T}}$. In fact, assume that this is not the case, i.e., we have an invisible directed edge $b\tuh u_{n-1}$ or an undirected edge $b\tut u$ in $\Mf_{\hat{T}}$. By Part~(1) of \cref{lem:prop_construct_MAG}, we can find a directed path from $b$ to $\hat{C}$ starting from $b\tuh u_{n-1}$ or an anterior path from $b$ to $\hat{C}$ starting from $b\tut u$ in $\Mf_{\hat{T}}$. This implies that there is a potentially anterior path from $b$ to $\hat{C}$ in $\Pf_T$ starting with an edge that is not visible. Recall that $\hat{C}\subseteq C\subseteq \Dc=\pant_{\Pf_{\Vc\sm B}}(A)$. This means that there is a potentially anterior path from $B$ to $A$ in $\Pf$ starting with an edge that is not visible, which violates \cref{cond:complete_I}. This proves the claim. Since $u$ is not in the c-component containing $b$, it is impossible to have $u\huh u_1\huh \cdots \huh u_{n-1}\huh u_n=b$ or $u\huh b$. In summary, we can conclude that there is a pc-connecting path of the form  $\pi:u\tuh u_1\huh \cdots \huh u_{n-1}\huh u_n= b$ when $n>2$, and $\pi:u\tuh u_1=b$ invisible when $n=1$, from $u$ to $b$ in $\Mf_{\hat{T}}$.  
   
   Note that a pc-connecting path from some node in $\Mf_{\hat{T}}$ to $b$ cannot have $u$ as an intermediate node. If this is not the case, then $u$ would be in the c-component of $b$ and this causes a contradiction to the fact that $u\in H_1$. For every $w_1,w_2\in \hat{T}\sm \{u\}$, assume a pc-connecting path between them in $\Mf_{\hat{T}}$ intersects $u$:
   \[
        w_1\suh v_1 \huh \cdots \huh  v_{n-1} \hus w_2
   \]
   where $v_j=u$ for some $1\le j\le n-1$. If we have $u\huh u_1$, then we have
   \[
    w_1\tuh v_1\huh \dots \huh v_j \huh u_1\huh \cdots \huh b.
   \]
   Deleting repeated nodes (if any) gives a pc-connecting path from $w_1$ to $b$ that is into $b$ in $\Mf_{\hat{T}\sm \{u\}}$. Consider the case where we have $u\tuh u_1$ invisible. Then $v_{j-1}$ must be adjacent to $u_1$. If we have $v_{j-1} \huh u_1$, then we can get a pc-connecting path from $w_1$ to $b$ that is into $b$ similar to before. If we have $v_{j-1}\tuh u_1$, then by \cref{lem:property_visible} we know that it must be invisible. This implies that $v_{j-2}$ must be adjacent to $u_1$ and we have $v_{j-2} \huh u_1$ or $v_{j-2}\tuh u_1$ invisible. Repeat the above argument until we reach $w_1$. Then we can conclude that there are pc-connecting paths from $w_1$ and $w_2$ to $b$ that are both into $b$ in $\Mf_{\hat{T}\sm \{u\}}$. Concatenating the two pc-connecting paths and deleting repeated nodes (if any) gives a pc-connecting path from $w_1$ to $w_2$ in $\Mf_{\hat{T}\sm \{u\}}$. Therefore, \cref{lem:prop_construct_MAG_II} still holds if we replace $\hat{T}$ with $\hat{T}\sm \{u\}$.  

   \textbf{Step~2.2: show Part~(1) of \cref{lem:prop_construct_MAG} with $(\hat{T},\hat{T}\sm\hat{C})\curvearrowleft (\hat{T}\sm \{u\},(\hat{T}\sm\hat{C})\sm \{u\})$ for $u\in H_1$}. We first assume that $u\in H_1\ne \emptyset$. Assume for contradiction that some node $z\in (\hat{T}\sm \hat{C}) \sm\{u\}$ violates Part~(1) of \cref{lem:prop_construct_MAG} in the induced subgraph $\Mf_{\hat{T}\sm \{u\}}$. Since Part~(1) of \cref{lem:prop_construct_MAG} is valid for $\hat{T}$ but not for $\hat{T}\sm \{u\}$, there are two cases to consider: 
   \begin{enumerate}[label=\textbf{Case~\arabic*:}, ref=Case~\arabic*, leftmargin=*]
    \item we have $u\in \ch_{\Mf_{\hat{T}}}(z)$ and $z\tuh u$ is invisible;
    \item node $z$ is in the same c-component with a child $\tilde{u}$ and the bidirected paths intersect $u$ in $\Mf_{\hat{T}}$.
   \end{enumerate}
   
    \textbf{Case~1}.  Note that $z\tuh u$ and $u\tuh u_1$ are invisible. \cref{lem:property_visible} gives that there is an edge $z\suh u_1$ that is not visible. Since there are no almost cycles in a MAG, $z\suh u_1$ must be an invisible directed edge $z\tuh u_1$. This means that Part~(1) of \cref{lem:prop_construct_MAG} still holds, which is a contradiction.

    \textbf{Case~2}. Let \[\pi:z=v_0^1\huh v_{1}^1\huh \cdots \huh v_{n_1-1}^1\huh u \huh v_1^2\huh \cdots \huh v_{n_2-1}^2\huh v_{n_2}^2=\tilde{u}\] be a bidirected path between $z$ and $\tilde{u}$ intersecting $u$ in $\Mf_{\hat{T}}$. Recall that $u\tuh u_1$ is invisible. If there is no $v_i^1$ such that $v_i^1\huh u_1$ is present (including $i=0$), then $v_i^1\in \pa_{\Mf_{\hat{T}}}(u_1)$. So we have invisible directed edge $z\tuh u_1$ (it is invisible because otherwise $u\tuh u_1$ would be visible) and therefore we have that Part~(1) of \cref{lem:prop_construct_MAG} still holds, which is a contradiction. Hence, there exists $v_i^1$ such that $v_i^1\huh u_1$ is in $\Mf_{\hat{T}}$ and therefore $u_1\in \dcc_{\Mf_{\hat{T}\sm \{u\}}}(z)$. If there exists $v_i^2$ such that we have $v_i^2 \huh u_1$, then there is a bidirected path from $z$ to $\tilde{u}$ in $\Mf_{\hat{T}\sm \{u\}}$. This contradicts the assumption that $z$ fails Part~(1) of \cref{lem:prop_construct_MAG} in $\Mf_{\hat{T}\sm \{u\}}$. Therefore, we have $v_i^2 \tuh u_1$ for all $1\leq i\leq n_2$ and $\tilde{u}\tuh u_1$ is invisible since $u\tuh u_1$ is invisible. We have $z\tuh u_1$ invisible, since $z\tuh \tilde{u}$ is present by assumption and $\tilde{u}\tuh u_1$ is invisible. This causes a contradiction to the assumption that there does not exist a child of $z$ in $\Mf_{\hat{T}\sm \{u\}}$ such that the edge between them is an invisible directed edge. 
    
    Based on the above argument, we can conclude that Part~(1) of \cref{lem:prop_construct_MAG} holds if we replace $\hat{T}$ with $\hat{T}\sm \{u\}$. 
    
    \textbf{Step~2.3: show Part~(2) of \cref{lem:prop_construct_MAG} with $(\hat{T},\hat{T}\sm\hat{C})\curvearrowleft (\hat{T}\sm \{u\},(\hat{T}\sm\hat{C})\sm \{u\})$ for $u\in H_1$}. For Part~(2) of \cref{lem:prop_construct_MAG}, one just needs to notice that $u\notin \hat{C}$.

    \textbf{Step~2.4: finish Step~2 using recursion}. We finish the proof of that Parts~(1)\&(2) of \cref{lem:prop_construct_MAG} hold with $(\hat{T},\hat{T}\sm\hat{C})\curvearrowleft (\hat{T}\sm H_1,(\hat{T}\sm\hat{C})\sm H_1)$ by recursively applying the argument in Steps~2.1, 2.2, and 2.3. 

    \textbf{Step~3: finish the proof of (2)}.  Let $u\in H_2$. Similar to Step~2.1, by Part~(2) of \cref{lem:prop_construct_MAG} and the choice of $c$, for every $w_1,w_2\in \hat{C}\sm\{u\}$, if a pc-connecting path between them in $\Mf_{\hat{C}}$ intersects $u$, then there must be pc-connecting paths from $w_1$ to $c$ and from $w_2$ to $c$ that are both into $c$ in $\Mf_{\hat{C}\sm \{u\}}$. Then there must be a pc-connecting path from $w_1$ to $w_2$ in $\Mf_{\hat{C}\sm \{u\}}$. So Part~(2) of \cref{lem:prop_construct_MAG} still holds if we replace $\hat{C}$ with $\hat{C}\sm \{u\}$.  A similar argument to Step~2.2, which is by contradiction, shows that Part~(1) of \cref{lem:prop_construct_MAG} still holds if we replace $\hat{T}$ and $\hat{C}$ with $\hat{T}\sm (H_1\cup \{u\})$ and $\hat{C}\sm \{u\}$, respectively. Therefore, Parts~(1)\&(2) of \cref{lem:prop_construct_MAG} hold with $(\hat{T},\hat{C})\curvearrowleft (\breve{T},\breve{C})$. This implies that for every $t\in \breve{T}\sm \breve{C}$, there exists $s\in \ch_{\Mf_{\breve{T}}}(t)$ such that $t\tuh s$ is invisible or there is a bidirected path connecting $t$ and $s$ in $\Mf_{\breve{T}}$. If $t,s\in \breve{T}\sm \breve{C}$ then there is a bidirected path between $t$ and $s$ by the definition of $\breve{T}$. Assume $t\in \breve{T}$ and $s\in \breve{C}$. Then we have $\pi_1:t\huh v_1\huh \cdots\huh v_{n-1}\huh b$. Assume that $t\tuh s$ is invisible. Then $v_1$ must be adjacent to $s$. If we have $v_1\huh s$, then we are done. So we assume that $v_1\tuh s$. Applying \cref{lem:orientation_invisible}, we know that $v_1\tuh s$ is invisible. Then we have $v_2\huh s$ or $v_2\tuh s$. Repeat the argument until we reach $b$. Again if we have $b\huh s$, then we are done. Therefore, we are left with the case where $b\tuh s$ is an invisible edge. This contradicts \cref{cond:complete_I}. Hence, we can conclude that there must be a bidirected path connecting $t$ and $s$ in $\Mf_{\breve{T}}$, which shows Part~(2) of \cref{lem:prop_construct_MAG_III}.  
    
    \textbf{Step~4: show (3)}. By the choice of node $b\in B$, we have $b\in\breve{T}\sm \breve{C}$, which proves Part~(3) of \cref{lem:prop_construct_MAG_III}. 

    \textbf{Step~5: show (4)}. Since $\breve{C}\subseteq \hat{C}$, it holds $\breve{C}\subseteq \ant_{\Mf_{\Vc\sm B}}(A)$. If $\breve{C}$ is not a singleton set, then it is easy to see that all the nodes in $\breve{C}$ have arrowheads on them in $\Mf$, since they are all in one single c-component. This then implies $\breve{C}\subseteq \anc_{\Mf_{\Vc\sm B}}(A)$. Now consider the case where $\breve{C}=\{c\}$ is a singleton set. If $\hat{C}\sm \breve{C}\ne \emptyset$, then  Part~(2) of \cref{lem:prop_construct_MAG} implies that there is a pc-connecting path from a node in $\hat{C}\sm\breve{C}$ to $c$ that is into $c$ by the choice of $c\in \hat{C}$ in Step~1. Therefore, there is an arrowhead on $c$ in $\Mf$. So $c\in \anc_{\Mf_{\Vc\sm B}}(A)$. If $\hat{C}=\breve{C}=\{c\}$, then by Part~(1) of \cref{lem:prop_construct_MAG} there is an arrowhead from a node in $\hat{T}\sm\hat{C}$ to $c$ in $\Mf$. Hence, we have $\breve{C}\subseteq \anc_{\Mf_{\Vc\sm B}}(A)$, which gives Part~(4) of \cref{lem:prop_construct_MAG_III}.

\end{proof}

\begin{lemma}\label{lem:intarget}
    Suppose the setting of \cref{thm:idp_complete} and \cref{cond:complete_I}. Then $B\cap T\ne \emptyset$.
\end{lemma}

\begin{proof}[Proof of \cref{lem:intarget}]

    Assume on the contrary that the conclusion does not hold, i.e., $B\cap T=\emptyset$.  Recall $\Dc= \pant_{\Pf_{\Vc\sm B}}(A)$. We have $\pde_{\Pf_{\Vc\sm B}}(\Vc\sm(\Dc\cup B))\subseteq \Vc\sm \Dc$. Otherwise, there exist 
    $u\in \Vc\sm (\Dc\cup B)$ and $v\in \Dc$ such that $v\in\pde_{\Pf_{\Vc\sm B}}(u)$. This implies $u\in \pant_{\Pf_{\Vc\sm B}}(A)$, which contradicts the definition of $\Dc$. 
    
    \textbf{Step~1: show $T\subseteq \Dc$ provided $B\cap T=\emptyset$}. We show $T\cap (\Vc\sm \Dc)=\emptyset$ by contradiction. In fact, if this is not the case, then there exists a $\Pf$-bucket $\Bb\subseteq T$ such that $\Bb\subseteq \Vc\sm \Dc$. This is because by the definition of $\Dc$, every $\Pf$-bucket disjoint from $B$ is entirely contained in either $\Dc$ or $\Vc\sm \Dc$; moreover, the construction of $T$ in $\idp$ does not split the $\Pf$-buckets. Consider a topological order $\Bb_1\prec \cdots \prec\Bb_n$ of the buckets in $\Pf$ contained in $T$. Let $\Bb_i$ be the bucket with the highest index contained in $\Vc\sm \Dc$. Note that $\Bb_i\subseteq T\sm C$, since $C\subseteq \Dc$. Then we have $\pde_{\Pf_T}(\Bb_i)\subseteq \Bb_i$, since otherwise we would have $\pde_{\Pf_T}(\Bb_i)\cap \Dc\ne \emptyset$, which contradicts that $\pde_{\Pf_{\Vc\sm B}}(\Vc\sm(\Dc\cup B))\subseteq \Vc\sm \Dc$. This implies that $\Bb_i$ satisfies the criterion in Rule~L2, which contradicts the fact that sIDP outputs $\textsc{Fail}$. Hence, we can conclude $T\subseteq \Dc$.

    \textbf{Step~2: derive contradiction}. First note that every bucket in $T$ is either in $C$ or in $T\sm C$. Let $\Bb_{i_1}\prec \cdots \prec \Bb_{i_m}$ be a topological order of the buckets in $\Pf_T$ contained in $T\sm C$. Since the sIDP fails, by \cref{lem:idp_version}, pick $b\in \Bb_{i_m}$ such that $\pcc_{\Pf_T}(b)\cap \pch_{\Pf_T}(b)\nsubseteq \Bb_{i_m}$. Let $c\in \lb\pcc_{\Pf_T}(b)\cap \pch_{\Pf_T}(b)\rb \sm \Bb_{i_m}$. Since $c\in \pch_{\Pf_T}(b)$, there must be an edge between nodes $b$ and $c$ and it cannot be $b\hus c$. The cases $b\ouo c$, $b\out c$, $b\tut c$, and $b\tuo$ are excluded by the fact that $c\notin\bu_{\Pf_T}(b)$. Therefore, we have $b\ouh c$ or $b\tuh c$.  
   
    Since $\Bb_{i_m}$ is the last bucket according to the topological order, we have $c\in C$. By Step~1, we have $C\subsetneq T\subseteq \Dc$. By the $\idp$, there exists a sequence of sets 
    \[
    C=C_0\subsetneq C_1 \subsetneq \cdots \subsetneq C_\ell=\Dc
    \] 
    such that $C_i=\re_{\Pf_{C_{i+1}}}(\Bb)$ or $C_i=\re_{\Pf_{C_{i+1}}}(C\sm \re_{\Pf_{C_{i+1}}}(\Bb))$ for some $\Bb\subsetneq C_{i+1} \subseteq \Dc$ and for all $0\leq i\leq \ell-1$. There exists $0\leq j\leq \ell-1$ such that $b\in C_{j+1}\sm C_j$. If we have $b\ouh c$ or $\Pf_{C_{j+1}}$-invisible directed edge $b\tuh c$, then applying \cref{lem:region} with $B\curvearrowleft C_{j}$ gives that $b\in C_j$, which is a contradiction. Therefore, the edge between nodes $b$ and $c$ must be $\Pf_{C_{j+1}}$-visible directed edge $b\tuh c$. Since $c\in \pcc_{\Pf_T}(b)$ and $b\tuh c$ is $\Pf_T$-visible, there must be a pc-connecting path 
    \[
    b=v_0\sus \cdots \sus v_{n}=c
    \]
    with $n>1$ in $\Pf_T$. The left part of the proof follows verbatim from that of \cite[Lemma~1]{jaber19idencom}. We reproduce the argument for completeness. If $v_{n-1}\in C$, then $v_{i}\in C$ for all $0\leq i\leq n$ by \cref{lem:region}. This contradicts the fact that $b\in T\sm C$. Therefore, $v_{n-1}\in T\sm C$. Similarly, the edge $v_{n-1}\hus c$ cannot be $v_{n-1}\huh c$. If we have $v_{n-1}\huo c$, then we have $b\tuh v_{n-1}$ or $b\ouh v_{n-1}$ by \cref{lem:head_circle}.  If we have $v_{n-1}\hut c$, then $b\tuh v_{n-1}$ by the invisibility of $c\tuh v_{n-1}$ and $\fci{2}$ and $\fci{8}$. These cases both lead to a contradiction, since $v_{n-1}\notin C$ is a possible child of $b\in \Bb_{i_m}$, which contradicts the fact that $\Bb_{i_m}$ is the last element in the topological order of the buckets in $\Pf_T$ contained in $T\sm C$. Overall, the initial assumption that $B\cap T=\emptyset$ is false and we finish the proof.

\end{proof}

\begin{lemma}\label{lem:prop_construct_MAG}
    Suppose the setting of \cref{thm:idp_complete} and \cref{cond:complete_I}. Let $\Mf$ be a MAG constructed from $\Pf$ according to Lemma~\ref{lem:construct_MAG}. Then the induced subgraph $\Mf_{\hat{T}}$ establishes the following properties:
    \begin{enumerate}
        \item for every $t\in \hat{T}\sm \hat{C}$, there exists $u\in \ch_{\Mf_{\hat{T}}}(t)$ such that the edge $t\tuh u$ is invisible, or there is a bidirected path connecting $t$ and $u$ in $\Mf_{\hat{T}}$;
        \item $\pcc_{\Mf_{\hat{C}}}(c)=\hat{C}$ for every $c\in \hat{C}$;
        \item $B\cap (\hat{T}\sm \hat{C})\ne \emptyset$; and
        \item $\hat{C}\subseteq \ant_{\Mf_{\Vc\sm B}}(A)$.
    \end{enumerate}
\end{lemma}

\begin{proof}[Proof of \cref{lem:prop_construct_MAG}]

\textbf{Show property~(1)}. Let $t\in\hat{T}\sm \hat{C}$ and $\Bb$ be the bucket in $\Pf$ where $t\in \Bb$. By Step~(3) of \cref{lem:construct_MAG}, we have $t\in \pcc_{\Pf_T}(z)$ where $z\in \pch_{\Pf_T}(t)$ and $z\notin \Bb$. Since $z\in \pch_{\Pf_T}(t)$, there must be an edge between $t$ and $z$ that is not into $t$. Since nodes $t$ and $z$ are in different buckets, the edge $t\sus z$ cannot be $t\tut z$, $t\tuo z$, $t\out z$, or $t\ouo z$. Assume that we have $t\ouh z$, or invisible directed edge $t\tuh z$ in $\Pf_T$. Then for every $v\in \bu_{\Pf_T}(z)$, by \cref{lem:head_circle,lem:property_visible}, we have edge $t\suh v$ in $\Pf_T$ that is not visible. Therefore, we have $t\suh u$ that is not visible where $u\in \hat{T}\cap \bu_{\Pf_T}(z)$. By the construction of $\Mf$, we have invisible directed edge $t\tuh u$ in $\Mf_{\hat{T}}$.

We now consider the case where visible directed edge $t\tuh z$ is in $\Pf_T$. Since $t\in \pcc_{\Pf_T}(z)$, by \cref{def:pc-component} there must be a collider path in $\Pf_T$ 
\[
t\suh v_1\huh \cdots \huh v_{n-1}\hus z
\] 
with $n>1$ and none of the edges are visible. WLOG, we can assume that $t\suh v_1$ is $t\huh v_1$. In fact, if $t\suh v_1$ is $t\ouh v_1$ or invisible $t\tuh v_1$, then we can argue similarly to the last part and find a node $u\in \hat{T}\cap \bu_{\Pf_T}(v_1)$ such that we have an invisible directed edge $t\tuh u$ in $\Mf_{\hat{T}}$. We have two cases: 
\begin{enumerate}
    \item [(i)] $v_{n-1}\hus z$ is not $v_{n-1}\huh z$ in $\Pf_T$;
    \item [(ii)] $v_{n-1}\hus z$ is $v_{n-1}\huh z$ in $\Pf_T$.
\end{enumerate}
 If $v_{n-1}\hus z$ is not $v_{n-1}\huh z$ (i.e., $v_{n-1}\hus z$ is $v_{n-1}\huo z$ or invisible $v_{n-1}\hut z$), then we have $t\suh v_{n-1}$, by \cref{lem:head_circle} ($v_{n-1}\huo z\hut t$) or definition of invisible edges ($v_{n-1}\hut z\hut t$ where $v_{n-1}\hut z$ is invisible). The edge $t\suh v_{n-1}$ cannot be $t\huh v_{n-1}$ by \cref{lem:podpath} (when $z\ouh v_{n-1}$) or the fact that PAGs do not have almost cycles (when $z\tuh v_{n-1}$). This implies that $v_{n-1}\in \pch_{\Pf_T}(t)$. Since $t\huh v_1\huh \cdots \huh v_{n-1}$ in $\Pf_T$, we apply \cref{lem:prop_bucket} to find $\hat{v}_i\in \hat{T}\cap \bu_{\Pf_T}(v_i)$ such that $t\huh \hat{v}_1\huh \cdots \huh \hat{v}_{n-1}\eqcolon u$ in $\Pf_{\hat{T}}$ and $u\in \ch_{\Mf_{\hat{T}}}(t)$ holds by the construction of $\Mf$ and \cref{lem:prop_bucket}. If we have $v_{n-1}\huh z$, then similarly we have $t\huh \hat{v}_1\huh \cdots \huh \hat{v}_{n-1} \huh \hat{z}\eqcolon u$ in $\Mf_{\hat{T}}$ and $u\in \ch_{\Mf_{\hat{T}}}(t)$. This finishes the proof of Property~(1).

\textbf{Show property~(2)}.  If $\hat{C}$ is a singleton set, then the claim trivially holds. So we can assume WLOG that $\hat{C}$ is not a singleton set.  For every $c\in \hat{C}$, we have $c\in \pcc_{\Pf_{C}}(c^*)$ by Step~(4) of \cref{lem:construct_MAG}. Since $\Bb_m$ is the last bucket according to a topological order over $\Pf_C$ and $c^*\in \Bb_m$, there is no edges of the form $c^*\ouh v$ or $c^*\tuh v$ in $\Pf_C$. For $c\in \hat{C}$ such that $c\notin \Bb_m$, there is no undirected edges connecting $c$ and $c^*$. So the pc-connecting path between $c$ and $c^*$ in $\Pf_C$ must be a collider path of the form \[c\suh v_1\huh \cdots \huh v_{n-1}\huh c^*\] where $n>1$ and all the edges are not visible. By \cref{lem:prop_bucket}, we can find a subset of nodes $\{u_{j}\}_{j=1}^\ell$ such that $\{u_{j}\}_{j=1}^\ell \subseteq \hat{C}$ and we have a pc-connecting path $c\suh u_{1}\huh \cdots \huh u_{\ell}\huh c^*$ in $\Pf$. By \cref{lem:orientation_invisible}, we have $c\suh u_{1}\huh \cdots \huh u_{\ell}\huh c^*$ in $\Mf_{\hat{C}}$ where all edges are not visible. This means that $c\in \pcc_{\Mf_{\hat{C}}}(c^*)$. Since for every $c_1,c_2\in \hat{C}$ the pc-connecting paths from $c_1$ to $c^*$ and from $c_2$ to $c^*$ are both into $c^*$, there is a pc-connecting path between $c_1$ and $c_2$ in $\Mf_{\hat{C}}$. Hence, we have $\pcc_{\Mf_{\hat{C}}}(c)=\hat{C}$ for every $c\in\hat{C}$.

\textbf{Show property~(3)}. By \cref{lem:intarget}, there exists a bucket $\Bb$ in $\Pf_T$ such that $\Bb\cap T\ne \emptyset$. Recall that by \cref{cond:complete_I}, we have either $\Bb\subseteq \Dc$ or $\Bb\subseteq \Vc\sm \Dc$. If we have $\Bb\subseteq \Dc$, then it implies $b\in \Dc$, which contradicts the definition of $\Dc$. Therefore, we have $\Bb\subseteq \Vc\sm \Dc$. By \cref{lem:construct_MAG} and the fact that $\Bb\subseteq T$, there exists a node $d^*\in \Bb$ such that $d^*\in \hat{T}\cap (\Vc\sm \Dc)$ and therefore $\hat{T}\cap (\Vc\sm \Dc)\ne \emptyset$.  Let $v_1\prec \cdots \prec v_n$ be a topological order over $\Mf_{\hat{T}\cap (\Vc\sm \Dc)}$. Note that $v_n\in\hat{T}\sm \hat{C}$, since $\hat{C}\subseteq C\subseteq \Dc$. By part (1) of \cref{lem:prop_construct_MAG}, there exists $u\in \ch_{\Mf_{\hat{T}}}(v_n)$ such that $u\in \pcc_{\Mf_{\hat{T}}}(v_n)$. Since $v_n$ is the last element according to the given topological order over $\Mf_{\hat{T}\cap (\Vc\sm \Dc)}$, we have $u\in \Dc$. This implies that $v_n\in \pant_{\Pf}(A)$. Since $v_n\notin \Dc=\pant_{\Pf_{\Vc\sm B}}(A)$ and $u\in \Dc$, it must be that $v_n\in B$.

\textbf{Show property~(4)}. 
Note that $\hat{C}\subseteq \Dc=\pant_{\Pf_{\Vc\sm B}}(A)$ and no pair of nodes in $\hat{C}$ are in the same $\Pf$-bucket. Let $c\in \hat{C}$ and 
\[
\pi:c=v_0\sus \cdots \sus v_{n}\in A
\]
be a shortest potentially anterior path from $c$ to $A$ in $\Pf_{\Vc\sm B}$. If there are no edges of the form $v_i\sut v_{i+1}$ on $\pi$, then $\pi$ is a potentially directed path from $c$ to $A$ and there is an uncovered potentially directed path $\varpi$ from $c$ to $A$ by \cite[Lemma~B.1]{Zhang08complete}. Let $\varpi_*$ denote the corresponding path of $\varpi$ in $\Mf$. Since we orient all edges with circles near $c$ out of $c$ in $\Mf$, $\varpi_*$ in $\Mf$ is an anterior directed path out of $c$ (note that if we have $v_{i-1}\ouo v_i\ouo v_{i+1}$ in $\Pf$ then we must have $v_{i-1}\tut v_i\tut v_{i+1}$ or $v_{i-1}\tuh v_i\tuh v_{i+1}$ since $v_{i-1}\tuh v_i\hut v_{i+1}$ would introduce an unshielded collider). So $c\in \ant_{\Mf_{\Vc\sm B}}(A)$. Now consider the case where $\pi$ contains an edge of the form $v_i\sut v_{i+1}$. Let $v_j\sut v_{j+1}$ be the last edge of this form on $\pi$ starting from $c$. Since $\Pf$ is a COPAG and $\pi$ is shortest, we do not have pattern $v_{j-1}\ouo v_{j}\tut v_{j+1}$ or $v_{j-1}\ouo v_{j}\out v_{j+1}$ on $\pi$ (by $\fci{6}$ and $\fci{7}$ and the fact that $v_{j-1}$ is non-adjacent to $v_{j+1}$). Therefore, the corresponding path in $\Mf_{\Vc\sm B}$ of subpath $\pi(c,v_j)$ must be undirected. Note that either $v_j\in A$ or there is a non-trivial potentially directed path from $v_j$ to $A$. Similar to the last part, we can find an anterior path from $v_j$ to $A$ in $\Mf_{\Vc\sm B}$. Then we can find an anterior path from $c$ to $A$ in $\Mf_{\Vc\sm B}$. Hence, $\hat{C}\subseteq \ant_{\Mf_{\Vc\sm B}}(A)$.

\end{proof}

\begin{lemma}\label{lem:prop_construct_MAG_II}
    Let $v\in \hat{T}$. Then $\pcc_{\Mf_{\hat{T}}}(v)=\hat{T}$.    
\end{lemma}

\begin{proof}[Proof of \cref{lem:prop_construct_MAG_II}]
   
   First, note that by Part~(1) of \cref{lem:prop_construct_MAG}, our MAG $\Mf_{\hat{T}}$ is not a purely undirected graph. Then the same argument of \cite[Lemma~7]{jaber19idencom} works modulo some minor modifications. For readers' convenience, we in the following reproduce the argument in detail. 
   
   \textbf{Step~1: reduce the problem to showing $v\in \pcc_{\Mf_{\hat{T}}}(v^*)$ for every node $v\in \hat{T}\sm \hat{C}$}. Since MAGs do not have cycles and arrowheads cannot meet undirected edges, we can pick a node $v^*$ in $\Mf_{\hat{T}}$ such that there are arrowheads but no tails near it. The goal is to show that $v\in \pcc_{\Mf_{\hat{T}}}(v^*)$ for every node $v$ in $\Mf_{\hat{T}}$. This will imply the result. Indeed, if this is true, then for every $v,w\in \hat{T}$ we have $v,w\in \pcc_{\Mf_{\hat{T}}}(v^*)$ and the pc-connecting paths $\pi_v$ from $v$ to $v^*$ and $\pi_w$ from $v^*$ to $w$ in $\Mf_{\hat{T}}$ both have arrowheads towards $v^*$ by the choice of $v^*$. Then the path $\pi_v\oplus \pi_w$ is a pc-connecting path from $v$ to $w$, which implies $\hat{T}\subseteq \pcc_{\Mf_{\hat{T}}}(v)$. Let $v\in \hat{T}$ be such that $v\ne v^*$. Note that by Part~(1) of \cref{lem:prop_construct_MAG}, $v^*\in \hat{C}$ since there are no tails near $v^*$. If $v\in \hat{C}$, then by Part~(2) of \cref{lem:prop_construct_MAG} we have $v\in \pcc_{\Mf_{\hat{T}}}(v^*)$. Therefore, we only need to consider the case where $v\in \hat{T}\sm \hat{C}$. In the following, we fix an arbitrary $v\in \hat{T}\sm \hat{C}$.
   
   \textbf{Step~2}. By Part~(1) of \cref{lem:prop_construct_MAG}, there is a shortest directed path \[\pi:v\tuh v_1\tuh \cdots \tuh v_{n-1}\tuh v_n=c\] for some $c\in \hat{C}$ and $v_1,\ldots, v_{n-1}\in \hat{T}$ and for every consecutive pair $v_i$, $v_{i+1}$, there is either an invisible edge $v_i\tuh v_{i+1}$ or a bidirected path between $v_i$ and $v_{i+1}$ in $\Mf_{\hat{T}}$. We shall show by induction that $v\in \pcc_{\Mf_{\hat{T}}}(v_i)$ for every $v_i$ (including node $c$) and that the pc-connection path between them is into $v_i$.  

   \textbf{Step~2.1: induction}. The base case is trivial since we have that the directed edge $v\tuh v_1$ is invisible or there is a bidirected path from $v$ to $v_1$ by Part~(1) of \cref{lem:prop_construct_MAG}. For induction, we assume that the conclusion holds for all $v_i$ with $1\leq i\leq k$. We shall prove the conclusion for $v_{k+1}$. By the induction hypothesis, we have $v\in \pcc_{\Mf_{\hat{T}}}(v_k)$ and path with $m\ge 1$ of the form
   \[
   \tilde{\pi}:v=u_0\suh u_1\huh \cdots \huh u_{m-1}\huh v_k.
   \]  
    If there is a bidirected path from $v_k$ to $v_{k+1}$, then we are done. So we assume that the directed edge $v_k\tuh v_{k+1}$ is invisible. Then $u_j$ are all adjacent to node $v_{k+1}$ with an arrowhead on $v_{k+1}$. If we have $u_j\huh v_{k+1}$ for some $j$, then we are done. So we assume that $u_j\in \pa_{\Mf_{\hat{T}}}(v_{k+1})$ for all $j\in\{0,\ldots,m-1\}$. If the directed edge $u_0=v\tuh v_{k+1}$ is invisible in $\Mf_{\hat{T}}$, then we are done. The goal is to show that $v\tuh v_{k+1}$ must be invisible in $\Mf_{\hat{T}}$. If we can show that $u_1\tuh v_{k+1}$ is invisible in $\Mf_{\hat{T}}$, then since the first edge $v\suh u_1$ of a pc-connecting path is not visible directed by definition, \cref{lem:property_visible} implies that $v\tuh v_{k+1}$ is also invisible.
    
    Therefore, the remaining task is to show that $u_1\tuh v_{k+1}$ is invisible in $\Mf_{\hat{T}}$. We argue by contradiction. Assume that this is not the case, i.e., the directed edge $u_1\tuh v_{k+1}$ is visible in $\Mf_{\hat{T}}$. Then there exists a node $d\in \hat{T}$ such that $d\suh u_1$ or $d\suh w_1\huh \cdots \huh w_{\ell-1}\huh u_1$ with $w_i\in \pa_{\Mf_{\hat{T}}}(v_{k+1})$ and $d$ is non-adjacent to $v_{k+1}$. Then we have 
   \[
   d\suh u_1\huh u_2\cdots \huh u_{m-1}\huh v_k \quad \text{ or } \quad d\suh w_1\huh \cdots \huh w_{\ell-1}\huh u_1\huh \cdots \huh v_k
   \] 
   with $w_i,u_j\in \pa_{\Mf_{\hat{T}}}(v_{k+1})$. Recall that $d$ is non-adjacent to $v_{k+1}$. Therefore, the directed edge $v_{k}\tuh v_{k+1}$ is visible, which  contradicts the fact that $v_{k}\tuh v_{k+1}$ is invisible. Hence, we have shown our claim that $u_1\tuh v_{k+1}$ is invisible in $\Mf_{\hat{T}}$. This finishes the proof of the induction step.

   \textbf{Step~2.2}. The above induction establishes that $v\in \pcc_{\Mf_{\hat{T}}}(c)$ and the pc-connecting path is into $c$:
   \[
   v\suh \tilde{v}_1\huh \cdots \huh \tilde{v}_{\tilde{m}-1}\huh c.
   \]
   If there is a bidirected path between $c$ and $v^*$ in $\Mf_{\hat{T}}$, then we are done. Otherwise, we have that $c\in \pcc_{\Mf_{\hat{T}}}(v^*)$ and the pc-connecting path is into $v^*$ but not into $c$. Let 
   \[
   \tilde{\pi}: c\tuh \tilde{u}_{1}\huh \cdots \huh \tilde{u}_{\tilde{n}-1}\huh v^*
   \] 
   be such a pc-connecting path. Note that we have $v\suh \tilde{v}_1\huh \cdots \huh \tilde{v}_{\tilde{m}-1}\huh c\tuh \tilde{u}_1$. Therefore, we can argue similarly to Step~2.1 and obtain that $v\in\pcc_{\Mf_{\hat{T}}}(\tilde{u}_1)$ and the pc-connecting path is into $\tilde{u}_1$. The path constructed by first concatenating the pc-connecting path between $v$ and $\tilde{u}_1$ and $\tilde{u}_1\huh \cdots \huh \tilde{u}_{\tilde{n}-1}\huh v^*$ and second eliminating repeated nodes (if any) gives a pc-connecting path from $v$ to $v^*$ that is into $v^*$. This finishes the proof.  
\end{proof}

\newpage

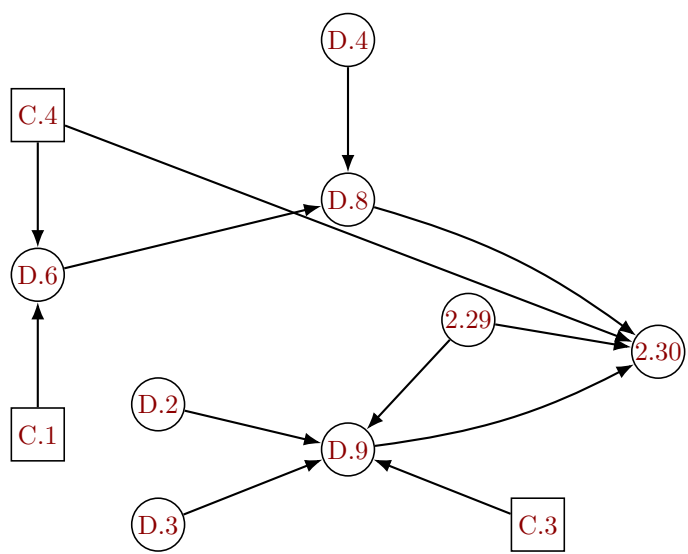
\begin{figure}[H]
\centering
\begin{tikzpicture}[
  node distance = 14mm and 20mm,
  every node/.style = {font=\small, align=center}
]

\node[ndout] (B6) {\ref{lem:sint_mag_II}};
\node[ndout, right=34mm of B6, yshift=10mm] (B8) {\ref{lem:sint_mag_III}};
\node[ndout, below=26mm of B8] (B9) {\ref{lem:sep_hint_mag}};
\node[ndout, right=34mm of B9, yshift=13mm] (T230) {\ref{thm:sep_hsint_mag}};

\node[ndint, above=14mm of B6, inner sep=2pt] (A14) {\ref{lem:visible_edge_II}};
\node[ndint, below=14mm of B6, inner sep=2pt] (A11) {\ref{lem:inducing_path}};

\node[ndout, above=14mm of B8] (B4) {\ref{lem:sint_mag_I}};

\node[ndout, above=10mm of B9, xshift=16mm] (S229) {\ref{prop:sep_mag}};

\node[ndout, left=18mm of B9, yshift=6mm] (B2) {\ref{lem:anc_selec_I}};
\node[ndout, left=18mm of B9, yshift=-10mm] (B3) {\ref{lem:anc_selec_II}};

\node[ndint, right=18mm of B9, yshift=-10mm, inner sep=2pt] (A13) {\ref{lem:visible_edge_I}};

\draw[arout] (A14) -- (B6);
\draw[arout] (A14) -- (T230);
\draw[arout] (A11) -- (B6);

\draw[arout] (B6) -- (B8);
\draw[arout] (B4) -- (B8);

\draw[arout] (S229) -- (B9);

\draw[arout] (B2) -- (B9);
\draw[arout] (B3) -- (B9);
\draw[arout] (A13) -- (B9);

\draw[arout] (B8) to[bend left=10] (T230);
\draw[arout] (B9) to[bend right=10] (T230);
\draw[arout] (S229) -- (T230);

\end{tikzpicture}
\caption{Proof structure of \cref{thm:sep_hsint_mag}.}
\label{fig:pf_structure_sep_hsint_mag}
\end{figure}

\begin{figure}[H]
\centering
\begin{tikzpicture}[
  node distance = 12mm and 18mm,
  every node/.style = {font=\small, align=center}
]

\node[ndout] (B10) {\ref{lem:soft_mani_SOPAG}};
\node[ndout, right=10mm of B10] (B11) {\ref{lem:0}};
\node[ndout, right=10mm of B11] (B13) {\ref{lem:1}};
\node[ndout, right=10mm of B13] (B15) {\ref{lem:2}};
\node[ndout, right=10mm of B15] (B16) {\ref{lem:3}};
\node[ndout, right=10mm of B16] (B18) {\ref{lem:4}};
\node[ndout, right=10mm of B18] (T239) {\ref{thm:sep_hsint_pag_mag}};
\node[ndout, above=18mm of B16] (B17) {\ref{lem:definite_hint}};

\node[ndint, below =14mm of B10, inner sep=2pt] (A22) {\ref{lem:property_visible}};

\node[ndint, above=18mm of B15, xshift=-10mm, inner sep=2pt] (A16) {\ref{lem:head_circle}};

\node[ndint, above=16mm of T239, inner sep=2pt] (A21) {\ref{lem:orientation_invisible}};
\node[ndint, below=16mm of T239, inner sep=2pt] (S229) {\ref{prop:sep_mag}};

\draw[arout] (A22) -- (B10);
\draw[arout] (A22) to[bend right=20,looseness=1.12] (B16);

\draw[arout] (B10) -- (B11);
\draw[arout] (B11) -- (B13);
\draw[arout] (B13) -- (B15);
\draw[arout] (B15) -- (B16);
\draw[arout] (B16) -- (B18);
\draw[arout] (B18) -- (T239);

\draw[arout] (B16) -- (B17);
\draw[arout] (B17) -- (B18);

\draw[arout] (A21) -- (T239);
\draw[arout] (S229) -- (T239);

\draw[arout] (A16) -- (B10);
\draw[arout] (A16) -- (B15);
\draw[arout] (A16) -- (B16);
\draw[arout] (A16) -- (B18);


\draw[arout] (B10) to[bend right=20,looseness=1.10] (B15);

\draw[arout] (B13) to[bend right=20,looseness=1.12] (B16);

\draw[arout] (B11) to[bend right=20,looseness=1.35] (B18);

\draw[arout] (B10) to[bend right=40,looseness=1.5] (B18);

\end{tikzpicture}
\caption{Proof structure of \cref{thm:sep_hsint_pag_mag}.}
\label{fig:pf_structure_sep_hsint_pag_mag}
\end{figure}
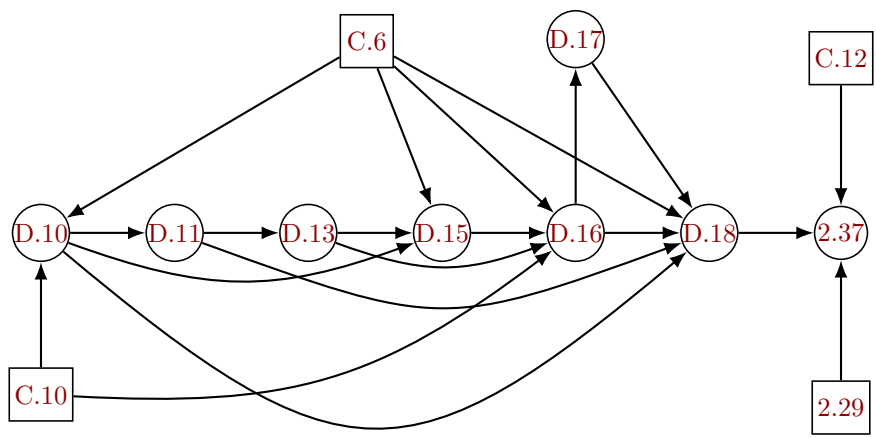

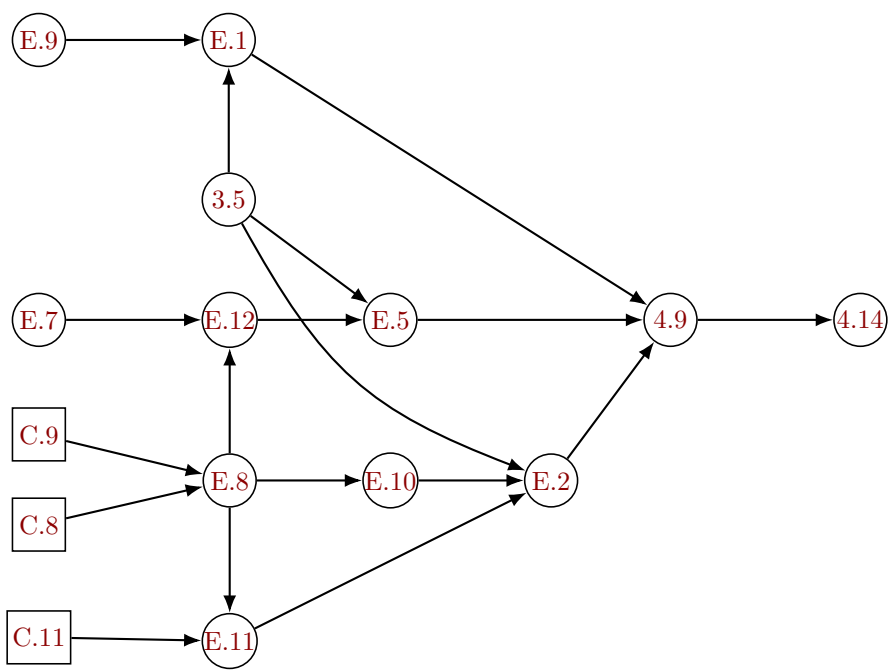
\begin{figure}[H]
\centering
\begin{tikzpicture}[
  node distance = 12mm and 18mm,
  every node/.style = {font=\small, align=center}
]

\node[ndout] (C2) {\ref{lem:IDP_sound_ci_I}};
\node[ndout, right=18mm of C2] (C3) {\ref{lem:IDP_sound_I}};
\node[ndout, below=14mm of C3] (N35) {\ref{thm:causal_calculus_mag_pag}};

\node[ndout, below=30mm of C2] (C1) {\ref{lem:pde_vs_pdet}};
\node[ndout, right=18mm of C1] (C5) {\ref{lem:IDP_sound_ci_IV}};
\node[ndout, right=14mm of C5] (C6) {\ref{lem:IDP_sound_III}};
\node[ndout, right=30mm of C6] (T47) {\ref{prop:sidp_rule}};
\node[ndout, right=18mm of T47] (T413) {\ref{thm:idp_sound}};

\node[ndout, below=14mm of C5] (C4) {\ref{lem:topo_order}};
\node[ndint, below=14mm of C1, yshift=6mm, inner sep=2pt] (A19) {\ref{lem:podpath}};
\node[ndint, below=14mm of C1, yshift=-6mm, inner sep=2pt] (A18) {\ref{lem:prop_bucket}};

\node[ndout, right=14mm of C4] (C7) {\ref{lem:IDP_sound_ci_2}};
\node[ndout, below=14mm of C4] (C8) {\ref{lem:IDP_sound_ci_3}};
\node[ndint, below=35mm of C1, inner sep=2pt] (A20) {\ref{lem:region}};

\node[ndout, right=14mm of C7] (C9) {\ref{lem:IDP_sound_II}};

\draw[arout] (C2) -- (C3);
\draw[arout] (N35) -- (C3);

\draw[arout] (C1) -- (C5);
\draw[arout] (C5) -- (C6);
\draw[arout] (C6) -- (T47);
\draw[arout] (T47) -- (T413);

\draw[arout] (A19) -- (C4);
\draw[arout] (A18) -- (C4);
\draw[arout] (C4) -- (C5);

\draw[arout] (C4) -- (C7);
\draw[arout] (C4) -- (C8);
\draw[arout] (A20) -- (C8);

\draw[arout] (C7) -- (C9);
\draw[arout] (C8) -- (C9);

\draw[arout] (N35) to (C6);
\draw[arout] (N35) to[bend right=20, looseness=1.15] (C9);

\draw[arout] (C9) to (T47);

\draw[arout] (C3) to (T47);

\end{tikzpicture}
\caption{Proof structure of \cref{thm:idp_sound,prop:sidp_rule}.}
\label{fig:pf_structure_idp_sound}
\end{figure}

\begin{figure}[H]
\centering
\begin{tikzpicture}[
  node distance=12mm and 18mm,
  every node/.style={font=\small,align=center}
]

\begin{scope}

\node[ndout] (C13) {\ref{lem:noniden}};

\node[ndout, above left=15mm and 5mm of C13] (C19) {\ref{lem:intarget}};
\node[ndout, right=15mm of C19] (C20) {\ref{lem:prop_construct_MAG}};
\node[ndout, right=15mm of C20] (C21) {\ref{lem:prop_construct_MAG_II}};
\node[ndout, left=15mm of C19] (C18) {\ref{lem:idp_version}};

\node[ndout, above=15mm of C19] (C16) {\ref{lem:construct_MAG}};
\node[ndout, above=15mm of C20] (C17) {\ref{lem:prop_construct_MAG_III}};

\node[ndout, above=50mm of C13] (Thm) {\ref{thm:idp_complete}};

\draw[arout] (C18) -- (C19);
\draw[arout] (C19) -- (C20);
\draw[arout] (C20) -- (C21);

\draw[arout] (C16) -- (Thm);
\draw[arout] (C17) -- (Thm);

\draw[arout] (C18) -- (C16);
\draw[arout] (C20) -- (C17);
\draw[arout] (C21) -- (C17);
\draw[arout] (C16) -- (C20);

\draw[arout] (C13) -- (C19);
\draw[arout] (C13) -- (C20);
\draw[arout] (C13) -- (C18);
\draw[arout] (C13) -- (C16);
\draw[arout] (C13) to[bend right=45, looseness=1.20] (C17);

\end{scope}

\coordinate (belowTop) at ($(current bounding box.south)+(13mm,-26mm)$);

\begin{scope}[shift={(belowTop)}]

\node[ndout] (bC19) at (0,0) {\ref{lem:intarget}};

\node[ndint, below=14mm of bC19, inner sep=2pt] (A28) {\ref{lem:region}};

\node[ndout, below=14mm of A28] (bC18) {\ref{lem:idp_version}};

\node[ndout, left=36mm of bC18] (bC16) {\ref{lem:construct_MAG}};
\node[ndint, right=16mm of bC16, inner sep=2pt] (A26) {\ref{lem:prop_bucket}};

\node[ndint, right=22mm of bC18, inner sep=2pt] (A24) {\ref{lem:head_circle}};

\node[ndout, below=28mm of bC18, xshift=10mm] (bC20) {\ref{lem:prop_construct_MAG}};
\node[ndint, above=10mm of bC20, inner sep=2pt] (A27) {\ref{lem:podpath}};

\node[ndint, below=12mm of bC20, inner sep=2pt] (A29) {\ref{lem:orientation_invisible}};
\node[ndout, left=25mm of A29] (bC17b) {\ref{lem:prop_construct_MAG_III}};
\node[ndout, right=10mm of A29] (bC13t) {\ref{lem:noniden}};
\node[ndint, right =10mm  of bC13t, inner sep=2pt] (A32) {\ref{prop:hedge}}; 
\node[ndout, above =10mm  of A32, inner sep=2pt] (thm) {\ref{thm:idp_complete}}; 
\node[ndint, left=26mm of bC20, inner sep=2pt] (A30) {\ref{lem:property_visible}};
\node[ndout, left=22mm of A30] (bC21) {\ref{lem:prop_construct_MAG_II}};

\draw[arout] (A28) -- (bC19);

\draw[arout] (A32) -- (bC13t);
\draw[arout] (A32) -- (thm);
\draw[arout] (A29) -- (bC13t);

\draw[arout] (A26) -- (bC16);
\draw[arout] (A30) -- (bC16);
\draw[arout] (A26) -- (bC18);
\draw[arout] (A30) -- (bC18);
\draw[arout] (A26) -- (bC20);

\draw[arout] (A24) -- (bC18);
\draw[arout] (A24) -- (bC19);
\draw[arout] (A24) -- (bC20);

\draw[arout] (A27) -- (bC20);
\draw[arout] (A27) -- (bC18);

\draw[arout] (A29) -- (bC20);
\draw[arout] (A29) -- (bC17b);

\draw[arout] (A30) -- (bC21);
\draw[arout] (A30) -- (bC20);
\draw[arout] (A30) -- (bC17b);

\end{scope}

\end{tikzpicture}
\label{fig:pf_structure_idp_complete}
\caption{Proof structure of \cref{thm:idp_complete}.}
\end{figure}

\newpage
\bibliographystyle{plainurl}
\bibliography{bibfile_cdscm}

@InProceedings{reisach2021beware_simulated_dag,
  title     = {Beware of the Simulated DAG! Causal Discovery Benchmarks May Be Easy to Game},
  author    = {Reisach, Alexander G. and Seiler, Christof and Weichwald, Sebastian},
  booktitle = {Advances in Neural Information Processing Systems},
  volume    = {34},
  pages     = {27772--27784},
  year      = {2021}
}

@Article{Peters2015structural_int_distance,
  title   = {Structural Intervention Distance (SID) for Evaluating Causal Graphs},
  author  = {Peters, J. and B{\"u}hlmann, P.},
  journal = {Neural Computation},
  volume  = {27},
  number  = {3},
  pages   = {771--799},
  year    = {2015},
  doi     = {10.1162/NECO_a_00708}
}

@InProceedings{gentzel2019case_interventional_eval,
  title     = {The Case for Evaluating Causal Models Using Interventional Measures and Empirical Data},
  author    = {Gentzel, Amanda and Garant, Dan and Jensen, David},
  booktitle = {Advances in Neural Information Processing Systems},
  volume    = {32},
  pages     = {11722--11732},
  year      = {2019}
}

@article{dang2026effect_level_validation,
  title         = {Effect-Level Validation for Causal Discovery},
  author        = {Dang, Hoang and Pham, Luan and Nguyen, Minh},
  year          = {2026},
  journal={arXiv.org preprint},
  volume={arXiv: 2602.08340 [cs.AI]},
  url           = {https://arxiv.org/abs/2602.08340}
}

@article{Chang2025PostSelection,
  title   = {Post-selection inference for causal effects after causal discovery},
  author  = {Chang, Ting-Hsuan and Guo, Zijian and Malinsky, Daniel},
  journal = {Biometrika},
  year    = {2025},
  pages   = {asaf073},
  doi     = {10.1093/biomet/asaf073},
  note    = {Advance article}
}

@article{Gradu2025ValidInference,
  title   = {Valid Inference After Causal Discovery},
  author  = {Gradu, Paula and Zrnic, Tijana and Wang, Yixin and Jordan, Michael I.},
  journal = {Journal of the American Statistical Association},
  year    = {2025},
  volume  = {120},
  number  = {550},
  pages   = {1127--1138},
  doi     = {10.1080/01621459.2024.2402089}
}

@article{Uhler2013Faithfulness,
  title   = {Geometry of the Faithfulness Assumption in Causal Inference},
  author  = {Uhler, Caroline and Raskutti, Garvesh and B{\"u}hlmann, Peter and Yu, Bin},
  journal = {The Annals of Statistics},
  year    = {2013},
  volume  = {41},
  number  = {2},
  pages   = {436--463},
  doi     = {10.1214/12-AOS1080}
}

@InProceedings{Jung21Estimate_causal_effec_on_markov_equiva,
  title     = {Estimating Identifiable Causal Effects on Markov Equivalence Class through Double Machine Learning},
  author    = {Jung, Yonghan and Tian, Jin and Bareinboim, Elias},
  booktitle = {Proceedings of the 38th International Conference on Machine Learning},
  pages     = {5168--5179},
  year      = {2021},
  editor    = {Meila, Marina and Zhang, Tong},
  volume    = {139},
  series    = {Proceedings of Machine Learning Research},
  month     = {18--24 Jul},
  publisher = {PMLR},
  url       = {https://proceedings.mlr.press/v139/jung21b.html}
}

@InProceedings{Bellot24bound_causal_effects_mark_equi,
  title     = {Towards Bounding Causal Effects under Markov Equivalence},
  author    = {Bellot, Alexis},
  booktitle = {Proceedings of the Fortieth Conference on Uncertainty in Artificial Intelligence},
  pages     = {308--332},
  year      = {2024},
  volume    = {244},
  series    = {Proceedings of Machine Learning Research},
  month     = {15--19 Jul},
  publisher = {PMLR},
  url       = {https://proceedings.mlr.press/v244/bellot24a.html}
}

@misc{Dawid2025CoxFoundationsLecture,
  author       = {Dawid, A. Philip},
  title        = {A Lifetime of Irrelevance: {Conditional Independence}: The Relevance of Irrelevance},
  howpublished = {Handout for the David R. Cox Foundations of Statistics Lecture 2025, University of Cambridge},
  year         = {2025},
  url          = {https://www.statslab.cam.ac.uk/~apd25/drcpresentationhandout.pdf},
  urldate      = {2026-02-28}
}

@article{Dawid1979ConditionalIndependence,
  title   = {Conditional Independence in Statistical Theory (with Discussion)},
  author  = {Dawid, A. Philip},
  journal = {Journal of the Royal Statistical Society: Series B (Methodological)},
  year    = {1979},
  volume  = {41},
  number  = {1},
  pages   = {1--31},
  doi     = {10.1111/j.2517-6161.1979.tb01052.x}
}

@article{HuVanDerPas2025SelectingValidAdjustmentSets,
  title         = {Selecting valid adjustment sets with uncertain causal graphs},
  author        = {Hu, Zhongyi and van der Pas, St{\'e}phanie},
  year          = {2025},
  journal={arXiv.org preprint},
  volume={arXiv: 2511.01662 [stat.ST]},
}

@inproceedings{PerkovicKalischMaathuis2017CPDAGBK,
  author    = {Emilija Perkovic and Markus Kalisch and Marloes H. Maathuis},
  title     = {Interpreting and Using {CPDAGs} With Background Knowledge},
  booktitle = {Proceedings of the Thirty-Third Conference on Uncertainty in Artificial Intelligence (UAI 2017), Sydney, Australia, August 11-15, 2017},
  editor    = {Gal Elidan and Kristian Kersting and Alexander Ihler},
  publisher = {AUAI Press},
  year      = {2017},
  url       = {http://auai.org/uai2017/proceedings/papers/120.pdf}
}

@inproceedings{Perkovic2020MPDAGID,
  title     = {Identifying causal effects in maximally oriented partially directed acyclic graphs},
  author    = {Perkovic, Emilija},
  booktitle = {Proceedings of the 36th Conference on Uncertainty in Artificial Intelligence (UAI)},
  editor    = {Peters, Jonas and Sontag, David},
  series    = {Proceedings of Machine Learning Research},
  volume    = {124},
  pages     = {530--539},
  year      = {2020},
  publisher = {PMLR},
  url       = {https://proceedings.mlr.press/v124/perkovic20a.html}
}

@article{WitteHenckelMaathuisDidelez2020EfficientAdjustment,
  title   = {On Efficient Adjustment in Causal Graphs},
  author  = {Witte, Janine and Henckel, Leonard and Maathuis, Marloes H. and Didelez, Vanessa},
  journal = {Journal of Machine Learning Research},
  year    = {2020},
  volume  = {21},
  number  = {246},
  pages   = {1--45},
  url     = {https://jmlr.org/papers/v21/20-175.html}
}

@article{Maathuis2009IDA,
  title   = {Estimating High-Dimensional Intervention Effects from Observational Data},
  author  = {Maathuis, Marloes H. and Kalisch, Markus and B{\"u}hlmann, Peter},
  journal = {The Annals of Statistics},
  year    = {2009},
  volume  = {37},
  number  = {6A},
  pages   = {3133--3164},
  doi     = {10.1214/09-AOS685}
}

@inproceedings{vanDerZander2014ConstructingSeparators,
  title     = {Constructing Separators and Adjustment Sets in Ancestral Graphs},
  author    = {van der Zander, Benito and Li{\'s}kiewicz, Maciej and Textor, Johannes},
  booktitle = {Proceedings of the Thirtieth Conference on Uncertainty in Artificial Intelligence (UAI 2014)},
  year      = {2014},
  publisher = {AUAI Press},
  pages     = {907--916}
}

@inproceedings{Hyttinen2015DoCalculusUnknown,
  title     = {Do-calculus when the True Graph Is Unknown},
  author    = {Hyttinen, Antti and Eberhardt, Frederick and J{\"a}rvisalo, Matti},
  booktitle = {Proceedings of the Thirty-First Conference on Uncertainty in Artificial Intelligence (UAI 2015)},
  year      = {2015},
  publisher = {AUAI Press},
  pages     = {395--404}
}

@article{Nandy2017JointInterventions,
  title   = {Estimating the Effect of Joint Interventions from Observational Data in Sparse High-Dimensional Settings},
  author  = {Nandy, Preetam and Maathuis, Marloes H. and Richardson, Thomas S.},
  journal = {The Annals of Statistics},
  year    = {2017},
  volume  = {45},
  number  = {2},
  pages   = {647--674},
  doi     = {10.1214/16-AOS1462}
}

@techreport{Dawid2007Fundamentals,
  author      = {Dawid, A. Philip},
  title       = {Fundamentals of Statistical Causality},
  institution = {Department of Statistical Science, University College London},
  type        = {Research Report},
  number      = {279},
  year        = {2007},
  pages       = {94},
}

@article{constantinou17extendedci,
 ISSN = {00905364},
 URL = {http://www.jstor.org/stable/26362953},
 author = {Panayiota Constantinou and A. Philip Dawid},
 journal = {The Annals of Statistics},
 number = {6},
 pages = {2618--2653},
 publisher = {Institute of Mathematical Statistics},
 title = {EXTENDED CONDITIONAL INDEPENDENCE AND APPLICATIONS IN CAUSAL INFERENCE},
 urldate = {2025-03-18},
 volume = {45},
 year = {2017}
}

@article{Elwert2014EndogenousSelectionBias,
  title   = {Endogenous Selection Bias: The Problem of Conditioning on a Collider Variable},
  author  = {Elwert, Felix and Winship, Christopher},
  journal = {Annual Review of Sociology},
  year    = {2014},
  volume  = {40},
  number  = {1},
  pages   = {31--53},
  doi     = {10.1146/annurev-soc-071913-043455}
}

@article{Munafo2018ColliderScope,
  title   = {Collider scope: when selection bias can substantially influence observed associations},
  author  = {Munaf{\`o}, Marcus R. and Tilling, Kate and Taylor, Amy E. and Evans, David M. and Davey Smith, George},
  journal = {International Journal of Epidemiology},
  year    = {2018},
  volume  = {47},
  number  = {1},
  pages   = {226--235},
  doi     = {10.1093/ije/dyx206}
}

@article{GreenlandPearlRobins1999CausalDiagrams,
  title   = {Causal diagrams for epidemiologic research},
  author  = {Greenland, Sander and Pearl, Judea and Robins, James M.},
  journal = {Epidemiology},
  year    = {1999},
  volume  = {10},
  number  = {1},
  pages   = {37--48},
  doi     = {10.1097/00001648-199901000-00008}
}

@article{Colombo12learning_highdim_dag,
 ISSN = {00905364, 21688966},
 URL = {http://www.jstor.org/stable/41713636},
 author = {Diego Colombo and Marloes H. Maathuis and Markus Kalisch and Thomas S. Richardson},
 journal = {The Annals of Statistics},
 number = {1},
 pages = {294--321},
 publisher = {Institute of Mathematical Statistics},
 title = {LEARNING HIGH-DIMENSIONAL DIRECTED ACYCLIC GRAPHS WITH LATENT AND SELECTION VARIABLES},
 urldate = {2026-02-26},
 volume = {40},
 year = {2012}
}

@article{Colombo14order_ind_constrint_based_causal_structure_learning,
  author  = {Diego Colombo and Marloes H. Maathuis},
  title   = {Order-Independent Constraint-Based Causal Structure Learning},
  journal = {Journal of Machine Learning Research},
  year    = {2014},
  volume  = {15},
  number  = {116},
  pages   = {3921--3962},
  url     = {http://jmlr.org/papers/v15/colombo14a.html}
}

@article{Ali09mark_equiv_ancestral_graphs,
author = {R. Ayesha Ali and Thomas S. Richardson and Peter Spirtes},
title = {{Markov equivalence for ancestral graphs}},
volume = {37},
journal = {The Annals of Statistics},
number = {5B},
publisher = {Institute of Mathematical Statistics},
pages = {2808 -- 2837},
year = {2009},
doi = {10.1214/08-AOS626},
URL = {https://doi.org/10.1214/08-AOS626}
}

@inproceedings{Pearl1994prob_calculus_action,
  author    = {Pearl, Judea},
  title     = {A Probabilistic Calculus of Actions},
  booktitle = {Proceedings of the Tenth Conference on Uncertainty in Artificial Intelligence (UAI-94)},
  year      = {1994},
  pages     = {454--462},
}

@article{chen2025foundationsstructuralcausalmodels,
      title={Foundations of Structural Causal Models with Latent Selection}, 
      author={Leihao Chen and Onno Zoeter and Joris M. Mooij},
      year={2025},
      url={https://arxiv.org/abs/2401.06925}, 
      journal={arXiv.org preprint},
      volume={arXiv: 2401.06925 [stat.ME]},
}

@article{Rosset_18_bound_card_hidden,
author = {Rosset, Denis and Gisin, Nicolas and Wolfe, Elie},
title = {Universal bound on the cardinality of local hidden variables in networks},
year = {2018},
issue_date = {September 2018},
publisher = {Rinton Press, Incorporated},
address = {Paramus, NJ},
volume = {18},
number = {11–12},
issn = {1533-7146},
journal = {Quantum Info. Comput.},
month = sep,
pages = {910–926},
numpages = {17}
}

@InProceedings{Hwang_2024_pos,
  title     = {On Positivity Condition for Causal Inference},
  author    = {Hwang, Inwoo and Choe, Yesong and Kwon, Yeahoon and Lee, Sanghack},
  booktitle = {Proceedings of the 41st International Conference on Machine Learning},
  pages     = {20818--20841},
  year      = {2024},
  volume    = {235},
  series    = {Proceedings of Machine Learning Research},
  month     = {21--27 Jul},
  publisher = {PMLR},
  url       = {https://proceedings.mlr.press/v235/hwang24a.html},
}

@article{Dawid93Hyper_Markov,
author = {Dawid, A. Philip and Steffen, L. Lauritzen},
title = {{Hyper Markov Laws in the Statistical Analysis of Decomposable Graphical Models}},
volume = {21},
journal = {The Annals of Statistics},
number = {3},
publisher = {Institute of Mathematical Statistics},
pages = {1272 -- 1317},
year = {1993},
doi = {10.1214/aos/1176349260},
URL = {https://doi.org/10.1214/aos/1176349260}
}

@article{Byrne15Structural_Markov,
author = {Simon Byrne and A. Philip Dawid},
title = {{Structural Markov graph laws for Bayesian model uncertainty}},
volume = {43},
journal = {The Annals of Statistics},
number = {4},
publisher = {Institute of Mathematical Statistics},
pages = {1647 -- 1681},
year = {2015},
doi = {10.1214/15-AOS1319},
URL = {https://doi.org/10.1214/15-AOS1319}
}

@article{Goudie19Joining_Splitting,
author = {Robert J. B. Goudie and Anne M. Presanis and David Lunn and Daniela De Angelis and Lorenz Wernisch},
title = {{Joining and Splitting Models with Markov Melding}},
volume = {14},
journal = {Bayesian Analysis},
number = {1},
publisher = {International Society for Bayesian Analysis},
pages = {81 -- 109},
year = {2019},
doi = {10.1214/18-BA1104},
URL = {https://doi.org/10.1214/18-BA1104}
}

@article{massa2010combining,
  title={Combining statistical models},
  author={Massa, M. Sofia and Lauritzen, Steffen L.},
  journal={Contemporary Mathematics},
  volume={516},
  pages={239--259},
  year={2010}
}

@article{bartolucci2000_mtp2_lrt,
  title   = {A Likelihood Ratio Test for {MTP}$_2$ within Binary Variables},
  author  = {Francesco Bartolucci and Antonio Forcina},
  year    = {2000},
  journal = {Annals of Statistics},
  volume  = {28},
  number  = {4},
  pages   = {1206--1218},
  doi     = {10.1214/aos/1015956713},
  url     = {https://doi.org/10.1214/aos/1015956713}
}

@article{lauritzen2021_mtp2_expfam_binary,
  title   = {Total Positivity in Exponential Families with Application to Binary Variables},
  author  = {Steffen L. Lauritzen and Caroline Uhler and Piotr Zwiernik},
  year    = {2021},
  journal = {Annals of Statistics},
  volume  = {49},
  number  = {3},
  pages   = {1436--1459},
  doi     = {10.1214/20-AOS2007},
  url     = {https://doi.org/10.1214/20-AOS2007}
}

@article{fallat2017_mtp2_markov_structures,
  title   = {Total Positivity in Markov Structures},
  author  = {Shaun Fallat and Steffen L. Lauritzen and Kayvan Sadeghi and Caroline Uhler and Nanny Wermuth and Piotr Zwiernik},
  year    = {2017},
  journal = {Annals of Statistics},
  volume  = {45},
  number  = {3},
  pages   = {1152--1184},
  doi     = {10.1214/16-AOS1478},
  url     = {https://doi.org/10.1214/16-AOS1478}
}

@article{lauritzen2019_mle_gaussian_mtp2,
  title   = {Maximum Likelihood Estimation in Gaussian Models under Total Positivity},
  author  = {Steffen L. Lauritzen and Caroline Uhler and Piotr Zwiernik},
  year    = {2019},
  journal = {Annals of Statistics},
  volume  = {47},
  number  = {4},
  pages   = {1835--1863},
  doi     = {10.1214/17-AOS1668},
  url     = {https://doi.org/10.1214/17-AOS1668}
}

@article{chen25notes,
      title={Notes on Forr\'e's Notion of Conditional Independence and Causal Calculus for Continuous Variables}, 
      author={Leihao Chen},
      journal={arXiv.org preprint},
      year={2025},
      volume={arXiv: 2603.24333 [stat.ST]},
      url={https://arxiv.org/abs/2603.24333},
}

@InProceedings{lee20gID,
  title = 	 {General Identifiability with Arbitrary Surrogate Experiments},
  author =       {Lee, Sanghack and Correa, Juan D. and Bareinboim, Elias},
  booktitle = 	 {Proceedings of The 35th Uncertainty in Artificial Intelligence Conference},
  pages = 	 {389--398},
  year = 	 {2020},
  editor = 	 {Adams, Ryan P. and Gogate, Vibhav},
  volume = 	 {115},
  series = 	 {Proceedings of Machine Learning Research},
  month = 	 {22--25 Jul},
  publisher =    {PMLR},
  url = 	 {https://proceedings.mlr.press/v115/lee20b.html}
}

@InProceedings{Yaroslav22revisite_gID,
  title = 	 {Revisiting the general identifiability problem},
  author =       {Kivva, Yaroslav and Mokhtarian, Ehsan and Etesami, Jalal and Kiyavash, Negar},
  booktitle = 	 {Proceedings of the Thirty-Eighth Conference on Uncertainty in Artificial Intelligence},
  pages = 	 {1022--1030},
  year = 	 {2022},
  volume = 	 {180},
  series = 	 {Proceedings of Machine Learning Research},
  month = 	 {Aug},
  url = 	 {https://proceedings.mlr.press/v180/kivva22a.html}
}

@inproceedings{Perkovic15complete,
  author    = {Perkovi\'{c}, Emilija and Textor, Johannes and Kalisch, Markus and Maathuis, Marloes H.},
  title     = {A complete generalized adjustment criterion},
  booktitle = {Proceedings of the 31st Conference on Uncertainty in Artificial Intelligence (UAI)},
  year      = {2015}
}

@inproceedings{Shpitser10validity,
  author    = {Ilya Shpitser and T. J. Van der Weele and James M. Robins},
  title     = {On the validity of covariate adjustment for estimating causal effects},
  booktitle = {Proceedings of the 26th Conference on Uncertainty in Artificial Intelligence (UAI)},
  year      = {2010}
}

@inproceedings{Pearl10confounding,
  author    = {Judea Pearl and Azaria Paz},
  title     = {Confounding equivalence in causal inference},
  booktitle = {Proceedings of the 26th Conference on Uncertainty in Artificial Intelligence (UAI)},
  year      = {2010}
}

@inproceedings{Pearl93aAspects,
  author    = {Judea Pearl},
  title     = {Aspects of graphical models connected with causality},
  booktitle = {Proceedings of the 49th Session of the International Statistical Institute},
  pages     = {391--401},
  year      = {1993}
}

@article{Richard01causal_inference_continuous,
author = {Richard D. Gill and James M. Robins},
title = {{Causal Inference for Complex Longitudinal Data: The Continuous Case}},
volume = {29},
journal = {The Annals of Statistics},
number = {6},
publisher = {Institute of Mathematical Statistics},
pages = {1785 -- 1811},
year = {2001},
doi = {10.1214/aos/1015345962},
URL = {https://doi.org/10.1214/aos/1015345962}
}

@article{forre2025mathematical,
  title={A Mathematical Introduction to Causality},
  author={Forr{\'e}, Patrick and Mooij, Joris M.},
  year={2025},
  url={https://staff.fnwi.uva.nl/j.m.mooij/articles/causality_lecture_notes_2025.pdf}
}

@book{kallenberg2017random,
  title={Random Measures, Theory and Applications},
  author={Kallenberg, Olav},
  isbn={9783319415987},
  series={Probability Theory and Stochastic Modelling},
  url={https://books.google.nl/books?id=i6WoDgAAQBAJ},
  year={2017},
  publisher={Springer International Publishing}
}

@article{bogachev20kant,
title = {Kantorovich problems and conditional measures depending on a parameter},
journal = {Journal of Mathematical Analysis and Applications},
volume = {486},
number = {1},
pages = {123883},
year = {2020},
issn = {0022-247X},
doi = {https://doi.org/10.1016/j.jmaa.2020.123883},
url = {https://www.sciencedirect.com/science/article/pii/S0022247X20300457},
author = {Vladimir I. Bogachev and Ilya I. Malofeev},
}

@article{shpitser23doesidalgorithmfail,
      title={When does the ID algorithm fail?}, 
      author={Ilya Shpitser},
      journal={arXiv.org preprint},
      year={2023},
      volume={arXiv:2307.03750 [stat.ME]},
      url={https://arxiv.org/abs/2307.03750},
}

@article{Mooij20JCI,
  author  = {Joris M. Mooij and Sara Magliacane and Tom Claassen},
  title   = {Joint Causal Inference from Multiple Contexts},
  journal = {Journal of Machine Learning Research},
  year    = {2020},
  volume  = {21},
  number  = {99},
  pages   = {1-108},
  url     = {http://jmlr.org/papers/v21/17-123.html}
}

@article{venkateswaran24completecausalexplanationexpert,
      title={Towards Complete Causal Explanation with Expert Knowledge}, 
      author={Aparajithan Venkateswaran and Emilija Perković},
      journal={arXiv.org preprint},
      year={2024},
      volume={arXiv:2407.07338 [stat.ML]},
      url={https://arxiv.org/abs/2407.07338},
}

@InProceedings{bryan20tieredbackground,
  title = 	 {On the Completeness of Causal Discovery in the Presence of Latent Confounding with Tiered Background Knowledge},
  author =       {Andrews, Bryan and Spirtes, Peter and Cooper, Gregory F.},
  booktitle = 	 {Proceedings of the Twenty Third International Conference on Artificial Intelligence and Statistics},
  pages = 	 {4002--4011},
  year = 	 {2020},
  volume = 	 {108},
  series = 	 {Proceedings of Machine Learning Research},
  month = 	 {26--28 Aug},
  publisher =    {PMLR},
  url = 	 {https://proceedings.mlr.press/v108/andrews20a.html}
}

@article{perkovic18complete,
  author  = {Emilija Perkovi\'c and Johannes Textor and Markus Kalisch and Marloes H. Maathuis},
  title   = {Complete Graphical Characterization and Construction of Adjustment Sets in Markov Equivalence Classes of Ancestral Graphs},
  journal = {Journal of Machine Learning Research},
  year    = {2018},
  volume  = {18},
  number  = {220},
  pages   = {1-62},
  url     = {http://jmlr.org/papers/v18/16-319.html}
}

@InProceedings{jaber19idencom,
  title = 	 {Causal Identification under {M}arkov Equivalence: Completeness Results},
  author =       {Jaber, Amin and Zhang, Jiji and Bareinboim, Elias},
  booktitle = 	 {Proceedings of the 36th International Conference on Machine Learning},
  pages = 	 {2981--2989},
  year = 	 {2019},
  volume = 	 {97},
  series = 	 {Proceedings of Machine Learning Research},
  month = 	 {Jun},
  url = 	 {https://proceedings.mlr.press/v97/jaber19a.html}
}

@inproceedings{jaber22causal,
 author = {Jaber, Amin and Ribeiro, Adele and Zhang, Jiji and Bareinboim, Elias},
 booktitle = {Advances in Neural Information Processing Systems},
 pages = {3679--3690},
 title = {Causal Identification under Markov equivalence: Calculus, Algorithm, and Completeness},
 url = {https://proceedings.neurips.cc/paper_files/paper/2022/file/17a9ab4190289f0e1504bbb98d1d111a-Paper-Conference.pdf},
 volume = {35},
 year = {2022}
}

@phdthesis{zhang2006causal,
  title={Causal inference and reasoning in causally insufficient systems},
  author={Zhang, Jiji},
  year={2006},
  school={Carnegie Mellon University}
}

@article{zhang2008causal,
  title={Causal reasoning with ancestral graphs},
  author={Zhang, Jiji},
  journal={Journal of Machine Learning Research},
  volume={9},
  pages={1437--1474},
  year={2008},
  publisher={Microtome Publishing}
}

@article{forre2021transitional,
  title={Transitional conditional independence},
  author={Forr{\'e}, Patrick},
  journal={arXiv.org preprint},
  volume={arXiv:2104.11547 [math.ST]},
  year={2021}
}

@article{hernan2004structural,
 author = {Miguel A. Hernán and Sonia Hernández-Díaz and James M. Robins},
 journal = {Epidemiology},
 number = {5},
 pages = {615--625},
 publisher = {Lippincott Williams & Wilkins},
 title = {A Structural Approach to Selection Bias},
 volume = {15},
 year = {2004}
}

@book{Hernan2020WhatIf,
  author    = {Hern{\'a}n, Miguel A. and Robins, James M.},
  title     = {Causal Inference: What If},
  year      = {2020},
  address   = {Boca Raton, FL},
  publisher = {Chapman \& Hall/CRC},
  note      = {Online book (continuously updated)},
  url       = {https://miguelhernan.org/whatifbook}
}

@inproceedings{bareinboim2015recovering,
  title={Recovering causal effects from selection bias},
  author={Bareinboim, Elias and Tian, Jin},
  booktitle={Proceedings of the AAAI Conference on Artificial Intelligence},
  volume={29},
  number={1},
  pages={2410–2416},
  year={2015}
}

@article{shpitser2008complete,
  title={Complete identification methods for the causal hierarchy},
  author={Shpitser, Ilya and Pearl, Judea},
  journal={Journal of Machine Learning Research},
  volume={9},
  pages={1941--1979},
  year={2008}
}

@article{richardson2002ancestral,
  title={Ancestral graph Markov models},
  author={Richardson, Thomas S. and Spirtes, Peter},
  journal={The Annals of Statistics},
  volume={30},
  number={4},
  pages={962--1030},
  year={2002},
  publisher={Institute of Mathematical Statistics}
}

@book{pearl2009causality, place={Cambridge}, edition={2nd}, title={Causality: Models, Reasoning, and Inference}, DOI={10.1017/CBO9780511803161}, publisher={Cambridge University Press}, author={Pearl, Judea}, year={2009}}

@article{bongers2021foundations,
  title={Foundations of structural causal models with cycles and latent variables},
  author={Bongers, Stephan and Forr{\'e}, Patrick and Peters, Jonas and Mooij, Joris M.},
  journal={The Annals of Statistics},
  volume={49},
  number={5},
  pages={2885--2915},
  year={2021},
  publisher={Institute of Mathematical Statistics}
}

@article{richardson2023nested,
  title={Nested Markov properties for acyclic directed mixed graphs},
  author={Richardson, Thomas S. and Evans, Robin J. and Robins, James M. and Shpitser, Ilya},
  journal={The Annals of Statistics},
  volume={51},
  number={1},
  pages={334--361},
  year={2023},
  publisher={Institute of Mathematical Statistics}
}

@inproceedings{forre2020causal,
  title={Causal calculus in the presence of cycles, latent confounders and selection bias},
  author={Forr{\'e}, Patrick and Mooij, Joris M.},
  booktitle={Proceedings of the 36th Conference on Uncertainty in Artificial Intelligence},
  pages={71--80},
  year={2020},
}

@article{richar03mark,
author = {Richardson, Thomas S.},
title = {Markov Properties for Acyclic Directed Mixed Graphs},
journal = {Scandinavian Journal of Statistics},
volume = {30},
number = {1},
pages = {145-157},
doi = {https://doi.org/10.1111/1467-9469.00323},
url = {https://onlinelibrary.wiley.com/doi/abs/10.1111/1467-9469.00323},
year = {2003}
}

@book{spirtes2001causation,
  title={Causation, prediction, and search},
  author={Spirtes, Peter and Glymour, Clark and Scheines, Richard},
  year={2001},
  publisher={MIT press}
}

@inproceedings{spirtes95causal,
author = {Spirtes, Peter and Meek, Christopher and Richardson, Thomas},
title = {Causal inference in the presence of latent variables and selection bias},
year = {1995},
booktitle = {Proceedings of the 11th Conference on Uncertainty in Artificial Intelligence},
pages = {499–506},
numpages = {8},
}

@incollection{spirtes99alogorithm,
    author = {Spirtes, P and Meek, C and Richardson, T},
    isbn = {9780262315821},
    title = "{An Algorithm for Causal Inference in the Presence of Latent Variables and Selection Bias}",
    booktitle = "{Computation, Causation, and Discovery}",
    publisher = {AAAI Press},
    year = {1999},
    month = {05},
    doi = {10.7551/mitpress/2006.003.0009},
    url = {https://doi.org/10.7551/mitpress/2006.003.0009},
    eprint = {https://direct.mit.edu/book/chapter-pdf/278764/9780262315821\_cah.pdf},
}

@article{Correa17causal, title={Causal Effect Identification by Adjustment under Confounding and Selection Biases}, volume={31}, url={https://ojs.aaai.org/index.php/AAAI/article/view/11060}, DOI={10.1609/aaai.v31i1.11060}, abstractNote={ &lt;p&gt; Controlling for selection and confounding biases are two of the most challenging problems in the empirical sciences as well as in artificial intelligence tasks. Covariate adjustment (or, Backdoor Adjustment) is the most pervasive technique used for controlling confounding bias, but the same is oblivious to issues of sampling selection. In this paper, we introduce a generalized version of covariate adjustment that simultaneously controls for both confounding and selection biases. We first derive a sufficient and necessary condition for recovering causal effects using covariate adjustment from an observational distribution collected under preferential selection. We then relax this setting to consider cases when additional, unbiased measurements over a set of covariates are available for use (e.g., the age and gender distribution obtained from census data). Finally, we present a complete algorithm with polynomial delay to find all sets of admissible covariates for adjustment when confounding and selection biases are simultaneously present and unbiased data is available. &lt;/p&gt; }, number={1}, journal={Proceedings of the AAAI Conference on Artificial Intelligence}, author={Correa, Juan and Bareinboim, Elias}, year={2017}, month={Feb.} }

@inproceedings{shpitser06identification,
author = {Shpitser, Ilya and Pearl, Judea},
title = {Identification of conditional interventional distributions},
year = {2006},
booktitle = {Proceedings of the 22ed Conference on Uncertainty in Artificial Intelligence},
pages = {437–444},
numpages = {8},

}

@article{pearl95causal,
 ISSN = {00063444},
 URL = {http://www.jstor.org/stable/2337329},
 author = {Judea Pearl},
 journal = {Biometrika},
 number = {4},
 pages = {669--688},
 publisher = {[Oxford University Press, Biometrika Trust]},
 title = {Causal Diagrams for Empirical Research},
 urldate = {2024-05-27},
 volume = {82},
 year = {1995}
}

@inproceedings{tian02general,
author = {Tian, Jin and Pearl, Judea},
title = {A general identification condition for causal effects},
year = {2002},
isbn = {0262511290},
publisher = {American Association for Artificial Intelligence},
address = {USA},
booktitle = {Eighteenth National Conference on Artificial Intelligence},
pages = {567–573},
numpages = {7},
location = {Edmonton, Alberta, Canada}
}

@inproceedings{huang06pearl,
author = {Huang, Yimin and Valtorta, Marco},
title = {Pearl's calculus of intervention is complete},
year = {2006},
booktitle = {Proceedings of the 22ed Conference on Uncertainty in Artificial Intelligence},
pages = {217–224},
numpages = {8},
}

@article{Huang2008OnTC,
  title={On the completeness of an identifiability algorithm for semi-Markovian models},
  author={Yimin Huang and Marco G. Valtorta},
  journal={Annals of Mathematics and Artificial Intelligence},
  year={2008},
  volume={54},
  pages={363-408},
  url={https://api.semanticscholar.org/CorpusID:52818662}
}

@inproceedings{shpister06joint,
title = "Identification of joint interventional distributions in recursive semi-Markovian causal models",
author = "Ilya Shpitser and Judea Pearl",
year = "2006",
month = nov,
day = "13",
language = "English (US)",
isbn = "1577352815",
pages = "1219--1226",
booktitle = "Proceedings of the 21st National Conference on Artificial Intelligence and the 18th",
}

@article{Abouei24sID, title={s-ID: Causal Effect Identification in a Sub-population}, volume={38}, url={https://ojs.aaai.org/index.php/AAAI/article/view/30011}, DOI={10.1609/aaai.v38i18.30011}, abstractNote={Causal inference in a sub-population involves identifying the causal effect of an intervention on a specific subgroup, which is distinguished from the whole population through the influence of systematic biases in the sampling process. However, ignoring the subtleties introduced by sub-populations can either lead to erroneous inference or limit the applicability of existing methods. We introduce and advocate for a causal inference problem in sub-populations (henceforth called s-ID), in which we merely have access to observational data of the targeted sub-population (as opposed to the entire population). Existing inference problems in sub-populations operate on the premise that the given data distributions originate from the entire population, thus, cannot tackle the s-ID problem. To address this gap, we provide necessary and sufficient conditions that must hold in the causal graph for a causal effect in a sub-population to be identifiable from the observational distribution of that sub-population. Given these conditions, we present a sound and complete algorithm for the s-ID problem.}, number={18}, journal={Proceedings of the AAAI Conference on Artificial Intelligence}, author={Abouei, Amir Mohammad and Mokhtarian, Ehsan and Kiyavash, Negar}, year={2024}, month={Mar.}, pages={20302-20310} }

@article{balke97bounds,
author = {Alexander Balke and Judea Pearl},
title = {Bounds on Treatment Effects from Studies with Imperfect Compliance},
journal = {Journal of the American Statistical Association},
volume = {92},
number = {439},
pages = {1171--1176},
year = {1997},
publisher = {Taylor \& Francis},
doi = {10.1080/01621459.1997.10474074},
URL = { 
    
        https://doi.org/10.1080/01621459.1997.10474074  
},
eprint = {     
        https://doi.org/10.1080/01621459.1997.10474074
}
}

@article{Heckman79SampleSB,
  title={Sample selection bias as a specification error},
  author={James J. Heckman},
  journal={Applied Econometrics},
  year={1979},
  volume={31},
  pages={129-137},
  url={https://api.semanticscholar.org/CorpusID:30028243}
}

@inproceedings{Blom19beyond,
  title     = {Beyond Structural Causal Models: Causal Constraints Models},
  author    = {Tineke Blom and Stephan Bongers and Joris M. Mooij},
  pages     = {585-594},
  url       = {https://proceedings.mlr.press/v115/blom20a.html},
  booktitle = {Proceedings of the 35th Uncertainty in Artificial Intelligence Conference ({UAI}-19)},
  year      = 2020,
  publisher = {PMLR},
  series    = {Proceedings of Machine Learning Research},
  volume    = 115,
}

@InProceedings{cooper95causal_discovery_selecion,
  title = 	 {Causal Discovery from Data in the Presence of Selection Bias},
  author =       {Cooper, Gregory F.},
  booktitle = 	 {Pre-proceedings of the Fifth International Workshop on Artificial Intelligence and Statistics},
  pages = 	 {140--150},
  year = 	 {1995},
  editor = 	 {Fisher, Doug and Lenz, Hans-Joachim},
  volume = 	 {R0},
  series = 	 {Proceedings of Machine Learning Research},
  month = 	 {04--07 Jan},
  publisher =    {PMLR},
  url = 	 {https://proceedings.mlr.press/r0/cooper95a.html},
  note =         {Reissued by PMLR on 01 May 2022.}
}

@article{Lauritzen02chaingraph,
author = {Lauritzen, Steffen L.  and Richardson, Thomas S. },
title = {Chain graph models and their causal interpretations (with discussion)},
journal = {Journal of the Royal Statistical Society: Series B (Statistical Methodology)},
volume = {64},
number = {3},
pages = {321--361},
doi = {https://doi.org/10.1111/1467-9868.00340},
year = {2002}
}

@article{Maathuis15generalized,
author = {Marloes H. Maathuis and Diego Colombo},
title = {{A generalized back-door criterion}},
volume = {43},
journal = {The Annals of Statistics},
number = {3},
publisher = {Institute of Mathematical Statistics},
pages = {1060 -- 1088},
year = {2015},
doi = {10.1214/14-AOS1295},
URL = {https://doi.org/10.1214/14-AOS1295}
}

@InProceedings{dawid10dag,
  title = 	 {Beware of the DAG!},
  author = 	 {Dawid, A. Philip},
  booktitle = 	 {Proceedings of Workshop on Causality: Objectives and Assessment at NIPS 2008},
  pages = 	 {59--86},
  year = 	 {2010},
  editor = 	 {Guyon, Isabelle and Janzing, Dominik and Schölkopf, Bernhard},
  volume = 	 {6},
  series = 	 {Proceedings of Machine Learning Research},
  address = 	 {Whistler, Canada},
  month = 	 {12 Dec},
  publisher =    {PMLR},
  url = 	 {https://proceedings.mlr.press/v6/dawid10a.html}
}

@article{Dawid21decision,
url = {https://doi.org/10.1515/jci-2020-0008},
title = {Decision-theoretic foundations for statistical causality},
author = {Dawid, A. Philip},
pages = {39--77},
volume = {9},
number = {1},
journal = {Journal of Causal Inference},
year = {2021},
}

@article{abouei24sIDlatent,
      title={Causal Effect Identification in a Sub-Population with Latent Variables}, 
      author={Amir Mohammad Abouei and Ehsan Mokhtarian and Negar Kiyavash and Matthias Grossglauser},
      journal={arXiv.org preprint},
      year={2024},
      volume={arXiv:2405.14547 [cs.LG]},
}

@article{Zhang08complete,
title = {On the completeness of orientation rules for causal discovery in the presence of latent confounders and selection bias},
journal = {Artificial Intelligence},
volume = {172},
number = {16},
pages = {1873-1896},
year = {2008},
issn = {0004-3702},
doi = {https://doi.org/10.1016/j.artint.2008.08.001},
url = {https://www.sciencedirect.com/science/article/pii/S0004370208001008},
author = {Jiji Zhang}
}
\end{document}